\documentclass[10pt]{article}
\usepackage[a4paper, top = 1.25in, bottom = 1.5in]{geometry}   

\usepackage{tocloft}
\RequirePackage{amsmath}
\RequirePackage{amsthm} 
\theoremstyle{definition} 
\newtheorem{definition}{Definition}[section] 
 
\theoremstyle{plain} 
\newtheorem{theorem}{Theorem}[section] 
\newtheorem{proposition}[theorem]{Proposition} 
\newtheorem{corollary}[theorem]{Corollary} 
\newtheorem{lemma}[theorem]{Lemma} 
\theoremstyle{definition} 
\newtheorem{example}[theorem]{Example}
\theoremstyle{remark} 
\newtheorem{remark}[theorem]{Remark}

\RequirePackage{mathptmx}      
\usepackage[hidelinks]{hyperref}

\RequirePackage{graphicx}
\RequirePackage{rotating} 
\RequirePackage{amsmath}
\RequirePackage{xurl}
\RequirePackage{graphicx} 
\RequirePackage{float} 
\RequirePackage{amssymb}
\RequirePackage{tikz-cd} 
\RequirePackage{stmaryrd} 
\RequirePackage{faktor} 
\RequirePackage{array}%
\RequirePackage{bbm} 
\newcommand{\R}{\mathbb R}

\newcommand{\V}{\mathbb V}
\newcommand{\W}{\mathbb W}
\newcommand{\Z}{\mathbb Z}

\RequirePackage{pbox} 
\RequirePackage{mathtools} 
\RequirePackage{xr} 
\newcommand{\supp}[1] { \mathrm{supp}(#1) } 
\newcommand{\loc}[2] { {#1}_{#2} } 
\newcommand{\locprob}[2] { \overline{{#1}_{#2}} } 
\newcommand{\tensor}[1] { {\left(#1\right)}^{\otimes 2} } 
\newcommand{\tensoro}[1] { {#1}^{\otimes 2} } 
\newcommand{\loccov}[2] { \Sigma_{#1}(#2) } 
\newcommand{\loccovnorm}[2] { \overline{\Sigma}_{#1}(#2) } 
\newcommand{\ball}[2] { \overline{\BB}\left(#1,#2\right) } 
\newcommand{\openball}[2] { \BB\left(#1,#2\right) } 
\newcommand{\openballM}[3] { \BB_{#3}\left(#1,#2\right) } 
\newcommand{\closedball}[2] { \overline{\BB}\left(#1,#2\right) } 
\newcommand{\closedballB} { \overline{\BB} } 
\newcommand{\openballG}[3] { \BB_{#3}\left(#1,#2\right) } 
\newcommand{\closedballG}[3] { \overline{\BB}_{#3}\left(#1,#2\right) } 
\newcommand{\closedballM}[3] { \overline{\BB}_{#3}\left(#1,#2\right) } 
\newcommand{\sphere}[2] { \partial\mathcal{B}\left(#1,#2\right) } 
\newcommand{\sphereG}[3] { \partial\mathcal{B}_{#3}\left(#1,#2\right) } 
\newcommand{\indicatrice}[1] { 1_{#1} } 
\newcommand{\normalreacho}[1] {  \lambda_0(#1) } 
\newcommand{\normalreach}[1] {  \lambda(#1) } 
\newcommand{\normalreachmap} {  \lambda } 
\newcommand{\normalreachmapo} {  \lambda_0 } 
\newcommand{\reach}[1] { \mathrm{reach}(#1) } 
\newcommand{\Tan}[2] { \mathrm{Tan}(#1, #2) } 
\newcommand{\dd} { \mathrm{d} } 
\newcommand{\Grass}[2] { \mathcal{G}_{#1}(#2) } 
\newcommand{\BB} { \mathcal{B} } 
\newcommand{\HH} { \mathcal{H} } 
\newcommand{\CC} { \mathcal{C} } 
\newcommand{\MM} { \mathcal{M} } 
\newcommand{\MMo} { \mathcal{M}_0 } 
\newcommand{\NNo} { \mathcal{N}_0 } 
\newcommand{\MMcheck} { \check \MM } 
\newcommand{\muo} { \mu_0 } 
\newcommand{\nuo} { \nu_0 } 
\newcommand{\mucheck} { \check \mu } 
\newcommand{\nucheck} { \check \nu } 
\newcommand{\muchecko} { \check{\mu}_0 } 
\newcommand{\imm} { u } 
\newcommand{\immcheck} { \check u } 
\newcommand{\matrixspace}[1] { \mathrm{M}(#1) } 
\newcommand{\Wasssymbol}[1] { \mathrm{W}_{#1} } 
\newcommand{\Wassun}[2] { \Wasssymbol{1}\left(#1,#2\right) } 
\newcommand{\Wassdeux}[2] { \Wasssymbol{2}\left(#1,#2\right) } 
\newcommand{\Wassp}[2] { \Wasssymbol{p}\left(#1,#2\right) } 
\newcommand{\Wassq}[2] { \Wasssymbol{q}\left(#1,#2\right) } 
\newcommand{\gammaWassersteinp} { \mathrm{W}_{p,\gamma} } 
\newcommand{\gammaWassersteindeux} { \mathrm{W}_{2,\gamma} } 
\newcommand{\fmax} { f_{\mathrm{max}} } 
\newcommand{\fmin} { f_{\mathrm{min}} } 
\newcommand{\Jmax} { J_{\mathrm{max}} } 
\newcommand{\Jmin} { J_{\mathrm{min}} } 
\newcommand{\volball}[1] { V_{#1} } 
\newcommand{\volsphere}[1] { S_{#1} } 
\newcommand{\petito}[1] { o(#1) } 
\newcommand{\grando}[1] { O(#1) } 
\newcommand{\eucN}[1] { \left\|#1\right\| } 
\newcommand{\eucNbig}[1] { \bigg\Vert #1 \bigg\Vert } 
\newcommand{\frobN}[1] { \left\|#1\right\|_\mathrm{F} } 
\newcommand{\opN}[1] { \|#1\|_{\mathrm{op}} } 
\newcommand{\gammaN}[1] { \left\|#1\right\|_\gamma } 
\newcommand{\eucP}[2] { \left\langle #1, #2\right\rangle } 
\newcommand{\geoD}[3] { d_{#3}(#1, #2) } 
\newcommand{\transp}[1] { {}^t#1 } 
\newcommand{\outerP}[1] { {#1}^{\otimes 2} } 

\newcommand{\CechF}[1] { V[#1] }
\newcommand{\CechFt}[2] { V^{#2}[#1] }
\newcommand{\CechFpers}[1] { \V[#1] }

\newcommand{\DTMF}[1] { W[#1] }
\newcommand{\DTMFt}[2] { W^{#2}[#1] }
\newcommand{\DTMFpers}[1] { \W[#1] }

\newcommand{\Hdist}[2] { \mathrm{d}_\mathrm{H}\left(#1, #2\right) } 
\newcommand{\NN} { \mathcal{N} } 
\newcommand{\vbb} { \mathbbm{v} } 

\newcommand{\X} { \mathbb{X} } 
\newcommand{\Y} { \mathbb{Y} } 
\renewcommand{\S}{\ifmmode\operatorname{\mathbb{S}}\else\origS\fi} 
\renewcommand{\P}{\ifmmode\operatorname{\mathbb{P}}\else\origS\fi} 
\renewcommand{\SS} {\ifmmode\operatorname{\mathcal{S}}\else\origS\fi} 
\RequirePackage{bbm}
\newcommand{\dist}[2] { \mathrm{dist}\left(#1, #2\right) } 
\newcommand{\med}[1] { \mathrm{med}\left( #1 \right) } 
\newcommand{\Cechmod}[1] { V\left[#1\right] } 
\newcommand{\barcode}[1] { \mathrm{Barcode}\left(#1\right) } 
\newcommand{\bdist}[2] { \mathrm{d}_\mathrm{b}\left(#1, #2\right) } 
\newcommand{\idist}[2] { \mathrm{d}_\mathrm{i}\left(#1, #2\right) } 
\newcommand{\secondF} { \mathrm{II}} 
\newcommand{\encadrer}[1] {\noindent\begin{center}\fbox{\begin{minipage}{.95\linewidth} #1 \end{minipage}} \end{center} }
\newcommand{\Wdist}[3] { \mathrm{W}_{#1}\left(#2, #3\right) } 
\newcommand{\Wdistgamma}[3] { \mathrm{W}_{#1,\gamma}\left(#2, #3\right) } 
\newcommand{\DTM}[1] { \mathrm{d}_{#1} }
\newcommand{\reachgamma}[1] { \mathrm{reach}_\gamma(#1) } 

\newcommand{\mini}[1] { \min\left(#1\right) }  

\begin{document}
\sloppy 

\begin{titlepage}
\begin{center}
\Large{
\textsc{ Recovering the homology of immersed manifolds }}
\vspace{.3cm}

\large{
Raphaël \textsc{Tinarrage} \\
}
\vspace{.3cm} 
\large{ 
Datashape, Inria Paris-Saclay -- LMO, Université Paris-Saclay
}
\end{center}
\vspace{.2cm} 

\paragraph{Abstract.} 
Given a sample of an abstract manifold immersed in some Euclidean space, we describe a way to recover the singular homology of the original manifold. It consists in estimating its tangent bundle---seen as subset of another Euclidean space---from a measure theoretical point of view, and in applying measure-based filtrations for persistent homology.  
We show that our construction is consistent and stable. The proof relies on two main ingredients.
First, we introduce and study the normal reach, a notion of reach adapted to immersed manifolds. It allows to quantify the deviation of geodesics around self-intersections. 
Secondly, we study the estimation of tangent spaces via local principal component analysis, with respect to the Wasserstein distance.
We illustrate our method on a few synthetic datasets, in the context of homology estimation and transverse manifolds clustering.

\paragraph{Numerical experiments.} 
A Python notebook can be found at \url{https://github.com/raphaeltinarrage/ImmersedManifolds/blob/master/Demo.ipynb}.
Some animations are gathered at \url{https://youtube.com/playlist?list=PL_FkltNTtklDlIFg1djM5XprlL8Ys0hW4}.



\vspace{.3cm} 
\tableofcontents
\end{titlepage}

\section{Introduction}
\label{intro}

A central challenge in Topological Data Analysis (TDA) consists in estimating the topology of a subset $\MM \subset \R^n$ based on a finite collection of points $X$ that lie in $\MM$ or close to. By estimating the topology of $\MM$, we mean inferring its homotopy type, or more simply inferring its singular homology groups. In what follows, the subset $\MM$ will be referred to as the \emph{underlying space}, and $X$ as the \emph{observation}. 

Inferring the homotopy type of $\MM$ may be done by constructing a homotopy equivalent simplicial complex. 
A usual method consists in considering the union of balls of radius $t \geq 0$ centered around every point of $X$, and in taking the nerve of this covering \cite{Hatcher_Algebraic}. This simplicial complex is called the \emph{\v{C}ech complex} of $X$ with parameter $t$. One can also consider the \emph{Vietoris-Rips complex} of $X$ with parameter $t$, defined as the clique complex of the underlying graph of the previous complex.   
The parameter $t$ is to be chosen in accordance with the Hausdorff distance $\Hdist{X}{\MM}$ and some geometric quantities associated to $\MM$, such as its reach \cite{niyogi2008finding,chazal2008smooth,kim2019homotopy} or its $\mu$-reach \cite{chazal2009sampling,attali2013vietoris}.
Several variations of this construction have been studied, for instance by letting the parameter $r$ vary across the points of $X$ \cite{chazal2008smooth,kim2019homotopy}, by considering ellipsoids instead of balls \cite{kalisnik2020finding}, or by using balls rectricted to $\MM$ \cite{kim2019homotopy}.
Besides the \v{C}ech and the Rips complex, one may also consider the \emph{$\alpha$-shape}, obtained by first building the Delaunay triangulation of $X$, and then keeping simplices that fit in an empty ball of radius $\alpha$. This construction yields a simplicial complex homotopy equivalent to the \v{C}ech complex \cite{edelsbrunner1993union,edelsbrunner1994three}. Developments of this construction include the \emph{witness complex} \cite{de2004topological,attali2007weak}, obtained by choosing a subset of ‘landmark’ points, or the \emph{tangential Delaunay complex} \cite{boissonnat2014manifold}, that incorporates tangent space information.

Besides, the problem of inference of homology groups of $\MM$ can be solved by computing a homotopy equivalent simplicial complex, such as those listed in the previous paragraph. However, other solutions to this problem have been proposed. They often consist in computing the image of the map induced in simplicial homology by a simplicial inclusion $K^s \hookrightarrow K^t$, where $K^s$ (resp. $K^t$) is the \v{C}ech or the Vietoris-Rips complex at time $s$ (resp. $t$).
The parameters $s$ and $t$ are still to be chosen in accordance with the Hausdorff distance $\Hdist{X}{\MM}$ and some geometric quantities of $\MM$, such as its weak feature size \cite{chazal2007stability,Chazal_Towards} or its convexity radius and distorsion \cite{fasy2018reconstruction}.

Another point of view on inference of homology groups, that allows to avoid the selection of the parameters $s$ and $t$, is \emph{persistent homology} \cite{edelsbrunner2010computational,boissonnat2018geometric}. It consists in building from $X$ an algebraic structure, called a \emph{persistence module}, which can be summarized in a \emph{persistence barcode}. The bars of the barcode can be interpreted as homological features of $X$ at different scales. 
These persistence modules are obtained from {\em filtrations}, that is, increasing families of subspaces built on top of $X$.
Among the many filtrations available to the user, the most used are the sublevel sets of the distance function to $X$, its simplicial equivalent the \emph{\v{C}ech filtration}, and its clique-complex version the \emph{Vietoris-Rips filtration}. 
The main theoretical advantage of these filtrations is their stability: small perturbations of $X$ in Hausdorff distance implies only small perturbations of the barcodes in bottleneck distance \cite{Chazal_Persistencemodules}. 
This stability allows to design statistical procedures for inferring the homology groups of $\MM$ from $X$ \cite{bubenik2010statistical,fasy2014confidence,turner2014frechet}.

A critical problem, both in the context of homotopy type inference and homology inference, is the presence of \emph{anomalous points} in $X$, that is, roughly speaking, points that cause the Hausdorff distance $\Hdist{X}{\MM}$ to be large. In presence of anomalous points, the results presented above cannot be used.
Among the attemps that have been made to overcome this issue, the filtration defined by the sublevel sets of the \emph{distance-to-measure} (DTM) introduced in \cite{Chazal_Geometricinference}, and some of its variants \cite{pwz-gikde-15}, have been proven to provide relevant information.
Unfortunately, from a practical perspective, the exact computation of the sublevel sets filtration of the DTM turn out to be far too expensive in most cases. 
To address this problem, the \emph{witnessed $k$-distance} \cite{guibas2013witnessed}, the \emph{weighted Vietoris-Rips complex filtration} \cite{Buchet_Efficient} and the \emph{DTM-filtrations} \cite{anai2020dtm} have been proposed.

In this paper, we address the problem of homotopy type and homological inference, by weakening the assumptions of \cite{niyogi2008finding}, where it is supposed that $\MM$ is a submanifold with positive reach.
Here, we consider that $\MM$ is an immersed manifold, not embedded.
That is to say, we suppose that there exists an abstract $\mathcal{C}^2$-manifold $\MMo$, immersed in the Euclidean space via a $\mathcal{C}^2$-immersion $\imm\colon \MMo \rightarrow \R^n$, whose image is $\MM$.
As before, the observation $X$ is a subset of $\R^n$, that we suppose close to $\MM$ in Hausdorff distance.
Throughout this paper, we will use the example of a circle immersed in the plane as a lemniscate, as represented in Figure \ref{Paper2:fig:lemniscate}.
Being an immersion, $\MM$ may self-intersect, and the sets $\MMo$ and $\MM$ may have different homotopy types.
The \v{C}ech filtration of $\MM$, or $X$, would reveal the homology of $\MM$, not that of $\MMo$.
Consequently, the usual approach based on the \v{C}ech filtration no longer applies here, and new methods must be developed.

\begin{figure}[H]
\centering
\begin{minipage}{.25\linewidth}
\centering
\includegraphics[width=.70\linewidth]{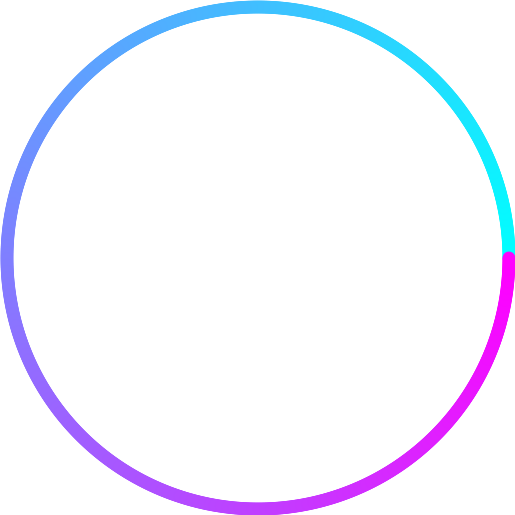}
\\$\MMo$
\end{minipage}
\begin{minipage}{.36\linewidth}
\centering
\vspace{.5cm}
\includegraphics[width=.9\linewidth]{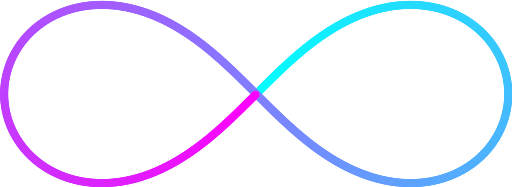}
\\ \vspace{.4cm} $\MM$
\end{minipage}
\begin{minipage}{.36\linewidth}
\centering
\vspace{.5cm}
\includegraphics[width=.9\linewidth]{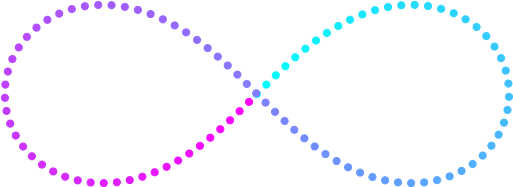}
\\ \vspace{.4cm} $X$
\end{minipage}
\caption{Left: The abstract manifold $\MMo$, a circle.
Middle: The immersion $\MM \subset \mathbb{R}^2$, known as the lemniscate of Bernoulli.
Right: The observation $X$.}
\label{Paper2:fig:lemniscate}
\end{figure}


\paragraph{Previous work.}
Among the works that involve immersed manifolds, let us cite \cite{wang2011spectral,gong2012robust,arias2017spectral}, which are set in the context where $\MM$ is a union of intersecting submanifolds. Hence $\MM$ is not a submanifold itself, but it is an immersed manifold, coming from an abstract manifold $\MMo$, made up of several connected components.
In these three works, the authors propose algorithms to classify the different components of $\MM$. In the context of the present paper, classifying the components of $\MM$ means finding the connected components of $\MMo$.
Each of these algorithms rely on the estimation of tangent spaces, so as to separate the set $\MM$ where it self-intersects. In other words, they estimate the tangent bundle of the manifold. This is a point of view that we also adopt.
We remark that, among these works, only \cite{arias2017spectral} provides mathematical proofs of consistency, for their Algorithms 2 and 3. We compare this method to ours at the end of this subsection.

Another related problem is the one of dimension estimation. In many manifold reconstruction algorithms that involve the estimation of tangent spaces, such as in \cite{wang2011spectral,boissonnat2014manifold,cheng2016tangent,aamari:hal-01521955,aamari2019nonasymptotic}, or in \cite[Algorithm 4]{arias2017spectral}, the dimension $d$ of the underlying manifold $\MM$ is given as an input of the algorithm. If $d$ is not known, a dimension estimator may be used, whether supposing that the input data exactly lies on $\MM$ \cite{singer2012vector,kim2016minimax}, or allowing the data to be corrupted by noise \cite{koltchinskii2000empirical,little2009estimation,mordohai2010dimensionality,gong2012robust}. Another strategy consists in designing tangent spaces estimators that does not require the dimension $d$, such as the empirical covariance matrix \cite[Algorithms 2 and 3]{arias2017spectral}. In the present paper, we generalize the definition of the empirical covariance matrix to any measure input, that we call \emph{local covariance matrices} (see Definition \ref{def:loccov}). We show that it is a consistent estimator of the tangent spaces (Proposition \ref{Paper2:prop:consistency}) and that is is robust to noise (see Equation \eqref{Paper2:eq:estimationloccovrobust}).

Our method is based on the stability of tangent space estimation via local covariance matrices.
Such a stability has already been studied in \cite{martinez2018shape}, and the stability of truncations of measures in \cite{pmlr-v97-memoli19a}.

\paragraph{Our contributions.}
In order to estimate the homology of a manifold from an immersion of it, we propose to estimate its tangent bundle, seen as a subset of another Euclidean space. 
As it turns out, in the process of estimating this tangent bundle, we will make errors, which will result in anomalous points.
This issue will be solved by using the DTM-filtrations, which require to use a measure theoretical framework \cite{Chazal_Geometricinference,anai2020dtm}.
Let us describe the method, in measure theoretical terms.

Let $\MMo$ be a compact $\mathcal{C}^2$-manifold of dimension $d$, and $\muo$ a Radon probability measure on $\MMo$ with full support.
Let $\imm \colon \MMo \rightarrow \R^n$ be a $\mathcal{C}^2$-immersion. 
We assume the following genericity condition: the immersion is such that self-intersection points correspond to different tangent spaces.
In other words, for every $x_0, y_0 \in \MMo$ such that $x_0 \neq y_0$ and $\imm(x_0)=\imm(y_0)$, the tangent spaces $d_{x_0} \imm( T_{x_0} \MMo)$ and $d_{y_0} \imm( T_{y_0} \MMo)$ of $\MMo$, seen in $\R^n$, are different.
As we will explain later, this condition ensures that the problem is well-posed (see Hypothesis \hyperref[hyp:1]{1}).
Now, define the image of the immersion $\MM = \imm(\MMo)$ and the pushforward measure $\mu = \imm_*\muo$. 
We consider the following problem: the input data is the measure $\mu$, or a close measure $\nu$. 
Our goal is to infer the singular homology of $\MMo$ (with coefficients in $\mathbb{Z}/2\mathbb{Z}$ for instance) from the data $\nu$. 
In practice, $\nu$ can be given as the empirical measure on a point cloud.
To answer this problem, we will build in this paper a persistence module such that the homology of $\MMo$ can be read on the corresponding persistence diagram.

To get back to $\MMo$, we proceed as follows: let $\matrixspace{\R^n}$ be the vector space of $n \times n$ matrices, and $\immcheck\colon \MMo \rightarrow \R^n \times \matrixspace{\R^n}$ the map 
\begin{align*}
\immcheck\colon x_0 \longmapsto \left( \imm(x_0), ~\frac{1}{d+2} p_{T_{u(x_0)} \MM} \right),
\end{align*}
where $p_{T_{u(x_0)} \MM}$ is the matrix of the orthogonal projection on the tangent space $T_{u(x_0)} \MM = d_{x_0} \imm( T_{x_0} \MMo) \subset \R^n$, written in the canonical basis of $\R^n$.
The term $\frac{1}{d+2}$ is a technical normalization factor that will be explained later (see Proposition \ref{Paper2:prop:consistency}).
Now, define the set $\MMcheck = \immcheck(\MMo)$. It is a submanifold of $\R^n \times \matrixspace{\R^n}$, $\mathcal{C}^1$-diffeomorphic to $\MMo$.
It is called the \emph{lift} of $\MMo$, or the \emph{lifted manifold}.
The space $\R^n \times \matrixspace{\R^n}$ is called the \emph{lift space}.
Figure \ref{Paper2:fig:lemniscate_PCA} provides a representation of the lifted manifold, when the input immersion is the lemniscate, as in Figure \ref{Paper2:fig:lemniscate}.

\begin{figure}[H]
\centering
\begin{minipage}{.49\linewidth}
\centering
\includegraphics[width=.55\linewidth]{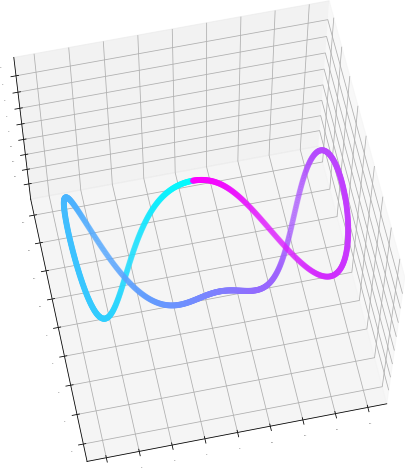}
\end{minipage}
\begin{minipage}{.49\linewidth}
\centering
\includegraphics[width=.75\linewidth]{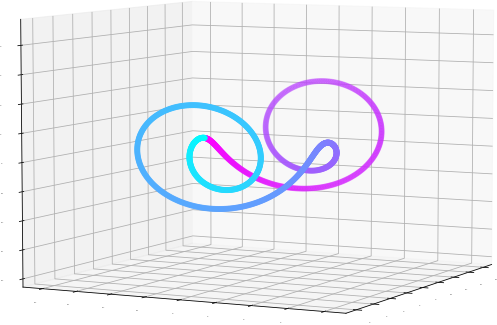}
\end{minipage}
\caption{Two views of the submanifold $\MMcheck \subset \R^2 \times \matrixspace{\R^2} \simeq \R^6$, projected in a 3-dimensional subspace via Principal Component Analysis (PCA). Observe that it does not self-intersect. The initial set $\MM$ is represented in Figure \ref{Paper2:fig:lemniscate}.}
\label{Paper2:fig:lemniscate_PCA}
\end{figure}

Suppose that one is able to estimate $\MMcheck$ from $\nu$. Then one could consider the persistent homology of a filtration based on $\MMcheck$---say the \v{C}ech filtration of $\MMcheck$ in the ambient space $\R^n \times \matrixspace{\R^n}$ for instance---and read the singular homology of $\MMo$ in the corresponding persistent barcode.
This is represented in Figure \ref{Paper2:fig:lemniscate_barcodes}.

\begin{figure}[H]
\begin{minipage}{.49\linewidth}
\centering
\includegraphics[width=.9\linewidth]{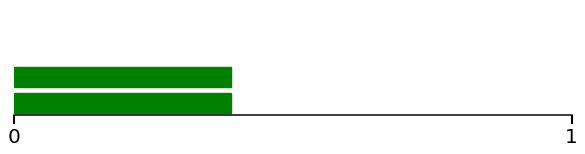}
\end{minipage}
\begin{minipage}{.49\linewidth}
\centering
\includegraphics[width=.9\linewidth]{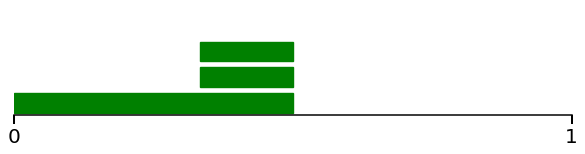}
\end{minipage}
\caption{Left: Persistence barcode of the 1-homology of the \v{C}ech filtration of $\MM $ in the ambient space $\R^2$. One reads the 1-homology of the lemniscate.
Right: Persistence barcode of the 1-homology of the \v{C}ech filtration of $\MMcheck$ in the lift space $\R^2 \times \matrixspace{\R^2}$. At the beginning of the barcode, one reads the 1-homology of a circle. Parameter $\gamma = 2$}.
\label{Paper2:fig:lemniscate_barcodes}
\end{figure}

Unfortunately, we won't be able to give a good estimation of $\MMcheck$. This is because the tangent spaces $T_{u(x_0)} \MM$, that we compute via local covariance matrices, won't be estimated correctly if $x$ is too close to a self-intersection of $\MM$.
In order to get around this issue, we adopt a measure theoretical point of view.
Instead of estimating the lifted submanifold $\MMcheck$, we propose to estimate the \emph{exact lifted measure} $\muchecko$, defined as the push-forward $\muchecko = \immcheck_* \muo$. 
It is a measure on the lift space $\R^n \times \matrixspace{\R^n}$ and has support $\MMcheck$. 

It is worth noting that $\MMcheck$ can be naturally seen as a submanifold of $\R^n \times \Grass{d}{\R^n}$, where $\Grass{d}{\R^n}$ denotes the Grassmannian of $d$-dimensional linear subspaces of $\R^n$. From this point of view, $\muchecko$ can be seen as a measure on $\R^n \times \Grass{d}{\R^n}$, i.e., a \emph{varifold}. This point of view has already been used in data analysis, such as in geometric inference \cite{buet2017varifold,buet2019weak} or in computational anatomy \cite{charon2013varifold}.
However, for computational reasons, we choose to work in the matrix space $\matrixspace{\R^n}$ instead of $\Grass{d}{\R^n}$.

Here is an alternative definition of $\muchecko$: 
for any test function $\phi\colon \R^n \times \matrixspace{\R^n} \rightarrow \R$,
\[\int \phi(x, A) \dd \muchecko(x, A) = \int_{\MMo} \phi\left(\imm(x_0),\frac{1}{d+2} p_{T_{u(x_0)} \MM} \right) \dd \muo(x_0). \]
Getting back to the observed measure $\nu$, we propose to estimate $\check{\mu}_0$ with the \emph{lifted measure} $\check{\nu}$, defined as follows: for any test function $\phi\colon \R^n \times \matrixspace{E} \rightarrow \R$,
\[\int \phi(x, A) \dd \check{\nu}(x, A) = \int_{\MM} \phi\bigg(x,\loccovnorm{\nu}{x}\bigg) \dd \nu(x), \]
where $\loccovnorm{\nu}{x}$ is \emph{normalized local covariance matrix} (see Definition \ref{def:loccov}). 
It depends on a parameter $r > 0$.
We prove that $\loccovnorm{\nu}{x}$ can be used to estimate the tangent spaces $\frac{1}{d+2} p_{T_{u(x_0)} \MM}$ of $\MM$.
However, this estimation is biased next to the self-intersection of $\MM$, as shown in Figure \ref{Paper2:fig:lemniscatecheckbias}. 
As a consequence, the support of $\nucheck$ is not close to $\MMcheck$ in Haudorff distance. 

\begin{figure}[H]
\begin{minipage}{.49\linewidth}
\centering
\includegraphics[width=.55\linewidth]{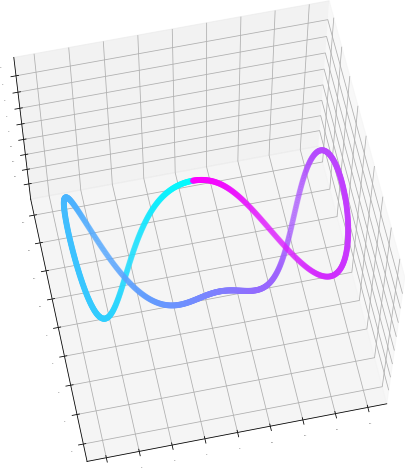}
\end{minipage}
\begin{minipage}{.49\linewidth}
\centering
\includegraphics[width=.55\linewidth]{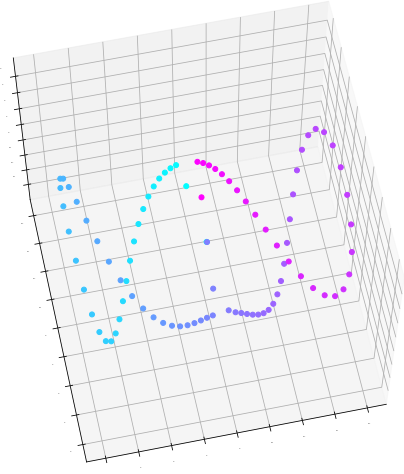}
\end{minipage}
\caption{Left: The set $\supp{\muchecko}=\MMcheck$, where $\mu$ is the uniform measure on $\MM$ (see Figure \ref{Paper2:fig:lemniscate}).
Right: The set $\supp{\nucheck}$, where $\nu$ is the empirical measure on $X$. Parameters $\gamma = 2$ and $r=0{.}1$.}
\label{Paper2:fig:lemniscatecheckbias}
\end{figure}

At this point, one could use an outliers-removal procedure, so as to recover $\MMcheck$. However, such a procedure depends critically on a choice of parameter, and is not reliable in practice.
Instead, and still from a measure theoretical point of view, we will prove that the measure $\nucheck$ is close to $\muchecko$ in \emph{Wasserstein distance} (see Theorem \ref{Paper2:th:estimation}).
This is true since only a few anomalous points are present
As a consequence, by using persistent homology for measures---such as the DTM-filtrations---the measure $\nucheck$ can be used to infer the homotopy type of $\MMcheck$, that is, of $\MMo$ (see Corollaries \ref{Paper2:cor:homotopytypeDTMcheck} and \ref{Paper2:cor:DTMfiltrcheck}). The barcodes of the DTM-filtration on $\nucheck$ are represented in Figure \ref{Paper2:fig:lemniscate_DTM}.

\begin{figure}[H]
\begin{minipage}{.49\linewidth}
\centering
\includegraphics[width=.95\linewidth]{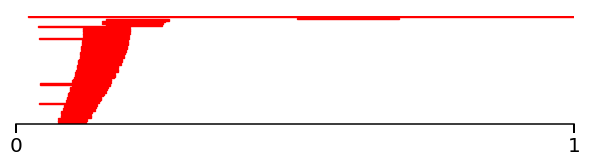}
\end{minipage}
\begin{minipage}{.49\linewidth}
\centering
\includegraphics[width=.95\linewidth]{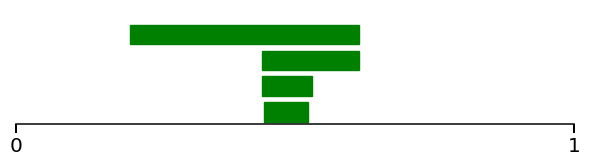}
\end{minipage}
\caption{Persistence barcodes of the 0-homology (left) and 1-homology (right) of the DTM-filtration of the lifted measure $\nucheck$. 
Observe that the homology of the circle is salient on these barcodes (one large red bar and one large green bar).
Parameters $\gamma = 2$, $r=0{.}1$ and $m=0{.}01$.}
\label{Paper2:fig:lemniscate_DTM}
\end{figure}

In order to quantify the quality of this approximation, we introduce a new geometric quantity: the \emph{normal reach} (see Definition \ref{Paper2:def:normalreach}).
It has been designed to play the role of the reach, when the subset considered is an immersed manifold.
We show that the normal reach gives a scale at which an immersed manifold can be seen as an embedded manifold (see Proposition \ref{Paper2:prop:normalreach}).

As a last remark, let us compare our method to \cite{arias2017spectral}.
In this paper, the input dataset is a point cloud $X \subset \R^n$, seen as a sample of the union of two intersecting submanifolds. Translated in our context, $\MMo$ is the disjoint union of two abstract manifolds, and $\MM$ is an immersion of it.
Their Algorithm 3 consists in estimating the tangent spaces $Q_x$ on top of each point $x \in X$, via a variation of the empirical covariance matrix. Then, the authors build a graph $G$, whose vertices are the input data points $x \in X$, and where an edge $[x,y]$ is added if the Euclidean positions are close enough ($\eucN{x-y}\leq \epsilon$) and if the tangent space estimations are close enough too ($\frobN{Q_x-Q_y}\leq \eta$). The output of the algorithm is then the connected components of $G$.
Unfortunately, due to the bad estimation of tangent spaces around self-intersections, the algorithm may treat the intersection points as a cluster of its own, hence returning more connecting components than wanted.
In order to circumvent this issue, their Algorithm 2 includes an outliers-removal step, so as to exclude points close to the self-intersection. Under a particular choice of parameters, it is shown that the algorithm returns exactly two clusters, accurately clustering the points away from the intersection.

In comparison, our method has been thought to estimate the singular homology of $\MMo$, not only its connected components. In this setting, the outliers-removal procedure is a crucial step. This is because removing too many points would cause the apparition of gaps in the lifted manifold $\MMcheck$, which would be complicated to fill. 
Instead of discarding outliers, our method incorporates a sort of hierarchical clustering, performed by the use of the DTM-filtrations. Indeed, in the DTM-filtration of the lifted measure $\nucheck$, the points are weighted according to their degree of anomalousness. This anomalousness is quantified via their local density in the lift space $\R^n \times \matrixspace{\R^n}$. The underlying idea is the following: since only a few points are close to the intersection, only a few points will have a bad tangent space estimation, hence their density will be small. A careful analysis will make this idea rigorous.
Another advantange of our method lies in the use of persistent homology: the output of our algorithm is a persistence barcode.
Hence we do not need to select precise connected components, or more generally, precise homological features.  
It is up to the user to read on this barcode the bars that seem to be relevant (in general, one chooses the longest bars).
This procedure is justified theoretically by Corollary \ref{Paper2:cor:DTMfiltrcheck}, which shows that the output barcode is stable.

\paragraph{Data availability.}
A Python notebook, containing numerical illustrations and codes used in this paper, can be found at \url{https://raphaeltinarrage.github.io/ImmersedManifolds}.

\paragraph{Outline.}
The rest of the paper is as follows. 
Sect. \ref{Paper2:sec:intro} gathers usual definitions related to Euclidean topology of compact sets, Riemannian geometry and persistent homology.
We also describe our model.
In Sect. \ref{Paper2:sec:reach} we introduce the normal reach, and derive certain probability bounds based on it.
In Sect. \ref{Paper2:sec:tangentspaceestimation}, we study the tangent space estimation of an immersed manifold via local covariance matrices.
We gather these results in Sect. \ref{Paper2:sec:topoinference} to obtain estimation guarantees for our method.

\paragraph{Notations and constants.}
We gather in Appendix \ref{app:notations} the notations that are used.
Moreover, throughout the paper, we will refer to constants that are collected in a table in Appendix \ref{app:constants}.
It is not necessary to read this table, since the constants will be introduced along the text.


\section{Preliminaries}
\label{Paper2:sec:intro}
\subsection{Euclidean and Riemannian geometry}
\label{subsec:riemannian}
In this subsection, we give some geometry results that will be useful in what follows. 
Here and in the rest of the paper, we will only consider compact manifolds and submanifolds without boundary, and measures that are Radon measures.
We refer the reader to \cite{federer1959curvature} for an exposition of the notion of reach, to \cite{do1992riemannian} for a presentation of Riemannian geometry, and to \cite{morgan2016geometric} for a gentle introduction to geometric measure theory.

\paragraph{Reach.}
Let $X$ be any subset of $\R^n$ and $y \in \R^n$ a point. The \emph{distance} from $y$ to $X$ is the quantity $$\dist{y}{X}  = \inf\{\eucN{x-y} \mid x \in X\}.$$
A \emph{projection} of $y$ on $X$ is a point $x \in X$ that minimizes the distance $\eucN{x-y}$.
The \emph{medial axis} of $X$ is the subset $\med{X} \subset \R^n$ which consists of points $y\in \R^n$ that admit at least two distinct projections on $X$:
\begin{align*}
\med{X} = \left\{ y \in \R^n \mid \exists x,x' \in X, ~x \neq x', ~\eucN{y-x}=\eucN{y-x'}=\dist{y}{X}  \right\}.
\end{align*}
The \emph{reach} of $X$ is 
\begin{equation*}
\reach{X} = \inf\left\{ \eucN{x-y} \mid x \in X, ~y \in \med{X} \right\}.
\end{equation*}
A useful property of sets with positive reach is the approximation by tangent spaces. 
For a general set $X$, we define its tangent cone at $x \in X$, denoted $\Tan{X}{x}$, as:
\begin{align*}
\{0\} \cup \left\{v \in \R^n \mid \forall \epsilon > 0, ~\exists y \in X \text{ s.t. } y \neq x, ~\eucN{y-x}<\epsilon, ~\eucNbig{ \frac{v}{\eucN{v}} - \frac{y-x}{\eucN{y-x}} } < \epsilon \right\}.
\end{align*}
Note that if $X$ is a submanifold, we recover the usual notion of tangent space.
The following characterization is fundamental in the study of sets with positive reach:
\begin{theorem}[{{\cite[Theorem 4.18(2)]{federer1959curvature}}}]
\label{background:th:federer}
A closed set $X\subset \R^n$ has positive reach $\tau$ if and only if for every $x,y \in X$, we have
\begin{align*}
\dist{y-x}{\Tan{X}{x}} \leq \frac{1}{2\tau} \eucN{y-x}^2.
\end{align*}
\end{theorem}
The reach is a quantity that controls both the local and global regularity of the set $X$.
When $X = \MM$ is a topological submanifold, having a positive reach implies that $\MM$ is of regularity $\CC^{1,1}$ \cite[Proposition 1.4]{lytchak2005almost}.
Conversely, a $\CC^{1,1}$-submanifold $\MM$ has a positive reach \cite[Theorem 4.19]{federer1959curvature}. 
Moreover, when $\MM$ is $\CC^2$, it can be shown that $\reach{\MM}$ is caused either by a bottleneck structure or by high curvature:
\begin{theorem}[{{\cite[Theorem 3.4]{aamari:hal-01521955}}}]
\label{background:th:aamari_reach}
A closed $\CC^2$-submanifold $\MM \subset \R^n$ with positive reach must satisfy at least one of the following two properties:
\begin{itemize}
\item \emph{Global case:} there exist $x,y \in \MM$ with $\eucN{x-y} = 2\reach{\MM}$ and $\frac{1}{2}(x+y) \in \med{\MM}$,
\item \emph{Local case:} there exists an arc-length parametrized geodesic $\gamma\colon I \rightarrow \MM$ with $\eucN{\ddot{\gamma}(0)} = \reach{\MM}^{-1}$.
\end{itemize}
\end{theorem}
In this paper, we will suppose that the manifold is of regularity $\CC^2$, so as to obtain uniform bounds on its second derivatives (see Hypothesis \hyperref[hyp:2]{2}). We do not study whether the results could be generalized to $\CC^{1,1}$ manifolds.

\paragraph{Riemannian structure on immersed manifolds.}
If $u\colon\MMo \rightarrow \MM \subset  \R^n$ is an immersion of a $\CC^2$-manifold, then $\MMo$ is naturally endowed with a Riemannian structure, by pulling back the inner product of $\R^n$. This makes $u$ an isometry. From now on, we will consider that $\MMo$ is given this Riemannian structure. 
We denote the (abstract) tangent space of $\MMo$ at $x_0$ as $T_{x_0} \MMo$, its image in $\R^n$ as $T_{u(x_0)} \MM = d_{x_0} \imm( T_{x_0} \MMo)$, and its orthogonal complement, the normal space, as $(T_{\imm(x_0)} \MM)^\bot$.
The geodesic distance between two points $x_0,y_0\in\MMo$ is denoted $\geoD{x_0}{y_0}{\MMo}$.
For any $x_0 \in \MMo$ and $r\geq 0$, we denote by $\openballM{x_0}{r}{\MM_0}$ (resp. $\closedballM{x_0}{r}{\MM_0}$) the open (resp. closed) geodesic ball of center $x_0$ and radius $r$ of $\MMo$.
Moreover, for any $v_0 \in T_{x_0}\MM_0$, we denote by $\openballM{v_0}{r}{T_{x_0}\MM_0}$ the open ball of center $v_0$ and radius $r$ of $T_{x_0}\MM_0$.


For every $x_0 \in \MMo$, one defines the \emph{second fundamental form} of $\MMo$ at $x_0$. 
It is a symmetric bilinear form 
\begin{align*}
\secondF_{x_0} \colon~ T_{x_0} \MMo \times T_{x_0} \MMo \longrightarrow (T_{u(x_0)} \MM)^\bot.
\end{align*}
Let $x_0 \in \MMo$, $v_0 \in T_{x_0} \MMo$ a unit vector, and consider an unit-speed geodesic $\gamma_0 \colon I \rightarrow \MMo$ such that $\gamma_0(0) = x_0$ and $\dot{\gamma_0}(0) = v_0$.
Let us denote by $\gamma$ the map $u \circ \gamma_0 \colon I \rightarrow \MM$.
The following relation can be found in \cite[Sect. 6]{niyogi2008finding} or \cite[Sect. 3]{boissonnat2019reach}:
\begin{equation}
\label{eq:supnormiscurvature}
\secondF_{x_0}(v_0, v_0) = \ddot{\gamma}(0).
\end{equation}
In particular, any bound on the operator norm $\opN{\secondF_{x_0}}$ of $\secondF_{x_0}$ implies a bound on $\eucN{\ddot{\gamma}(0)}$. 
From now, we suppose that the operator norms $\opN{\secondF_{x_0}}$ are bounded by a constant $\rho > 0$ (see Hypothesis \hyperref[hyp:1]{1}). For instance, if $\MMo$ is an embedded manifold, then $\rho$ can be chosen as its reach \cite[Proposition 6.1]{niyogi2008finding}. In general, if $\MMo$ is a compact $\CC^2$-manifold, such a global upper bound $\rho$ exists.
Let us list a few useful results.

\begin{lemma}
\label{Paper2:lem:Federergeod}
Let $x_0 \in \MM_0$ and $\gamma_0\colon I \rightarrow \MM_0$ an arc-length parametrized geodesic starting from $x_0$. Let $\gamma = \imm \circ \gamma_0$, $v = \dot \gamma(0)$ and $x = u(x_0)$. For all $t \in I$, we have
\begin{enumerate}
\item $\eucN{\gamma(t) - (x+tv)} \leq \frac{\rho}{2}t^2.$\label{lem:Federergeod:point1}
\end{enumerate}
Consequently, for every $y_0 \in \MM_0$, denoting $\delta = \geoD{x_0}{y_0}{\MM_0}$ and $y=u(y_0)$, we have
\begin{enumerate}
\itemsep.2cm
\setcounter{enumi}{1}
\item $\dist{y-x}{T_x \MM} \leq \frac{\rho}{2} \delta^2$,\label{lem:Federergeod:point2}
\item $(1-\frac{\rho}{2} \delta)\delta \leq \eucN{x-y}$.\label{lem:Federergeod:point3}
\end{enumerate}
Concerning the immersion $\imm\colon \MMo \rightarrow \MM$, we deduce that
\begin{enumerate}
\setcounter{enumi}{3}
\itemsep.05cm
\item the map $\imm$ is injective on the open geodesic ball $\openballG{x_0}{\frac{2}{\rho}}{\MMo}$, \label{lem:Federergeod:point4}
\item for every $y_0 \in \openballG{x_0}{\frac{1}{\rho}}{\MMo}$ such that $y_0 \neq x_0$, the vector $y-x$ is not orthogonal to $T_x \MM$ nor $T_y \MM$. \label{lem:Federergeod:point5}
\end{enumerate}
\end{lemma}
\noindent
The first point of this lemma can be found in \cite[Equation (5)]{niyogi2008finding}, and the other points follow directly. 
Note that stronger versions of these results can be found in \cite{boissonnat2019reach}.

We now state a technical lemma. It gives how much time it takes for a geodesic to exit a Euclidean ball (represented in Figure \ref{fig:exitball}).
It is a direct consequence of Lemma \ref{Paper2:lem:Federergeod} and its proof is deferred to Appendix \ref{sec:appendix_intro}. 
\begin{lemma}
\label{Paper2:lem:distancetocenter}
Let $x_0, y_0 \in \MM_0$ and $\gamma_0\colon I \rightarrow \MM_0$ an arc-length parametrized geodesic with $\gamma_0(0) = y_0$. Define $x = \imm(x_0)$, $y = \imm(y_0)$, $\gamma = u\circ \gamma_0$, $v = \dot \gamma(0)$ and $l = \eucN{y-x}$.
Suppose that $l < \frac{1}{\rho}$ and $\eucP{v}{y-x} = 0$.
\begin{enumerate}
\item The map $t \mapsto \eucN{\gamma(t)-x}$ is increasing on $[0, T_1]$ where $T_1 = \frac{\sqrt{2}}{\rho} \sqrt{2-\sqrt{3+\rho^2 l^2}}$.\label{lem:distancetocenter:point1}
\end{enumerate}
Let $r$ be such that $l \leq r < \frac{1}{2\rho}$ and define 
$$T_2 = \frac{\sqrt{2}}{\rho}\sqrt{1-\rho r - \sqrt{ 1 - 2\rho r + \rho^2 l^2}} ~~~~\text{and}~~~~ T_2' = \frac{\sqrt{2}}{\rho}\sqrt{1-\rho r + \sqrt{ 1 - 2\rho r + \rho^2 l^2}}.$$
\begin{enumerate}
\setcounter{enumi}{1}
\itemsep.2cm
\item If $t \in (T_2,T_2')$, then $\eucN{\gamma(t)-x}>r$.
Moreover, $T_2 \leq 2 \sqrt{r^2 - l^2}$.
\label{lem:distancetocenter:point2}
\item If $l=0$, then $T_2 = \frac{1}{\rho}(1-\sqrt{1-2\rho r})$ and $T_2' = \frac{1}{\rho}(1+\sqrt{1-2\rho r})$.
\label{lem:distancetocenter:point4}
\end{enumerate}
Moreover, let $s$ be such that $0 \leq s \leq  r$ and define 
$$a = \inf\{t \geq 0 \mid \eucN{\gamma(t)-x}\geq s\}~~~~~~\text{and}~~~~~~b = \inf\{t \geq 0 \mid \eucN{\gamma(t)-x}\geq r\}.$$
\begin{enumerate}
\setcounter{enumi}{3}
\itemsep.2cm
\item $b-a \leq \sqrt{6}\sqrt{r^2-s^2}$.
\label{lem:distancetocenter:point5}
\item If $l=0$, then $b - a \leq 2(r-s)$.
\label{lem:distancetocenter:point6} 
\end{enumerate}
\end{lemma}

\begin{figure}[H]
\begin{minipage}{.49\linewidth}
\centering
\includegraphics[width=.6\linewidth]{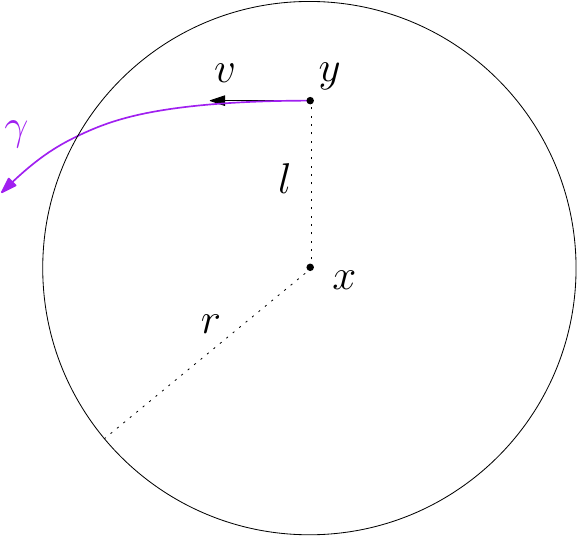}
\end{minipage}
\begin{minipage}{.49\linewidth}
\centering
\includegraphics[width=.6\linewidth]{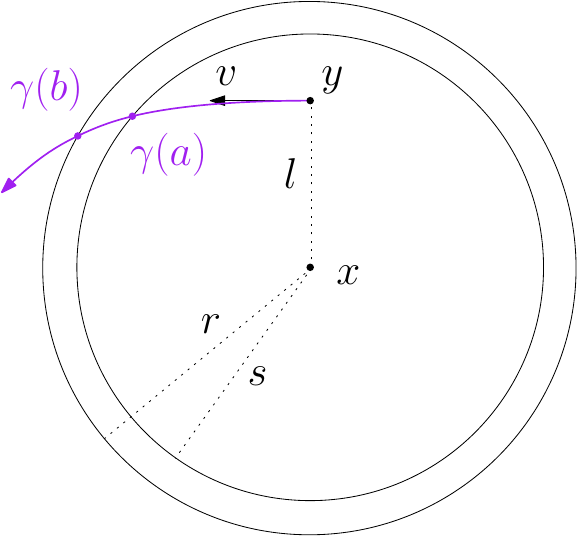}
\end{minipage}
\caption{Illustration of Lemma \ref{Paper2:lem:distancetocenter} Point \ref{lem:distancetocenter:point1} (left) and Point \ref{lem:distancetocenter:point5} (right).}
\label{fig:exitball}
\end{figure}

Last, the \emph{exponential map} of $\MMo$ at $x_0$ will be denoted 
\begin{align*}
\exp_{x_0}^{\MMo} \colon T_{x_0} \MMo \rightarrow \MMo.
\end{align*}
According to \cite[Corollary 4 Point 1]{alexander2006gauss}, the map $ \exp_{x_0}^{\MMo}$ is injective on the open ball $\openballM{0}{\frac{\pi}{\rho}}{T_{x_0}\MM_0}$ of $T_{x_0}\MM_0$, and is a diffeomorphism onto its image $\openballM{0}{\frac{\pi}{\rho}}{\MMo}$.
Moreover, for any $x_0 \in \MMo$ and $v_0 \in T_{x_0}\MMo$, the $d$-dimensional Jacobian of $\exp_{x_0}^{\MMo}$ at $v_0$ is defined as 
\begin{align*}
J_{v_0} = \sqrt{ \det \left(A^t \cdot A \right) },
\end{align*}
where $A = d_{v_0} \exp_{x_0}^{\MMo}$ is the differential of the exponential map, seen as a $d \times n$ matrix.
As shown by the following result, the Jacobian of the exponential map is linked with the bound $\rho$ on the operator norms $\opN{\secondF_{x_0}}$.

\begin{lemma}{(\cite[Proposition III.22]{Aamari_thesis})}
\label{Paper2:lem:regularityexp0}
Let $x_0\in\MMo$ and $v_0 \in T_{x_0}\MMo$ such that $\eucN{v_0} = r < \frac{\pi}{2\sqrt{2}\rho}$. The Jacobian $J_{v_0}$ of $\exp_{x_0}^{\MMo}$ at $v_0$ satisfies
\begin{align*}
\bigg(1-\frac{\left(r \rho \right)^2}{6}\bigg)^d \leq J_{v_0} \leq \bigg(1+(r\rho)^2\bigg)^d.
\end{align*}
\end{lemma} 


\paragraph{Coarea formula.}
For any measure $\tau$ on a probability space $\Omega$ and any measurable map $h\colon \Omega \rightarrow \Omega'$, the \emph{push-forward} of $\tau$ by $h$ is the measure $h_*\tau$ defined via $$h_*\tau(A) = \tau(h^{-1}(A))$$ for any measurable set $A \subset \Omega'$.
The \emph{transfer} property refers to the following fact: for any integrable map $\phi\colon \Omega'\rightarrow \R$, we have $$\int \phi\cdot \dd h_*\tau = \int \phi\circ h \cdot\dd \tau.$$
Now, suppose that $\MMo$ and $\NNo$ are Riemannian manifolds, of respective dimensions $d\geq d'$, and let $\HH^d_{\MMo}$ and $\HH^{d'}_{\NNo}$ denote their corresponding Hausdorff measures. 
The coarea formula allows to reformulate integrals on $\MMo$ as integrals on $\NNo$. 
\begin{theorem}{(\cite[Chapter 3]{morgan2016geometric})}
\label{th:areaformula}
Let $f \colon \MMo \rightarrow \NNo$ be a differentiable map. For any $x_0\in\MMo$, let $J_{x_0}$ denote its Jacobian.
Let $\phi\colon \MMo \rightarrow [0,+\infty)$ be a measurable function. We have:
$$\int_{\MMo} \phi(x_0) J_{x_0} \cdot\dd \HH^d_{\MMo}(x_0) = \int_{\NNo}\left(\int_{x_0\in f^{-1}(\{y_0\})} \phi(x_0)\cdot\dd\HH^{d-d'}_{\NNo}(x_0) \right)\dd \HH^{d'}_{\NNo}(y_0).$$
\end{theorem}
A useful consequence of this theorem is the following: suppose that $\MMo$ is endowed with a Radon probability measure $\muo$ that admits a density $h\colon \MMo \rightarrow [0,+\infty)$ against $\HH^d_{\MMo}$. Consider the push-forward measure $\nuo = f_*\muo$.
If the Jacobian $J_{x_0}$ of $f$ never vanishes, then the push-forward measure $\nuo$ admits a density $g\colon \NNo \rightarrow [0,+\infty)$ against $\HH^{d'}_{\NNo}$, where
\begin{equation}
\label{eq:density_coarea}
g(y_0)=\int_{x_0\in f^{-1}(\{y_0\})}  h(x_0) J_{x_0}^{-1}\cdot\dd\HH^{d-d'}_{\NNo}(x_0).
\end{equation}
In particular, when $d=d'$, we have $g(y_0)=\sum_{x_0\in f^{-1}(\{y_0\})}  h(x_0) J_{x_0}^{-1}$. 

\subsection{Persistent homology}
In this subsection, we write down the definitions of persistence modules, and their associated pseudo-distances, as presented in \cite{Chazal_Persistencemodules}. 
We refer the interested reader to \cite{edelsbrunner2010computational,boissonnat2018geometric} for a thorough description.
Let $T \subset \R$ be an interval, $E=\R^n$ a Euclidean space and $k$ a field.

\paragraph{Persistence modules.}
A {\em persistence module} over $T$ is a pair $(\V, \vbb)$ where $\V = (V^t)_{t\in T}$ is a family of $k$-vector spaces, and $\vbb = (v_s^t)_{s\leq t \in T}$ is a family of linear maps $v_s^t\colon V^s \leftarrow V^t$ such that:
\begin{itemize}
\itemsep.15cm
\item for every $t\in T$, $v_t^t\colon V^t \rightarrow V^t$ is the identity map,
\item for every $r, s,t\in T$ such that $r\leq s\leq t$, we have $v_s^t \circ v_r^s = v_r^t$.
\end{itemize}
When there is no risk of confusion, we may denote a persistence module by $\V$ instead of $(\V, \vbb)$.
Given $\epsilon \geq 0$, an {\em $\epsilon$-morphism} between two persistence modules $\V$ and $\W$ is a family of linear maps $\phi = (\phi_t\colon \V^t \rightarrow \W^{t+\epsilon})_{t \in T}$ such that the following diagram commutes for every $s \leq t \in T$:
\begin{center}
\begin{tikzcd}
V^s \arrow["\phi_s", d] \arrow[r, "v_s^t"] & \arrow["\phi_t", d] V^t \\
W^{s+\epsilon} \arrow[r, "w_{s+\epsilon}^{t+\epsilon}"] & W^{t+\epsilon}
\end{tikzcd}
\end{center}
If $\epsilon = 0$ and each $\phi_t$ is an isomorphism, the family $(\phi_t)_{t \in T}$ is an {\em isomorphism} of persistence modules.
An {\em $\epsilon$-interleaving} between two persistence modules $\V$ and $\W$ is a pair of $\epsilon$-morphisms $(\phi_t\colon V^t \rightarrow W^{t+\epsilon})_{t \in T}$ and $(\psi_t\colon W^t \rightarrow V^{t+\epsilon})_{t \in T}$ such that the following diagrams commute for every $t \in T$: 
\begin{center}
\begin{minipage}[t]{0.4\textwidth}
\centering
\begin{tikzcd}
V^t \arrow[dr, "\phi_t"] \arrow[rr, "v_t^{t+2\epsilon}"] & & V^{t+2\epsilon} \\
& W^{t+\epsilon} \arrow[ur, "\psi_{t+\epsilon}"] & 
\end{tikzcd}
\end{minipage}
\begin{minipage}[t]{0.4\textwidth}
\centering
\begin{tikzcd}
& V^{t+\epsilon} \arrow[dr, "\phi_{t+\epsilon}"] & \\
W^t \arrow[ur, "\psi_t"] \arrow[rr, "w_t^{t+2\epsilon}"] & & W^{t+2\epsilon} 
\end{tikzcd}
\end{minipage}
\end{center}
The \emph{interleaving pseudo-distance} between $\V$ and $\W$ is defined as 
$$\idist{\V}{\W} = \inf \{\epsilon \geq 0 \mid  \V \text{ and } \W \text{ are } \epsilon \text{-interleaved}\}.$$

\paragraph{Persistence barcodes.}
A persistence module $(\V, \vbb)$ is said to be {\em pointwise finite-dimensional} if for every $t \in T$, $V^t$ is finite-dimensional. If this property is satisfied, we can define a notion of persistence barcode \cite{botnan2020decomposition}. It comes from the algebraic decomposition of the persistence module into interval modules. Moreover, given two pointwise finite-dimensional persistence  modules $\V, \W$ with persistence barcodes $\barcode{\V}$ and $\barcode{\W}$, the so-called isometry theorem states that 
\begin{align*}
\bdist{\barcode{\V}}{\barcode{\W}} = \idist{\V}{\W},
\end{align*}
where $\idist{\cdot}{\cdot}$ denotes the interleaving distance between persistence modules, and $\bdist{\cdot}{\cdot}$ denotes the bottleneck distance between barcodes.

More generally, the persistence module $(\V, \vbb)$ is said to be {\em $q$-tame} if for every $s,t \in T$ such that $s < t$, the map $v_s^t$ has finite rank. The $q$-tameness of a persistence module ensures that we can still define a notion of persistence barcode, even though the module may not be decomposable into interval modules. Moreover, the isometry theorem still holds \cite{Chazal_Persistencemodules}.

\paragraph{Filtrations of sets and simplicial complexes.}
A family of subsets $\X=(X^t)_{t \in T}$ of $E$ is a {\em filtration} if it is non-decreasing for the inclusion, i.e., for any $s, t \in T$ such that $s \leq t$, we have $X^s \subset X^t$. 
Given $\epsilon\geq 0$, two filtrations $\X=(X^t)_{t \in T}$ and $\Y=(Y^t)_{t \in T}$ of $E$ are {\em $\epsilon$-interleaved} if, for every $t \in T$, we have $X^t \subset Y^{t+\epsilon}$ and $Y^t \subset X^{t+\epsilon}$. 
The interleaving pseudo-distance between $\X$ and $\Y$ is defined as the infimum of such $\epsilon$:
\[\idist{\X}{\Y} = \inf \{ \epsilon \geq 0 \mid  \X \text{  and  }  \Y  \text{  are  }\epsilon\text{-interleaved} \}.\]
Filtrations of simplicial complexes and their interleaving distance are similarly defined: 
given a simplicial complex $S$, a {\em filtration of $S$} is a non-decreasing family $\S = (S^t)_{t \in T}$ of subcomplexes of $S$. The interleaving pseudo-distance between two filtrations $(S_1^t)_{t \in T}$ and $(S_2^t)_{t \in T}$ of $S$ is the infimum of the $\epsilon \geq 0$ such that they are $\epsilon$-interleaved, i.e. for any $t \in T$, we have $S_1^{t} \subset S_2^{t+\epsilon}$ and  $S_2^{t} \subset S_1^{t+\epsilon}$.

\paragraph{Relation between filtrations and persistence  modules.}
Applying the singular homology functor to a set filtration gives rise to a persistence module whose linear maps between homology groups are induced by the inclusion maps between sets. As a consequence, if two filtrations are $\epsilon$-interleaved, then their associated persistence modules are also $\epsilon$-interleaved, the interleaving homomorphisms being induced by the interleaving inclusion maps. As a consequence of the isometry theorem, if the modules are $q$-tame, then the bottleneck distance between their persistence barcodes is upper bounded by $\epsilon$ \cite{Chazal_Persistencemodules}.
The same remarks hold when applying the simplicial homology functor to simplicial filtrations.

\subsection{Persistent homology for measures}
\label{subsec:pers_measures}
In this subsection we define the distance-to-measure (DTM), based on \cite{Chazal_Geometricinference}, and the DTM-filtrations, based on \cite{anai2020dtm}.
Let $T = \R^+$ and $E = \R^n$ endowed with the standard Euclidean norm.  

\paragraph{Wasserstein distances.}
Given two probability measures $\mu$ and $\nu$ over $E$, a transport plan between $\mu$ and $\nu$ is a probability measure $\pi$ over $E \times E$ whose marginals are $\mu$ and $\nu$. 
Let $p \geq 1$.
The \textit{$p$-Wasserstein distance between $\mu$ and $\nu$} is defined as 
\[\Wassp{\mu}{\nu} = \left(\inf_\pi \int_{E\times E} \|x-y\|^p d\pi(x,y) \right)^\frac{1}{p},\]
where the infimum is taken over all the transport plans $\pi$.
If $q$ is such that $p\leq q$, then an application of Jensen's inequality shows that $\Wassp{\mu}{\nu} \leq \Wassq{\mu}{\nu}$.

\paragraph{Distance-to-measure (DTM).} 
Let $\mu$ be a probability measure over $E$, and $m \in [0,1)$ a parameter.
The DTM associated to $\mu$ with parameter $m$ is the function $\DTM{\mu, m}\colon E\rightarrow \R$ defined as: 
\begin{align*}
\DTM{\mu, m}(x) = \sqrt{\frac{1}{m}\int_0^{m} \delta_{\mu,t}^2(x)\dd t} 
~~~~~~~\text{where}~~~~~~
\delta_{\mu, m}(x) = \inf \left\{r\geq0 \mid \mu\left(\closedball{x}{r}\right)>m\right\},
\end{align*}
and where $\closedball{x}{r}$ denotes the closed ball of center $x$ and radius $r$ of $E$.
When $m$ is fixed and there is no risk of confusion, we may write $\DTM{\mu}$ instead of $\DTM{\mu, m}$.
Among the important properties of the DTM, it has been shown that it is $1$-Lipschitz \cite[Corollary 3.7]{Chazal_Geometricinference}.
Moreover, it is stable in Wasserstein distance \cite[Theorem 3.5]{Chazal_Geometricinference}: for any probability measures $\mu$ and $\nu$, we have 
\begin{align}
\label{eq:stab_DTM}
\|\DTM{\mu, m} - \DTM{\nu, m}\|_\infty \leq m^{-\frac{1}{2}}\Wassdeux{\mu}{\nu}.
\end{align}
If $f\colon E \rightarrow \R$ is any function and $t \in \R$, we will denote the \emph{$t$-sublevel set} of $f$ as $f^t = f^{-1}((-\infty,t])$.
The following theorem shows that the sublevel sets $\DTM{\mu, m}^t$ of $\DTM{\mu, m}$ can be used to estimate the homotopy type of $\supp{\mu}$.
\begin{theorem}[{\cite[Corollary 4.11, case $\mu=1$]{Chazal_Geometricinference}}]
\label{th:DTM-homotopytype}
Consider two probability measures $\mu, \nu$ on $E$ and $m \in (0,1)$. Denote $K = \supp{\mu}$.
Suppose that $\reach{K} = \tau > 0$, and that $\mu$ satisfies the following hypothesis for $r < \left(\frac{m}{a}\right)^\frac{1}{d}$: $\forall x \in K, \mu(\openball{x}{r}) \geq a r^d$.
Suppose that $\Wassdeux{\mu}{\nu} \leq  m^\frac{1}{2} \big(\frac{\tau}{9} - (\frac{m}{a})^{\frac{1}{d}} \big)$.
Define $\epsilon = (\frac{m}{a})^\frac{1}{d} + m^{-\frac{1}{2}} \Wassdeux{\mu}{\nu}$ and choose $t \in [4 \epsilon, \tau - 3\epsilon]$.
Then $\DTM{\nu, m}^t$ and $K$ are homotopic equivalent.
\end{theorem}

\paragraph{DTM-filtrations.}
These filtrations have been introduced in \cite{anai2020dtm} and are defined as follows: consider a probability measure $\mu$ on $E$ and a parameter $m \in [0,1)$. 
For every $t \in T$, consider the set 
\begin{equation}
\label{eq:def_DTM_filtration}
\DTMFt{\mu}{t} = \bigcup_{x \in \supp{\mu}} \closedball{x}{t - \DTM{\mu, m}(x)},
\end{equation}
where $\closedball{x}{r}$ denotes the closed ball of center $x$ and of radius $r$ of $E$ if $r\geq0$, or denotes the empty set if $r <0$.
The family $\DTMF{\mu}= \left(\DTMFt{\mu}{t}\right)_{t\geq 0}$ is a filtration of $E$.
It is called the \emph{DTM-filtration} with parameters $(\mu, m,1)$.
By applying the singular homology functor, we obtain a persistence module, denoted $\DTMFpers{\mu}$.
If $\supp{\mu}$ is bounded, then $\DTMFpers{\mu}$ is $q$-tame.
Moreover, it has been proven that the DTM-filtrations are stable with respect to the input measure:

\begin{theorem}[{\cite[Theorem 4.5]{anai2020dtm}}]
\label{th:DTM-filtr-stability}
Consider two measures $\mu, \nu$ on $E$ with supports $X$ and $Y$. Let $\mu', \nu'$ be two measures with compact supports $\Gamma$ and $\Omega$ such that $\Gamma \subset X$ and $\Omega \subset Y$.
Then the interleaving distance $d_i(\CechF{X,d_\mu},\CechF{Y,d_\nu})$ between the DTM-filtrations $\DTMF{\mu}$ and $\DTMF{\nu}$ is upper bounded by
\[ m^{-\frac{1}{2}}W_2(\mu,\mu') + m^{-\frac{1}{2}}W_2(\mu',\nu')  + m^{-\frac{1}{2}}W_2(\nu',\nu) + c(\mu',m) + c(\nu',m),\]
where for any measure $\tau$, we define the quantity $c(\tau,m) = \sup_{x\in \supp{\tau}} d_{\tau, m}(x)$.
\end{theorem} 
%
%
%
Under a regularity assumption on $\mu$, we can restate Theorem \ref{th:DTM-filtr-stability} without mentioning the intermediate measures $\mu'$ and $\nu'$.
The proof is given in Appendix \ref{sec:appendix_intro}.

\begin{corollary}
\label{Paper2:cor:DTM-filtrations}
Consider two probability measures $\mu, \nu$ on $E$, $m \in (0,1)$ and denote $w = \Wassdeux{\mu}{\nu}$.
Suppose that $w \leq \frac{1}{4}$, and that $\mu$ satisfies the following for $r < \left(\frac{m}{a}\right)^\frac{1}{d}$: $\forall x \in \supp{\mu}, \mu(\openball{x}{r}) \geq a r^d$.
Then
\begin{align*}
\idist{\DTMF{\mu}}{\DTMF{\nu}} 
\leq c_{\ref{Paper2:cor:DTM-filtrations:index}} \left(\frac{w}{m}\right)^\frac{1}{2} + c_{\ref{Paper2:cor:DTM-filtrations:index}}' m^\frac{1}{d},
\end{align*}
with $c_{\ref{Paper2:cor:DTM-filtrations:index}} = 8\mathrm{diam}(\supp{\mu}) + 5$ and $c_{\ref{Paper2:cor:DTM-filtrations:index}}' = 2a^{-\frac{1}{d}}$.
\end{corollary}

\subsection{Model and hypotheses}
\label{Paper2:subsec:model}
\paragraph{Model.}
We consider an abstract $\CC^2$-manifold $\MM_0$ of dimension $d \geq 1$, $E = \R^n$ the Euclidean space and a $\CC^2$-immersion $\imm\colon \MM_0 \rightarrow E$. 
We denote $\MM = \imm(\MM_0)$.
Moreover, for any $x_0 \in \MMo$, we write $x$ for $u(x_0)$, $T_{x_0} \MM_0$ for the (abstract) tangent space of $\MM_0$ at $x_0$, and $T_{x} \MM$ for $d_{x_0}u (T_{x_0} \MM_0)$, which is an affine subspace of $E$.
Let $\immcheck$ be the map
\begin{align*}
\immcheck\colon ~ \MM_0 &\longrightarrow E\times \matrixspace{E} \\
x_0 &\longmapsto \left(x, ~\frac{1}{d+2}p_{T_{x} \MM}\right),
\end{align*}
where $p_{T_{x} \MM}$ is the orthogonal projection matrix on $T_{x} \MM$, and $\matrixspace{E}$ the space of $n\times n$ matrices. Note that the map $\immcheck$ is a $\CC^1$-immersion since $\imm$ is $\CC^2$.
We define the \emph{lifted manifold} as $\check \MM = \check{u}(\MM_0)$. 
We also consider a probability measure $\muo$ on $\MMo$, and define $\mu = \imm_* \muo$ and $\muchecko = \immcheck_* \muo$.
These several sets and measures fit in the following commutative diagrams:
\begin{center}
\begin{minipage}{.49\linewidth}
\[
\begin{tikzcd}
\MM_0 \arrow[dr, "\imm"] \arrow[rr, "\check \imm"] & & \MMcheck \arrow[dl, "\mathrm{proj}"]  \\
& \MM &
\end{tikzcd}
\]
\end{minipage}
\begin{minipage}{.49\linewidth}
\[
\begin{tikzcd}
\muo \arrow[dr, "\imm_*"] \arrow[rr, "\check \imm_*"] & & \muchecko \arrow[dl, "\mathrm{proj_*}"]  \\
& \mu &
\end{tikzcd}
\]
\end{minipage}
\end{center}
\noindent
As explained in the introduction, the aim of this work is to estimate the homotopy type of $\MMo$, or its homology groups, from the measure $\mu$, or from a close measure $\nu$.
We detail our method in Subsect. \ref{Paper2:subsec:overview}, and show that, by using DTM-filtrations, the problem boils down to estimating the measure $\muchecko$ from $\nu$. 

Besides, we endow $\MMo$ with the Riemannian structure given by the immersion $u$.
For every $x_0 \in \MMo$, the second fundamental form of $\MMo$ at $x_0$ is denoted $\secondF_{x_0}$,
and the exponential map is denoted $\exp_{x_0}^{\MMo}$.
We shall also consider the map $\exp_{x}^{\MM}\colon T_x \MM \rightarrow \MM$, the \emph{exponential map seen in $\MM$}, defined as $\imm \circ \exp_{x_0}^{\MMo} \circ (d_{x_0} \imm)^{-1}$.

\paragraph{Notation conventions.}
In the rest of this paper, symbols with 0 as a subscript shall refer to quantities associated to $\MMo$. For instance, a point of $\MMo$ may be denoted $x_0$, and a curve on $\MMo$ may be denoted $\gamma_0$. 
Symbols with a caron accent shall refer to quantities associated to $\MMcheck$, such as a point $\check x$, or a curve $\check \gamma$. 
Symbols with no such subscript or accent shall refer to quantities associated to $\MM$, such as $x$ or $\gamma$. 
In order to simplify the notations, we consider the following convention: 
\noindent
\begin{center}
\fbox{\begin{minipage}{.95\linewidth}
Dropping the 0 subscript to a symbol shall correspond to applying the map $\imm$.
\vspace{.05cm}

Dropping the 0 subscript to a symbol and adding a caron accent shall correspond to applying the map $\immcheck$.
\end{minipage}}
\end{center}
\noindent
For instance, if $x_0$ is a point of $\MMo$, then $x$ represents $\imm(x_0)$, and $\check x$ represents $\immcheck(x_0)$.
Similarly, if $\gamma_0 \colon I \rightarrow \MMo$ is a map, then $\gamma$ represents $\imm \circ \gamma_0$, and $\check \gamma$ represents $\check \imm \circ \gamma_0$.
Note that it is possible to have $x = y$ but $T_x \MM \neq T_y \MM$. When writing $T_x \MM$, we will always refer to an implicit point $x_0 \in \MMo$.


\paragraph{Hypotheses.}
Throughout the paper, we shall refer to the four hypotheses listed below.
The first one is a transversity-like condition.

\encadrer{
\textbf{Hypothesis \hyperref[hyp:1]{1}. }\label{hyp:1}
For every $x_0, y_0 \in \MMo$ such that $x_0 \neq y_0 $ and $x = y$, we have $T_{x} \MM \neq T_{y} \MM$.
}
\noindent
This Hypothesis \hyperref[hyp:1]{1} ensures that the $\CC^1$-immersion $\check \imm$ is injective, hence that it is a $\CC^1$-diffeomorphism, since its domain $\MMo$ is compact. As a consequence, the lifted manifold $\check \MM$ is a submanifold of $E\times \matrixspace{E}$, with the same homotopy type than $\MMo$. This allows to recover the homology of $\MMo$ from $\MMcheck$.

\encadrer{
\textbf{Hypothesis \hyperref[hyp:2]{2}. }\label{hyp:2}
The operator norm of the second fundamental form of $\MMo$ at each point is bounded by $\rho>0$. 
}
\noindent
In Hypothesis \hyperref[hyp:2]{2}, we consider that $\MMo$ is endowed the Riemannian structure given by the immersion $\imm$.
According to Equation \eqref{eq:supnormiscurvature}, this hypothesis implies the following key property: if $\gamma_0\colon I \rightarrow \MMo$ is an arc-length parametrized geodesic of class $\CC^2$, then for all $t \in I$, we have $\eucN{\ddot \gamma(t)} \leq \rho$. In particular, we can use the Lemmas \ref{Paper2:lem:Federergeod} and \ref{Paper2:lem:distancetocenter}.

\encadrer{
\textbf{Hypothesis \hyperref[hyp:3]{3}. }\label{hyp:3}
The measure $\mu_0$ admits a density $f_0$ on $\MMo$. Moreover, $f_0$ is $L_0$-Lipschitz (with respect to the geodesic distance) and bounded by $f_{\mathrm{min}}, f_{\mathrm{max}} > 0$.  
}
\noindent
In Hypothesis \hyperref[hyp:3]{3}, we consider that $\MMo$ is endowed with the volume measure $\HH^d_{\MMo}$, that is, the measure obtained by pulling back the $d$-dimensional Hausdorff measure $\HH^d$ on $E$ via the immersion $\imm$. Note that this may not be the uniform measure on $\MM$ (the volume is not renormalized). By assumption, $\mu_0$ is a probability measure, hence the integral $\int f_0 \cdot\dd \HH^d_{\MMo}$ is equal to $1$.


In order to state the fourth hypothesis, we need the notion of \emph{normal reach}, that we will define in Subsect. \ref{Paper2:subsec:normalreach}.
Roughly speaking, the normal reach is a map $\normalreachmapo\colon \MMo \rightarrow [0,+\infty)$ that indicates how close the point $x_0$ is from a self-intersection.
We remind the reader that we use the sublevel set notation $\normalreachmapo^r = \normalreachmapo^{-1}([0,r])$.
\encadrer{
\textbf{Hypothesis \hyperref[hyp:4]{4}. }\label{hyp:4}
There exists $c_{\ref{Paper2:hyp:normalreach:index}} \geq 0$ and $r_{\ref{Paper2:hyp:normalreach:index}} >0$ such that, for every $r \in [0, r_{\ref{Paper2:hyp:normalreach:index}} )$, $\muo(\normalreachmapo^r)\leq c_{\ref{Paper2:hyp:normalreach:index}} r$.
}
\noindent
This hypothesis will only be used in the last part of this paper, when gathering our results about tangent space estimation and stability of measures.
Thanks to it, we will be able to subdivide $\MMo$ in two sets that involve a different analysis: the points with small normal reach (points in $\lambda_0^r$) and points with large normal reach (in $\MMo\setminus\lambda_0^r$).

The following table lists which hypotheses will be invoked in each subsection of the paper.


\begin{center}
\begin{tabular}{l|c|c|c|c|c|c|c|c}
Subsection  
&\ref{Paper2:subsec:normalreach} 
&\ref{subsec:probabilistic_bounds} 
&\ref{Paper2:subsec:quantif_normal_reach} 
&\ref{Paper2:subsec:consistency}
&\ref{Paper2:sec:appendix:tangentspaceestimation}
&\ref{Paper2:subsec:stability}
&\ref{subsec:approx_th} 
&\ref{Paper2:sec:topoinference} \\
\hline
Hyp. used 
& \hyperref[hyp:2]{2} 
& \hyperref[hyp:2]{2}, \hyperref[hyp:3]{3}
& \hyperref[hyp:1prime]{1'}, \hyperref[hyp:2]{2}, \hyperref[hyp:3]{3}
& \hyperref[hyp:2]{2}, \hyperref[hyp:3]{3} 
& \hyperref[hyp:2]{2}, \hyperref[hyp:3]{3} 
& \hyperref[hyp:2]{2}, \hyperref[hyp:3]{3} 
& \hyperref[hyp:1]{1}, \hyperref[hyp:2]{2}, \hyperref[hyp:3]{3}, \hyperref[hyp:4]{4}
& \hyperref[hyp:1]{1}, \hyperref[hyp:2]{2}, \hyperref[hyp:3]{3}, \hyperref[hyp:4]{4}
\end{tabular}
\end{center}
In Subsect. \ref{Paper2:sec:appendix:tangentspaceestimation} and \ref{Paper2:subsec:stability} we will introduce to a new set of hypotheses: \hyperref[hyp:5]{5}, \hyperref[hyp:6]{6}, and \hyperref[hyp:7]{7}. 
We will show that these hypotheses are consequences of \hyperref[hyp:2]{2} and \hyperref[hyp:3]{3}. However, referring to these new hypotheses will simplify the exposition, and will allow to state our results in a more general setting.

Concerning the naturality of the hypotheses, note that Hypothesis \hyperref[hyp:1]{1} is necessary to ensure that our problem is well-posed. To see this, consider the subset $\MM \subset \R^2$ consisting of two tangent circles. As depicted in Figure \ref{fig:hyp1}, $\MM$ may be the immersion of the abstract manifold $\MMo$ being the disjoint union of two circles, or of $\MMo'$ being a circle. 
In this case, the observation of $\MM$ cannot discriminate between $\MMo$ and $\MMo'$.
However, these immersions would not satisfy Hypothesis \hyperref[hyp:1]{1}.
Note that, in this example, the immersion $\MM_0'\rightarrow \MM$ is only $\CC^1$, but one can easily design a similar example of regularity $\CC^2$.

\begin{figure}[H]
\centering
\includegraphics[width=.95\linewidth]{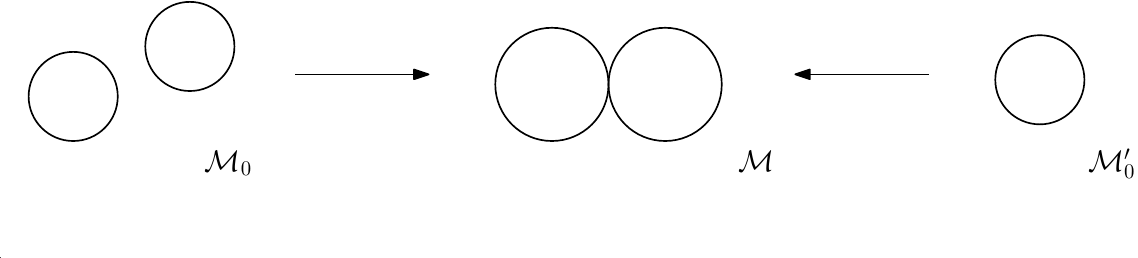}
\caption{According to Hypothesis \hyperref[hyp:1]{1}, $\MM$ cannot be two tangent circle.}
\label{fig:hyp1}
\end{figure}

In the literature, works often consider submanifolds, and Hypothesis \hyperref[hyp:2]{2} is usually stated as a lower bound on the reach $\tau$. 
Our hypothesis, stated as an upper bound on the norm $\rho$ of the second fundamental forms of the immersion, is weaker.
Indeed, under the assumption that the immersion is $\CC^2$, we have $\rho \leq \frac{1}{\tau}$.
The advantage of our formulation is that, in the case of an immersed manifold, the reach may be zero, hence cannot be used.
Note that Hypothesis \hyperref[hyp:2]{2} is equivalent to the following property: there exists a function $\alpha\colon [0,+\infty) \rightarrow [0,+\infty)$ such that $\lim_{r \rightarrow 0} \alpha = 1$, and such that for any $x_0 \in \MMo$, the image of the geodesic ball $\imm(\closedballG{x_0}{r}{\MMo})$ has reach lower bounded by $\alpha(r)\frac{1}{\rho}$. 
The fact that Hypothesis \hyperref[hyp:2]{2} implies this statement is a consequence of the proof of \cite[Proposition 6.1]{niyogi2008finding}, and the converse is a consequence of the proof of Proposition \ref{Paper2:prop:normalreach}.

The introduction of constants in Hypothesis \hyperref[hyp:3]{3} will allow derive explicit bounds for our method. Note that we do not suppose that the measure $\mu$ is given as an input, but only a close measure $\nu$ with respect to the Wasserstein distance. There is no hypothesis concerning the measure $\nu$.

Last, we think that Hypothesis \hyperref[hyp:4]{4} is a consequence of Hypotheses \hyperref[hyp:1]{1}, \hyperref[hyp:2]{2} and \hyperref[hyp:3]{3}, but we have not been able to prove it yet.
As a partial result, we prove that it is a consequence of Hypotheses \hyperref[hyp:1prime]{1'}, \hyperref[hyp:2]{2} and \hyperref[hyp:3]{3}, where Hypothesis \hyperref[hyp:1prime]{1'} is a strenghtening of Hypothesis \hyperref[hyp:1]{1} (see Proposition \ref{prop:quantif_normal_reach}).
We also show that Hypothesis \hyperref[hyp:4]{4} is a consequence of Hypotheses \hyperref[hyp:1]{1}, \hyperref[hyp:2]{2}, \hyperref[hyp:3]{3} and $\dim(\MMo)=1$ (see Remark \ref{rem:quantif_normalreach_dim1}).

\section{Reach of an immersed manifold}
\label{Paper2:sec:reach}
In this section, we introduce a new notion of reach, adapted to the immersed manifolds, and derive technical results that will be useful in the rest of the paper.
As an introduction, we consider an embedded manifold $u\colon \MMo \rightarrow \MM \subset \R^n$ with positive reach $\tau$.
Let $x_0, y_0$ be two points of $\MMo$. 
We wish to compare their geodesic distance $\geoD{x_0}{y_0}{\MM_0}$ and their Euclidean distance $\eucN{y-x}$.
A first inequality is true in general: 
\[\eucN{y-x} \leq \geoD{x_0}{y_0}{\MM_0}.\]
\noindent
Moreover, if they are close enough in geodesic distance---say $\geoD{x_0}{y_0}{\MM_0} \leq \tau$ for instance---then the inequality $\rho \leq \frac{1}{\tau}$ and Lemma \ref{Paper2:lem:Federergeod} Point \ref{lem:Federergeod:point3} yields
\[\geoD{x_0}{y_0}{\MM_0} \leq 2 \eucN{x-y}.\]
This section is devoted to obtaining such a converse inequality when the manifold $\MMo$ is only immersed, not embedded. In this case, the condition $\geoD{x_0}{y_0}{\MM_0} \leq \frac{1}{\tau}$ has to be turned into an upper bound on $\geoD{x_0}{y_0}{\MM_0}$ that depends on $x_0$ and $y_0$ (as we will obtain in Lemma \ref{Paper2:lem:comparisoneucgeod2}).

\subsection{Normal reach}
\label{Paper2:subsec:normalreach}
We consider an immersion $\imm \colon \MMo \rightarrow \MM \subset E$ which satisfies Hypothesis \hyperref[hyp:2]{2}.

\begin{definition}
\label{Paper2:def:normalreach}
For every $x_0 \in \MM_0$, let $\Lambda(x_0) = \left\{y_0 \in \MM_0 \mid  y_0 \neq x_0, ~x-y \bot T_y \MM \right\}$.
The \emph{normal reach} of $\MM_0$ at $x_0$ is defined as:
\[ \normalreacho{x_0} = \inf_{y_0 \in \Lambda(x_0)} \eucN{x-y}.\]
Observe that if $x_0, y_0$ are distinct points of $\MMo$ with $x=y$, then $x-y$ is orthogonal to any vector, hence $\lambda_0(x_0) = \eucN{x-y} = 0$.
Hence we can define the \emph{normal reach seen in $\MM$}, denoted $\lambda \colon \MM \rightarrow \mathbb{R}$, as 
\begin{align*}
\normalreach{x} = \left\{
    \begin{array}{ll}
        \normalreacho{\imm^{-1}(x)} & \text{if } x \text{ has only one preimage,} \\
        0 & \text{else.}
    \end{array}
\right.
\end{align*}
It satisfies the relation $\lambda_0 = \lambda \circ \imm$. 
\end{definition}

\begin{figure}[H]
\centering
\begin{minipage}{.49\linewidth}
\centering
\includegraphics[width=.95\linewidth]{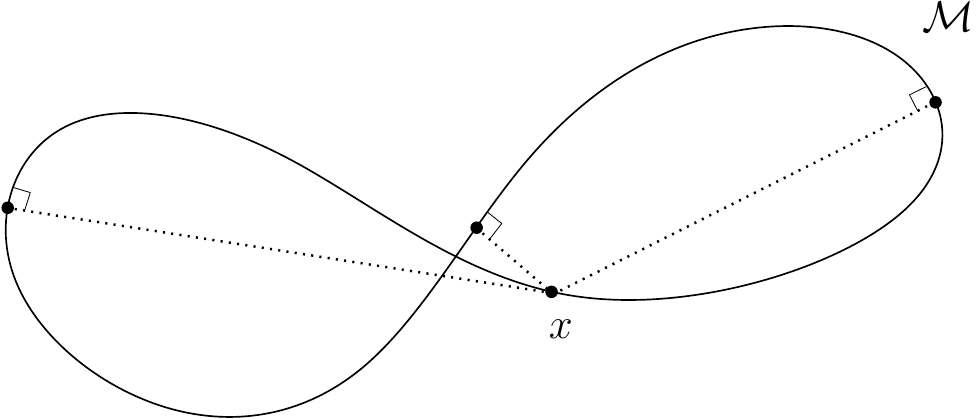}
\end{minipage}
\begin{minipage}{.49\linewidth}
\centering
\includegraphics[width=.95\linewidth]{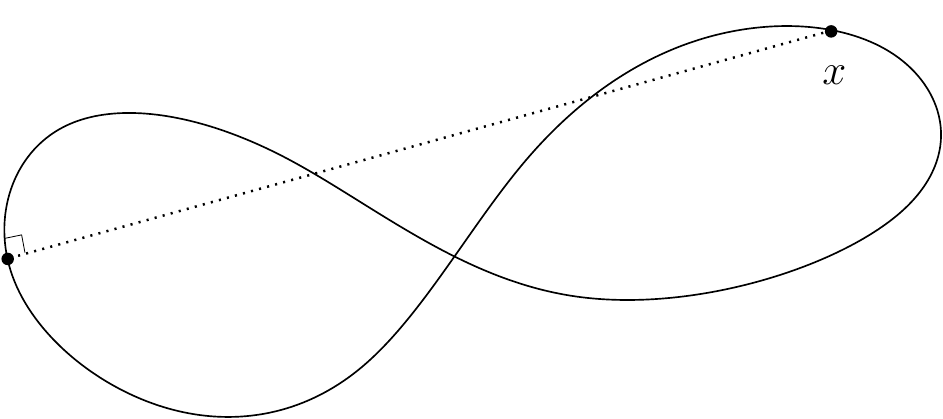}
\end{minipage}
\caption{The set $\Lambda(x_0)$ from Definition \ref{Paper2:def:normalreach}, for two different points $x_0$.}
\end{figure}

Note that $\Lambda(x_0)$ is closed, hence the infimum of Definition \ref{Paper2:def:normalreach} is attained.
Indeed, we can write $\Lambda(x_0) = L \setminus\{x_0\}$, with $L = \{y_0 \in \MM_0 \mid  x-y \bot T_y \MM \}$. The set $L$ is closed since it is the preimage of $\{0\}$ by the continuous map $y_0 \mapsto \eucN{p_{T_y \MM}(x-y)}$. Furthermore, $\{x_0\}$ is an isolated point of $\Lambda(x_0)$, since Lemma \ref{Paper2:lem:Federergeod} Point \ref{lem:Federergeod:point5} says that, for every $y_0$ in the geodesic ball $\openballM{x_0}{\frac{1}{\rho}}{\MM_0}$ such that $y_0 \neq x_0$, the vector $x-y$ is not orthogonal to $T_y \MM$, hence $y_0 \notin L$.

\begin{example}
\label{ex:normal_reach_lemniscate}
Suppose that $\MM$ is the lemniscate of Bernoulli, with diameter 2.
Figure \ref{Paper2:fig:normal_reach_lemniscate} represents the values of the normal reach $\lambda\colon \MM \rightarrow \R$.
Observe that $\lambda$ is not continuous.
\begin{figure}[H]
\centering
\includegraphics[width=.45\linewidth]{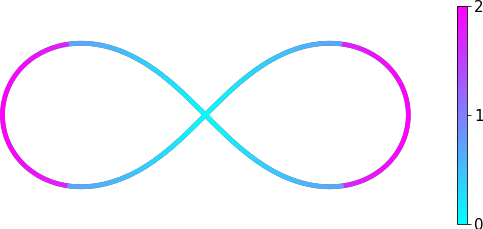}
\caption{Values of the normal reach on the lemniscate of Bernoulli.}
\label{Paper2:fig:normal_reach_lemniscate}
\end{figure}
\end{example}

\begin{remark}
\label{rem:semi-continuous}
The normal reach $\lambda_0$ is \emph{lower semi-continuous}, that is, for any sequence $(x_0^n)_{n \geq 0}$ of $\MMo$ converging to a point $x_0 \in \MMo$, we have $\liminf_{n\rightarrow \infty} \lambda_0(x_0^n) \geq \lambda_0(x_0)$.
Indeed, by definition, we can choose for every $n \geq 0$ a point $y_0^n \in \MMo$ such that $x_0^n \neq y_0^n$, $x^n-y^n\bot T_{y^n} \MM$ and $\lambda_0(x_0^n) = \eucN{x^n-y^n}$. Moreover, since $\MMo$ is compact, the sequence $(y_0^n)_{n \geq 0}$ admits an accumulation point $y_0$ such that $\liminf_{n\rightarrow \infty} \lambda_0(x_0^n) = \eucN{x-y}$. By continuity, we have $x-y \bot T_y \MM$. Moreover, since $\geoD{x_0^n}{y_0^n}{\MMo} \geq \frac{1}{\rho}$ for all $n \geq 0$ by Lemma \ref{Paper2:lem:Federergeod} Point \ref{lem:Federergeod:point5}, we deduce that $x_0 \neq y_0$. Consequently, we have $y_0 \in \Lambda(x_0)$, and $\lambda_0(x_0) \leq \eucN{x-y}$.
\end{remark}


Here is a key property of the normal reach:
\begin{lemma}
\label{Paper2:lem:connected}
Let $x_0 \in \MM_0$. Let $r \geq 0$ such that $r < \normalreach{x}$. 
Then $\imm^{-1}\left(\closedball{x}{r}\right)$ is connected.
\end{lemma}

\begin{figure}[H]
\centering
\includegraphics[width=.85\linewidth]{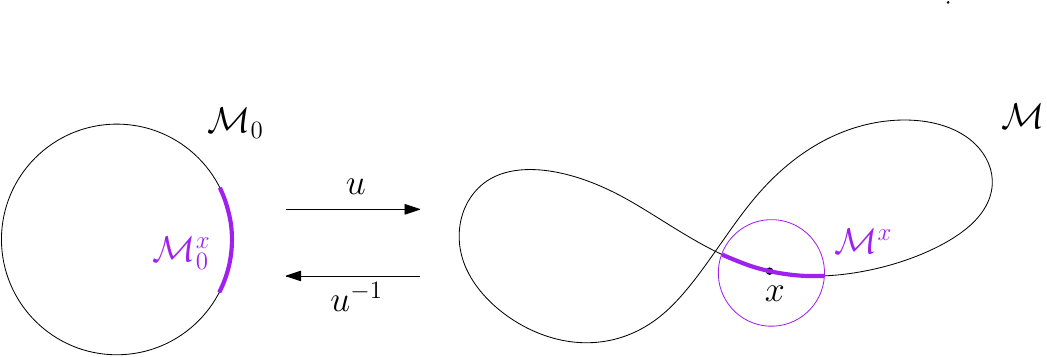}\\
\vspace{-0.5cm}
\includegraphics[width=.85\linewidth]{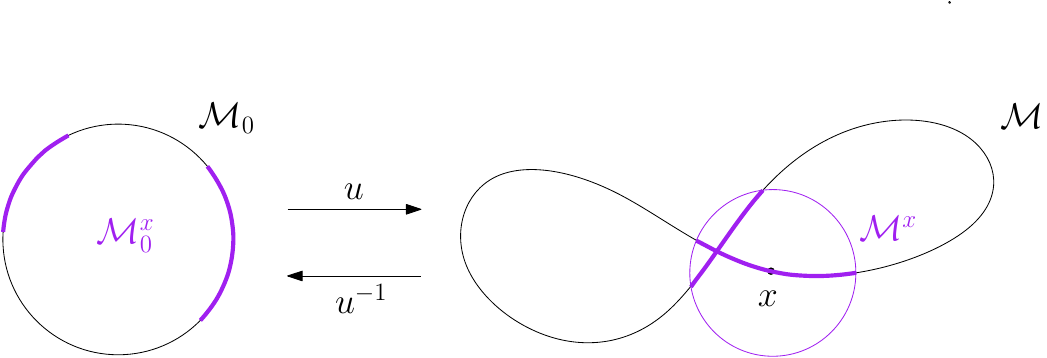}
\caption{Top: the set $\imm^{-1}\left(\closedball{x}{r}\right)$, with $r < \lambda(x)$, is connected.
Bottom: when $r \geq \lambda(x)$, it may not be connected.}
\end{figure}


\begin{proof}
Write $\MM^x =\MM \cap \closedball{x}{r}$ and $\MM^x_0 = \imm^{-1}(\MM^x)$.
By contradiction, suppose that $\MM^x_0$ is not connected.
Let $C \subset \MM^x_0$ be a connected component which does not contain $x_0$. Since $C$ is compact, we can consider a minimizer $y_0$ of $\{ \eucN{x-y} \mid  y_0 \in C \}$. 
It is clear that $y$ satisfies $x-y \bot T_y \MM$, otherwise it would not be a local minimizer.
Now, the properties $x-y \bot T_y \MM$ and $x_0 \neq y_0$ imply that $\eucN{x-y} \geq \lambda(x)$, which contradicts $r < \lambda(x)$.
\end{proof}

The following lemma is the counterpart of \cite[Proposition 6.3]{niyogi2008finding} for the normal reach. It allows to compare the geodesic  and Euclidean distances by only imposing a condition on the last one.
\begin{lemma}
\label{Paper2:lem:comparisoneucgeod2}
Let $x_0, y_0 \in \MM_0$. Denote $r = \eucN{x-y}$ and $\delta = \geoD{x_0}{y_0}{\MMo}$. 
Suppose that $r < \mini{\frac{1}{2 \rho}, \normalreach{x} }$.
Then 
\[\delta \leq c_{\ref{Paper2:lem:comparisoneucgeod2:index}}(\rho r) r \mathrm{~~~~~~~where~~~~~~~} c_{\ref{Paper2:lem:comparisoneucgeod2:index}}(t)=\frac{1}{t}\left(1 - \sqrt{1 - 2t}\right).\]
In other words, the following inclusion holds: $\imm^{-1}(\closedball{x}{r} ) \subset \closedballM{x_0}{c_{\ref{Paper2:lem:comparisoneucgeod2:index}}(\rho r) r}{\MM_0}$.
\end{lemma}
\noindent
Note that, for $t < \frac{1}{2}$, we have the inequalities $1 \leq c_{\ref{Paper2:lem:comparisoneucgeod2:index}}(t) \leq 1+2t < 2$.

\begin{proof}
Denote $\MM^x = \MM \cap \closedball{x}{r}$ and $\MM^x_0 = \imm^{-1}(\MM^x)$.
Let $C_0^x$ be the connected component of $x_0$ in $\MM^x_0$.
Let us show that $C_0^x$ is included in the closed geodesic ball $\closedballM{x_0}{c_{\ref{Paper2:lem:comparisoneucgeod2:index}}(\rho r) r}{\MM_0}$.

Let $\epsilon>0$ be such that $r < r + \epsilon < \frac{1}{2\rho}$.
For any tangent vector $v_0 \in T_{x_0} \MMo$ of unit norm, consider the unit-speed geodesic $\gamma_0 \colon I \rightarrow \MMo$ with $\gamma(0) = x_0$ and $\dot \gamma(0) = v_0$.
According to Lemma \ref{Paper2:lem:distancetocenter} Points \ref{lem:distancetocenter:point2} and \ref{lem:distancetocenter:point4}, for any $t$ in $\left(\frac{1}{\rho}(1-\sqrt{1-2\rho r}, \frac{1}{\rho}(1+\sqrt{1-2\rho r})\right)$, we have $\eucN{x-\gamma(t)}>r$. Consequently, $C_0^x$ and the geodesic sphere $\sphereG{x_0}{t}{\MM_0}$ are disjoint. 
We deduce that $C_0^x \subset \openballM{x_0}{t}{\MMo}$.
In particular, we have $C_0^x \subset\closedballM{x_0}{t^*}{\MM_0}$, where
$$
t^* = \frac{1}{\rho}(1-\sqrt{1-2\rho r} )= c_{\ref{Paper2:lem:comparisoneucgeod2:index}}(\rho r) r.
$$

Besides, $\MM_0^x$ is connected by Lemma \ref{Paper2:lem:connected}. We deduce that $\MM_0^x \subset \openballM{x_0}{c_{\ref{Paper2:lem:comparisoneucgeod2:index}}(\rho r) r}{\MMo}$.
In particular, $y_0$ must satisfy $\geoD{x_0}{y_0}{\MM_0} \leq c_{\ref{Paper2:lem:comparisoneucgeod2:index}}(\rho r) r$, and we deduce the result.
\end{proof}

Following the same idea, we can prove the following lemma.
It states that normal reach $\lambda_0(x_0)$ can be understood as the minimal distance $\eucN{x-y}$ for points $y_0$ far enough from $x_0$ in $\MMo$ 

\begin{lemma}
\label{lem:exists_projection}
Let $x_0,y_0 \in \MMo$ such that $\eucN{x-y}< \frac{1}{2\rho}$ and $\geoD{x_0}{y_0}{\MMo} \geq 4 \eucN{x-y}$. 
Then there exists a $z_0 \in \MMo$ such that $\geoD{y_0}{z_0}{\MMo} < 2\eucN{x-y}$, $\eucN{x-z} \leq \eucN{x-y}$ and $x-z \bot T_z\MM$.
Consequently, $\lambda_0(x_0) \leq \eucN{x-y}$.
\end{lemma}

\begin{proof}
Denote $r = \eucN{x-y}$ and $\MMo^x = u^{-1}(\closedball{x}{r})$.
Let $C_0^x$ denote the connected component of $x_0$ in $\MMo^x$.
As we have seen in the proof of Lemma \ref{Paper2:lem:comparisoneucgeod2}, $C_0^x$ is included in the closed geodesic ball $\closedballM{x_0}{c_{\ref{Paper2:lem:comparisoneucgeod2:index}}(\rho r) r}{\MM_0}$. 
Since $r < \frac{1}{2\rho}$, we have the inequality $c_{\ref{Paper2:lem:comparisoneucgeod2:index}}(\rho r) < 2$, and we deduce $C_0^x \subset \openballM{x_0}{2r}{\MMo}$.
Now, let $C_0^y$ denote the connected component of $y_0$ in $\MMo^x$.
Similarly, we have $C_0^y \subset \openballM{x_0}{2r}{\MMo}$.
Since $\geoD{x_0}{y_0}{\MMo} \geq 4r$, we deduce that $C_0^y$ and $C_0^x$ are disjoint. Now, let $z_0$ be a minimizer of $z_0 \mapsto \eucN{x-z}$ on $C_0^y$. We have $z_0 \neq x_0$. Moreover, $x-z \bot T_z \MM$, otherwise $z$ would not be a minimizer. Hence, by definition of the normal reach, $\lambda_0(x_0) \leq \eucN{x-z} \leq \eucN{x-y}$.
\end{proof}

The following proposition connects the normal reach to the usual notion of reach, in the case where $\MMo$ is embedded.

\begin{proposition}
\label{Paper2:prop:normalreachtoreach}
Suppose that $\imm\colon \MM_0 \rightarrow \MM \subset E$ is a $\CC^2$-embedding. Let $\tau>0$ be the reach of $\MM$. We have 
\[ \tau = \mini{\frac{1}{\rho_*},\frac{1}{2} \lambda_*}, \]
where $\rho_*$ is the supremum of the operator norm of the second fundamental forms of $\MMo$, and $\lambda_* = \inf_{x \in \MM}\normalreach{x}$ is the infimum of the normal reach.
\end{proposition}

\begin{proof}
We first prove that $\tau \geq \mini{ \frac{1}{\rho_*} , \frac{1}{2} \lambda_*}$.
According to Theorem \ref{background:th:aamari_reach}, two cases may occur: the reach is either caused by a bottleneck or by curvature.
In the first case, there exists $x, y \in \MM$ and $z \in \med{\MM}$ with $\eucN{x-y} = 2\tau$ and $\eucN{x-z} = \eucN{y-z} = \tau$. We deduce that $x-y \bot T_y \MM$. Hence by definition of $\normalreach{x}$, 
\[\normalreach{x} \leq \eucN{x-y} = 2 \eucN{x-z} \leq 2 \tau .\]
In the second case, there exists $x \in \MM$ and an arc-length parametrized geodesic $\gamma\colon I \rightarrow \MM$ such that $\gamma(0) = x$ and $\eucN{ \ddot{\gamma}(0) } = \frac{1}{\tau}$.
But $\eucN{\ddot \gamma(0)} \leq \rho_*$, hence $\frac{1}{\tau} \leq \rho_*$.
This disjunction shows that $\tau \geq \mini{\frac{1}{\rho}, \frac{1}{2} \lambda_\mathrm{min}}$.

We now prove that $\tau \leq \mini{\frac{1}{\rho_*}, \frac{1}{2} \lambda_*}$.
The inequality $\tau \leq \frac{1}{\rho_*}$ appears in \cite[Proposition 6.1]{niyogi2008finding}. 
To prove $\tau \leq \frac{1}{2}\lambda_*$, consider any $x_0 \in \MM_0$. Let $y_0 \in \Lambda(x_0)$ such that $\eucN{x - y}$ is minimal.
Using Theorem \ref{background:th:federer} and the property $x-y \bot T_y \MM$, we immediately have
\begin{align*}
\tau 
&\leq \frac{\eucN{x-y}^2}{2 \dist{y-x}{T_y \MM}}
= \frac{\eucN{x-y}}{2} 
= \frac{\lambda(x)}{2},
\end{align*}
and the result follows.
\end{proof}

\begin{remark}
\label{rem:generalization_charac_reach}
We can generalize Proposition \ref{Paper2:prop:normalreachtoreach} as follows: suppose that $\imm\colon \MM_0 \rightarrow \MM \subset E$ is a $\CC^2$-immersion (potentially an embedding) that satisfies Hypothesis \hyperref[hyp:1]{1}. Let $\tau\geq0$ be the reach of $\MM$. We have 
\[ \tau = \mini{\frac{1}{\rho_*}, \frac{1}{2} \lambda_*}. \]
In fact, if $\imm$ is not an embedding, we can show that $\tau = 0$ and $\lambda_* = 0$. 
Indeed, if $x \in \MM$ is a point that admits several preimages by $\imm$, we have seen that $\lambda(x) = 0$. On the other hand, by Hypothesis \hyperref[hyp:1]{1}, the tangent cone $\Tan{\MM}{x}$ is not an affine subspace but an union of several affine subspaces. In this case, we see that Theorem \ref{background:th:federer} cannot hold, hence that $\tau = 0$.

Note that if Hypothesis \hyperref[hyp:1]{1} is not satisfied, it is possible to have $\tau > 0$ but $ \lambda_* = 0$. This would be the case for any non-injective immersion $\imm\colon \MM_0 \rightarrow \MM$ such that its image $\MM$ is a $\CC^{2}$-submanifold.
%
\end{remark}
\medbreak

As shown by the previous remark, when $\MM$ is not a submanifold, global quantities such as the reach $\tau$ or the minimal normal reach $\lambda_*$ are zero.
However, as shown by the following proposition, the normal reach gives a scale at which $\MM$ still behaves well. 
Note that we shall not make use of this result in the rest of the paper. 
\begin{proposition}
Assume that $\MMo$ satisfies Hypothesis \hyperref[hyp:2]{2}.
Let $x \in \MM_0$ and $r < \mini{\frac{1}{4\rho}, \normalreach{x}}$. 
Then $\MM \cap \closedball{x}{r}$ is a set of reach at least $\frac{1-2\rho r}{\rho}$. 
\label{Paper2:prop:normalreach}
\end{proposition}

\begin{figure}[H]
\centering
\includegraphics[width=.55\linewidth]{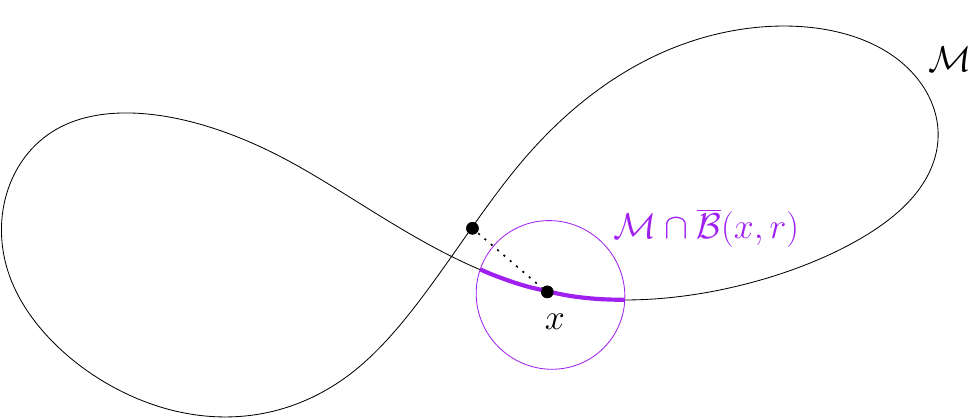}
\caption{The set $\MM \cap \closedball{x}{r}$ has positive reach.}
\end{figure}

\begin{proof}
Denote $\MM^x = \MM \cap \closedball{x}{r}$ and $\MM^x_0 = \imm^{-1}(\MM^x)$.
In order to give a bound on the reach of $\MM^x$, we will use the characterization of Theorem \ref{background:th:federer}. 
First, let us prove that for every $y_0, z_0 \in \MM^x_0$, 
\begin{align*}
\dist{z-y}{T_y \MM} \leq \frac{\rho}{2(1-2\rho r)} \eucN{z-y}^2.
\end{align*}
Let $y_0, z_0 \in \MM_0^x$, and $\delta = \geoD{y_0}{z_0}{\MMo}$. 
Lemma \ref{Paper2:lem:Federergeod} Point \ref{lem:Federergeod:point3} gives $\delta \leq \frac{1}{1 - \frac{\rho}{2} \delta} \eucN{y-z}$. 
Moreover, $\delta \leq \geoD{y_0}{x_0}{\MMo} + \geoD{x_0}{z_0}{\MMo} \leq 2 c_{\ref{Paper2:lem:comparisoneucgeod2:index}}(\rho r) r$. Hence,
\begin{align*}
\frac{1}{1 - \frac{\rho}{2} \delta}
&\leq \frac{1}{1 - c_{\ref{Paper2:lem:comparisoneucgeod2:index}}(\rho r) \rho r}
= \frac{1}{\sqrt{1-2\rho r}},
\end{align*}
and we deduce that
\begin{equation}
\label{Paper2:eq:proofreachmajoration}
\delta \leq \frac{1}{\sqrt{1-2\rho r}} \eucN{y-z}.
\end{equation}
Besides, Lemma \ref{Paper2:lem:Federergeod} Point \ref{lem:Federergeod:point2} gives $\dist{z-y}{T_y \MM} \leq \frac{\rho}{2} \delta^2$, and combining these two inequalities yields
$\dist{z-y}{T_y \MM} \leq \frac{\rho}{2(1-2\rho r)} \eucN{z-y}^2$.


Secondly, let us prove that 
\begin{equation}
\dist{z-y}{\Tan{\MM^x}{y}} \leq  \frac{\rho}{2(1-2\rho r)} \eucN{z-y}^2,
\label{Paper2:eq:proofreach}
\end{equation}
where $\Tan{\MM^x}{y}$ is the tangent cone at $y$ of the closed set $\MM^x$. 
According to Equation \eqref{Paper2:eq:proofreachmajoration}, it is enough to prove that $\Tan{\MM^x}{y} = T_y \MM$. We shall prove that $\MM$ does not self-intersect in $\closedball{x}{r}$.
According to Lemma \ref{Paper2:lem:comparisoneucgeod2}, we have $\MM_0^x \subset \openballM{x_0}{\frac{1}{\rho}}{\MMo}$. 
Using Lemma \ref{Paper2:lem:Federergeod} Point \ref{lem:Federergeod:point4}, we get that $\imm$ is injective on $\MM_0^x$, as wanted.
%
%
To conclude the proof, it follows from Theorem \ref{background:th:federer} and Equation \eqref{Paper2:eq:proofreach} that $\MM^x$ has reach at least $\frac{1-2\rho r}{\rho}$.
\end{proof}

\subsection{Probabilistic bounds under normal reach conditions}
\label{subsec:probabilistic_bounds}
We now consider $\MMo$ and $\muo$ which satisfy Hypotheses \hyperref[hyp:2]{2} and \hyperref[hyp:3]{3}.
The aim of this subsection is to provide a quantitative control of the measure $\mu = \imm_* \mu_0$, that is, bounds on the measure of balls and annuli (see Propositions \ref{Paper2:prop:probabilitybounds} and \ref{Paper2:prop:probabilityboundssqrt}).
We do so by pulling-back $\mu$ on the tangent spaces $T_x \MM$, where it is simpler to compute integrals (see Lemma \ref{Paper2:lem:densityg}).

Recall that the exponential map of $\MMo$ at a point $x_0 \in \MMo$ is denoted
\begin{align*}
\exp_{x_0}^{\MMo} \colon T_{x_0} \MMo \rightarrow \MMo.
\end{align*}
To ease the reading of this subsection, we introduce the \emph{exponential map seen in $\MM$}, denoted $\exp^{\MM}_{x} \colon T_{x} \MM \rightarrow \MM$.
It is defined as
\begin{align*}
\exp^{\MM}_{x} =  u \circ \exp_{x_0}^{\MMo} \circ (d_{x_0} \imm)^{-1}.
\end{align*}
Note that the map $\exp^{\MM}_{x}$ is well-defined, even if $x$ is a self-intersection point of $\MM$. Indeed, $\exp^{\MM}_{x}$ will always refer implicitely to a choice of point $x_0$ such that $x = \imm(x_0)$. This is consistent with the notation conventions of Subsect. \ref{Paper2:subsec:model}.
This map fits in the following commutative diagram:
\[
\begin{tikzcd}
\MM_0 \arrow[r, "\imm"] 
&\MM 
\\
T_{x_0} \MMo \arrow[r, "d_{x_0} \imm"] \arrow[u, "\exp^{\MMo}_{x_0}"]
& T_{x} \MM \arrow[u, dashed, "\exp^{\MM}_{x}" right]
\end{tikzcd}
\]
We also define the map $\overline \exp^{\MM}_{x}$ as the restriction of $\exp^{\MM}_{x}$ to the closed ball $\closedballM{0}{\frac{2}{\rho}}{T_{x}\MM}$ of $T_{x}\MM$.
It is injective by Lemma \ref{Paper2:lem:Federergeod} Point \ref{lem:Federergeod:point4}, and its image is $\imm\left(\closedballM{x_0}{\frac{2}{\rho}}{\MMo}\right)$.
Moreover, for any $r \leq \mini{\frac{1}{2 \rho}, \normalreach{x} }$, $\imm\left(\closedballM{x_0}{\frac{2}{\rho}}{\MMo}\right)$ contains $\MM\cap\closedball{x}{r}$ by Lemma \ref{Paper2:lem:comparisoneucgeod2}, hence we can consider its inverse
\begin{equation}
\label{eq:exp_map_seen_in_M}
(\overline \exp^{\MM}_{x})^{-1} \colon \MM\cap\closedball{x}{r} \longrightarrow T_x\MM.
\end{equation}
The next lemma gathers previous results.
We remind the reader that the $d$-dimensional Jacobian has been defined in Subsect. \ref{subsec:riemannian}.

\begin{lemma}
\label{Paper2:lem:regularityexp}
Let $x_0 \in \MMo$ and $r < \mini{\frac{1}{2\rho} , \lambda(x)}$.
Denote $\closedballB_0 = \left(\overline \exp^{\MM}_x\right)^{-1}\left(\closedball{x}{r}\right)$.
We have the inclusions
\begin{align*}
\closedballG{0}{r}{T_{x}\MM} 
~\subset~ \closedballB_0 
~\subset~ \closedballG{0}{c_{\ref{Paper2:lem:comparisoneucgeod2:index}}(\rho r)r }{T_{x}\MM}.
\end{align*}
Moreover, for all $v \in \closedballB_0$, the Jacobian $J_v$ of $\overline\exp^{\MM}_x$, 
is bounded by 
\begin{align*}
\bigg(1-\frac{(r\rho)^2}{6}\bigg)^d \leq J_v \leq \big(1+(r\rho)^2\big)^d,
\end{align*}
and these terms are bounded by $\Jmin = (\frac{23}{24})^d$ and $\Jmax = (\frac{5}{4})^d$.
\end{lemma}
\begin{proof}
The inclusions come from Lemma \ref{Paper2:lem:comparisoneucgeod2}.
The bounds on the Jacobian come from Lemma \ref{Paper2:lem:regularityexp0} and the fact that $c_{\ref{Paper2:lem:comparisoneucgeod2:index}}(\rho r)r \leq 2 r \leq \frac{1}{\rho} \leq \frac{\pi}{2\sqrt{2}\rho}$ when $r < \frac{1}{2\rho}$.
\end{proof}

We now study the measure $\mu$. By definition, it is the push-forward of $\muo$ by $\imm$.
By applying the coarea formula, and in particular Equation \eqref{eq:density_coarea}, we obtain that $\mu$ admits the following density against $\HH^d_\MM$, the $d$-dimensional Hausdorff measure restricted to $\MM$:
\begin{align*}
f(x) = \sum_{ x_0 \in \imm^{-1}(\{x\}) } f_0 ( x_0 ).
\end{align*}
Indeed, in this case, the Jacobian is always 1, since $\MMo$ has been given the pull-back Riemannian metric.
Note that if $x$ has only one preimage by $u$---i.e., if $\normalreach{x}>0$---then $f(x) = f_0 \circ \imm^{-1}(x)$.
In the rest of the paper, we will only use $f$ on points $x$ such that $\lambda(x)>0$.
This is motivated by the fact that in Sect. \ref{Paper2:sec:tangentspaceestimation} and \ref{Paper2:sec:topoinference}, we will assume Hypothesis \hyperref[hyp:4]{4}, which gives that the measure of the set $\{x_0 \in \MMo \mid \lambda_0(x_0)=0\}$ is zero.
Moreover, we have a Lipschitz-like property for the density $f$, valid as long as the points are chosen far enough from the self-intersection of $\MM$:

\begin{lemma}
\label{Paper2:rem:lipschitz}
For all $x_0, y_0 \in \MMo$ such that $\eucN{x-y} < \mini{\frac{1}{2 \rho} , \lambda(x)}$, we have
\[ |f(x) - f(y)| \leq 2L_0 \eucN{x-y}. \]
\end{lemma}

\begin{proof}
Recall that, by Hypothesis \hyperref[hyp:3]{3}, the density $f_0$ is $L_0$-Lipschitz with respect to the geodesic distance: for all $x_0, y_0 \in \MMo$, we have  $|f_0(x_0) - f_0(y_0)| \leq L_0 \cdot \geoD{x_0}{y_0}{\MMo}$.
To prove the lemma, we start with the case where $y$ has only one preimage by $\imm$, so that we can write $f(y) = f_0 \circ \imm^{-1}(y)$. 
Since $\eucN{x-y} < \lambda(x)$ by assumption, we have $0<\lambda(x)$, hence $x$ also has only one preimage.
Now we have
\begin{align*}
\left|f(x) - f(y)\right| &= \left|f_0 \circ \imm^{-1}(x) - f_0 \circ \imm^{-1}(y)\right|  \\
&\leq L_0 \cdot \geoD{\imm^{-1}(x)}{\imm^{-1}(y)}{\MMo}  \\
&\leq 2 L_0 \eucN{x-y},
\end{align*}
where we used Lemma $\ref{Paper2:lem:comparisoneucgeod2}$ on the last inequality.
Now we prove that $\eucN{x-y} < \mini{\frac{1}{2 \rho} , \lambda(x)}$ implies that $y$ has only one preimage. Let $r = \eucN{x-y}$, and suppose by contradiction that $y_0, y_0'$ are two distinct preimages. According to Lemma \ref{Paper2:lem:Federergeod} Point \ref{lem:Federergeod:point4}, $\geoD{y_0}{y_0'}{\MMo} \geq \frac{2}{\rho}$. But Lemma \ref{Paper2:lem:comparisoneucgeod2} says that $\imm^{-1}(\openball{x}{r}) \subset \openballM{x_0}{2 r}{\MM_0} \subset \openballM{x_0}{\frac{1}{\rho}}{\MM_0}$, which yields the contradiction $\geoD{y_0}{y_0'}{\MMo} \leq \geoD{y_0}{x_0}{\MMo} + \geoD{y_0'}{x_0}{\MMo} < \frac{1}{\rho}$.
\end{proof}

We now state the key lemma of this subsection, that allows to go from a measure on $\MMo$ to a measure on $\MM$. 

\begin{lemma}
\label{Paper2:lem:densityg}
Let $x_0 \in \MMo$ and $r < \mini{\frac{1}{2\rho} , \lambda(x)}$.
Consider $\loc{\mu}{x}$, the measure $\mu$ restricted to $\closedball{x}{r}$, define $\closedballB_0 = \left(\overline \exp^{\MM}_x\right)^{-1}\left(\closedball{x}{r}\right)$ and the push-forward
\begin{align*}
\loc{\nu}{x} = \left(\overline \exp_x^{\MM}\right)^{-1}_* \loc{\mu}{x},
\end{align*}
where $(\overline \exp^{\MM}_{x})^{-1}$ has been defined in Equation \eqref{eq:exp_map_seen_in_M}.
The measure $\loc{\nu}{x}$ admits the following density against the $d$-dimensional Hausdorff measure on $T_x \MM$:
\begin{align*}
g(v) &= f\big(\overline \exp_{x}^{\MM}(v)\big) \cdot J_{v} \cdot \indicatrice{\closedballB_0}(v).
\end{align*}
Moreover, for all $v \in \closedballB_0$, the map $g$ satisfies 
$|g(v) - g(0)| \leq c_{\ref{Paper2:lem:densityg:index}} r,$
where $c_{\ref{Paper2:lem:densityg:index}} = 4 L_0 \Jmax + \frac{d}{2}\rho \fmax$.
\end{lemma}

\begin{figure}[H]
\centering
\includegraphics[width=.95\linewidth]{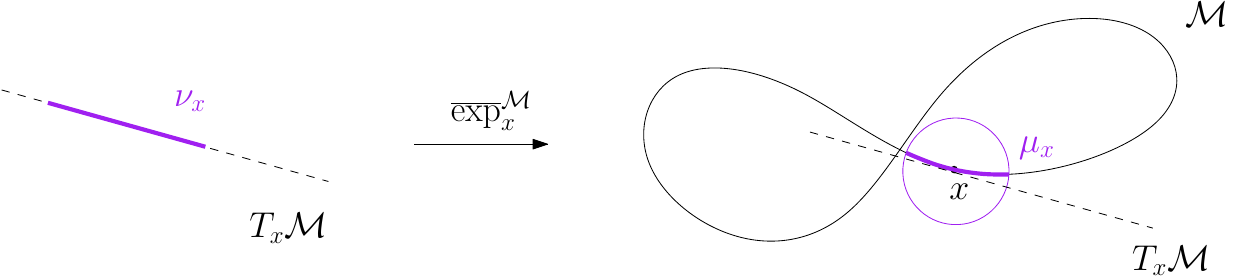}
\caption{Measures involved in Lemma \ref{Paper2:lem:densityg}.}
\end{figure}

\begin{proof}
The expression of $g$ comes from the coarea formula (Equation \eqref{eq:density_coarea}) applied to the map 
$(\overline \exp^{\MM}_{x})^{-1} \colon \MM\cap\closedball{x}{r} \rightarrow \closedballB_0$,
and the measure $\loc{\nu}{x} = \left(\overline \exp_x^{\MM}\right)^{-1}_* \loc{\mu}{x}$.
To prove the inequality, observe that we can decompose
\begin{align*}
g(v) - g(0)
&= f\big( \exp^{\MM}_x(v) \big) J_v - f\big( \exp^{\MM}_x(0) \big)J_{0}\\
&= \bigg[ f\big(\exp^{\MM}_x(v)\big)-f\big(\exp^{\MM}_x(0)\big) \bigg] J_{v} + \left(J_v-J_{0}\right)f\big(\exp^{\MM}_x(0)\big).
\end{align*}
On the one hand, using Lemma \ref{Paper2:rem:lipschitz}, we get
\begin{align*}
\left|f\left(\overline \exp^{\MM}_x(v)\right)-f\left(\overline \exp^{\MM}_x(0)\right)\right| 
&\leq 2L_0 \eucN{\overline \exp^{\MM}_x(v)-\overline \exp^{\MM}_x(0)} \\
&= 2L_0 \eucN{\imm \circ \exp^{\MMo}_{x_0}(v) - \imm \circ \exp^{\MMo}_{x_0}(0)} \\
&\leq 2L_0 \cdot \geoD{\overline \exp^{\MMo}_{x_0}(v) }{ x_0 }{\MMo}
= 2L_0 \eucN{v}.
\end{align*}
On the other hand, $J_0 = 1$ and $\big(1-\frac{(r\rho)^2}{6}\big)^d \leq J_{v} \leq \big(1+(r\rho)^2\big)^d$ yield $|J_v-J_{0}| \leq d (\rho r)^2 \leq \frac{d}{2}\rho r$.
Using the triangle inequality we see that
\begin{align*}
|g(v) - g(0)|
&\leq 2L_0 \eucN{v} \Jmax + \fmax \frac{d}{2}\rho r
\leq \left(4L_0 \Jmax + \fmax \frac{d}{2}\rho\right) r,
\end{align*}
as wanted.
\end{proof}

\begin{remark}
\label{Paper2:rem:densityg0}
In the same vein as Lemma \ref{Paper2:lem:densityg}, define $\overline \exp_{x_0}^{\MMo}$ to be the map $\exp_{x_0}^{\MMo}$ restricted to $\closedballM{0}{\frac{2}{\rho}}{T_{x_0} \MMo}$. 
For any $x_0 \in \MMo$, let $\mu_0^{x_0}$ be the measure $\mu_0$ restricted to $\closedballM{x_0}{\frac{2}{\rho}}{\MMo}$, and define the measure 
\begin{align*}
\nu_0 = (\overline \exp_{x_0}^{\MMo})^{-1} \mu_0^{x_0}.
\end{align*}
Using the area formula, one shows that $\nu_0$ admits the following density over the $d$-dimensional Hausdorff measure on $T_{x_0} \MMo$:
\begin{align*}
g_0(v) = f_0\left(\overline \exp_{x_0}^{\MMo}(v)\right) \cdot J_v \cdot \indicatrice{\closedballM{0}{\frac{2}{\rho}}{T_{x_0} \MMo}}(v).
\end{align*}
\end{remark}
\medbreak

Now we can use the density $g$ of Lemma \ref{Paper2:lem:densityg} to derive explicit bounds on $\mu$. We remind the reader that $\volball{d}$ denote the volume of the unit ball of $\R^d$.
\begin{proposition}
\label{Paper2:prop:probabilitybounds}
Let $x_0 \in \MM_0$, $r \leq \mini{\frac{1}{2\rho} , \normalreach{x}}$ and $s \in [0,r]$.
We have
\begin{enumerate}
\itemsep.15cm
\item $\mu\left(\closedball{x}{r}\right) \geq c_{\ref{Paper2:hyp:muA:index}} r^d$, \label{prop:probabilitybounds:point1}
\item $ \left| \frac{\mu(\closedball{x}{r})}{\volball{d}r^d} - f(x) \right| \leq c_{\ref{Paper2:prop:probabilitybounds:index}} r,$\label{prop:probabilitybounds:point2}
\item $\mu\left(\closedball{x}{r} \setminus \closedball{x}{s}\right) \leq c_{\ref{Paper2:hyp:muB:index}}  r^{d-1}(r-s).$\label{prop:probabilitybounds:point3}
\end{enumerate}
with $c_{\ref{Paper2:hyp:muA:index}} = \fmin \Jmin \volball{d}$, $c_{\ref{Paper2:prop:probabilitybounds:index}} = c_{\ref{Paper2:lem:densityg:index}} + \fmax \Jmax  d 2^{d} \rho$ and $c_{\ref{Paper2:hyp:muB:index}} = d 2^{d}\fmax \Jmax \volball{d}$.
\end{proposition}

\begin{figure}[H]
\centering
\begin{minipage}{.49\linewidth}
\centering
\includegraphics[width=.95\linewidth]{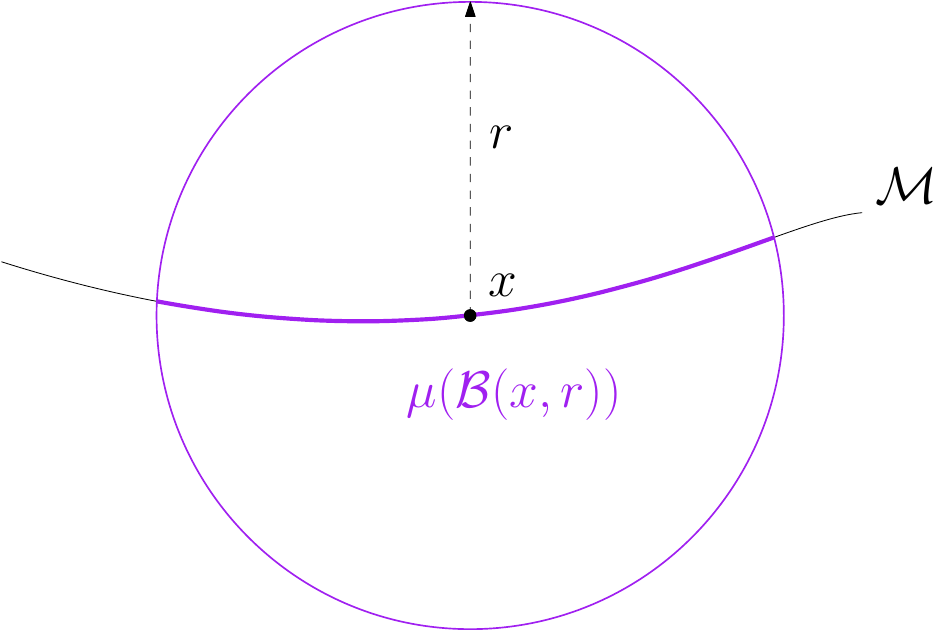}
\end{minipage}
\begin{minipage}{.49\linewidth}
\centering
\includegraphics[width=.95\linewidth]{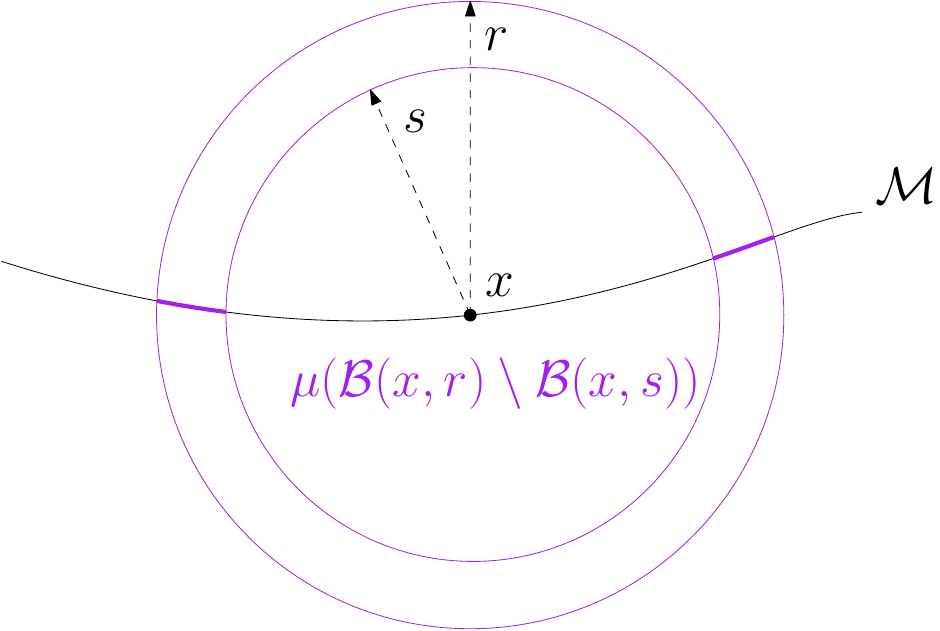}
\end{minipage}
\caption{Representation of Proposition \ref{Paper2:prop:probabilitybounds} Point \ref{prop:probabilitybounds:point1} (left) and Point \ref{prop:probabilitybounds:point3} (right).}
\end{figure}

\begin{proof}
Consider the map $(\overline \exp_x^{\MM})^{-1}$ defined in Equation \eqref{eq:exp_map_seen_in_M} and the measure $\loc{\nu}{x} = \left(\overline \exp_x^{\MM}\right)^{-1}_* \loc{\mu}{x}$ as defined in Lemma \ref{Paper2:lem:densityg}.
In the following, we write $T= T_x \MM$, and $\closedballB_0 = \left(\overline  \exp_{x}^{\MM}\right)^{-1}\left( \closedball{x}{r} \right)$.

\medbreak \noindent \emph{Point \ref{prop:probabilitybounds:point1}.} 
By definition of $\nu_x$, we have
$\mu\left(\closedball{x}{r}\right)
= \nu_x\big( \closedballB_0 \big)$.
Writing down the density $g$ of $\nu_x$ yields
\begin{align*}
\nu_x\big(\closedballB_0\big)
= \int_{ \closedballB_0 } g(v) \dd \HH^d(v).
\end{align*}
According to the expression of $g$ in Lemma \ref{Paper2:lem:densityg}, we have $g \geq \fmin \Jmin$. Therefore,
\begin{align*}
\int_{ \closedballB_0 } g(v) \dd \HH^d(v)
\geq \int_{ \closedballB_0 } \fmin \Jmin \dd \HH^d(v)
= \fmin \Jmin \HH^d\big( \closedballB_0 \big).
\end{align*}
Besides, since $\closedballB_0 \supset \closedballM{0}{r}{T} $, we have
\begin{align*}
\HH^d\big( \closedballB_0 \big) \geq \HH^d\left( \closedballM{0}{r}{T}\right) = \volball{d}r^d.
\end{align*}
We finally obtain $\nu_x\big(\closedballB_0\big) \geq \fmin \Jmin \volball{d}r^d$.


\medbreak \noindent \emph{Point \ref{prop:probabilitybounds:point2}.} 
Observe that $\int_{\closedballM{0}{r}{T}} f(x) \dd \HH^d(v) = f(x) \volball{d} r^d$. Hence
\begin{align}
\label{terms:probabilitybounds} \bigg|\mu(\closedball{x}{r}) - f(x) \volball{d} r^d\bigg|
&=\left | \int_{ \closedballB_0 } g(v) \dd \HH^d(v) 
- \int_{\closedballM{0}{r}{T}} f(x) \dd \HH^d(v) \right|\nonumber \\
&\leq \underbrace{\left|\int_{\closedballM{0}{r}{T}} (f(x)-g(v)) \dd \HH^d(v)\right|}_{\text{A}} 
+ \underbrace{\left|\int_{ \closedballB_0 \setminus \closedballM{0}{r}{T}} g(v) \dd \HH^d(v)\right |}_{\text{B}}. 
\end{align}
To bound Term \hyperref[terms:probabilitybounds]{A}, notice that $g(0) = f(\exp_x^{\MM}(0)) J_0 = f(x)$. Hence we can write:
\begin{align*}
\left|\int_{\closedballM{0}{r}{T}} (f(x)-g(v)) \dd \HH^d(v)\right| 
&\leq \int_{\closedballM{0}{r}{T}} \big|g(0)-g(v)\big| \dd \HH^d(v) .
\end{align*}
Now, Lemma \ref{Paper2:lem:densityg} gives $|g(v) - g(0)| \leq c_{\ref{Paper2:lem:densityg:index}} r$, and we eventually obtain the inequality $\left|\int_{\closedballM{0}{r}{T}} (f(x)-g(v)) \dd \HH^d(v)\right| \leq c_{\ref{Paper2:lem:densityg:index}} r \volball{d} r^d$.

On the other hand, we bound Term \hyperref[terms:probabilitybounds]{B} thanks to the inclusion $\closedballB_0 \subset \closedballM{0}{c_{\ref{Paper2:lem:comparisoneucgeod2:index}}(\rho r) r}{T}$. 
Denote $\mathcal{A} = \closedballM{0}{c_{\ref{Paper2:lem:comparisoneucgeod2:index}}(\rho r)r}{T} \setminus \closedballM{0}{r}{T}$.
We have $\closedballB_0 \setminus \closedballM{0}{r}{T}\subset \mathcal{A}$, hence
\begin{align*}
\int_{ \closedballB_0 \setminus \closedballM{0}{r}{T}} g(v) \dd \HH^d(v)
&\leq \int_{ \mathcal{A} } g(v) \dd \HH^d(v) 
\leq \fmax \Jmax  \HH^d(\mathcal{A}).
\end{align*}
Moreover, we have
\begin{align*}
\HH^d(\mathcal{A}) 
&= \HH^d\left(\closedballM{0}{c_{\ref{Paper2:lem:comparisoneucgeod2:index}}(\rho r)r}{T}\right) - \HH^d\left(\closedballM{0}{r}{T}\right)
= \volball{d} \left(c_{\ref{Paper2:lem:comparisoneucgeod2:index}}(\rho r)^d - 1\right)r^d.
\end{align*}
We can use $c_{\ref{Paper2:lem:comparisoneucgeod2:index}}(\rho r) \leq 1+2 \rho r \leq 2$ and the inequality $a^d-1 \leq d(a-1)a^{d-1}$, where $a\geq 1$, to get
\begin{align*}
\left(c_{\ref{Paper2:lem:comparisoneucgeod2:index}}(\rho r)^d - 1\right)
&\leq  d \cdot \left(c_{\ref{Paper2:lem:comparisoneucgeod2:index}}(\rho r) - 1\right) \cdot c_{\ref{Paper2:lem:comparisoneucgeod2:index}}(\rho r)^{d-1} 
\leq d \cdot 2 \rho r \cdot 2^{d-1}.
\end{align*}
We finally deduce the following bound on Term \hyperref[terms:probabilitybounds]{B}:
\begin{align*}
\int_{ \closedballB_0 \setminus \closedballM{0}{r}{T}} g(v) \dd \HH^d(v)
\leq \fmax \Jmax \volball{d}r^d d \cdot \rho r 2^{d}.
\end{align*}
Gathering Terms \hyperref[terms:probabilitybounds]{A} and \hyperref[terms:probabilitybounds]{B}, we obtain
\begin{align*}
\left|\mu(\closedball{x}{r}) - f(x) \volball{d} r^d\right|
\leq r\left(c_{\ref{Paper2:lem:densityg:index}} + \fmax \Jmax  d \rho 2^{d}\right) \volball{d} r^d.
\end{align*}

\medbreak \noindent \emph{Point \ref{prop:probabilitybounds:point3}.}
Let us write
\begin{align*}
\mu\left(\closedball{x}{r} \setminus \closedball{x}{s}\right)
&= \nu_x\big( \big(\overline \exp_x^{\MM}\big)^{-1} \left(\closedball{x}{r} \setminus \closedball{x}{s}\right) \big) \\
&=  \int_{\left(\overline \exp_x^{\MM}\right)^{-1} \left(\closedball{x}{r} \setminus \closedball{x}{s}\right)} g(v) \dd \HH^d(v).
\end{align*}
In spherical coordinates, this integral reads
\begin{equation}
\label{Paper2:eq:probabilitybounds_point3}
\int_{\left(\overline \exp_x^{\MM}\right)^{-1} \left(\closedball{x}{r} \setminus \closedball{x}{s}\right)} g(v) \dd \HH^d(v)
=  \int_{v \in \sphereG{0}{1}{T}} \int_{t=a(v)}^{b(v)} g(tv) t^{d-1} \dd t \dd v,
\end{equation}
where $a(v)$ and $b(v)$ are defined as follows: for every $v \in T_x \MM$ of unit norm, let $\gamma_0$ be an arc-length parametrized geodesic with $\gamma_0(0) = x$ and $\dot \gamma_0(0) = v$, and set $a(v)$ and $b(v)$ to be the first positive values such that $\eucN{\gamma(a(v)) - x} = s$ and $\eucN{\gamma(b(v)) - x} = r$.
\begin{figure}[H]
\centering
\includegraphics[width=.4\linewidth]{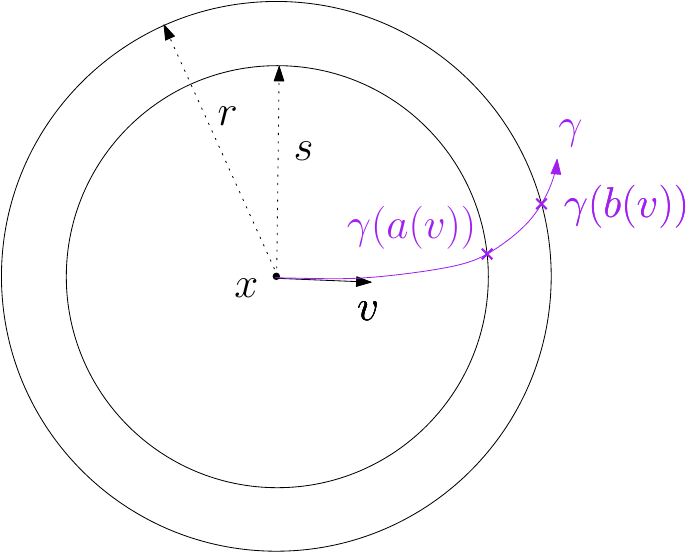}
\caption{Illustration of $a(v)$ and $b(v)$ in Equation \eqref{Paper2:eq:probabilitybounds_point3}.}
\end{figure}
\noindent
For any $v \in \sphereG{0}{1}{T}$, Lemma \ref{Paper2:lem:distancetocenter} Point \ref{lem:distancetocenter:point2} gives $b(v) \leq 2r$, hence
\begin{align*}
\int_{t=a(v)}^{b(v)} g(tv) t^{d-1} \dd t 
\leq \int_{t=a(v)}^{b(v)} \fmax \Jmax (2r)^{d-1} \dd t.
\end{align*}
Moreover, according to Lemma \ref{Paper2:lem:distancetocenter} Point \ref{lem:distancetocenter:point6}, we have $b(v) - a(v) \leq 2(r-s)$, hence
\begin{align*}
\int_{t=a(v)}^{b(v)} \fmax \Jmax (2r)^{d-1} \dd t
&= (b(v)-a(v)) \fmax \Jmax (2r)^{d-1} \dd t\\
&\leq 2(r-s) \fmax \Jmax (2r)^{d-1}.
\end{align*}
From these last two equations we deduce
\begin{align*}
\int_{v \in \sphere{0}{1}} \int_{t=a(v)}^{b(v)} g(tv) t^{d-1} \dd t \dd v
&\leq2 (r-s) \fmax \Jmax (2r)^{d-1} \int_{v \in \sphere{0}{1}}  \dd v \\
&= 2 (r-s) \fmax \Jmax (2r)^{d-1} d \volball{d}.
\end{align*}
Going back to Equation \eqref{Paper2:eq:probabilitybounds_point3}, we obtain
\begin{align*}
\mu\left(\closedball{x}{r} \setminus \closedball{x}{s}\right)
&= 2^{d} d \volball{d} \fmax \Jmax (r-s) r^{d-1},
\end{align*}
which concludes the proof.
\end{proof}

In Sect. \ref{Paper2:sec:tangentspaceestimation}, we will study the estimation of tangent spaces of $\MM$ thanks to the normal reach. By using the previous proposition, we will be able to give precise bounds around points $x\in \MM$ with large normal reach $\lambda(x)$. However, for points with small normal reach, we won't be able to use it. Therefore we need a version of Proposition \ref{Paper2:prop:probabilitybounds} without normal reach condition. This is the aim of the following result.
\begin{proposition}
\label{Paper2:prop:probabilityboundssqrt}
Let $x_0 \in \MM_0$, $r \leq \frac{1}{2\rho}$ and $s \in [0,r]$. We have
\begin{enumerate}
\itemsep.15cm
\item $\mu\left(\closedball{x}{r}\right) \geq c_{\ref{Paper2:hyp:muA:index}} r^d$  \label{prop:probabilityboundssqrt:point1}
\item $\mu\left(\closedball{x}{r} \setminus \closedball{x}{s}\right) \leq c_{\ref{Paper2:hyp:muBsqrt:index}} r^{d-\frac{1}{2}} (r-s)^\frac{1}{2}$ \label{prop:probabilityboundssqrt:point3}
\end{enumerate}
with $c_{\ref{Paper2:hyp:muA:index}} =\fmin \Jmin \volball{d}$ and $c_{\ref{Paper2:hyp:muBsqrt:index}}= \frac{\fmax \Jmax}{\fmin \Jmin }( \frac{\rho}{\sqrt{4 - \sqrt{13}}})^d d 2^{2d} \sqrt{3}$.
\end{proposition}
Note that Point \ref{prop:probabilityboundssqrt:point1} is similar to Proposition \ref{Paper2:prop:probabilitybounds} Point \ref{prop:probabilitybounds:point1}, and that Point \ref{prop:probabilityboundssqrt:point3} is a weaker form of Proposition \ref{Paper2:prop:probabilitybounds} Point \ref{prop:probabilitybounds:point3}. There is no equivalent of Proposition \ref{Paper2:prop:probabilitybounds} Point \ref{prop:probabilitybounds:point2} without normal reach condition.

\begin{proof}
\label{Paper2:appendix:lem:boundsqrt}
Let $\MM^x = \MM \cap \closedball{x}{r}$ and $\MM_0^x = \imm^{-1}(\MM^x)$. 
Lemma \ref{Paper2:lem:comparisoneucgeod2} does not apply: it is not true that $\MM_0^x \subset \closedballM{x_0}{ c_{\ref{Paper2:lem:comparisoneucgeod2:index}}(\rho r) r }{\MMo}$.
However, we can decompose $\MM_0^x$ in connected components $C_0^i, i \in I$. They are represented in Figure \ref{fig:appendix:connected_components}.
\begin{figure}[H]
\centering
\includegraphics[width=.8\linewidth]{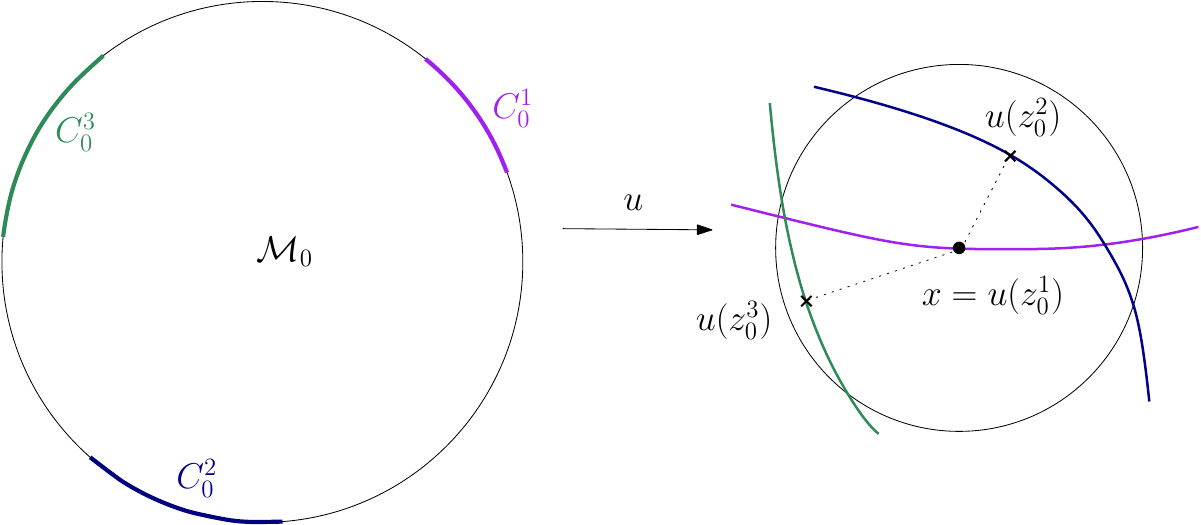}
\caption{The connected components $C_0^i$, $i\in I$.}
\label{fig:appendix:connected_components}
\end{figure}
\noindent
For every $i \in I$, let $z_0^i$ be a minimizer of $z_0 \mapsto \eucN{z-x}$ on $C_0^i$. We have $x-z^i \bot T_{z^i} \MM$. Following the same proof as Lemma \ref{Paper2:lem:comparisoneucgeod2}, one shows that $C_0^i$ is included in the geodesic ball $\closedballM{z_0^i}{\frac{1}{\rho}}{\MMo}$.
Hence we can consider $\mu_0^i$, the measure $\mu_0$ restricted to $C_0^i$, and define $\nu_0^i = (\overline \exp_{z_0}^{\MMo})^{-1}_* \mu_0^i$, as in Remark \ref{Paper2:rem:densityg0}. 
The measure $\nu_0^i$ admits $g_0^i$ as a density over the $d$-dimensional Hausdorff measure on $T_{z_0^i} \MMo$, where
\begin{align*}
g_0^i(v) = f_0\big(\overline \exp_{z_0}^{\MMo}(v)\big) \cdot J_v \cdot \indicatrice{ (\overline \exp_{z_0}^{\MMo})^{-1}(C_0^i) }(v).
\end{align*}

\medbreak \noindent \emph{Point \ref{prop:probabilityboundssqrt:point1}.}
By definition of $\mu$, can write
\begin{align*}
\mu(\closedball{x}{r})
= \mu_0(\imm^{-1}(\closedball{x}{r}))
&= \sum_{i \in I} \mu_0( C_0^i).
\end{align*}
Denote by $0 \in I$ the index of the connected component of $\MM_0^x$ which contains $x_0$. We have $C_0^{0} \supset \closedballM{x_0}{r}{\MMo}$. As in the proof of Proposition \ref{Paper2:prop:probabilitybounds} Point \ref{prop:probabilitybounds:point1}, we deduce that 
\begin{align*}
\mu_0( C_0^{0}) 
&\geq \int_{ (\overline \exp_{z_0}^{\MMo})^{-1}(C_0^0) } g_0^{0} \cdot\dd \HH^d \\ 
&\geq \fmin \Jmin \HH^d\big( (\overline \exp_{z_0}^{\MMo})^{-1}(C_0^0) \big)
= \fmin \Jmin \volball{d}r^d.
\end{align*}
Therefore, $\mu(\closedball{x}{r}) \geq \fmin \Jmin \volball{d}r^d$.

\medbreak \noindent \emph{Point \ref{prop:probabilityboundssqrt:point3}.}
For any $i \in I$, define $D_0^i = C_0^i \cap \imm^{-1}(\closedball{x}{r}\setminus \closedball{x}{s})$.
Let us show that
\begin{equation}
\label{eq:sqrtPoint3}
\mu_0(D_0^i) \leq \fmax \Jmax 2^{d-1} \sqrt{6}  d \volball{d} \cdot r^{d-1}  \sqrt{r^2-s^2}.
\end{equation}
As in Equation \eqref{Paper2:eq:probabilitybounds_point3}, we write this measure as
\begin{align*}
\int_{ (\overline \exp_{z_0^i}^{\MMo})^{-1}( D_0^i ) } g_0^i(y) \dd \HH^d(y)
&=  \int_{v \in \sphere{0}{1}} \int_{t=a(v)}^{b(v)} g_0^i(tv) t^{d-1} \dd t \dd v,
\end{align*}
where $a(v)$ and $b(v)$ are defined as in the proof of Proposition \ref{Paper2:prop:probabilitybounds} Point \ref{prop:probabilitybounds:point3}: for every $v \in \subset T_{z^i} \MM$ of unit norm, let $\gamma_0$ be an arc-length parametrized geodesic with $\gamma_0(0) = z_0^i$ and $\dot \gamma_0(0) = v$, and set $a(v)$ and $b(v)$ to be the first positive values such that $\eucN{\gamma(a(v)) - x} \geq s$ and $\eucN{\gamma(b(v)) - x} = r$.
For any $v \in \sphereG{0}{1}{T}$, Lemma \ref{Paper2:lem:distancetocenter} Point \ref{lem:distancetocenter:point2} gives $b(v) \leq 2r$, and Lemma \ref{Paper2:lem:distancetocenter} Point \ref{lem:distancetocenter:point5} gives $b(v) - a(v) \leq \sqrt{6}\sqrt{r^2-s^2}$. We deduce that
\begin{align*}
\int_{t=a(v)}^{b(v)} \fmax \Jmax (2r)^{d-1} \dd t
&\leq\sqrt{6} \sqrt{r^2-s^2} \fmax \Jmax (2r)^{d-1}.
\end{align*}
Therefore,
\begin{align*}
\int_{v \in \sphere{0}{1}} \int_{t=a(v)}^{b(v)} g_0^i(tv) t^{d-1} \dd t \dd v
&\leq \sqrt{6}\sqrt{r^2-s^2} \fmax \Jmax (2r)^{d-1} d \volball{d},
\end{align*}
which yields Equation \eqref{eq:sqrtPoint3}.

We now gather the connected components $D_0^i$.
Since $\imm^{-1}(\closedball{x}{r} \setminus \closedball{x}{s}) = \bigcup_{i\in I} D_0^i$, we have
\begin{align*}
\mu(\closedball{x}{r} \setminus \closedball{x}{s}) 
= \sum_{i \in I} \mu_0(D_i).
\end{align*}
Using Equation \eqref{eq:sqrtPoint3} we get
\begin{align*}
\mu(\closedball{x}{r} \setminus \closedball{x}{s}) 
\leq |I| \fmax \Jmax 2^{d-1} \sqrt{6}  d \volball{d} \cdot r^{d-1}  \sqrt{r^2-s^2},
\end{align*}
where $|I|$ is the cardinal of $I$.
Let us show that $|I| \leq \frac{1}{\fmin \Jmin  \volball{d}}( \frac{2\rho}{ \alpha })^d$, with $\alpha = \sqrt{4 - \sqrt{13}}$, which will conclude the proof.

Let $i,j \in I$ such that $i \neq j$. We first show that $\geoD{z^i_0}{z^j_0}{\MMo} \geq \frac{\alpha}{\rho}$.
Let $\gamma_0\colon [0, T] \rightarrow \MMo$ be a geodesic from $z^i_0$ to $z^j_0$.
Consider the map $\phi\colon t \mapsto \eucN{\gamma(t) - x}$.
Since $C_0^i$ and $C_0^j$ are disjoint connected components, there must be a $t^* < T$ such that $\eucN{\gamma(t^*) - x_0} > r$.
Moreover, according to Lemma \ref{Paper2:lem:distancetocenter} Point \ref{lem:distancetocenter:point1}, $\phi$ is increasing on $[0, T_1]$ where $T_1 = \frac{\sqrt{2}}{\rho} \sqrt{2-\sqrt{3+\rho^2 l^2}}$. 
Since $\phi(T) \leq r$, we deduce that $T$ is greater than $T_1$. 
Note that the assumption $r \leq \frac{1}{2\rho}$ yields $T_2 \geq \frac{\alpha}{\rho}$. Hence we obtain the bound $$\geoD{z^i_0}{z^j_0}{\MMo} = T \geq T_1 \geq \frac{\alpha}{\rho}.$$
This implies that the geodesic balls $\openballM{z_0^i}{ \frac{\alpha}{2 \rho} }{\MM_0}$, $i \in I$, are disjoint. Therefore, 
\begin{align*}
1 
&\geq \mu_0\bigg(\bigcup_i \openballM{z_0^i}{ \frac{\alpha}{2\rho} }{\MM_0} \bigg)
\geq |I|\fmin \Jmin \volball{d} \bigg( \frac{\alpha}{2 \rho} \bigg)^d,
\end{align*}
and we deduce that $|I| \leq \frac{1}{\fmin \Jmin \volball{d}}\left( \frac{2\rho}{\alpha}\right)^d$.
\end{proof}

\subsection{Sublevel sets of the normal reach}
\label{Paper2:subsec:quantif_normal_reach}

In this subsection, we assume the Hypotheses \hyperref[hyp:2]{2} and \hyperref[hyp:3]{3}, as well as  Hypothesis \hyperref[hyp:1prime]{1'}, stated in the next paragraph. This last hypothesis can be seen as a strengthening of Hypothesis \hyperref[hyp:1]{1}.
Our goal is to give an upper bound on $\mu_0(\lambda_0^t)$, the measure of the set of points $x_0 \in \MMo$ with normal reach not greater than $t$ (see Proposition \ref{prop:quantif_normal_reach}).
This proves a result announced in Subsect. \ref{Paper2:subsec:model}: Hypothesis \hyperref[hyp:4]{4} is a consequence of Hypotheses \hyperref[hyp:1prime]{1'}, \hyperref[hyp:2]{2} and \hyperref[hyp:3]{3}. We close this subsection with a remark concerning generalizations of this result. 
Since Hypothesis \hyperref[hyp:4]{4} trivially holds when the immersion is an embedding (with $r_{\ref{Paper2:hyp:normalreach:index}} = \min \lambda_0$ and $c_{\ref{Paper2:hyp:normalreach:index}} = 0$), we shall also suppose that $\imm$ is not an embedding.

First, we say that a finite collection $A$ of linear subspaces of $E$ is in \emph{general position} if  
$$
\mathrm{codim}\big(\bigcap_{V \in A} V\big) = \sum_{V \in A} \mathrm{codim}(V),
$$
where we define $\mathrm{codim}(V) = \dim(E) - \dim(V)$.
Now, we say that the immersion $\imm\colon \MMo \rightarrow \MM$ is \emph{self-transverse} (also called \emph{completely regular} in \cite{herbert1981multiple}) 
if for any point $x \in \MM$, the collection of tangent spaces $\{T_y\MM \mid y_0 \in \MMo, ~ x = y \}$ is in \emph{general position}.
Suppose that $\imm$ is self-transverse, and denote by $\NNo$ be the self-intersections of $\MMo$:
$$
\NNo = \{x_0 \in \MMo \mid \exists y_0 \in \MMo,~ x_0 \neq y_0,~ x = y\}.
$$
Its image is denoted $\NN = \imm(\NNo)$.
Equivalently, $\NNo$ is the set of points with zero normal reach, that is, $\NNo = \lambda_0^{-1}(\{0\})$. We also have $\NN=\lambda^{-1}(\{0\})$.
In general, $\NNo$ and $\NN$ are not submanifolds, but only (closed) immersed manifolds.
The subset $\NNo$ can be decomposed as a disjoint union 
$$
\NNo = \bigsqcup_{i \geq 2} \NNo^{(i)}
~~~~~~~~\text{where}~~~~~~~~
\NNo^{(i)} = \{x_0 \in \MMo \mid \lvert\imm^{-1}(\{x\})\rvert = i \},
$$
and where $\lvert\cdot\rvert$ denotes the cardinal. In other words, $\NNo^{(i)}$ is the set of points of $\MMo$ whose image is shared by exactly $i$ distinct points of $\MMo$. 
Each $\NNo^{(i)}$ is a submanifold of $\MMo$, not necessarily closed, of dimension $i \dim(\MMo) - (i-1) \dim(E)$ \cite[Lemma 2.3]{herbert1981multiple}.
Moreover, the tangent spaces of $\NN^{(i)} = \imm(\NNo^{(i)})$ can be described as:
\begin{equation}
\label{eq:tangent_N}
T_x \NN^{(i)} = \bigcap_{y_0 \in \imm^{-1}(x)}T_y \MM.
\end{equation}
In order to state the proofs of this subsection, we shall make the following assumption: the immersion $\imm$ only has \emph{double points}, that is, $\NNo$ is equal to $\NNo^{(2)}$. We shall refer to this assumption as
\encadrer{
\textbf{Hypothesis \hyperref[hyp:1prime]{1'}. }\label{hyp:1prime}
The immersion $\imm$ is self-transverse, and only has double points.
}
\noindent
In this case, $\NNo$ is a submanifold of $\MMo$, of dimension $2 \dim(\MMo) - \dim(E)$. 
Most of the examples we will consider later in the paper satisfy this assumption. They are curves in the plane (Examples \ref{Paper2:ex:lemniscate_overview}, \ref{ex:lemniscate_DTM} and \ref{ex:olympics}) or surfaces in the space (Examples \ref{ex:torus}, \ref{ex:Klein}).

We will also need a few quantities related to the immersion.
Let $\mathcal{D}_0$ be the set of critical points of the Euclidean distance on $\MMo$, that is,
\begin{equation}
\label{Paper2:eq:quantif_normal_reach_D}
\mathcal{D}_0 = \left\{ (x_0,y_0)\in \MMo\times\MMo \mid  x_0 \neq y_0, ~ x-y \bot T_y \MM ~\text{ and }~ x-y \bot T_x \MM\right\}.
\end{equation}
Also, let $\mathcal{C}_0$ be the set of double points of $\MMo$:
\begin{equation}
\label{Paper2:eq:quantif_normal_reach_C}
\mathcal{C}_0 = \left\{ (x_0,y_0)\in \MMo\times\MMo \mid x_0 \neq y_0 ~\text{ and }~ x=y \right\}.
\end{equation}
Note that the projection of $\mathcal{C}_0$ on the first coordinate is $\NNo$. Moreover, we have $\CC_0 \subset \mathcal{D}_0$, and these sets are compact. Since $\CC_0$ is an isolated subset of $\mathcal{D}_0$ by Lemma \ref{Paper2:lem:Federergeod} Point \ref{lem:Federergeod:point4}, we have that $\mathcal{D}_0 \setminus \CC_0 $ also is compact.
Consider the quantity 
\begin{equation}
\label{Paper2:eq:quantif_normal_reach_Delta}
\Delta = \inf \left\{\eucN{x-y} \mid  (x_0,y_0) \in \mathcal{D}_0 \setminus \mathcal{C}_0\right\}.
\end{equation}
The constant $\Delta$ can be understood as the minimal length of the nonzero bottlenecks of $\MM$.
From the compactness of $\mathcal{D}_0 \setminus \mathcal{C}_0$ we deduce that $\Delta>0$. 
Moreover, we define
\begin{equation}
\label{Paper2:eq:quantif_normal_reach_Delta0}
\Delta_0 = \inf \left\{\eucN{x-y} \mid  x_0 \in \NNo,~ y_0 \in \MMo, ~ x \neq y, ~x-y \bot T_y \MM\right\}.
\end{equation}
It is a measure of regularity around the self-intersections of $\MM$.
Using Lemma \ref{Paper2:lem:Federergeod} Point \ref{lem:Federergeod:point5}, one proves that this infimum is taken over a compact set, hence that $\Delta_0>0$.
Last, we will need a measure a similarity between linear subspaces.
If $U,V$ denote two linear subspaces of $E$, let their minimal angle be
\begin{align*}
\angle (U,V) = \inf\left \{ \arccos\left(\frac{\lvert\eucP{u}{v}\rvert}{\eucN{u}\eucN{v}}\right) \mid u \in U,~ v \in V, ~ u,v \in (U \cap V)^\bot \right\},
\end{align*}
where $\inf \emptyset = 0$ by convention. Note that $\angle (U,V) > 0$ when $U \neq V$. 
Now, define
\begin{align}
\Theta = \inf \left\{\angle(T_x\MM,T_y\MM) \mid  (x_0,y_0) \in \mathcal{C}_0\right\}.
\label{Paper2:eq:quantif_normal_reach_Theta}
\end{align}
According to the self-transversality hypothesis and the compactness of $\CC_0$, we have $\Theta > 0$. These constants are represented in Figure \ref{fig:quantities}.

\begin{figure}[H]
\centering
\includegraphics[width=.55\linewidth]{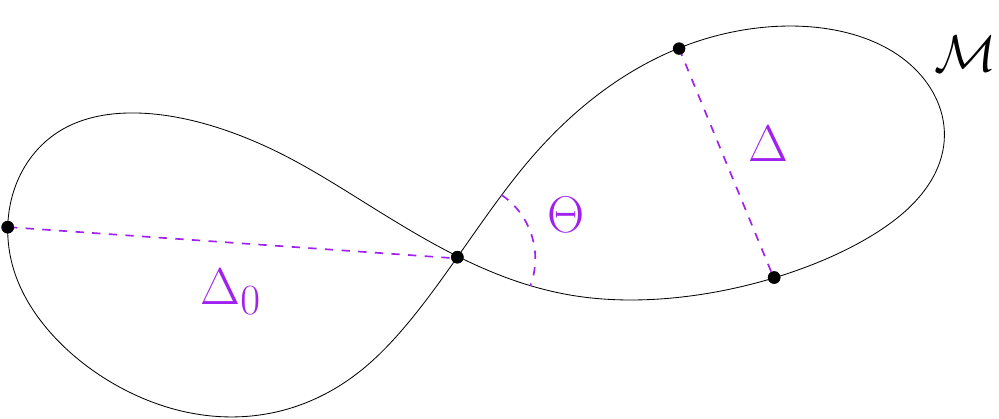}
\caption{The constants $\Delta$, $\Delta_0$ and $\Theta$ associated to $\MM$.}
\label{fig:quantities}
\end{figure}

In order to bound the measure $\mu_0\left(\lambda_0^t\right)$, we will prove that the  sublevel set $\lambda_0^t$ is included in a thickening of $\NNo$.
By bounding the measure of this thickening, we will obtain the main result 
(Proposition \ref{prop:quantif_normal_reach}).
We start with a lemma which describes the situation around self-intersection points of $\MMo$. 

\begin{lemma}
\label{lem:sublevelreach_theta}
Let $(x_0,y_0) \in \mathcal{C}_0$ (defined in Equation \eqref{Paper2:eq:quantif_normal_reach_C}). 
Let $\gamma_0 \colon I\rightarrow \MMo$ (resp. $\gamma_0'$) be an arc-length parametrized geodesic starting from $x_0$ (resp. from $y_0$), and denote $v = \dot \gamma(0)$  (resp. $v' = \dot \gamma'(0)$). 
Let $\theta = \arccos(\lvert\eucP{v}{v'}\rvert)$ be their angle.
Let $\delta, \delta' \geq 0$ such that $\delta' \leq \delta \leq \frac{\sin(\theta)}{2 \rho}$. Then we have
$$\eucN{\gamma(\delta)-\gamma'(\delta')} \geq \frac{\sin(\theta)}{2} \delta.$$
As a consequence, if $v$ is orthogonal to $T_x \MM \cap T_y \MM$, then the distance from $\gamma(\delta)$ to $\imm(\closedballM{y_0}{\delta}{\MMo})$ is lower bounded by $\frac{\sin(\Theta)}{2} \delta$.
\end{lemma}

\begin{proof}
Let us introduce $\overline x = x + \delta v $ and $\overline y = y + \delta'v'$, as represented in Figure \ref{fig:lemma_intersection}.
\begin{figure}[H]
\centering
\includegraphics[width=.85\linewidth]{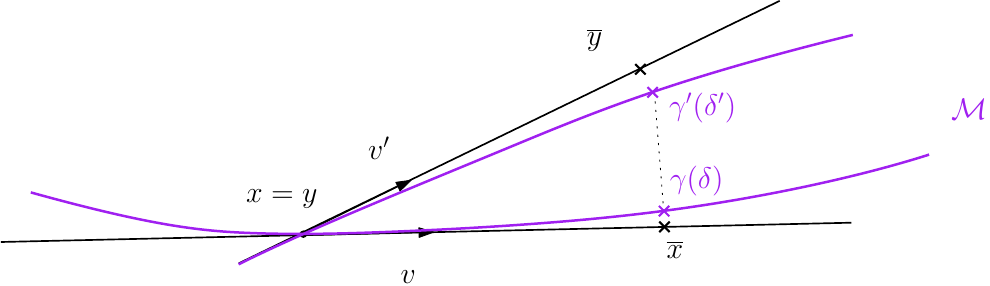}
\caption{Situation in Lemma \ref{lem:sublevelreach_theta}.}
\label{fig:lemma_intersection}
\end{figure}
The triangle inequality yields
\begin{align*}
\eucN{\gamma(\delta)-\gamma'(\delta')} 
&\geq \eucN{\overline x-\overline y} - \eucN{\gamma(\delta)-\overline x} - \eucN{\gamma'(\delta') - \overline y}. 
\end{align*}
According to Lemma \ref{Paper2:lem:Federergeod} Point \ref{lem:Federergeod:point1}, we have $\eucN{\gamma(\delta)-\overline x} \leq \frac{\rho}{2} \delta^2$ and $\eucN{\gamma'(\delta') - \overline y} \leq \frac{\rho}{2} {\delta'}^2 \leq \frac{\rho}{2} {\delta}^2$.
Moreover, $\eucN{\overline x-\overline y}$ is not lower than $\eucN{\overline x - z}$, where $z$ is the projection of $\overline x$ on the line spanned by $v'$.
Elementary trigonometry shows that $\eucN{\overline x - z} = \sin(\theta) \delta$.
Hence the previous equation yields
\begin{align*}
\eucN{\gamma(\delta)-\gamma'(\delta')} 
&\geq \sin(\theta) \delta - \frac{\rho}{2}\delta^2 - \frac{\rho}{2}\delta^2 
= \sin(\theta) \delta \left(1 - \frac{\rho}{\sin(\theta)}\delta\right),
\end{align*}
and we conclude with $\delta \leq \frac{\sin(\theta)}{2 \rho}$.
\end{proof}

The following lemma shows that, around $\NN$, the immersed manifold $\MM$ is a union of two transversally intersecting pieces.

\begin{lemma}
\label{lem:sublevelreach_mx}
For any $r < \mini{\frac{1}{2\rho}, \Delta_0}$ and $x \in \NN$, the set $\imm^{-1}(\closedball{x}{r})$ is made up of two connected components, and we have 
$$
\imm^{-1}(\closedball{x}{r}) \subset \bigcup_{y_0 \in\imm^{-1}( \{x\})} \closedballM{y_0}{2 r}{\MMo}.
$$
\end{lemma}

\begin{proof}
%
Consider $\MMo^x = \imm^{-1}(\closedball{x}{r})$ and $C_0^i$, $i \in I$, its connected components, as represented in Figure \ref{fig:lemma_twopieces}.
Let us denote $C_0^0$ the connected component that contains $x_0$.
\begin{figure}[H]
\centering
\includegraphics[width=.85\linewidth]{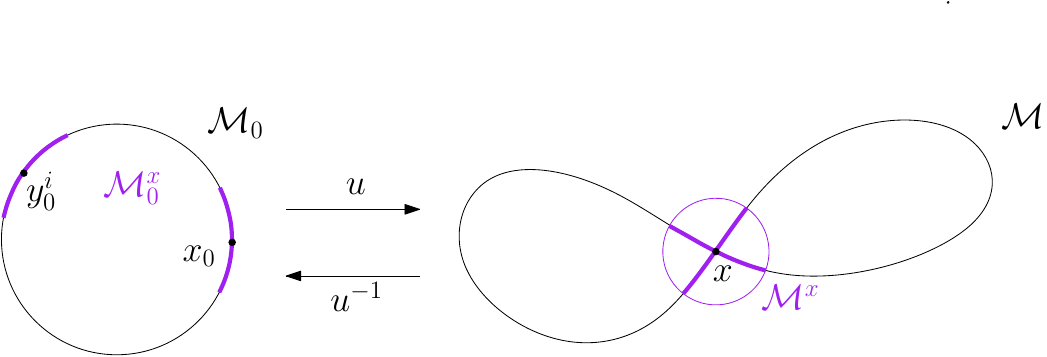}
\caption{Situation in Lemma \ref{lem:sublevelreach_mx}.}
\label{fig:lemma_twopieces}
\end{figure}
For any $i \in I \setminus \{0\}$, let $y_0^i$ be a minimizer of $y_0 \mapsto \eucN{y-x}$ on $C_0^i$. 
It satisfies $x-y^i\bot T_{y^i}\MM$.
Since $r$ has been chosen lower than $\Delta_0$ (defined in Equation \eqref{Paper2:eq:quantif_normal_reach_Delta0}), we must have $y^i=x$, that is to say, $y_0^i \in \imm^{-1}(\{x\})$.
Using that $\imm^{-1}(\{x\})$ consits of two elements, we deduce that $\MMo^x$ is made up of two connected components.

Now, as we have seen in the proof of Lemma \ref{Paper2:lem:comparisoneucgeod2}, these connected components satisfy $C_0^i \subset \closedballM{y_0^i}{c_{\ref{Paper2:lem:comparisoneucgeod2:index}}(\rho r)r}{\MMo}$.
The result follows from $c_{\ref{Paper2:lem:comparisoneucgeod2:index}}(\rho r)< 2$.
\end{proof}

We can now connect the normal reach to the distance to $\NNo$.
We first prove that, close to a self-intersection point, the normal reach is lower bounded by the geodesic distance to that point.

\begin{lemma}
\label{lem:sublevelset_theta}
Let $x_0 \in \MMo$ and denote by $\delta = \geoD{x_0}{\NNo}{\MMo}$ the geodesic distance from $x_0$ to $\NNo$. Suppose that $\delta < \mini{\frac{1}{4\rho},\frac{\sin(\Theta)}{2}, \frac{\Delta_0}{2}}$. 
Then $\lambda_0(x_0) \geq \frac{\sin(\Theta)}{2}\delta$.
\end{lemma}

\begin{proof}
Let $y_0$ be a projection of $x_0$ on $\NNo$, that is, a point that minimizes the geodesic distance $\geoD{x_0}{y_0}{\MMo}$ with $y_0 \in \NNo$.
Let $\gamma_0\colon I \rightarrow \MMo$ be a geodesic from $y_0$ to $x_0$, and denote $v_0 = \dot \gamma_0(0)$. Note that $v \bot T_{y}\NN$, otherwise $y_0$ would not be a minimizer.
Denote $r = \eucN{x-y}$ and $\MMo^y = \imm^{-1}(\closedball{y}{2r})$.
Let also $y_0'$ be the other point of $\MMo$ such that $y' = y$.

First, let us show that $\lambda_0(x_0) \leq r$.
According to Lemma \ref{lem:sublevelreach_mx}, $\MMo^y$ consists of two connected components, $C_0$ that contains $y_0$, and $C_0'$ that contains $y_0'$.
Now, consider a minimizer of $z_0 \mapsto \eucN{z-x}$ on $C_0'$.
This point satisfies $x-z \bot T_z \MM$ and $x_0 \neq z_0$. Thus, 
$$\lambda_0(x_0) \leq \eucN{x-z} \leq \eucN{x-y}\leq r,$$ 
as announced.
Moreover, using the inequality $\closedball{x}{r} \subset \closedball{y}{2r}$, we deduce that $z$ realizes the normal reach of $x_0$, that is, $\lambda_0(x_0) = \eucN{x-z}$.

To conclude, consider the geodesic $\gamma_0$ defined above.
We have seen that $v$ is orthogonal to the tangent space $T_{y}\NN$. Since $T_{y}\NN = T_x \MM\cap T_y \MM$ by Equation \eqref{eq:tangent_N}, we have $v \bot (T_x \MM\cap T_y \MM)$, hence we can apply the consequence of Lemma \ref{lem:sublevelreach_theta} to get 
$$\lambda_0(x_0) = \eucN{x-z} = \eucN{\gamma(\delta)-z} \geq \frac{\sin(\Theta)}{2} \delta,$$
as wanted. 
\end{proof}

%
%

The following lemma is a converse of Lemma \ref{lem:sublevelset_theta}: points with small normal reach are close to the self-intersection submanifold $\NNo$.

\begin{lemma}
\label{lem:sublevelset_lambda}
Let $x_0 \in \MMo$ such that $\lambda_0(x_0) \leq \mini{\frac{\sin(\Theta)}{8\rho}, \frac{\sin(\Theta)^2}{4}, \frac{\Delta_0\sin(\Theta)}{4},\Delta}$.
Then $\lambda_0(x_0) \geq  \frac{\sin(\Theta)}{2} \geoD{x_0}{\NNo}{\MMo}$.
\end{lemma}

\begin{proof}
Put $r = \lambda_0(x_0)$, and denote the sublevel set $\lambda_0^{r} = \lambda_0^{-1}([0,t])$. 
Let $C_0$ denote the connected component of $x_0$ in $\lambda_0^r$.
We have seen in Remark \ref{rem:semi-continuous} that the normal reach $\lambda_0$ is lower semi-continuous. Hence $C_0$ is closed, and $\lambda_0$ attains a minimum on it.
Let $y_0$ be a minimizer of $\lambda_0$ on $C_0$.
Let us prove that $\lambda_0(y_0) = 0$ by contradiction.
Suppose that $\lambda_0(y_0) > 0$, and let $z_0 \in \MMo$ be such that $z_0 \neq y_0$, $y-z \bot T_z \MM$ and $\lambda_0(y_0) = \eucN{y-z}$.
These points are represented in the following figure.
\begin{figure}[H]
\centering
\includegraphics[width=.45\linewidth]{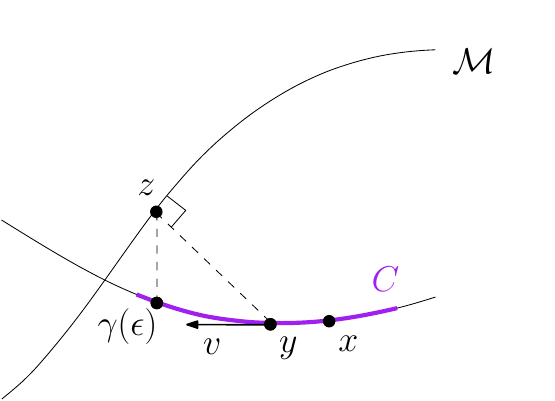}
\caption{Situation in Lemma \ref{lem:sublevelset_lambda}, supposing by contradiction that $\lambda_0(y_0)>0$.}
\end{figure}

Since $\eucN{y-z} = \lambda_0(y_0) \leq \lambda_0(x_0) < \Delta$, where $\Delta$ has been defined in Equation \eqref{Paper2:eq:quantif_normal_reach_Delta}, the vector $y-z$ is not orthogonal to $T_y \MM$. Let $v_0 \in T_{y_0} \MMo$ such that $\eucP{v}{z-y}>0$, and consider an arc-length parametrized geodesic $\gamma\colon I \rightarrow \MMo$ with $\gamma_0(0) = y_0$ and $\dot\gamma_0(0) = v_0$. 
We will show that $\lambda_0(\gamma_0(\epsilon))<\lambda_0(y_0)$ for $\epsilon > 0$ small enough, contradicting the minimality of $y_0$.

On the one hand, for $\epsilon > 0$ small enough, the bound $\eucP{v}{z-y}>0$ yields 
\begin{equation}
\eucN{z - \gamma(\epsilon)} < \eucN{z - y} = \lambda_0(y_0)
\label{eq:lemmasmallreach1}
\end{equation}

On the other hand, Lemma \ref{Paper2:lem:Federergeod} Point \ref{lem:Federergeod:point5} and $y-z \bot T_z \MM$ gives $\geoD{z_0}{y_0}{\MMo} \geq \frac{1}{\rho}$.
Together with the assumption $\lambda_0(x_0) < \frac{\sin(\Theta)}{8\rho} < \frac{1}{4\rho}$, we deduce that $\geoD{z_0}{y_0}{\MMo} > 4 \lambda_0(x_0)$.
By continuity, for $\epsilon>0$ small enough, we also have
\begin{equation}
\geoD{z_0}{\gamma_0(\epsilon)}{\MMo} > 4 \lambda_0(x_0).
\label{eq:lemmasmallreach2}
\end{equation}

Now, we deduce from $\lambda_0(x_0) \geq \lambda_0(y_0)$ and Equation \eqref{eq:lemmasmallreach1} and \eqref{eq:lemmasmallreach2} that $\geoD{z_0}{\gamma_0(\epsilon)}{\MMo} > 4 \eucN{z-\gamma(\epsilon)}$.
Therefore we can apply Lemma \ref{lem:exists_projection} on $z_0$ and $\gamma_0(\epsilon)$ to get 
$$\lambda_0(\gamma_0(\epsilon)) \leq \eucN{z-\gamma(\epsilon)} < \lambda_0(y_0),$$ 
which contradicts the minimality of $y_0$. We conclude that $\lambda_0(y_0) = 0$.

Next, let $\gamma_0\colon [0,T] \rightarrow \MMo$ be a path from $y_0$ to $x_0$ in $C_0$. 
Let us show that, for all $t \in [0,T]$,
\begin{equation}
\label{eq:sublevelset_path}
\lambda_0(\gamma_0(t)) 
\geq \frac{\sin(\Theta)}{2} \geoD{\gamma_0(t)}{\NNo}{\MMo}.
\end{equation}
According to Lemma \ref{lem:sublevelset_theta}, it is enough to show that $\geoD{\gamma_0(t)}{\NNo}{\MMo} < c$, where $c=\mini{\frac{1}{4\rho},\frac{\sin(\Theta)}{2}, \frac{\Delta_0}{2}}$.
By contradiction, suppose that $\geoD{\gamma_0(t)}{\NNo}{\MMo} \geq c$ for some $t$. Since $\geoD{\gamma_0(0)}{\NNo}{\MMo} = \geoD{y_0}{\NNo}{\MMo} = 0$, we can consider the first value $t \in [0,T]$ such that $\geoD{\gamma_0(t)}{\NNo}{\MMo} = c$. 
Lemma \ref{lem:sublevelset_theta} then gives $\lambda_0(\gamma_0(t)) 
\geq \frac{\sin(\Theta)}{2} c$. Besides, by definition of $C_0$, we have $\lambda_0(\gamma_0(t)) \leq \lambda_0(x_0)$. We deduce that
$$
\lambda_0(x_0) \geq \frac{\sin(\Theta)}{2} \cdot \mini{\frac{1}{4\rho},\frac{\sin(\Theta)}{2}, \frac{\Delta_0}{2}},
$$
which contradicts the assumptions of the lemma. 
We now obtain the result from Equation \eqref{eq:sublevelset_path} at $t = T$.
\end{proof}

We now prove the main result of this subsection.

\begin{proposition}
\label{prop:quantif_normal_reach}
Suppose that the immersion $\imm$ satisfies Hypotheses \hyperref[hyp:1prime]{1'}, \hyperref[hyp:2]{2} and \hyperref[hyp:3]{3}. Let $\alpha = \dim(E) - \dim(\MMo)$.
For every $r < r_{\ref{prop:quantif_normal_reach:index}}$, we have 
\begin{align*}
\mu_0(\lambda_0^r) \leq  c_{\ref{prop:quantif_normal_reach:index}} r^\alpha + \grando{r^{\alpha+1}},
\end{align*}
where $r_{\ref{prop:quantif_normal_reach:index}} = \mini{\frac{\sin(\Theta)}{8\rho}, \frac{\sin(\Theta)^2}{4}, \frac{\Delta_0\sin(\Theta)}{4},\Delta}$ and $c_{\ref{prop:quantif_normal_reach:index}} = \left(\frac{2}{\sin(\theta)}\right)^\alpha V_{\alpha}\fmax\HH^{d'}_{\MMo}(\NNo)$.
\end{proposition}

\begin{proof}
Let $d'$ be the dimension of $\NNo$, and $\alpha$ its codimension in $\MMo$. Since $\dim(\NNo) = 2 \dim(\MMo) - \dim(E)$ by Hypothesis \hyperref[hyp:1prime]{1'}, we have $\alpha = \dim(E) - \dim(\MMo)$.
Besides, according to Lemma \ref{lem:sublevelset_lambda}, the sublevel set $\lambda_0^r$ is included in the geodesic thickening $\NNo^{t}$ of $\NNo$, where $t =  \frac{2}{\sin(\Theta)} r$. 
According to Weyl's Tube Formula \cite[Theorem 9.23]{gray2012tubes}, the volume of $\NNo^{t}$ is 
\begin{align*}
\HH^d_{\MMo}(\NNo^{t}) &= V_{\alpha}\cdot \HH^{d'}_{\MMo}(\NNo) \cdot t^{\alpha} + \grando{t^{\alpha+1}}.
\end{align*}
where $\HH^{d'}_{\MMo}(\NNo)$ is the $d'$-dimensional volume of $\NNo$, and $V_\alpha$ the volume of the unit ball in $\R^\alpha$.
Using the density of $\muo$ given by Hypothesis \hyperref[hyp:2]{2}, we can write $\muo(\NNo^{t})\leq \fmax\HH^d_{\MMo}(\NNo^{t})$.
Hence 
$$
\mu_0(\NNo^t) \leq \fmax\HH^d_{\MMo}(\NNo^{t}) = \fmax\HH^d_{\MMo}\left(\NNo^{\frac{2}{\sin(\Theta)} r}\right),
$$
and the result follows.
\end{proof}

\begin{remark}
\label{rem:quantif_normalreach_generalize}
It seems reasonable to think that Proposition \ref{prop:quantif_normal_reach} is still valid when replacing Hypothesis \hyperref[hyp:1prime]{1'} with the weaker Hypothesis \hyperref[hyp:1]{1}, at least with $\alpha = 1$.
In this case, the quantities $\Delta$, $\Delta_0$ and $\Theta$ are still well-defined, and Lemmas \ref{lem:sublevelreach_mx} and \ref{lem:sublevelreach_theta} hold.
However, Lemmas \ref{lem:sublevelset_theta} and \ref{lem:sublevelset_lambda} may not be true anymore.
An illustration of this is given by the two following intersecting surfaces of $\R^4$: 
\begin{align*}
\MM^{(1)} = \{(a,b,0,0)\mid a,b \in \R\} ~~~~~~\text{and}~~~~~~ \MM^{(2)} = \{(a,0,c,a^2)\mid a,c \in \R\}.
\end{align*}
These submanifolds intersect at $\MM^{(1)} \cap \MM^{(2)} = \{0\}$, but their intersection is not transverse, since their tangent spaces does not span the fourth canonical basis vector of $\R^4$.
The distance from a point $x = (a,0,0,a^2)$ of $\MM^{(2)}$ to $\MM^{(1)}$ is $a^2$. Besides, for small values of $a$, the geodesic distance from $x$ to $0$ on $\MM^{(2)}$ is approximately $a$. But there is no constant $c$ such that $a^2 \geq c\cdot a$.
Hence Lemma \ref{lem:sublevelset_theta} does not hold anymore.
Nonetheless, the subset of point of $\MM^{(2)}$ at distance at most $r$ from $\MM^{(1)}$ is 
$$
\left\{(a,0,c,a^2) \in \MM^{(2)} \mid a^4 + c^2 \leq r^2 \right\}.
$$
This is approximately the rectangle $\left\{(a,0,c,a^2) \in \MM^{(2)} \mid \lvert a \rvert \leq \sqrt{r}, ~ \lvert c \rvert \leq r\right\}$, whose volume is $r^{\frac{3}{2}}$. We see here that Proposition \ref{prop:quantif_normal_reach} holds with $\alpha = \frac{3}{2}$.
\end{remark}

\begin{remark}
\label{rem:quantif_normalreach_dim1}
We can also see that Proposition \ref{prop:quantif_normal_reach} is true under the following conditions: Hypotheses \hyperref[hyp:1]{1}, \hyperref[hyp:2]{2}, \hyperref[hyp:3]{3} and $\dim(\MMo) = 1$. Indeed, in this case, the self-intersecting manifold $\NNo$ is a finite subset of $\MMo$. Knowing this fact, the proofs of Lemmas \ref{lem:sublevelreach_theta}, \ref{lem:sublevelreach_mx}, \ref{lem:sublevelset_theta}, \ref{lem:sublevelset_lambda} and Proposition \ref{prop:quantif_normal_reach} can be used without modification.
In this case, the result reads $\mu_0(\lambda_0^r) \leq  c_{\ref{prop:quantif_normal_reach:index}} r + \grando{r^{2}}$.
\end{remark}

\section{Tangent space estimation}
\label{Paper2:sec:tangentspaceestimation}

We now come back to our original setting: $\MMo$ is a manifold of dimension $d \geq 1$, immersed in $E = \R^n$ via $\imm\colon \MMo \rightarrow \MM$. Moreover, $\MMo$ is endowed with a measure $\muo$. The push-forward measure is denoted $\mu = \imm_* \muo$.
In this section, we show that one can estimate the tangent spaces of $\MM$, based on the measure $\mu$, or a close measure $\nu$, via the computation of local covariance matrices.

\subsection{Local covariance matrices and lifted measures}
\label{Paper2:subsec:defgammaN}

We remind the reader that the aim of this work is to estimate the homotopy type of $\MMo$, or its homology groups, from the measure $\nu$.
As explained in the introduction and in Subsect. \ref{Paper2:subsec:model}, our strategy consists in estimating the \emph{lifted manifold} 
$$\MMcheck = \left\{\left(x, ~\frac{1}{d+2}p_{T_x\MM} \right) \mid x_0 \in \MMo  \right\},$$
where $p_{T_{x} \MM}$ is the matrix of the orthogonal projection on the tangent space $T_{x} \MM$, seen as an element of $\matrixspace{E}$, the space of $n \times n$ matrices. 
The normalization term $\frac{1}{d+2}$ has been chosen in accordance with Proposition \ref{Paper2:prop:consistency}, stated in the next subsection, and makes our method independant of the dimension $d$.
Note that the set $\MMcheck$ can also be described as the image of $\MMo$ under the map
\begin{align*}
\label{Paper2:eq:muchecko_disintegration}
\immcheck\colon x_0 \longmapsto \left( x, ~\frac{1}{d+2} p_{T_{x} \MM} \right).
\end{align*}
Using Hypothesis \hyperref[hyp:1]{1}, we deduce that $\MMcheck$ is diffeomorphic to $\MMo$, hence that their homotopy types and homology groups coincide.
Adopting a measure theoretical point of view on the problem, we will actually estimate the \emph{exact lifted measure}, defined as the push-forward $\muchecko = \immcheck_* \muo$. We explain in Subsect. \ref{Paper2:subsec:overview} how one can infer the homotopy type and homology groups of $\MMo$ from $\muchecko$.
Note that this measure can also be defined as 
\begin{equation}
\label{Paper2:eq:muchecko_disintegration}
\muchecko = (\imm_* \mu_0)(x_0) \otimes \left\{\delta_{ \frac{1}{d+2} p_{T_{x} \MM} }\right\}
\end{equation}
by disintegration of measure.
Here is another alternative definition of $\muchecko$: 
for any smooth $\phi\colon E \times \matrixspace{E} \rightarrow \R$ with compact support,
\begin{equation}
\label{Paper2:eq:muchecko}
\int \phi(x, A) \dd \muchecko(x, A) = \int \phi\left(\imm(x_0),\frac{1}{d+2} p_{T_{x} \MM} \right) \dd \muo(x_0). 
\end{equation}
In order to approximate $\muchecko$, we have to propose an estimator of the tangent spaces. We consider the following construction. If $x$ in any vector of $E=\R^n$, seen as a row vector, the \emph{tensor product} is defined as the $n\times n$ matrix $\tensoro{x} = x^t \cdot x$.

\begin{definition}
\label{def:loccov}
Let $\nu$ be any probability measure on $E$.
Let $r>0$ and $x \in \supp{\nu}$. 
The \emph{local covariance matrix of $\nu$ around $x$ at scale $r$} is the following matrix:
\[ \loccov{\nu}{x} = \int_{\ball{x}{r}} \tensor{x-y} \frac{\dd \nu(y)}{\nu(\ball{x}{r})}.\]
We also define the \emph{normalized local covariance matrix} as $\loccovnorm{\nu}{x} = \frac{1}{r^2} \loccov{\nu}{x}$.
\end{definition}
Note that $\loccov{\nu}{x}$ and $\loccovnorm{\nu}{x}$ depend on $r$, which is not made explicit in the notation.
The normalization factor $\frac{1}{r^2}$ of the normalized local covariance matrix is justified by Proposition \ref{Paper2:prop:consistency}.
Moreover, we introduce the following notations: for every $r>0$ and $x \in \supp{\nu}$,
\begin{itemize}
\itemsep.15cm
\item $\loc{\nu}{x}$ is the restriction of $\nu$ to the ball $\ball{x}{r}$,
\item $\locprob{\nu}{x} = \frac{1}{\nu(\ball{x}{r})} \loc{\nu}{x}$ is the corresponding probability measure. 
\end{itemize} 
Thus the local covariance matrix can be written as $\loccov{\nu}{x} = \int \tensor{x-y} \dd \locprob{\nu}{x}(y)$.

We note that such notions have already been studied in the context of Topological Data Analysis. The collection of probability measures $\{ \locprob{\nu}{x} \}_{x \in \supp{\nu}}$ is called in \cite[Sect. 3.3]{pmlr-v97-memoli19a} the local truncation of $\nu$ at scale $r$. 
The map $x \mapsto \loccov{\nu}{x}$ is called in \cite[Sect. 2.2]{martinez2018shape} the multiscale covariance tensor field of $\nu$ associated to the truncation kernel.

We now propose an estimator of the lifted measure $\muchecko$, inspired by Equations \eqref{Paper2:eq:muchecko_disintegration} and \eqref{Paper2:eq:muchecko}.

\begin{definition}
\label{Paper2:def:lifted_measure}
For any measure $\nu$ on $E$, we denote by $\check{\nu}$ the measure on $E \times \matrixspace{E}$ defined by 
\begin{align*}
\check \nu = \nu(x) \otimes \left\{\delta_{\loccovnorm{\nu}{x}}\right\}.
\end{align*}
It is called the \emph{lifted measure} associated to $\nu$.
In other words, for every smooth $\phi\colon E \times \matrixspace{E} \rightarrow \R$ with compact support, we have
\begin{equation*}
\int \phi(x, A) \dd \check{\nu}(x, A) = \int \phi\bigg(x,\loccovnorm{\nu}{x}\bigg) \dd \nu(x).
\end{equation*}
In accordance with the local covariance matrices, the lifted measure $\nucheck$ depends on the parameter $r$ which is not made explicit in the notation.
\end{definition}

In order to compare these measures, we consider a Wasserstein-type distance on the space $E \times \matrixspace{E}$.
Fix $\gamma > 0$, and let $\gammaN{\cdot}$ be the Euclidean norm on $E \times \matrixspace{E}$ defined as
\begin{equation}
\label{Paper2:eq:gammaN}
\gammaN{(x,A)}^2 = \eucN{x}^2 + \gamma^2 \frobN{A}^2,
\end{equation}
where $\eucN{\cdot}$ represents the usual Euclidean norm on $E$ and $\frobN{\cdot}$ represents the Frobenius norm on $\matrixspace{E}$.
Let $p \geq 1$.
We denote by $\gammaWassersteinp(\cdot, \cdot)$ the $p$-Wasserstein distance with respect to this metric. 
By definition, if $\alpha, \beta$ are probability measures on $E\times \matrixspace{E}$, then $\gammaWassersteinp(\alpha, \beta)$ can be written as
\begin{equation}
\label{Paper2:eq:gammaWassersteinp}
\gammaWassersteinp(\alpha, \beta)
= \inf_{\pi} \bigg( \int_{(E \times \matrixspace{E})^2} 
\gammaN{(x,A)-(y,B)}^p \dd \pi\big((x,A),(y,B)\big) \bigg)^\frac{1}{p},
\end{equation}
where the infimum is taken over all measures $\pi$ on $(E \times \matrixspace{E})^2$ with marginals $\alpha$ and $\beta$.

The parameter $\gamma$ of the norm $\gammaN{\cdot}$ has been designed to balance the importance given to the Euclidean information ($E$-coordinate) and matrix information ($\matrixspace{E}$-coordinate) in $E\times\matrixspace{E}$. The more $\gamma$ is large, the more the matrix information will be relatively important. Since there is no canonical choice of $\gamma$, it will remain as a free parameter in the rest of the paper. In the experiments of the next section, we will choose the value $\gamma = 1$ or $2$, for it seemed relevant in practice.
For a discussion about how this parameter may influence the persistent homology of the lifted manifold $\MMcheck$, we refer the reader to \cite[Subsect. 4.4]{tinarrage2021computing}.

We subdivise the rest of this section in four subsections.
They respectively consists in showing that
\begin{itemize}
\itemsep.2cm
\item \textbf{Consistency:} if $\mu_0$ is a measure satisfying the Hypotheses \hyperref[hyp:2]{2} and \hyperref[hyp:3]{3}, then $\gammaWassersteinp(\muchecko, \mucheck)$ is small (Proposition \ref{Paper2:prop:consistencycheck}),
\item \textbf{Stability of the localized measures:} in addition, if $\nu$ is a measure on $E$ such that $\Wassp{\mu}{\nu}$ is small, then so is $\Wasssymbol{1}(\locprob{\mu}{y}, \locprob{\nu}{y})$ (Lemmas \ref{Paper2:lem:Wstabxx} and \ref{Paper2:lem:Wstabxxsqrt}),
\item \textbf{Stability of the lifted measures:} consequently, $\gammaWassersteinp(\mucheck, \check \nu)$ is also small (Proposition \ref{Paper2:prop:stability})
\item \textbf{Approximation:} under the previous hypotheses, $\gammaWassersteinp(\muchecko, \check \nu)$ is small (Theorem \ref{Paper2:th:estimation}).
\end{itemize}
These measures fit in a commutative diagram:
\begin{center}
\begin{minipage}{.49\linewidth}
\[
\begin{tikzcd}
\MM_0 \arrow[dr, "\imm"] \arrow[rr, "\check \imm"] & & E \times \matrixspace{E} \arrow[dl, "\mathrm{proj}"]  \\
& E &
\end{tikzcd}
\]
\end{minipage}
\begin{minipage}{.49\linewidth}
\[
\begin{tikzcd}[column sep=tiny]
\muo \arrow[drrrrr, "\imm_*"] \arrow[rrrrrrrrrr, "\check \imm_*"] &&&&&&&&& & \muchecko   & \mucheck & \nucheck\\
&&&&& \mu \arrow[urrrrr,swap, "g_*"] \arrow[urrrrrr, bend right=20,swap, "(f_{\mu})_*"]&&&& & & &\nu \arrow[u,swap, "(f_{\nu})_*",bend right=0]
\end{tikzcd}
\]
\end{minipage}
\end{center}
where the maps $g$, $f_\mu$ and $f_\nu\colon E \rightarrow E \times \matrixspace{E}$ are defined as
\begin{align*}
g\colon x \longmapsto \left(x, \frac{1}{d+2} p_{T_{x} \MM} \right),
~~~~~~~
f_\mu\colon x \longmapsto \bigg(x,\loccovnorm{\mu}{x}\bigg),
~~~~~~~
f_\nu\colon x \longmapsto \bigg(x,\loccovnorm{\nu}{x}\bigg).
\end{align*}
Note that the map $g$ is well-defined only on points $x \in \MM$ that are not self-intersection points, i.e., points $x$ such that $\lambda(x)>0$. Under Hypothesis \hyperref[hyp:4]{4}, $g$ is well-defined $\mu$-almost surely.
The maps $f_\mu$ and $f_\nu$ are defined respectively on $\supp{\mu}$ and $\supp{\nu}$.

\subsection{Consistency of the estimation}
\label{Paper2:subsec:consistency}
In this subsection, we assume that $\MMo$ and $\mu_0$ satisfy Hypotheses \hyperref[hyp:2]{2} and \hyperref[hyp:3]{3}.
We first show that the normalized covariance matrix approximates the tangent spaces of $\MM$, as long as the parameter $r$ is chosen smaller than the normal reach. A similar result appears in \cite[Lemma 13]{arias2017spectral} in the case where $\MM$ is a submanifold and $\mu$ is the uniform distribution on $\MM$.
Based on this result, we deduce that the lifted measure $\mucheck$ is close to the exact lifted measure $\muchecko$. The quality of this approximation depends on the measure of the set of points with small normal reach, i.e., points where the tangent spaces are not well-estimated.

\begin{proposition}
\label{Paper2:prop:consistency}
Let $x_0 \in \MM_0$ and $r < \mini{\normalreach{x}, \frac{1}{2\rho}}$. Denote by $p_{T_x \MM}$ the orthogonal projection matrix on the tangent space $T_x \MM$. We have 
\[ \frobN{ \loccovnorm{\mu}{x} - \frac{1}{d+2} p_{T_x \MM} } \leq c_{\ref{Paper2:lem:loccovconsistency:index}} r, \]
where $c_{\ref{Paper2:lem:loccovconsistency:index}} = 6 \rho + 4\frac{c_{\ref{Paper2:lem:densityg:index}}}{\fmin \Jmin} + \frac{\fmax }{\fmin \Jmin} 2^d d \rho + \frac{c_{\ref{Paper2:prop:probabilitybounds:index}}}{\fmin \Jmin}$.
\end{proposition}

\begin{proof}
According to \cite[Lemma 11]{arias2017spectral}, the matrix $r^2 \frac{1}{d+2} p_{T_x \MM}$ is equal to
\begin{align*}
\Sigma_* = \int_{\closedballM{0}{r}{T_x \MM}} \tensoro{y} \cdot\frac{\dd \HH^d(y)}{\volball{d}r^d}.
\end{align*}
Hence the proposition reduces to $\frobN{ \loccov{\mu}{x} - \Sigma_* } 
\leq c_{\ref{Paper2:lem:loccovconsistency:index}} r^3.$
%
%
%
%
%
%
Let us write $T= T_x \MM$, $\closedballB = \closedball{x}{r}$ and $\closedballB_0 = (\overline  \exp_{x}^{\MM})^{-1}( \closedballB )$, 
where $(\overline{\exp}_{x}^{\MM})^{-1}$ as been defined in Equation \eqref{eq:exp_map_seen_in_M}.
We consider the following intermediate matrices:
\begin{align*}
\Sigma_1 &= \int_{\closedballB} \tensor{\left(\overline \exp_x^\MM\right)^{-1}(x')} \dd \locprob{\mu}{x}(x'), \\
\Sigma_2 &= \int_{\closedballB_0} g(0) \tensoro{y} \cdot\frac{\dd \HH^d (y)}{|\loc{\mu}{x}|}, \\
\Sigma_3 &= \int_{\closedballG{0}{r}{T}} g(0) \tensoro{y} \cdot\frac{\dd \HH^d (y)}{|\loc{\mu}{x}|}.
\end{align*}
The triangle inequality now yields:
\begin{align}
\label{terms:consistency}
\frobN{\loccov{\mu}{x} - \Sigma_*}
&\leq \underbrace{\frobN{\loccov{\mu}{x} - \Sigma_1}}_{\text{A}} 
+ \underbrace{\frobN{\Sigma_1 - \Sigma_2}}_{\text{B}} 
+ \underbrace{\frobN{\Sigma_2 - \Sigma_3}}_{\text{C}} 
+ \underbrace{\frobN{\Sigma_3 - \Sigma_*}}_{\text{D}}.
\end{align}

\medbreak \noindent \emph{Term \hyperref[terms:consistency]{A}.} 
By definition of the local covariance matrix, we have
\begin{align*}
\loccov{\mu}{x} 
=\int_{\ball{x}{r}} \tensor{x-x'} \locprob{\mu}{x}(x').
\end{align*}
We use the upper bound
\begin{align*}
\frobN{ \Sigma_\mu(x) - \Sigma_1 }
&\leq \int_{\closedball{x}{r} } \frobN{ \tensor{x-x'} - \tensor{\left(\overline \exp_x^\MM\right)^{-1}(x')} } \dd \locprob{\mu}{x}(x') \\
&\leq \sup_{x' \in \MM\cap\closedball{x}{r}} \frobN{ \tensor{x-x'} - \tensor{\left(\overline \exp_x^\MM\right)^{-1}(x')} }.
\end{align*}
Let $x' \in \MM\cap\openball{x}{r}$.
According to Lemma \ref{Paper2:lem:comparisoneucgeod2}, we have $\eucN{\left(\overline \exp_x^\MM\right)^{-1}(x')} \leq 2 r$. Moreover, $\eucN{x-x'} \leq r$, and Lemma \ref{Paper2:lem:outer}, stated in the following subsection, gives
\begin{equation}
\label{Paper2:eq:consistency_covariance_proof2} 
\frobN{ \tensor{x-x'} - \tensor{\left(\overline \exp_x^\MM\right)^{-1}(x')} } \leq (r+2r) \eucN{(x'-x)-\left(\overline \exp_x^\MM\right)^{-1}(x')}.
\end{equation}
Now, let us justify that 
\begin{equation}
\label{Paper2:eq:consistency_covariance_proof}
\eucN{(x'-x)-\left(\overline \exp_x^\MM\right)^{-1}(x')} \leq \frac{\rho}{2} \geoD{x_0}{x'_0}{\MM_0}^2.
\end{equation}
If we write $x' = \gamma(\delta)$ with $\gamma$ a geodesic such that $\gamma(0) = x$ and $\delta = \geoD{x_0}{x_0'}{\MMo}$, then $\left(\overline \exp_x^\MM\right)^{-1}(x') = \delta \dot \gamma(0)$, and we get
\begin{align*}
\eucN{(x'-x)-\left(\overline \exp_x^\MM\right)^{-1}(x')}
&= \eucN{\gamma(\delta)-\left(x+ \delta \dot \gamma(0)\right)}
\leq \frac{\rho}{2} \delta^2,
\end{align*}
where we used Lemma \ref{Paper2:lem:Federergeod} Point \ref{lem:Federergeod:point1} for the last inequality. Hence Equation \eqref{Paper2:eq:consistency_covariance_proof} is true.
Combined with Lemma \ref{Paper2:lem:comparisoneucgeod2}, which gives $\geoD{x_0}{x_0'}{\MMo} \leq 2 \eucN{x - x'} \leq 2r$, we obtain
\begin{align*}
\frobN{ \tensor{x-x'} - \tensor{\left(\overline \exp_x^\MM\right)^{-1}(x')} } \leq  \frac{\rho}{2} (2 r)^2 = 2 \rho r^2.
\end{align*}
We now use Equation \eqref{Paper2:eq:consistency_covariance_proof2} to deduce $\frobN{ \Sigma_\mu(x) - \Sigma_1 } \leq (r+2r) 2 \rho r^2 = 6 \rho r^3$.

\medbreak \noindent \emph{Term \hyperref[terms:consistency]{B}.} 
By transfer, we can write $\Sigma_1$ as
\begin{align*}
\Sigma_1 
= \int_{\closedballB} \tensor{\left(\overline \exp_x^\MM\right)^{-1}(x')} \frac{\dd \HH^d (y)}{|\loc{\mu}{x}|}
&= \int_{ \closedballB_0 } g(y) \tensoro{y} \cdot\frac{\dd \HH^d (y)}{|\loc{\mu}{x}|}.
\end{align*}
We deduce the upper bound
\begin{align*}
\frobN{ \Sigma_1 - \Sigma_2 }
&\leq \int_{ \closedballB_0 } \big|g(0) - g(y)\big| \eucN{ \tensoro{y} } \frac{\dd \HH^d(y) }{|\loc{\mu}{x}|}.
\end{align*}
According to Lemma \ref{Paper2:lem:outer}, $\eucN{ \tensoro{y} } = \eucN{y}^2 \leq (2r)^2$, and Lemma \ref{Paper2:lem:densityg} gives $|g(y) - g(0)| \leq c_{\ref{Paper2:lem:densityg:index}} r$. 
Therefore,
\begin{align*}
\frobN{ \Sigma_1 - \Sigma_2 }
&\leq 4r^2 \cdot c_{\ref{Paper2:lem:densityg:index}} r \cdot \frac{ \HH^d\big(\closedballB_0\big)} {|\loc{\mu}{x}|} .
\end{align*}
To conclude, note that $|\loc{\mu}{x}| \geq \fmin \Jmin \HH^d\big( \closedballB_0\big)$ by Proposition \ref{Paper2:prop:probabilitybounds} Point \ref{prop:probabilitybounds:point1}, hence we obtain $
\frobN{ \Sigma_1 - \Sigma_2 }
\leq 4\frac{c_{\ref{Paper2:lem:densityg:index}}}{\fmin \Jmin} r^3$.

\medbreak \noindent \emph{Term \hyperref[terms:consistency]{C}.} 
As for the previous terms, we use the upper bound
\begin{align*}
\frobN{ \Sigma_2 - \Sigma_3 } 
&\leq \int_{ \closedballM{0}{r}{T} \setminus \closedballB_0} \frobN{ g(0) \cdot \tensoro{y} } \frac{\dd \HH^d(y)}{|\loc{\mu}{x}|}.
\end{align*}
On the one hand, $\frobN{ g(0) \cdot \tensoro{y} } \leq g(0) \cdot r^2 \leq \fmax r^2$, and we get
\begin{align*}
\frobN{ \Sigma_2 - \Sigma_3 } 
&\leq \fmax r^2 \frac{\HH^d \left( \closedballM{0}{r}{T} \setminus \closedballB_0 \right) }{{|\loc{\mu}{x}|}}.
\end{align*}
On the other hand, since $\closedballB_0 \subset \closedballG{x}{c_{\ref{Paper2:lem:comparisoneucgeod2:index}}(\rho r)r  }{T}$, we have
\begin{align*}
\HH^d \left(\closedballB_0 \setminus \closedballM{0}{r}{T} \right)
&= (c_{\ref{Paper2:lem:comparisoneucgeod2:index}}(\rho r)r)^d \volball{d} - r^d \volball{d}.
\end{align*}
The inequality $a^d - 1 \leq d (a-1) a^{d-1}$, where $a \geq 1$, gives 
\begin{align*}
\left(c_{\ref{Paper2:lem:comparisoneucgeod2:index}}(\rho r)r\right)^d \volball{d} - r^d \volball{d}
\leq \volball{d}r^d \cdot d(c_{\ref{Paper2:lem:comparisoneucgeod2:index}}(\rho r)-1) 2^{d-1}.
\end{align*}
Combined with the inequalities $c_{\ref{Paper2:lem:comparisoneucgeod2:index}}(\rho r) \leq 1 +2\rho r$ and $|\loc{\mu}{x}| \geq \fmin \Jmin \volball{d} r^d$, we get
\begin{align*}
\frobN{ \Sigma_2 - \Sigma_3 } 
&\leq \frac{\fmax }{\fmin \Jmin} 2^d d \rho r^3. 
\end{align*}

\medbreak \noindent \emph{Term \hyperref[terms:consistency]{D}.} 
Let us write $\Sigma^*$ as
\begin{align*}
\Sigma_* = \int_{\closedballM{0}{r}{T_x \MM}} \tensoro{y} \cdot\frac{|\loc{\mu}{x}|}{\volball{d}r^d} \cdot\frac{\dd \HH^d(y)}{|\loc{\mu}{x}|}.
\end{align*}
Hence we have
\begin{align*}
\frobN{\Sigma_3 - \Sigma_*} 
& \leq \int_{ \closedballG{0}{r}{T} } \left| \frac{|\loc{\mu}{x}|}{\volball{d}r^d} -  f(x) \right| \frobN{ \tensoro{y} }  \frac{\dd \HH^d(y)}{|\loc{\mu}{x}|}.
\end{align*}
According to Proposition \ref{Paper2:prop:probabilitybounds} Point \ref{prop:probabilitybounds:point2}, $\left| \frac{|\loc{\mu}{x}|}{\volball{d}r^d} -  f(x) \right| \leq c_{\ref{Paper2:prop:probabilitybounds:index}} r$. 
Moreover, $\frobN{ \tensoro{y} } \leq r^2$ and $\int_{\closedballG{0}{r}{T}} \frac{\dd \HH^d(y)}{|\loc{\mu}{x}|} \leq \frac{1}{\fmin \Jmin}$.
Therefore, $\frobN{ \Sigma_3 - \Sigma_* } \leq \frac{c_{\ref{Paper2:prop:probabilitybounds:index}}}{\fmin \Jmin} r^3. $
We deduce the result by summing Terms \hyperref[terms:consistency]{A}, \hyperref[terms:consistency]{B}, \hyperref[terms:consistency]{C} and \hyperref[terms:consistency]{D}.
\end{proof}

We now deduce a result concerning the lifted measures $\mucheck$ and $\muchecko$ (defined in Subsect. \ref{Paper2:subsec:defgammaN}).
We remind the reader that the notation $\lambda^r$ refers to the sublevel set $\normalreachmap^{-1}([0,r])$. Hence the quantity $\mu(\lambda^r)$ is the measure of the set of points $x \in \MM$ such that $\lambda(x)\leq t$.

\begin{proposition}
\label{Paper2:prop:consistencycheck}
Let $r < \frac{1}{2 \rho}$. Then
\begin{align*}
\gammaWassersteinp(\check{\mu}, \check{\mu}_0) \leq \gamma\left(2 \mu(\lambda^r)^\frac{1}{p} + c_{\ref{Paper2:lem:loccovconsistency:index}} r\right).
\end{align*}
\end{proposition}

\begin{proof}
Define the map $\phi\colon \MM_0 \rightarrow \left(E \times \matrixspace{E}\right) \times \left(E \times \matrixspace{E}\right)$ as
\begin{align*}
\phi \colon x_0 \mapsto \left(\bigg(x,\loccovnorm{\mu}{x}\bigg), \left(x,\frac{1}{d+2} p_{T_x \MM}\right)\right),
\end{align*}
and consider the measure $\pi = \phi_* \mu_0$. It is a transport plan between $\mucheck$ and $\muchecko$.
By definition of the Wasserstein distance, 
\begin{align*}
\gammaWassersteinp^p(\check{\mu}, \check{\mu}_0) \leq \int \gammaN{\left(x,T\right)-\left(x',T'\right)}^p \dd \pi\left(\left(x,T\right),\left(x',T'\right)\right), 
\end{align*}
hence we can use this transport plan and write
\begin{align*}
\gammaWassersteinp^p(\check{\mu}, \check{\mu}_0) 
&\leq \int \gammaN{ \bigg(x, \frac{1}{r^2}\loccov{\mu}{x}\bigg) - \left(x, \frac{1}{d+2} p_{T_x \MM}\right) }^p \dd \mu(x) \\
&= \gamma^p \int \frobN{\frac{1}{r^2}\loccov{\mu}{x} - \frac{1}{d+2} p_{T_x \MM}}^p \dd \mu(x).
\end{align*}
We split this last integral into the sets $A = \normalreachmap^r$ and $B = E \setminus \normalreachmap^r$. 

On $A$, we use the upper bound $\frobN{\frac{1}{r^2}\loccov{\mu}{x} - \frac{1}{d+2} p_{T_x \MM}} \leq \frobN{\frac{1}{r^2}\loccov{\mu}{x}}+\frobN{\frac{1}{d+2} p_{T_x \MM}} \leq 1+1$ to obtain 
\[\int_A \frobN{\frac{1}{r^2}\loccov{\mu}{x} - \frac{1}{d+2} p_{T_x \MM}}^p \dd \mu(x) \leq 2^p \mu(A).\]
On $B$, we use Proposition \ref{Paper2:prop:consistency} to get 
\begin{align*}
\int_B \frobN{\frac{1}{r^2}\loccov{\mu}{x} - \frac{1}{d+2} p_{T_x \MM}}^p \dd \mu(x) 
&\leq (c_{\ref{Paper2:lem:loccovconsistency:index}}r)^p.
\end{align*}
Combining these two inequalities yields $\gammaWassersteinp^p(\check{\mu}, \check{\mu}_0) \leq \gamma^p( 2^p \mu(A) + (c_{\ref{Paper2:lem:loccovconsistency:index}} r)^p )$. 
Using the inequality $(a+b)^\frac{1}{p} \leq a^\frac{1}{p} + b^\frac{1}{p}$, where $a,b \geq 0$, we deduce
\begin{align*}
\gammaWassersteinp(\check{\mu}, \check{\mu}_0) 
&\leq \gamma \left(2 \mu(A)^\frac{1}{p} + c_{\ref{Paper2:lem:loccovconsistency:index}}r\right),
\end{align*}
which is the result.
\end{proof}

\subsection{Stability of localization of measures}
\label{Paper2:sec:appendix:tangentspaceestimation}
This technical subsection is dedicated to proving stability results for localization of measures. 
Throughout the subsection, we consider two measures $\mu$ and $\nu$ on $E$.
We show that, under some hypotheses on $\mu$, an upper bound on the Wasserstein distance $\Wasssymbol{p}(\mu,\nu)$ gives an upper bound on the the Wassertsein distance $\Wasssymbol{1}(\locprob{\mu}{y}, \locprob{\nu}{y})$ between their localized measures (see Lemmas \ref{Paper2:lem:Wstabxx} and \ref{Paper2:lem:Wstabxxsqrt}). We close this subsection with a comment about the sharpness of our bounds.

The results of this subsection only rely on the following hypotheses about $\mu$:
\encadrer{
\textbf{Hypothesis \hyperref[hyp:5]{5}. }\label{hyp:5}
$\exists c_{\ref{Paper2:hyp:muA:index}}>0, \forall x \in \supp{\mu}$, $\forall t \in [0, \frac{1}{2\rho})$, 
\[\mu(\ball{x}{t}) \geq c_{\ref{Paper2:hyp:muA:index}} t^d.\] 
}
\encadrer{
\textbf{Hypothesis \hyperref[hyp:6]{6}.}\label{hyp:6}
$\exists c_{\ref{Paper2:hyp:muB:index}}>0, \forall x \in \supp{\mu}$, $\exists \lambda(x) \geq 0$, $\forall s,t \in \left[0, \mini{\lambda(x),\frac{1}{2\rho}}\right)$ s.t. $s\leq t$, 
\[\mu(\ball{x}{t}\setminus \ball{x}{s}) \leq c_{\ref{Paper2:hyp:muB:index}} t^{d-1} (t-s).\]
}
\encadrer{
\textbf{Hypothesis \hyperref[hyp:7]{7}.}\label{hyp:7}
$\exists c_{\ref{Paper2:hyp:muBsqrt:index}}>0, \forall x \in \supp{\mu}$, $\forall s,t \in [0, \frac{1}{2\rho})$ s.t. $s\leq t$, 
\[\mu(\ball{x}{t}\setminus \ball{x}{s}) \leq c_{\ref{Paper2:hyp:muBsqrt:index}} t^{d-\frac{1}{2}} (t-s)^\frac{1}{2}.\]
}
\noindent
Note that these Hypotheses \hyperref[hyp:5]{5}, \hyperref[hyp:6]{6} and \hyperref[hyp:7]{7} are consequences of the initial Hypotheses \hyperref[hyp:2]{2} and \hyperref[hyp:3]{3}. Indeed, as stated in Propositions \ref{Paper2:prop:probabilitybounds} and \ref{Paper2:prop:probabilityboundssqrt}, these new hypotheses hold with $\normalreach{x}$ being the normal reach of $\MM$ at $x$, and with the constants $c_{\ref{Paper2:hyp:muA:index}} = \fmin \Jmin \volball{d}$,
~$c_{\ref{Paper2:hyp:muB:index}} = d 2^{d}\fmax \Jmax \volball{d}$ 
~and 
\begin{align*}
c_{\ref{Paper2:hyp:muBsqrt:index}} = \frac{\fmax \Jmax}{\fmin \Jmin }\left( \frac{\rho}{\sqrt{4 - \sqrt{13}}}\right)^d d 2^{2d} \sqrt{3}.
\end{align*}
In order to state the results of this subsection in a more general setting, we will only invoke the Hypotheses \hyperref[hyp:5]{5}, \hyperref[hyp:6]{6} and \hyperref[hyp:7]{7}.
We first state a lemma that will be useful in what follows.

\begin{lemma}
\label{Paper2:lem:outer}
For every $x,y \in E$, we have $\frobN{\outerP{x} - \outerP{y}} \leq (\eucN{x}+\eucN{y})\eucN{x-y}$. 
\end{lemma}
\begin{proof}
We apply the triangle inequality to $x \transp{x} - y \transp{y} = (x-y) \transp{x} + y\transp{(x-y)}$: 
\begin{align*}
\frobN{x \transp{x} - y \transp{y}} 
\leq \frobN{(x-y) \transp{x}} + \frobN{y \transp{(x-y)}} 
&\leq \eucN{x-y}\eucN{x} + \eucN{y}\eucN{x-y} \\
&= (\eucN{x}+\eucN{y})\eucN{x-y},
\end{align*} 
which gives the bound.
\end{proof}

Next, let us compare a measure and its submeasures. 
If $\mu$ is a measure of positive mass (potentially with $|\mu| \neq 1$), we remind the reader that the notation $\overline{\mu}$ refers to the corresponding probability measure $\frac{1}{\mu(E)} \mu$. Moreover, a \emph{submeasure} of $\mu$ is a measure $\mu'$ such that for all measurable set $A \subset E$, we have $\mu'(A) \leq \mu(A)$.

\begin{lemma}
\label{Paper2:lem:transportsubmeasure}
Let $\mu$ be any measure of positive mass, and let $\mu'$ be a submeasure of $\mu$ with $|\mu'|>0$.
Suppose that $\supp{\mu}$ is included in a ball $\closedball{x}{r}$.
Then
\begin{align*}
\Wasssymbol{p}\left(\overline{\mu}, \overline{\mu'}\right) \leq 2\left(1-\frac{|\mu'|}{|\mu|}\right)^\frac{1}{p} r.
\end{align*}
In particular, if $\mu$ is a probability measure, then
$\Wasssymbol{p}\left(\mu, \overline{\mu'}\right) \leq 2(1-|\mu'|)^\frac{1}{p} r$.
\end{lemma}

\begin{proof}
We start with the second inequality.
Consider the intermediate probability measure $\omega = \mu' + (1-|\mu'|) \delta_x$, where $\delta_x$ is the Dirac mass (represented in Figure \ref{fig:appendix:intermediatemeasure}).
We shall use the triangle inequality $\Wasssymbol{p}(\mu, \overline{\mu'}) \leq \Wasssymbol{p}(\mu, \omega) + \Wasssymbol{p}(\omega, \overline{\mu'})$.
We can write
\begin{itemize}
\itemsep.15cm
\item $\mu = \mu' + (\mu - \mu')$,
\item $\omega = \mu' + (1-|\mu'|) \delta_x$,
\item $\overline{\mu'} = \mu' + (\overline{\mu'}-\mu')$.
\end{itemize}
\begin{figure}[H]
\centering
\includegraphics[width=.8\linewidth]{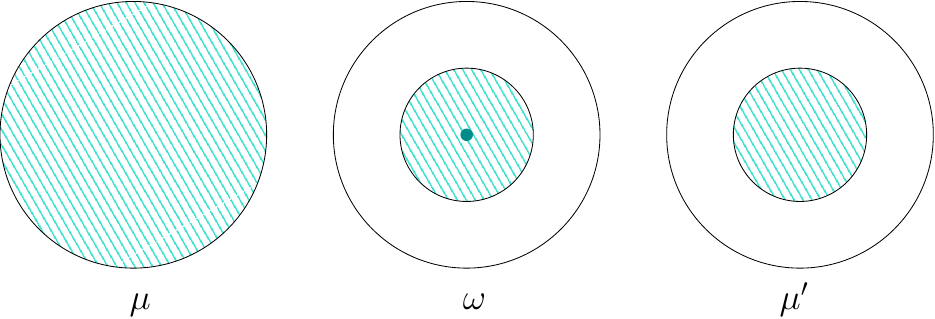}
\caption{The measures involved in the proof of Lemma \ref{Paper2:lem:transportsubmeasure}. 
A hatched area represents the support of the measure, and a point represents a Dirac mass.}
\label{fig:appendix:intermediatemeasure}
\end{figure}
Observe that $\mu$ and $\omega$ admits $\mu'$ as a common submeasure of mass $|\mu'|$. Therefore we can build a transport plan between $\mu$ and $\omega$ where only a mass $1-|\mu'|$ of $\mu$ is moved to $x$. In other words,
\begin{align*}
\Wasssymbol{p}(\mu, \omega) \leq (1-|\mu'|)^\frac{1}{p} r.
\end{align*}
Similarly, one shows that $\Wasssymbol{p}\left(\omega, \overline{\mu'}\right) \leq (1-|\mu'|)^\frac{1}{p} r$.

Now let us prove the first inequality. Since $\mu'$ is a submeasure of $\mu$ of mass $|\mu'|$, then $\frac{1}{|\mu|} \mu'$ is a submeasure of $\overline{\mu} = \frac{1}{|\mu|}\mu$ of mass $\frac{1}{|\mu|} |\mu'|$. We then apply the previous inequality.
\end{proof}

We now compare the localized measures of $\mu$, (defined in Subsect. \ref{Paper2:subsec:defgammaN}). 

\begin{lemma}
\label{Paper2:lem:Wstabxy}
Let $x \in \supp{\mu}$.
Suppose that $x$ satisfies Hypotheses \hyperref[hyp:5]{5} and \hyperref[hyp:6]{6} with $r < \mini{\lambda(x), \frac{1}{2\rho}}$.
Let $y \in E$ such that $\eucN{x-y} < \frac{r}{4}$. 
Then $|\loc{\mu}{x}|> 0$, $|\loc{\mu}{y}| > 0$ and 
\begin{align*}
\Wasssymbol{1}\left(\locprob{\mu}{x}, \locprob{\mu}{y}\right) \leq c_{\mathrm{\ref{Paper2:lem:Wstabxy:index}}} \eucN{x-y},
\end{align*}
with $c_{\mathrm{\ref{Paper2:lem:Wstabxy:index}}} = 2\left(1 + 4\frac{5^{d-1}}{3^d}\right)\frac{c_{\ref{Paper2:hyp:muB:index}}}{c_{\ref{Paper2:hyp:muA:index}}}$.
\end{lemma}

\begin{proof}
It is clear that $|\loc{\mu}{y}| > 0$ since $\mu(\ball{y}{r}) \geq \mu(\ball{x}{r-\eucN{x-y}})$ and $x \in \supp{\mu}$.
Let us show the inequality $\Wasssymbol{1}(\locprob{\mu}{x}, \locprob{\mu}{y}) \leq c_{\mathrm{\ref{Paper2:lem:Wstabxy:index}}} \eucN{x-y}$ by studying the measure $\mu$ on the intersection $\ball{x}{r} \cap \ball{y}{r}$.
Let $\loc{\mu}{x,y}$ be the restriction of $\mu$ to $\ball{x}{r} \cap \ball{y}{r}$, and $\locprob{\mu}{x,y}$ the corresponding probability measure.
The triangle inequality gives:
\begin{align}
\label{terms:Wstabxy}
\Wasssymbol{1}(\locprob{\mu}{x}, \locprob{\mu}{y}) \leq 
\underbrace{ \Wasssymbol{1}(\locprob{\mu}{x}, \locprob{\mu}{x,y} ) }_{\text{A}}
+ \underbrace{ \Wasssymbol{1}(\locprob{\mu}{x,y}, \locprob{\mu}{y}) }_{\text{B}}.
\end{align}

\medbreak\noindent\emph{Term \hyperref[terms:Wstabxy]{A}.}
Let us show that $\Wasssymbol{1}(\locprob{\mu}{x}, \locprob{\mu}{x,y}) \leq 2\frac{c_{\ref{Paper2:hyp:muB:index}}}{c_{\ref{Paper2:hyp:muA:index}}}\eucN{x-y}$.
Note that $\locprob{\mu}{x,y}$ is a submeasure of $\locprob{\mu}{x}$. According to Lemma \ref{Paper2:lem:transportsubmeasure}, we have
\begin{align*}
\Wasssymbol{1}(\locprob{\mu}{x}, \locprob{\mu}{x,y}) 
&\leq 2 \left(1-\frac{|\loc{\mu}{x,y}|}{|\loc{\mu}{x}|}\right) r = 2 \frac{|\loc{\mu}{x}|-|\loc{\mu}{x,y}|}{|\loc{\mu}{x}|} r.
\end{align*}
We know from Hypothesis \hyperref[hyp:5]{5} that $|\loc{\mu}{x}| \geq c_{\ref{Paper2:hyp:muA:index}} r^d$. 
On the other hand, 
\begin{align*}
|\loc{\mu}{x}|-|\loc{\mu}{x,y}| 
&= \mu(\ball{x}{r}) - \mu(\ball{x}{r} \cap \ball{y}{r})\\
&\leq \mu(\ball{x}{r}) -  \mu( \ball{x}{r - \eucN{x-y}}),
\end{align*}
hence we can apply Hypothesis \hyperref[hyp:6]{6} to get $|\loc{\mu}{x}|-|\loc{\mu}{x,y}| \leq c_{\ref{Paper2:hyp:muB:index}} r^{d-1} \eucN{x-y}$.
We finally obtain 
\begin{align*}
\Wasssymbol{1}(\locprob{\mu}{x}, \locprob{\mu}{x,y}) \leq  2\frac{c_{\ref{Paper2:hyp:muB:index}} r^{d-1} \eucN{x-y}}{c_{\ref{Paper2:hyp:muA:index}} r^d}r = 2\frac{c_{\ref{Paper2:hyp:muB:index}}}{c_{\ref{Paper2:hyp:muA:index}}}\eucN{x-y}.
\end{align*}

\medbreak\noindent\emph{Term \hyperref[terms:Wstabxy]{B}.}
Similarly, Lemma \ref{Paper2:lem:transportsubmeasure} yields
\begin{align*}
\Wasssymbol{1}(\locprob{\mu}{y}, \locprob{\mu}{x,y}) 
&\leq 2 \frac{|\loc{\mu}{y}|-|\loc{\mu}{x,y}|}{|\loc{\mu}{y}|} r.
\end{align*}
Let us show that we still have $|\loc{\mu}{y}| \geq a' r^d$  and $|\loc{\mu}{y}|-|\loc{\mu}{x,y}| \leq b' r^{d-1} \eucN{x-y}$ with the constants $a' = (\frac{3}{4})^d c_{\ref{Paper2:hyp:muA:index}}$ and $b' = 2(\frac{5}{4})^{d-1} c_{\ref{Paper2:hyp:muB:index}}$.
The first inequality comes from Hypothesis \hyperref[hyp:5]{5}:
\begin{align*}
\mu(\ball{y}{r})
\geq \mu(\ball{x}{r - \eucN{x-y}})
\geq c_{\ref{Paper2:hyp:muA:index}}(r - \eucN{x-y})^d
\end{align*}
and $\eucN{x-y} \leq \frac{r}{4}$.
The second inequality comes from Hypothesis \hyperref[hyp:6]{6}:
\begin{align*}
\mu(\ball{y}{r}) - \mu(\ball{x}{r} \cap \ball{y}{r})
&\leq \mu(\ball{x}{r + \eucN{x-y}}) - \mu(\ball{x}{r - \eucN{x-y}}) \\
&\leq c_{\ref{Paper2:hyp:muB:index}}(r+\eucN{x-y})^{d-1}2\eucN{x-y}
\end{align*}
and $\eucN{x-y} \leq \frac{r}{4}$.
To conclude,
\begin{align*}
\Wasssymbol{1}(\locprob{\mu}{y}, \locprob{\mu}{x,y}) 
&\leq 2 \frac{2(\frac{5}{4})^{d-1} r^{d-1} c_{\ref{Paper2:hyp:muA:index}} \eucN{x-y}}{2 (\frac{3}{4})^d c_{\ref{Paper2:hyp:muB:index}} r^d} r = 8 \frac{5^{d-1}}{3^d} \frac{c_{\ref{Paper2:hyp:muB:index}}}{c_{\ref{Paper2:hyp:muA:index}}}\eucN{x-y}. 
\end{align*}
We obtain the result by summing Terms \hyperref[terms:Wstabxy]{A} and \hyperref[terms:Wstabxy]{B}.
\end{proof}

The following lemma is the counterpart of Lemma \ref{Paper2:lem:Wstabxy} when replacing Hypothesis \hyperref[hyp:6]{6} with the weaker Hypothesis \hyperref[hyp:7]{7}.

\begin{lemma}
\label{Paper2:lem:Wstabxysqrt}
Let $x \in \supp{\mu}$.
Suppose that $x$ satisfies Hypotheses \hyperref[hyp:5]{5} and \hyperref[hyp:7]{7} at $x$ with $r<\frac{1}{2\rho}$.
Let $y \in E$ such that $\eucN{x-y} < \frac{r}{4}$. 
Then $|\loc{\mu}{x}|, |\loc{\mu}{y}| > 0$, and 
\begin{align*}
\Wasssymbol{1}(\locprob{\mu}{x}, \locprob{\mu}{y}) \leq c_{\mathrm{\ref{Paper2:lem:Wstabxysqrt:index}}} r^\frac{1}{2} \eucN{x-y}^{\frac{1}{2}},
\end{align*}
with $c_{\mathrm{\ref{Paper2:lem:Wstabxysqrt:index}}} = \left(2 + \frac{2^\frac{5}{2} 5^{d-\frac{1}{2}} }{3^d} \right)\frac{c_{\ref{Paper2:hyp:muBsqrt:index}}}{c_{\ref{Paper2:hyp:muA:index}}}$.
\end{lemma}

\begin{proof}
The proof is similar to Lemma \ref{Paper2:lem:Wstabxy} with slight modifications. 
We still consider
\begin{align}
\label{terms:Wstabxysqrt}
\Wasssymbol{1}(\locprob{\mu}{x}, \locprob{\mu}{y}) \leq 
\underbrace{ \Wasssymbol{1}(\locprob{\mu}{x}, \locprob{\mu}{x,y} ) }_{\text{A}}
+ \underbrace{ \Wasssymbol{1}(\locprob{\mu}{x,y}, \locprob{\mu}{y}) }_{\text{B}}.
\end{align}

\medbreak\noindent\emph{Term \hyperref[terms:Wstabxysqrt]{A}.} 
We have $\Wasssymbol{1}(\locprob{\mu}{x}, \locprob{\mu}{x,y}) 
\leq 2 \frac{|\loc{\mu}{x}|-|\loc{\mu}{x,y}|}{|\loc{\mu}{x}|} r$.
Hypothesis \hyperref[hyp:5]{5} still gives $|\loc{\mu}{x}| \geq c_{\ref{Paper2:hyp:muA:index}} r^d$. 
But Hypothesis \hyperref[hyp:7]{7} now yields
\begin{align*}
|\loc{\mu}{x}|-|\loc{\mu}{x,y}| 
&\leq \mu(\ball{x}{r}) -  \mu( \ball{x}{r - \eucN{x-y}}) \\
&\leq c_{\ref{Paper2:hyp:muBsqrt:index}} r^{d-\frac{1}{2}}\eucN{x-y}^\frac{1}{2}.
\end{align*}
We eventually obtain $\Wasssymbol{1}(\locprob{\mu}{x}, \locprob{\mu}{x,y}) \leq  2\frac{c_{\ref{Paper2:hyp:muBsqrt:index}}}{c_{\ref{Paper2:hyp:muA:index}}}r^\frac{1}{2} \eucN{x-y}^\frac{1}{2}$.

\medbreak\noindent\emph{Term \hyperref[terms:Wstabxysqrt]{B}.} 
In order to bound $\Wasssymbol{1}(\locprob{\mu}{y}, \locprob{\mu}{x,y}) 
\leq 2 \frac{|\loc{\mu}{y}|-|\loc{\mu}{x,y}|}{|\loc{\mu}{y}|} r$,
Hypothesis \hyperref[hyp:5]{5} still gives $|\loc{\mu}{x}| \geq (\frac{3}{4})^d c_{\ref{Paper2:hyp:muA:index}} r^d$, and Hypothesis \hyperref[hyp:7]{7} yields
\begin{align*}
|\loc{\mu}{y}|-|\loc{\mu}{x,y}|
&\leq \mu(\ball{x}{r + \eucN{x-y}}) - \mu(\ball{x}{r - \eucN{x-y}}) \\
&\leq c_{\ref{Paper2:hyp:muBsqrt:index}}(r+\eucN{x-y})^{d-\frac{1}{2}}(2\eucN{x-y})^\frac{1}{2},
\end{align*}
which is not greater than $c_{\ref{Paper2:hyp:muBsqrt:index}}(\frac{5}{4} r)^{d-\frac{1}{2}}(2\eucN{x-y})^\frac{1}{2}$.
We finally get
$$\Wasssymbol{1}(\locprob{\mu}{y}, \locprob{\mu}{x,y}) 
\leq 2 \frac{c_{\ref{Paper2:hyp:muBsqrt:index}}(\frac{5}{4} r)^{d-\frac{1}{2}}(2\eucN{x-y})^\frac{1}{2}}{(\frac{3}{4})^d c_{\ref{Paper2:hyp:muA:index}} r^d}r \leq \frac{2^\frac{5}{2} 5^{d-\frac{1}{2}} c_{\ref{Paper2:hyp:muBsqrt:index}}}{3^d c_{\ref{Paper2:hyp:muA:index}}} r^\frac{1}{2} \eucN{x-y}^\frac{1}{2},$$
and we obtain the result by adding Terms  \hyperref[terms:Wstabxysqrt]{A} and  \hyperref[terms:Wstabxysqrt]{B}.
\end{proof}

We can now compare the localized measures of two probability measures.

\begin{lemma}
\label{Paper2:lem:Wstabxx}
Let $w = \Wasssymbol{p}(\mu,  \nu)$.
Let $y \in E$. 
Suppose that there exists $x \in \supp{\mu}$ such that $\eucN{x-y} \leq \alpha$ with $\alpha = (\frac{w}{r^{d-1}})^\frac{1}{2}$, and that $\mu$ satisfies Hypotheses \hyperref[hyp:5]{5} and \hyperref[hyp:6]{6} at $x$ with $r < \mini{\lambda(x),\frac{1}{2\rho}}$.
Assume that $w\leq \mini{c_{\ref{Paper2:hyp:muA:index}} , 1}(\frac{r}{4})^{d+1}$.
Then
\begin{align*}
\Wasssymbol{1}(\locprob{\mu}{y}, \locprob{\nu}{y}) 
&\leq c_{\ref{Paper2:lem:Wstabxx:index}} \alpha,
\end{align*}
with $c_{\ref{Paper2:lem:Wstabxx:index}} = 
\frac{2^{d-1}}{c_{\ref{Paper2:hyp:muA:index}}}
+ 2\frac{12 \cdot 5^{d-1}c_{\ref{Paper2:hyp:muB:index}}+1}{3^d c_{\ref{Paper2:hyp:muA:index}}}
+ 2^{d+3} \frac{(\frac{3}{2})^{d-1}c_{\ref{Paper2:hyp:muB:index}}+1}{c_{\ref{Paper2:hyp:muA:index}}}$.
\end{lemma}

\begin{proof}
Let $\pi$ be an optimal transport for $\Wasssymbol{p}(\mu, \nu)$.
Define $\loc{\pi}{y}$ to be the restriction of the measure $\pi$ to the set $\ball{y}{r} \times \ball{y}{r} \subset E \times E$. Its marginals ${p_1}_* \loc{\pi}{y}$ and ${p_2}_* \loc{\pi}{y}$ are submeasures of $\loc{\mu}{y}$ and $\loc{\nu}{y}$.
We shall use the triangle inequality:
\begin{align}
\label{terms:Wstabxx}
\Wasssymbol{1}(\locprob{\mu}{y}, \locprob{\nu}{y}) 
\leq 
\underbrace{ \Wasssymbol{1}(\overline{\mu_y}, \overline{{p_1}_* \loc{\pi}{y}}) }_{\text{A}}
+ \underbrace{ \Wasssymbol{1}(\overline{{p_1}_* \loc{\pi}{y}}, \overline{{p_2}_* \loc{\pi}{y}}) }_{\text{B}}
+ \underbrace{ \Wasssymbol{1}(\overline{{p_2}_* \loc{\pi}{y}}, \overline{\nu_y}) }_{\text{C}}
\end{align}
Before examinating each of these terms, note that we have 
\begin{equation}
\label{Paper2:eq:lemxx1}
|\loc{\pi}{y}| = |{p_1}_* \loc{\pi}{y}| = |{p_2}_* \loc{\pi}{y}| \geq \mu(\ball{y}{r-\alpha} ) - \frac{w}{\alpha}
\end{equation}
\begin{equation}
\label{Paper2:eq:lemxx2}
|\loc{\nu}{y}| \leq \mu(\ball{y}{r + \alpha} ) +  \frac{w}{\alpha}\end{equation}
\begin{equation}
\label{Paper2:eq:lemxx3}
|\loc{\nu}{y}| \geq \mu(\ball{y}{r - \alpha} ) - \frac{w}{\alpha}
\end{equation}
The first equation can be proven as follows:
\begin{align*}
\mu(\ball{y}{r-\alpha}) 
&= \pi(\ball{y}{r-\alpha} \times E) \\
&= \pi(\ball{y}{r-\alpha} \times \ball{y}{r}) + \pi(\ball{y}{r-\alpha} \times \ball{y}{r}^c) 
\end{align*}
On the one hand, $\pi(\ball{y}{r-\alpha} \times \ball{y}{r}) \leq \pi(\ball{y}{r} \times \ball{y}{r}) \leq |\loc{\pi}{y}|$.
On the other hand, Markov inequality yields 
\begin{align*}
\pi(\ball{y}{r-\alpha} \times \ball{y}{r}^c) 
\leq \pi(\{(z,z'), \eucN{z-z'} \geq \alpha\})
\leq \frac{1}{\alpha} \int \eucN{z-z'} \dd \pi(z,z'),
\end{align*}
and Jensen inequality gives
\begin{align*}
\frac{1}{\alpha} \int \eucN{z-z'} \dd \pi(z,z') 
&\leq \frac{1}{\alpha} \bigg( \int \eucN{z-z'}^p \dd \pi(z,z')\bigg)^\frac{1}{p} = \frac{w}{\alpha}.
\end{align*}
We deduce that $\mu(\ball{y}{r-\alpha})\leq |\loc{\pi}{y}| + \frac{w}{\alpha}$, which gives Equation \eqref{Paper2:eq:lemxx1}.
Equations \eqref{Paper2:eq:lemxx2} and \eqref{Paper2:eq:lemxx3} can be proven similarly. 
In addition to these preliminaries, note that the assumption $w\leq \mini{c_{\ref{Paper2:hyp:muA:index}} , 1}(\frac{r}{4})^{d+1}$ yields
\begin{equation}
\label{Paper2:eq:lemxx4}
\alpha \leq \frac{r}{4}
\end{equation}
\begin{equation}
\label{Paper2:eq:lemxx5}
\frac{w}{\alpha} \leq \frac{c_{\ref{Paper2:hyp:muA:index}}}{2} \left(\frac{r}{2}\right)^d
\end{equation}
We now study the Terms \hyperref[terms:Wstabxx]{B}, \hyperref[terms:Wstabxx]{A} and \hyperref[terms:Wstabxx]{C}. 
\medbreak \noindent \emph{Term \hyperref[terms:Wstabxx]{B}.} 
Since $\locprob{\pi}{y} = \frac{\loc{\pi}{y}}{|\loc{\pi}{y}|}$ is a transport plan between $\overline{{p_1}_* \loc{\pi}{y}}$ and $\overline{{p_2}_* \loc{\pi}{y}}$, we have
\begin{align*}
\Wasssymbol{1}(\overline{{p_1}_* \pi_y}, \overline{{p_2}_* \pi_y})
\leq \int \eucN{z-z'} \frac{\dd \loc{\pi}{y}(z,z')}{|\pi_y|}
\leq \frac{1}{|\pi_y|} \int \eucN{z-z'} \dd \pi(z,z'). 
\end{align*}
Moreover, Jensen inequality yields $\int \eucN{z-z'} \dd \pi(z,z') \leq w$.
Hence 
\begin{align*}
\Wasssymbol{1}(\overline{{p_1}_* \pi_y}, \overline{{p_2}_* \pi_y}) \leq \frac{w}{|\pi_y|}.
\end{align*}
Let us prove that $|\pi_y| \geq \frac{c_{\ref{Paper2:hyp:muA:index}}}{2}(\frac{r}{2})^d$. 
According to Equation \eqref{Paper2:eq:lemxx1}, $|\loc{\pi}{y}|\geq \mu(\ball{y}{r-\alpha} ) - \frac{w}{\alpha}$.
Now, note that $\mu(\ball{y}{r-\alpha} ) \geq  \frac{c_{\ref{Paper2:hyp:muA:index}}}{2^d} r^d$. Indeed, using Hypothesis \hyperref[hyp:5]{5},
\begin{align*}
\mu(\ball{y}{r-\alpha} )
\geq \mu(\ball{x}{r-\alpha-\eucN{x-y}} )
\geq c_{\ref{Paper2:hyp:muA:index}}(r-\alpha-\eucN{x-y})^d,
\end{align*}
and we conclude with $\eucN{x-y} \leq \alpha \leq \frac{r}{4}$.
Now, using Equation \eqref{Paper2:eq:lemxx5}, we get
\begin{align*}
|\loc{\pi}{y}| 
&\geq \mu(\ball{y}{r-\alpha} ) - \frac{w}{\alpha} \\
&\geq c_{\ref{Paper2:hyp:muA:index}}\left(\frac{r}{2}\right)^d - \frac{c_{\ref{Paper2:hyp:muA:index}}}{2}\left(\frac{r}{2}\right)^d
\geq \frac{c_{\ref{Paper2:hyp:muA:index}}}{2}\left(\frac{r}{2}\right)^d .
\end{align*}
Finally, since $\alpha = \left(\frac{w}{r^{d-1}}\right)^\frac{1}{2}$ and $\alpha \leq \frac{r}{4}$, we obtain
\begin{align*}
\Wasssymbol{1}(\overline{{p_1}_* \pi_y}, \overline{{p_2}_* \pi_y}) 
\leq \frac{w}{|\loc{\pi}{y}|}
\leq \frac{w}{\frac{c_{\ref{Paper2:hyp:muA:index}}}{2}(\frac{r}{2})^d}
= \frac{2^{d+1}}{c_{\ref{Paper2:hyp:muA:index}}} \alpha^2 \frac{1}{r} 
\leq \frac{2^{d-1}}{c_{\ref{Paper2:hyp:muA:index}}}\alpha.
\end{align*}

\medbreak \noindent \emph{Term \hyperref[terms:Wstabxx]{A}.} 
According to Lemma \ref{Paper2:lem:transportsubmeasure}, we have
\begin{align}
\label{eq:wstabxx:1}
\Wasssymbol{1}(\overline{\mu_y}, \overline{{p_1}_* \pi_y}) 
\leq 2 \frac{|\mu_y|-|{p_1}_* \pi_y|}{|\mu_y|} r.
\end{align}
We can use Equation \eqref{Paper2:eq:lemxx1} to get
\begin{align*}
|\mu_y|-|{p_1}_* \pi_y| 
&\leq \mu(\ball{y}{r}) - \mu(\ball{y}{r-\alpha}) + \frac{w}{\alpha} \\
&\leq \mu(\ball{x}{r+\eucN{x-y}}) - \mu(\ball{x}{r-\alpha-\eucN{x-y}}) + \frac{w}{\alpha}.
\end{align*}
Moreover, by Hypothesis \hyperref[hyp:6]{6}, we have
\[\mu(\ball{x}{r+\eucN{x-y}}) - \mu(\ball{x}{r-\alpha-\eucN{x-y}}) \leq c_{\ref{Paper2:hyp:muB:index}} (r+\eucN{x-y})^{d-1} (2\eucN{x-y}+\alpha),\]
which is not greater than $c_{\ref{Paper2:hyp:muB:index}} (\frac{5}{4} r)^{d-1} 3 \alpha$ since $\eucN{x-y} \leq \alpha \leq \frac{r}{4}$.
Besides, $\frac{w}{\alpha} = r^{d-1} \alpha$, and we obtain
\begin{align*}
|\mu_y|-|{p_1}_* \pi_y| 
&\leq \left(3 \left(\frac{5}{4}\right)^{d-1} c_{\ref{Paper2:hyp:muB:index}} + 1\right) r^{d-1} \alpha.
\end{align*}
Finally, thanks to Hypothesis \hyperref[hyp:5]{5}, we write
\begin{align*}
|\mu_y| 
= \mu(\ball{y}{r}) 
&\geq \mu(\ball{x}{r-\eucN{x-y}}) \\
&\geq c_{\ref{Paper2:hyp:muA:index}} (r-\eucN{x-y})^d 
\geq c_{\ref{Paper2:hyp:muA:index}}\left(\frac{3}{4}\right)^d r^d
\end{align*}
and we obtain
\begin{align*}
\frac{|\mu_y|-|{p_1}_* \pi_y|}{|\mu_y|} 
&\leq \frac{((3 (\frac{5}{4})^{d-1} c_{\ref{Paper2:hyp:muB:index}}  + 1) r^{d-1} }{c_{\ref{Paper2:hyp:muA:index}}(\frac{3}{4})^d r^d}  \alpha
= \frac{1}{r} \cdot \frac{12 \cdot 5^{d-1}c_{\ref{Paper2:hyp:muB:index}}+1}{3^d c_{\ref{Paper2:hyp:muA:index}}} \alpha.
\end{align*}
Combined with Equation \eqref{eq:wstabxx:1}, we deduce
\begin{align*}
\Wasssymbol{1}(\overline{\mu_y}, \overline{{p_1}_* \pi_y}) 
\leq 2\frac{12 \cdot 5^{d-1}c_{\ref{Paper2:hyp:muB:index}}+1}{3^d c_{\ref{Paper2:hyp:muA:index}}} \alpha.
\end{align*}

\medbreak \noindent \emph{Term \hyperref[terms:Wstabxx]{C}.} 
It is similar to Term \hyperref[terms:Wstabxx]{A}.
First, one shows that 
\begin{align}
\label{eq:wstabxx:2}
\Wasssymbol{1}(\overline{\nu_y}, \overline{{p_2}_* \pi_y}) \leq 2 \frac{|\nu_y|-|{p_2}_* \pi_y|}{|\nu_y|} r.
\end{align}
Using Equations \eqref{Paper2:eq:lemxx1} and \eqref{Paper2:eq:lemxx2} we get
\begin{align*}
|\nu_y|-|{p_2}_* \pi_y| 
&\leq \mu(\ball{y}{r+\alpha})+\frac{w}{\alpha} - \mu(\ball{y}{r-\alpha}) + \frac{w}{\alpha} \\
&\leq \mu(\ball{x}{r+\eucN{x-y}+\alpha}) - \mu(\ball{x}{r-\alpha-\eucN{x-y}}) + 2\frac{w}{\alpha}.
\end{align*}
By Hypothesis \hyperref[hyp:6]{6}, we have
\begin{align*}
&\mu(\ball{x}{r+\eucN{x-y}+\alpha}) - \mu(\ball{x}{r-\alpha-\eucN{x-y}}) \\
&\leq c_{\ref{Paper2:hyp:muB:index}} (r+\eucN{x-y}+\alpha)^{d-1} (2\eucN{x-y}+2\alpha)
\end{align*}
which is not greater than $c_{\ref{Paper2:hyp:muB:index}} (\frac{3}{2} r)^{d-1} 4\alpha$ since $\eucN{x-y} \leq \alpha \leq \frac{r}{4}$.
Moreover, $\frac{w}{\alpha} = r^{d-1} \alpha$, and we obtain
\begin{align*}
|\nu_y|-|{p_2}_* \pi_y|  
\leq ( 4 (\frac{3}{2})^{d-1} c_{\ref{Paper2:hyp:muB:index}}  + 2 ) r^{d-1} \alpha.
\end{align*}
We have seen that
\begin{align*}
|\nu_y| &\geq \mu(\ball{y}{r-\alpha})-\frac{w}{\alpha} 
\geq \frac{c_{\ref{Paper2:hyp:muA:index}}}{2}\left(\frac{r}{2}\right)^d.
\end{align*}
Hence
\begin{align*}
\frac{|\nu_y|-|{p_2}_* \pi_y|}{|\nu_y|}
\leq \frac{(4 (\frac{3}{2})^{d-1} c_{\ref{Paper2:hyp:muB:index}}  + 2)r^{d-1} }{ \frac{c_{\ref{Paper2:hyp:muA:index}}}{2} (\frac{r}{2})^d }\alpha
= \frac{1}{r} \cdot 2^{d+2} \frac{(\frac{3}{2})^{d-1}b+1}{c_{\ref{Paper2:hyp:muA:index}}} \alpha,
\end{align*}
and we deduce from Equation \eqref{eq:wstabxx:2} that
\begin{align*}
\Wasssymbol{1}(\overline{\mu_y}, \overline{{p_1}_* \pi_y}) 
&\leq 2^{d+3} \frac{(\frac{3}{2})^{d-1}c_{\ref{Paper2:hyp:muB:index}}+1}{c_{\ref{Paper2:hyp:muA:index}}} \alpha.
\end{align*}
To conclude, summing up the Terms \hyperref[terms:Wstabxx]{A}, \hyperref[terms:Wstabxx]{B} and \hyperref[terms:Wstabxx]{C} gives
$\Wasssymbol{1}(\locprob{\mu}{y}, \locprob{\nu}{y}) 
\leq c_{\ref{Paper2:lem:Wstabxx:index}} \alpha$
with 
\begin{align*}
c_{\ref{Paper2:lem:Wstabxx:index}} = 
\frac{2^{d-1}}{c_{\ref{Paper2:hyp:muA:index}}}
+ 2\frac{12 \cdot 5^{d-1}c_{\ref{Paper2:hyp:muB:index}}+1}{3^d c_{\ref{Paper2:hyp:muA:index}}}
+ 2^{d+3} \frac{(\frac{3}{2})^{d-1}c_{\ref{Paper2:hyp:muB:index}}+1}{c_{\ref{Paper2:hyp:muA:index}}},
\end{align*}
as wanted.
\end{proof}

As before, we prove a version of Lemma \ref{Paper2:lem:Wstabxx} where Hypothesis \hyperref[hyp:6]{6} is replaced by the weaker Hypothesis \hyperref[hyp:7]{7}.

\begin{lemma}
\label{Paper2:lem:Wstabxxsqrt}
Let $w = \Wasssymbol{p}(\mu,  \nu)$.
Let $y \in E$. 
Suppose that there exists $x \in \supp{\mu}$ such that $\eucN{x-y} \leq \alpha$ with $\alpha = (\frac{w}{r^{d-1}})^\frac{1}{2}$, and that $\mu$ satisfies Hypotheses \hyperref[hyp:5]{5} and \hyperref[hyp:7]{7} at $x$ with $r < \frac{1}{2\rho}$.
Assume that $w\leq \mini{c_{\ref{Paper2:hyp:muA:index}}, 1}(\frac{r}{4})^{d+1}$.
Then
\begin{align*}
\Wasssymbol{1}(\locprob{\mu}{y}, \locprob{\nu}{y}) 
&\leq c_{\ref{Paper2:lem:Wstabxxsqrt:index}} r^\frac{1}{2} \alpha^\frac{1}{2},
\end{align*}
with $c_{\ref{Paper2:lem:Wstabxxsqrt:index}} = 
\frac{2^{d-2}}{c_{\ref{Paper2:hyp:muA:index}}}
+ \frac{4 \cdot 3^\frac{1}{2} 5^{d-\frac{1}{2}} c_{\ref{Paper2:hyp:muBsqrt:index}}+4^{d-\frac{1}{2}} }{3^d c_{\ref{Paper2:hyp:muA:index}}}
+2\cdot 4^d\frac{2 c_{\ref{Paper2:hyp:muBsqrt:index}} (\frac{3}{2})^{d-\frac{1}{2}}+1 }{3^d c_{\ref{Paper2:hyp:muA:index}}}$.
\end{lemma}

\begin{proof}
The proof is similar as Lemma \ref{Paper2:lem:Wstabxx}. Let us highlight the modifications.
Since $\alpha \leq \frac{r}{4}$ and $\frac{w}{\alpha} = r^{d-1}\alpha$, we have the inequalities
\begin{equation*}
\alpha^\frac{1}{2} \leq \frac{1}{2}r^\frac{1}{2}
~~~~~~~~~\text{and}~~~~~~~~~
\frac{w}{\alpha} \leq \frac{1}{2}r^{d-\frac{1}{2}}\alpha^\frac{1}{2}.
\end{equation*}
We still write the triangle inequality:
\begin{align}
\label{terms:Wstabxxsqrt}
\Wasssymbol{1}(\locprob{\mu}{y}, \locprob{\nu}{y}) 
\leq 
\underbrace{ \Wasssymbol{1}(\overline{\mu_y}, \overline{{p_1}_* \loc{\pi}{y}}) }_{\text{A}}
+ \underbrace{ \Wasssymbol{1}(\overline{{p_1}_* \loc{\pi}{y}}, \overline{{p_2}_* \loc{\pi}{y}}) }_{\text{B}}
+ \underbrace{ \Wasssymbol{1}(\overline{{p_2}_* \loc{\pi}{y}}, \overline{\nu_y}) }_{\text{C}}
\end{align}
where $\pi$ is an optimal transport plan for $\Wasssymbol{p}(\mu, \nu)$.

\medbreak \noindent \emph{Term \hyperref[terms:Wstabxxsqrt]{B}.} 
The argument to obtain $\Wasssymbol{1}(\overline{{p_1}_* \pi_y}, \overline{{p_2}_* \pi_y}) \leq \frac{2^{d-1}}{c_{\ref{Paper2:hyp:muA:index}}}\alpha
$ is unchanged, and we use $\alpha^\frac{1}{2} \leq \frac{1}{2}r^\frac{1}{2}$ to get
\begin{align*}
\Wasssymbol{1}(\overline{{p_1}_* \pi_y}, \overline{{p_2}_* \pi_y}) 
\leq \frac{2^{d-2}}{c_{\ref{Paper2:hyp:muA:index}}}\alpha^\frac{1}{2} r^\frac{1}{2}.
\end{align*}

\medbreak \noindent \emph{Term \hyperref[terms:Wstabxxsqrt]{A}.} 
Using Hypothesis \hyperref[hyp:7]{7}, we have
\begin{align*}
&\mu(\ball{x}{r+\eucN{x-y}}) - \mu(\ball{x}{r-\alpha-\eucN{x-y}})\\ 
&\leq c_{\ref{Paper2:hyp:muBsqrt:index}} (r+\eucN{x-y})^{d-\frac{1}{2}} (2\eucN{x-y}+\alpha))^\frac{1}{2}
 \\
&\leq c_{\ref{Paper2:hyp:muBsqrt:index}} \left(\frac{5}{4} r\right)^{d-\frac{1}{2}} 3^\frac{1}{2} \alpha^\frac{1}{2}.
\end{align*}
And since $\frac{w}{\alpha} \leq \frac{1}{2}r^{d-\frac{1}{2}}\alpha^\frac{1}{2}$, we get
\begin{align*}
|\mu_y|-|{p_1}_* \pi_y| 
&\leq \mu(\ball{x}{r+\eucN{x-y}}) - \mu(\ball{x}{r-\alpha-\eucN{x-y}}) + \frac{w}{\alpha} \\
&\leq \bigg(c_{\ref{Paper2:hyp:muBsqrt:index}} \left(\frac{5}{4} \right)^{d-\frac{1}{2}} 3^\frac{1}{2}  + \frac{1}{2}\bigg) r^{d-\frac{1}{2}} \alpha^\frac{1}{2}. 
\end{align*}
Finally, we use
\begin{align*}
|\mu_y| 
= \mu(\ball{y}{r}) 
&\geq \mu(\ball{x}{r-\eucN{x-y}}) \\
&\geq c_{\ref{Paper2:hyp:muA:index}} (r-\eucN{x-y})^d 
\geq c_{\ref{Paper2:hyp:muA:index}}\left(\frac{3}{4}\right)^d r^d
\end{align*}
to obtain
\begin{align*}
\frac{|\mu_y|-|{p_1}_* \pi_y|}{|\mu_y|} 
&\leq \frac{ ((c_{\ref{Paper2:hyp:muBsqrt:index}} (\frac{5}{4} )^{d-\frac{1}{2}} 3^\frac{1}{2}  + \frac{1}{2}) r^{d-\frac{1}{2}} }{c_{\ref{Paper2:hyp:muA:index}}(\frac{3}{4})^d r^d}  \alpha^\frac{1}{2}
= \frac{1}{r^\frac{1}{2}} \cdot \frac{2 \cdot 3^\frac{1}{2} 5^{d-\frac{1}{2}} c_{\ref{Paper2:hyp:muBsqrt:index}}+4^{d-\frac{1}{2}} }{3^d c_{\ref{Paper2:hyp:muA:index}}} \alpha^\frac{1}{2}
\end{align*}
and we deduce
\begin{align*}
\Wasssymbol{1}(\overline{\mu_y}, \overline{{p_1}_* \pi_y}) 
&\leq 2 \frac{|\mu_y|-|{p_1}_* \pi_y|}{|\mu_y|} r
\leq \frac{4 \cdot 3^\frac{1}{2} 5^{d-\frac{1}{2}} c_{\ref{Paper2:hyp:muBsqrt:index}}+4^{d-\frac{1}{2}} }{3^d c_{\ref{Paper2:hyp:muA:index}}} r^\frac{1}{2} \alpha^\frac{1}{2}.
\end{align*}

\medbreak \noindent \emph{Term \hyperref[terms:Wstabxxsqrt]{C}.} 
We use Hypothesis \hyperref[hyp:7]{7} to get
\begin{align*}
&\mu(\ball{x}{r+\eucN{x-y}+\alpha}) - \mu(\ball{x}{r-\alpha-\eucN{x-y}}) \\
&\leq c_{\ref{Paper2:hyp:muBsqrt:index}} (r+\eucN{x-y}+\alpha)^{d-\frac{1}{2}} (2\eucN{x-y}+2\alpha)^\frac{1}{2}\\
&\leq 2 c_{\ref{Paper2:hyp:muBsqrt:index}} \left(\frac{3}{2}r\right)^{d-\frac{1}{2}} \alpha^\frac{1}{2}.
\end{align*}
And since $\frac{w}{\alpha} \leq \frac{1}{2}r^{d-\frac{1}{2}}\alpha^\frac{1}{2}$, we get
\begin{align*}
|\nu_y|-|{p_2}_* \pi_y| 
&\leq \mu(\ball{x}{r+\eucN{x-y}+\alpha}) - \mu(\ball{x}{r-\alpha-\eucN{x-y}}) + 2\frac{w}{\alpha}\\
&\leq \bigg(2 c_{\ref{Paper2:hyp:muBsqrt:index}} \left(\frac{3}{2}\right)^{d-\frac{1}{2}}+1\bigg) r^{d-\frac{1}{2}} \alpha^\frac{1}{2}. 
\end{align*}
Finally, we use
\begin{align*}
|\mu_y| 
= \mu(\ball{y}{r}) 
&\geq \mu(\ball{x}{r-\eucN{x-y}}) \\
&\geq c_{\ref{Paper2:hyp:muA:index}} (r-\eucN{x-y})^d 
\geq c_{\ref{Paper2:hyp:muA:index}}\left(\frac{3}{4}\right)^d r^d
\end{align*}
to obtain
\begin{align*}
\frac{|\mu_y|-|{p_1}_* \pi_y|}{|\mu_y|} 
&\leq \frac{ (2 c_{\ref{Paper2:hyp:muBsqrt:index}} (\frac{3}{2})^{d-\frac{1}{2}}+1) r^{d-\frac{1}{2}} }{c_{\ref{Paper2:hyp:muA:index}}(\frac{3}{4})^d r^d}  \alpha^\frac{1}{2}
= \frac{1}{r^\frac{1}{2}} \cdot 4^d\frac{2 c_{\ref{Paper2:hyp:muBsqrt:index}} (\frac{3}{2})^{d-\frac{1}{2}}+1 }{3^d c_{\ref{Paper2:hyp:muA:index}}} \alpha^\frac{1}{2}
\end{align*}
and we deduce
\begin{align*}
\Wasssymbol{1}(\overline{\mu_y}, \overline{{p_1}_* \pi_y}) 
&\leq 2 \frac{|\mu_y|-|{p_1}_* \pi_y|}{|\mu_y|} r
\leq 2\cdot 4^d\frac{2 c_{\ref{Paper2:hyp:muBsqrt:index}} (\frac{3}{2})^{d-\frac{1}{2}}+1 }{3^d c_{\ref{Paper2:hyp:muA:index}}} r^\frac{1}{2} \alpha^\frac{1}{2}. 
\end{align*}
We finally obtain the result by summing Terms \hyperref[terms:Wstabxxsqrt]{A}, \hyperref[terms:Wstabxxsqrt]{B} and \hyperref[terms:Wstabxxsqrt]{C}.
\end{proof}

\begin{remark}
\label{Paper2:rem:inequality_tight}
Let us comment the inequality of Lemma \ref{Paper2:lem:Wstabxx} with $p=1$, valid for all $r$ such that $\Wassun{\mu}{\nu} \leq \mini{a , 1}(\frac{r}{4})^{d+1}$:
\begin{equation}
\label{Paper2:eq:estimationWlocrobust0}
\Wasssymbol{1}(\locprob{\mu}{y}, \locprob{\nu}{y}) 
\leq c_{\ref{Paper2:lem:Wstabxx:index}} \left(\frac{\Wassun{\mu}{\nu}}{r^{d-1}}\right)^\frac{1}{2}.
\end{equation}
If $r$ is assumed to be constant, the behavior of $\Wasssymbol{1}(\locprob{\mu}{y}, \locprob{\nu}{y})$, when $\Wassun{\mu}{\nu}$ goes to 0, is 
$\Wasssymbol{1}(\locprob{\mu}{y}, \locprob{\nu}{y}) 
\lesssim \Wassun{\mu}{\nu}^\frac{1}{2}$.
On the other hand, if $r$ is supposed to follow the worst case, i.e., $r$ is of order $\Wassun{\mu}{\nu}^\frac{1}{d+1}$, then $\Wasssymbol{1}(\locprob{\mu}{y}, \locprob{\nu}{y})$ is of order 
$\Wasssymbol{1}(\locprob{\mu}{y}, \locprob{\nu}{y}) 
\lesssim \Wassun{\mu}{\nu}^\frac{1}{d+1}$.

A similar stability result already appears in \cite[Theorem 4.3]{pmlr-v97-memoli19a}, where $\mu$ and $\nu$ are two probability measures on a bounded set $X$, and $\mu$ satisfy the following condition: $\forall x \in X, \forall s,r \leq 0$ s.t. $s \leq r$, we have
$\frac{\mu( \closedball{x}{r} )}{ \mu( \closedball{x}{s} ) } \leq (\frac{r}{s})^d$.
The theorem states that, denoting $D = \mathrm{diam(X)}$, for all $x \in X$,
\begin{align*}
\Wassun{\locprob{\mu}{x}}{\locprob{\nu}{x}}
\leq (1+2r)
\left[\frac{\Wassun{\mu}{\nu}^\frac{1}{2}}{\mini{1, (\frac{r}{D})^d}}
+ \left(1 + \frac{\Wassun{\mu}{\nu}^\frac{1}{2}}{r}\right)^d -1  \right].
\end{align*} 
When $r\leq D$ and $\Wassun{\mu}{\nu}$ goes to zero, we obtain that $\Wassun{\locprob{\mu}{x}}{\locprob{\nu}{x}}$ is of order 
\begin{align*}
\Wassun{\locprob{\mu}{x}}{\locprob{\nu}{x}} \leq (1+2r)D^d \left(\frac{\Wassun{\mu}{\nu}}{r^{2d}}\right)^\frac{1}{2}.
\end{align*}
The exponent on $r$ is greater here than in Equation \eqref{Paper2:eq:estimationWlocrobust0}.

Let us show that the order we obtained, $(\frac{\Wassun{\mu}{\nu}}{r^{d-1}})^\frac{1}{2}$, is optimal.
More precisely, let us show that, for every $d\geq 1$, $r > 0$ and $\epsilon>0$ fixed, there exists measures $\mu$ and $\nu$ on $\R^d$ that satisfies the assumptions of Lemma \ref{Paper2:lem:Wstabxx}, but such that
\begin{align*}
\Wasssymbol{1}(\locprob{\mu}{y}, \locprob{\nu}{y}) 
&\geq c_d \left(\frac{\Wassun{\mu}{\nu}}{r^{d-1}}\right)^\frac{1}{2} - \epsilon
\end{align*}
with $c_d = \frac{1}{d+1}\left(\frac{2 d }{\volball{d}}\right)^\frac{1}{2}$.
We consider the following example.
Let $\mu = \HH^{d}_{[0,1]^d}$ be the Lebesgue measure on the hypercube $[0,1]^d$.
Denote $y = \left(\frac{1}{2}, \dots, \frac{1}{2}\right)$ its center, $B = \openball{y}{r}$ the open ball, and $A$ the annulus defined as
\begin{align*}
A = \openball{y}{r+\epsilon} \setminus \openball{y}{r}
\end{align*}
where $0 < \epsilon < r < \frac{1}{4}$.
In the following, $r$ stays fixed, and $\epsilon$ will go to zero. 
Consider the probability measure 
\begin{align*}
\nu = \HH^{d}_{[0,1]^d \setminus A} + \frac{\volball{d}(r+\epsilon)^d-\volball{d}r^d}{\volsphere{d-1}r^{d-1}}\HH^{d-1}_{\sphere{y}{r}}.
\end{align*}
Let $\locprob{\mu}{y}$ and $\locprob{\nu}{y}$ be the localized probability measures associated to $\mu$ and $\nu$ with parameter $r$.
We shall show that 
\begin{center}
$\Wasssymbol{1}(\mu, \nu)$ is of order $r^{d-1}\epsilon^2$ ~~~~~and~~~~~ $\Wasssymbol{1}(\locprob{\mu}{y}, \locprob{\nu}{y})$ is of order $\epsilon$
\end{center}
when $\epsilon\rightarrow 0$.
These measures are depicted in Figure \ref{fig:appendix:measuresremark}.
\begin{figure}[H]
\centering
\includegraphics[width=1\linewidth]{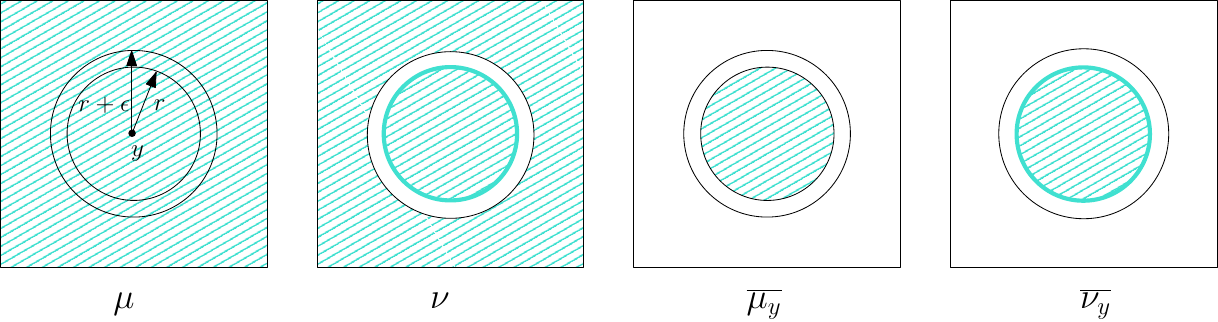}
\caption{The measures involved in the example. 
A hatched area represents the $d$-dimensional Hausdorff measure $\HH^d$, and a bold circle represents the $(d-1)$-dimensional Hausdorff measure $\HH^{d-1}$.
}
\label{fig:appendix:measuresremark}
\end{figure}

\medbreak\noindent
\emph{Step 1: Study of $\Wasssymbol{1}(\mu, \nu)$.} 
An optimal transport plan between $\mu$ and $\nu$ is given by transporting the submeasure $\HH^{d}_A$ of $\mu$ onto the submeasure $\frac{\volball{d}(r+\epsilon)^d-\volball{d}r^d}{\volsphere{d-1}r^{d-1}}\HH^{d-1}_{\sphere{y}{r}}$ of $\nu$ via the map
$x \mapsto \frac{r}{\eucN{x}}x$. 
Consequently, the Wasserstein distance is
\begin{align*}
\Wasssymbol{1}(\mu, \nu)
&= \int_A \eucN{x-\frac{r}{\eucN{x}} x } \frac{\volball{d}(r+\epsilon)^d-\volball{d}r^d}{\volsphere{d-1}r^{d-1}} \dd \HH^{d}(x).
\end{align*}
A change of coordinates shows that 
\begin{align*}
\int_A \eucN{x-\frac{r}{\eucN{x}}x } \dd \HH^{d}(x) 
&= \int_{\sphere{0}{1}} \int_r^{r+\epsilon} (t-r) t^{d-1} \dd \HH^{1}(t)  \dd \HH^{d-1}(v).
\end{align*}
Let us split the integral as $$\int_r^{r+\epsilon} (t-r) t^{d-1} \dd \HH^{1}(t) = \int_r^{r+\epsilon} t^d  \dd \HH^{1}(t) -  \int_r^{r+\epsilon} rt^{d-1} \dd \HH^{1}(t).$$
On the one hand, we have
\begin{align*}
\int_r^{r+\epsilon} t^d \dd \HH^{1}(t)
&= \frac{1}{d+1}\left(\left( r+\epsilon\right)^{d+1} - r^{d+1} \right) \\
&= r^{d} \epsilon + \frac{d}{2} r^{d-1} \epsilon^2 + \petito{\epsilon^2},
\end{align*}
where the Little-O notation refers to $\epsilon \to 0$.
On the other hand,
\begin{align*}
\int_r^{r+\epsilon} rt^{d-1} \dd \HH^{1}(t)
&= r\left( r^{d-1} \epsilon + \frac{d-1}{2} r^{d-2} \epsilon^2 + \petito{ \epsilon }^2 \right) \\
&= r^{d} \epsilon + \frac{d-1}{2} r^{d-1} \epsilon^2 +  \petito{\epsilon^2}.
\end{align*}
We deduce that $\int_r^{r+\epsilon} (t-r) t^{d-1} \dd \HH^{1}(t) = \frac{1}{2}r^{d-1} \epsilon^2 + \petito{\epsilon^2}$, and 
\begin{align*}
\int_A \eucN{x-\frac{r}{\eucN{x}}x } \dd \HH^{d}(x) 
&= \frac{\volsphere{d-1}}{2} r^{d-1} \epsilon^2 + \petito{ \epsilon^2}.
\end{align*}
In other words, $\Wasssymbol{1}(\mu, \nu) = \frac{d\volball{d}}{2}r^{d-1}\epsilon^2 + \petito{\epsilon^2}$.

\medbreak\noindent
\emph{Step 2: Study of $\Wasssymbol{1}(\locprob{\mu}{y}, \locprob{\nu}{y})$.}
Consider the measures
\begin{align*}
\locprob{\mu}{x} &= \frac{1}{\volball{d}r^d} \HH^{d}_B =\left(\frac{1}{\volball{d}(r+\epsilon)^d} +\frac{\volball{d}(r+\epsilon)^d - \volball{d}r^d}{\volball{d}(r+\epsilon)^d \volball{d}r^d }\right)\HH^{d}_B,\\
\locprob{\nu}{x} &= \frac{1}{\volball{d}(r+\epsilon)^d}\left( \HH^{d}_B + \frac{\volball{d}(r+\epsilon)^d-\volball{d}r^d}{\volsphere{d-1}r^{d-1}} \HH^{d-1}_{\sphere{y}{r}}\right).
\end{align*}
Let us compute the Wasserstein distance $\Wasssymbol{1}(\locprob{\mu}{y}, \locprob{\nu}{y})$.
As before, an optimal transport plan is given by transporting the submeasure $\frac{\volball{d}(r+\epsilon)^d - \volball{d}r^d}{\volball{d}(r+\epsilon)^d \volball{d}r^d }\HH^{d}_B$ of $\locprob{\mu}{x}$ onto the submeasure $\frac{\volball{d}(r+\epsilon)^d - \volball{d}r^d}{\volball{d}(r+\epsilon)^d \volsphere{d-1}r^{d-1} }\HH^{d-1}_{\sphere{y}{r}}$ of $\locprob{\nu}{x}$.
We have:
\begin{align*}
\Wasssymbol{1}(\locprob{\mu}{y}, \locprob{\nu}{y})
&= \int_B \eucN{x-\frac{r}{\eucN{x}} x } 
\frac{\volball{d}(r+\epsilon)^d - \volball{d}r^d}{\volball{d}(r+\epsilon)^d \volball{d}r^d } \dd \HH^{d}(x) \\
\end{align*}
A change of coordinates yields
\begin{align*}
\int_B \eucN{x-\frac{r}{\eucN{x}} x } \dd \HH^{d}(x)= \frac{\volsphere{d-1}}{d(d+1)}r^{d+1}.
\end{align*}
Besides, we have 
\begin{align*}
\frac{\volball{d}(r+\epsilon)^d - \volball{d}r^d}{\volball{d}(r+\epsilon)^d \volball{d}r^d } 
&=\frac{d \volball{d}r^{d-1}\epsilon + \grando{\epsilon^2}}{\volball{d}(r+\epsilon)^d \volball{d}r^d }
= \frac{d}{\volball{d}}\frac{\epsilon}{r^{d+1}} + \grando{\epsilon^2}.
\end{align*}
We deduce that
\begin{align*}
\Wasssymbol{1}(\locprob{\mu}{y}, \locprob{\nu}{y})
=\frac{\volsphere{d-1}}{d(d+1)}\frac{d}{\volball{d}} \epsilon + \grando{\epsilon^2}
=\frac{d}{d+1} \epsilon + \grando{\epsilon^2}.
\end{align*}

\medbreak\noindent
\emph{Step 3: Comparison of the distances.}
Using $\Wasssymbol{1}(\mu, \nu)
= \frac{d\volball{d}}{2}r^{d-1}\epsilon^2 + \petito{\epsilon^2}$ and $\Wasssymbol{1}(\locprob{\mu}{y}, \locprob{\nu}{y})
=\frac{d}{d+1} \epsilon + \grando{\epsilon^2}$,
we get
\begin{align*}
\frac{\Wasssymbol{1}(\locprob{\mu}{y}, \locprob{\nu}{y})^2}{\Wasssymbol{1}(\mu, \nu)}
&= c_d \frac{1}{r^{d-1}} + \grando{\epsilon}
~~~~~~~\text{ where }~~~~~~~c_d = \frac{\left( \frac{d}{d+1}  \right)^2}{\frac{d\volball{d}}{2}} 
= \frac{2 d }{(d+1)^2\volball{d}}.
\end{align*}
Together with $\Wasssymbol{1}(\mu, \nu)^\frac{1}{2} = \grando{\epsilon}$, we deduce the result:
\begin{align*}
\Wasssymbol{1}(\locprob{\mu}{y}, \locprob{\nu}{y})
&=c^\frac{1}{2}\left(\frac{\Wasssymbol{1}(\mu, \nu)}{r^{d-1}}\right)^\frac{1}{2} + \grando{\epsilon^2}. 
\end{align*}
\end{remark}

\subsection{Stability of the estimation}
\label{Paper2:subsec:stability}
In this subsection we study the stability of the normalized local covariance matrix operator $\mu \mapsto \loccovnorm{\mu}{\cdot}$ (see Definition \ref{def:loccov}) with respect to the $\Wasssymbol{p}$ metric on measures. 

As an introduction to the problem, let $\mu$ and $\nu$ be two probability measures, $x \in \supp{\mu} \cap \supp{\nu}$, and consider the Frobenius distance $\frobN{ \loccovnorm{\mu}{x} - \loccovnorm{\nu}{x} }$ between the normalized local covariance matrices. 
One shows that this distance is related to the 1-Wasserstein distance between the localized probability measures $\locprob{\mu}{x}$ and $\locprob{\nu}{x}$ via the following inequality (see Equation \eqref{Paper2:eq:lemmastabilityfrobnorm} in the proof of Lemma \ref{Paper2:lem:WchecktointW}):
\begin{equation*}
\frobN{ \loccovnorm{\mu}{x} - \loccovnorm{\nu}{x} }
\leq \frac{2}{r} \Wasssymbol{1}(\locprob{\mu}{x}, \locprob{\nu}{x}).
\end{equation*}
Without any assumption on the measures, it is not true that $\Wassun{\locprob{\mu}{x}}{\locprob{\nu}{x}}$ goes to 0 as $\Wasssymbol{1}(\mu, \nu)$ does (see Remark \ref{rem:notconvergence}).
However, if we assume that $\mu$ satisfies the Hypotheses \hyperref[hyp:5]{5} and \hyperref[hyp:6]{6}, that $x$ satisfies $\lambda(x) > 0$ and that $r$ is chosen such that $$ 4 \left(\frac{\Wassun{\mu}{\nu}}{\mini{c_{\ref{Paper2:hyp:muA:index}} , 1}}\right)^{\frac{1}{d+1}} \leq  r < \mini{\normalreach{x},\frac{1}{2\rho}},$$ then we have seen in Lemma \ref{Paper2:lem:Wstabxx} that
\begin{equation}
\label{Paper2:eq:estimationWlocrobust}
\Wassun{\locprob{\mu}{x}}{\locprob{\nu}{x}}
\leq c_{\ref{Paper2:lem:Wstabxx:index}} \bigg(\frac{\Wassun{\mu}{\nu}}{r^{d-1}}\bigg)^\frac{1}{2}.
\end{equation}
As a consequence of this inequality, estimating local covariance matrices is robust in Wasserstein distance:
\begin{equation}
\label{Paper2:eq:estimationloccovrobust}
\frobN{ \loccovnorm{\mu}{x} - \loccovnorm{\nu}{x} }
\leq 2c_{\ref{Paper2:lem:Wstabxx:index}} \bigg(\frac{\Wasssymbol{1}(\mu, \nu)}{r^{d+1}}\bigg)^\frac{1}{2}.
\end{equation}
We point out that another result of this kind bounds the distance $\frobN{ \loccovnorm{\mu}{x} - \loccovnorm{\nu}{x} }$ with the $\infty$-Wasserstein distance $\Wasssymbol{\infty}(\mu, \nu)$  \cite[Theorem 3]{martinez2018shape}. Namely, if $\mu$ and $\nu$ are fully supported probability measures with densities upper bounded by $l>0$ and supports included in $X \subset \mathbb{R}^d$, denoting $D = \mathrm{diam(X)}$, we have 
\begin{align*}
\frobN{ \loccovnorm{\mu}{x} - \loccovnorm{\nu}{x} } \leq l A \Wasssymbol{\infty}(\mu, \nu),
\end{align*}
where $A = \frac{d}{d+2} \frac{(r+D)^{d+1}}{Dr^d} + \frac{(2r+D)(r+D)^d}{r^d} + \frac{2d}{d+2} \frac{(r+D)^{d+2}}{Dr^d}$.

\begin{remark}
\label{rem:notconvergence}
In general, for $x \in \supp{\mu} \cap \supp{\nu}$, it is not true that $\frobN{ \loccovnorm{\mu}{x} - \loccovnorm{\nu}{x} }$ goes to zero as $\Wasssymbol{1}(\mu, \nu)$ goes to zero. 
To see this, one can consider $\epsilon > 0$, and the measures on $\R$ defined as
\begin{center}
$\mu = \frac{1}{2}(\delta_0 + \delta_1)$ ~~~~~~~~and~~~~~~~~ $\nu = \frac{1}{2}(\delta_0 + \delta_{1+\epsilon})$,
\end{center}
where $\delta_x$ denotes the Dirac mass on $x$.
Choose the scale parameter $r = 1$.
Restricting the measures $\mu$ and $\nu$ to the ball $\ball{0}{1}$ of $\R$ gives $\locprob{\mu}{0} = \frac{1}{2}(\delta_0 + \delta_1)$ and $\locprob{\nu}{0} =\delta_0$.
According to Definition \ref{def:loccov}, we deduce that the local covariance matrices are
\begin{center}
$\loccov{\mu}{0} = \frac{1}{2} 1^{\otimes 2}$ ~~~~~~~~and~~~~~~~~ $\loccov{\nu}{0} = 0$.~~~~~~~~~~~~
\end{center}
Hence $\frobN{ \loccovnorm{\mu}{x} - \loccovnorm{\nu}{x} } = \frac{1}{2}$, independently of $\epsilon$.
But, on the other hand, we have $\Wasssymbol{1}(\mu, \nu) = \frac{1}{2}\epsilon$, which goes to zero.

Similarly, $\gammaWassersteinp
(\check{\mu}, \check{\nu})$ does not have to go to zero when $\Wasssymbol{1}(\mu, \nu)$ does. Indeed, a similar computation shows that the local covariance matrices at $1$ are
\begin{center}
$\loccov{\mu}{1} = \frac{1}{2} 1^{\otimes 2}$ ~~~~~~~~and~~~~~~~~ $\loccov{\nu}{1} = 0$,~~~~~~~~~~~~
\end{center}
and we deduce that the lifted measures are
\begin{center}
$\check{\mu} = \frac{1}{2}\left(\delta_{\left(0,\frac{1}{2} 1^{\otimes 2}\right)} + \delta_{\left(1,\frac{1}{2} 1^{\otimes 2}\right)}\right)$ ~~~~~~~and~~~~~~~ $\check{\nu} = \frac{1}{2}\left(\delta_{(0,0)} + \delta_{(1+\epsilon,0)}\right)$.
~~~~~~~~~~~~~
\end{center}
Using the optimal transport plan $\pi$ between $\check{\mu}$ and $\check{\nu}$ that sends $\delta_{\left(0,\frac{1}{2} 1^{\otimes 2}\right)}$ to $\delta_{(0,0)}$
and $\delta_{\left(1,\frac{1}{2} 1^{\otimes 2}\right)}$ to $\delta_{(1+\epsilon,0)}$, we get 
\begin{align*}
\gammaWassersteinp^p(\check{\mu}, \check{\nu}) 
&= \frac{1}{2}\gammaN{\left(0,\frac{1}{2} 1^{\otimes 2}\right)-\bigg(0,0\bigg)}^p  + \frac{1}{2}\gammaN{\left(1,\frac{1}{2} 1^{\otimes 2}\right)-\bigg(1+\epsilon,0\bigg)}^p \\
&= \frac{1}{2}\bigg( \left(\frac{\gamma}{2}\right)^p + \left(\epsilon^2 + \gamma^2\frac{1}{4}\right)^\frac{p}{2} \bigg)
\geq \left(\frac{\gamma}{2}\right)^p.
\end{align*}
Again, we see that $\gammaWassersteinp^p(\check{\mu}, \check{\nu})$ is lower bounded, independently of $\epsilon$.
Hence $\gammaWassersteinp(\check{\mu}, \check{\nu})$ does not go to zero as $\Wasssymbol{1}(\mu, \nu)$ does.
However, under regularity assumptions on $\mu$, the following proposition states that it is the case.
\end{remark} 
\medbreak

\begin{proposition}
\label{Paper2:prop:stability}
Let $\mu$ and $\nu$ be two probability measures on $E$.
Suppose that $\mu$ satisfies Hypotheses \hyperref[hyp:5]{5}, \hyperref[hyp:6]{6} and \hyperref[hyp:7]{7}.
Define $w = \Wasssymbol{p}(\mu, \nu)$.
Suppose that $r \leq \mini{\frac{1}{2\rho} , 1}$ and $w\leq \mini{c_{\ref{Paper2:hyp:muA:index}}, 1} (\frac{r}{4})^{d+1}$. 
Then
\begin{align*}
\gammaWassersteinp(\check{\mu}, \check{\nu})
&\leq 2 w 
+ \gamma c_{\ref{Paper2:prop:stability:index}} \left(\frac{w}{r^{d+1}}\right)^\frac{1}{2} 
+ \gamma c_{\ref{Paper2:prop:stability:index}} ' \mu(\lambda^r)^\frac{1}{p} \left(\frac{w}{r^{d+1}} \right)^\frac{1}{4},
\end{align*}
with $c_{\ref{Paper2:prop:stability:index}} = 4(1 + c_{\ref{Paper2:lem:intlocwasserstein:index}})$
and $c_{\ref{Paper2:prop:stability:index}} ' =4c_{\ref{Paper2:lem:Wstabxxsqrt:index}}$.
\end{proposition}

\begin{proof}
According to Lemma \ref{Paper2:lem:WchecktointW} stated below, we have 
\begin{align*}
\gammaWassersteinp(\check{\mu}, \check{\nu})
&\leq 2^\frac{p-1}{p} \left(1+\frac{2 \gamma}{r}\right) w + 2^\frac{p-1}{p} \frac{2\gamma}{r} \bigg(\int \Wasssymbol{1}^p(\locprob{\mu}{x}, \locprob{\nu}{y}) \dd {\pi}(x,y) \bigg)^\frac{1}{p}. 
\end{align*}
Let $\alpha = \left(\frac{w}{r^{d-1}}\right)^\frac{1}{2}$. Lemma \ref{Paper2:lem:intlocwasserstein}, also stated below, gives
\begin{align*}
\bigg( \int \Wasssymbol{1}(\locprob{\mu}{x}, \locprob{\nu}{y}) \dd {\pi}(x,y) \bigg)^\frac{1}{p}
&\leq 2^\frac{p-1}{p} \bigg(
c_{\ref{Paper2:lem:Wstabxxsqrt:index}}r^\frac{1}{2} \mu(\lambda^r)^\frac{1}{p} \alpha^\frac{1}{2} + c_{\ref{Paper2:lem:intlocwasserstein:index}}\alpha
\bigg)
\end{align*}
Combining these inequalities yields
\begin{align*}
\gammaWassersteinp(\check{\mu}, \check{\nu})
&\leq 2^\frac{p-1}{p} w + 2^\frac{p-1}{p} \frac{2 \gamma}{r}\left(w+2^\frac{p-1}{p} c_{\ref{Paper2:lem:intlocwasserstein:index}} \alpha\right) + \left(2^\frac{p-1}{p}\right)^2\frac{2 \gamma}{r} c_{\ref{Paper2:lem:Wstabxxsqrt:index}} r^\frac{1}{2} \mu(\lambda^r)^\frac{1}{p} \alpha^\frac{1}{2} \\
&\leq 2 w + 2 \cdot 2 \gamma \left(\frac{w}{r}+2 c_{\ref{Paper2:lem:intlocwasserstein:index}}\frac{\alpha}{r} \right) + 2^2 \cdot 2\gamma c_{\ref{Paper2:lem:Wstabxxsqrt:index}} \mu(\lambda^r)^\frac{1}{p} \left(\frac{\alpha}{r}\right)^\frac{1}{2},
\end{align*}
where we used $2^\frac{p-1}{p} \leq 2$.
Beside, since $r \leq 1$, we have $w\leq 1$, and 
$$w = \left(\frac{w}{r^{d-1}}\right)^\frac{1}{2} r^\frac{d-1}{2}w^\frac{1}{2} \leq \left(\frac{w}{r^{d-1}}\right)^\frac{1}{2} = \alpha.$$
Consequently, we have
\begin{align*}
\gammaWassersteinp(\check{\mu}, \check{\nu})
&\leq 2^\frac{p-1}{p} w + 2^\frac{p-1}{p} 2 \gamma \left(1+2^\frac{p-1}{p} c_{\ref{Paper2:lem:intlocwasserstein:index}} \right) \frac{\alpha}{r} + \left(2^\frac{p-1}{p}\right)^2 2\gamma c_{\ref{Paper2:lem:Wstabxxsqrt:index}} \mu(\lambda^r)^\frac{1}{p} \left(\frac{\alpha}{r}\right)^\frac{1}{2}.
\end{align*}
We obtain the result by replacing $\frac{\alpha}{r}$ with $\left(\frac{w}{r^{d+1}}\right)^\frac{1}{2}$.
\end{proof}

Let us interpret the inequality given by Proposition \ref{Paper2:prop:stability}:
\begin{align*}
\gammaWassersteinp(\check{\mu}, \check{\nu})
&\leq 2 \Wasssymbol{p}(\mu, \nu) 
+ \gamma c_{\ref{Paper2:prop:stability:index}} \left(\frac{\Wasssymbol{p}(\mu, \nu)}{r^{d+1}}\right)^\frac{1}{2} 
+ \gamma c_{\ref{Paper2:prop:stability:index}} ' \mu(\lambda^r)^\frac{1}{p} \left(\frac{\Wasssymbol{p}(\mu, \nu)}{r^{d+1}} \right)^\frac{1}{4},
\end{align*}
The lifted measures $\check{\mu}$ and $\check{\nu}$ are defined on the lift space $E\times\matrixspace{E}$. 
Hence the Wasserstein distance $\gammaWassersteinp(\check{\mu}, \check{\nu})$ may witness a difference with respect to the Euclidean coordinate ($E$-coordinate) or the matrix coordinate ($\matrixspace{E}$-coordinate).
We can interpret this inequality as follows:
\begin{itemize}
\itemsep.01cm
\item the first term $2\Wasssymbol{p}(\mu, \nu)$ is to be seen as the initial Euclidean error between the measures $\mu$ and $\nu$,
\item the second term $\gamma c_{\ref{Paper2:prop:stability:index}} \left(\frac{\Wasssymbol{p}(\mu, \nu)}{r^{d+1}}\right)^\frac{1}{2}$ corresponds to the local errors $\Wasssymbol{1}(\locprob{\mu}{x}, \locprob{\nu}{y})$ in $\matrixspace{E}$ when comparing the normalized covariance matrices of points away from the self-intersections of $\MM$,
\item the third term $\gamma c_{\ref{Paper2:prop:stability:index}} ' \mu(\lambda^r)^\frac{1}{p} \left(\frac{\Wasssymbol{p}(\mu, \nu)}{r^{d+1}} \right)^\frac{1}{4}$ stands for the error in $\matrixspace{E}$ on points close to the self-intersections of $\MM$. The quantity of such points is measured via $\mu(\lambda^r)$, the measure of the $r$-sublevel set of the normal reach.
\end{itemize}
As a consequence of this proposition, the map $\mu \mapsto \check{\mu}$, seen as a map between spaces of measures endowed with the Wassertein metric, is continuous on the set of measures $\mu$ which satisfy Hypotheses \hyperref[hyp:5]{5}, \hyperref[hyp:6]{6} and \hyperref[hyp:7]{7} with $\frac{1}{2\rho} \geq r$.

We now give the lemmas used in the proof of this Proposition \ref{Paper2:prop:stability}.

\begin{lemma}
\label{Paper2:lem:WchecktointW}
Let $\pi$ be an optimal transport plan for $\Wasssymbol{p}(\mu, \nu)$. Then
\begin{align*}
\gammaWassersteinp(\check{\mu}, \check{\nu})
&\leq 2^\frac{p-1}{p} \left(1+\frac{2 \gamma}{r}\right) \Wasssymbol{p}(\mu,\nu) + 2^\frac{p-1}{p} \frac{2\gamma}{r} \bigg(\int \Wasssymbol{1}^p(\locprob{\mu}{x}, \locprob{\nu}{y}) \dd {\pi}(x,y) \bigg)^\frac{1}{p}. 
\end{align*}
\end{lemma}

\begin{proof}
We first prove the following fact: for every $x \in \supp{\mu}$ and $y \in \supp{\nu}$,
\begin{equation}
\label{Paper2:eq:lemmastabilityfrobnorm}
\frobN{\loccov{\mu}{x} - \loccov{\nu}{y}} \leq 2r\left(\eucN{x-y} + \Wasssymbol{1}(\locprob{\mu}{x}, \locprob{\nu}{y}) \right).
\end{equation}
Let $\rho$ be any transport plan between $\locprob{\mu}{x}$ and $\locprob{\nu}{y}$.
We have
\begin{align}
\label{Paper2:eq:lemmastabilityfrobnorm2}
\loccov{\mu}{x} - \loccov{\nu}{y} &=  \int \tensor{x-y} \dd \locprob{\mu}{x}(x') - \int \tensor{y-y'} \dd \locprob{\mu}{y}(y') \nonumber\\
&= \int \big(\tensor{x-x'}-\tensor{y-y'}\big) \dd \rho(x',y').
\end{align}
For any $x' \in \closedball{x}{r}$ and $y' \in \closedball{y}{r}$, we can use Lemma \ref{Paper2:lem:outer} to get
\begin{align*}
\frobN{\tensor{x-x'}-\tensor{y-y'}} 
\leq (r+r)(\eucN{x-y} + \eucN{x'-y'}).
\end{align*}
Therefore, Equation \eqref{Paper2:eq:lemmastabilityfrobnorm2} yields
\begin{align*}
\frobN{\loccov{\mu}{x} - \loccov{\nu}{y}} 
&\leq \int 2r (\eucN{x-y} + \eucN{x'-y'}) \dd \rho(x',y') \\
&\leq 2r \left(\eucN{x-y} + \Wasssymbol{1}(\locprob{\mu}{x}, \locprob{\nu}{y})\right).
\end{align*}

Now, a transport plan $\pi$ for $\Wasssymbol{p}(\mu, \nu)$ begin given, we build a transport plan $\check{\pi}$ for $(\check{\mu}, \check{\nu})$ as follows: for every $\phi\colon (E \times \matrixspace{E})^2 \rightarrow \R$ with compact support, let $\check{\pi}$ satisfy
\[\int \phi(x, A, y, B) \dd \check{\pi}(x, A, y, B) = \int \phi\left(x, \loccovnorm{\mu}{x}, y, \loccovnorm{\nu}{y}\right) \dd \pi(x,y). \]
We have the upper bound
\begin{align}
\label{Paper2:eq:lemmastabilityfrobnorm3}
\gammaWassersteinp^p(\check{\mu}, \check{\nu}) 
&\leq \int \gammaN{ (x,A) - (y,B) }^p \dd \check{\pi}(x,A,y,B)\nonumber \\
&=  \int \left(\eucN{ x - y }^2 + \gamma^2 \frobN{\loccovnorm{\mu}{x} - \loccovnorm{\nu}{y}}^2 \right)^\frac{p}{2} \dd {\pi}(x,y) \nonumber\\ 
&\leq \int \left(\eucN{ x - y } + \gamma \frobN{\loccovnorm{\mu}{x} - \loccovnorm{\nu}{y}} \right)^p \dd {\pi}(x,y)
\end{align}
Besides, Equation \eqref{Paper2:eq:lemmastabilityfrobnorm} gives 
\begin{align*}
\frobN{\loccovnorm{\mu}{x} - \loccovnorm{\nu}{y}} \leq \frac{1}{r^2} \frobN{\loccov{\mu}{x} - \loccov{\nu}{y}} \leq \frac{2}{r} \left(\eucN{x-y} + \Wasssymbol{1}(\locprob{\mu}{x}, \locprob{\nu}{y})\right).
\end{align*} 
We can use the inequality $(a+b)^p \leq 2^{p-1}(a^p+b^p)$, where $a,b \geq 0$, to deduce
\begin{align*}
&\left(\eucN{ x - y } + \gamma \frobN{\loccovnorm{\mu}{x} - \loccovnorm{\nu}{y}} \right)^p
\leq \left(\eucN{ x - y } + \gamma \frac{2}{r} \left(\eucN{x-y} + \Wasssymbol{1}(\locprob{\mu}{x}, \locprob{\nu}{y})\right) \right)^p \\
&~~~~~~~~~~~~~~~~~~~~~~~~~~~~~~~\leq 2^{p-1} \left(\left(1+\frac{2 \gamma}{r} \right)\eucN{ x - y }\right)^p + 2^{p-1}\left(\frac{2 \gamma}{r} \Wasssymbol{1}(\locprob{\mu}{x}, \locprob{\nu}{y}\right)^p.
\end{align*}
By inserting this inequality in Equation \eqref{Paper2:eq:lemmastabilityfrobnorm3} we obtain
\begin{align*}
\gammaWassersteinp^p(\check{\mu}, \check{\nu})
&\leq  2^{p-1} \int \left(\left(1+\frac{2 \gamma}{r} \right)\eucN{ x - y }\right)^p + \left(\frac{2 \gamma}{r} \Wasssymbol{1}(\locprob{\mu}{x}, \locprob{\nu}{y}) \right)^p \dd {\pi}(x,y) \\
&= 2^{p-1} \left(1+\frac{2 \gamma}{r}\right)^p \Wasssymbol{p}^p(\mu,\nu) + 2^{p-1} \left(\frac{2 \gamma}{r}\right)^p \int \Wasssymbol{1}^p(\locprob{\mu}{x}, \locprob{\nu}{y}) \dd {\pi}(x,y),
\end{align*}
which yields the result.
\end{proof}

\begin{lemma}
Let $w = \Wasssymbol{p}(\mu, \nu)$ and define $\alpha = (\frac{w}{r^{d-1}})^\frac{1}{2}$.
Suppose that $r \leq \frac{1}{2\rho}$ and $w\leq \mini{c_{\ref{Paper2:hyp:muA:index}}, 1} (\frac{r}{4})^{d+1}$.
Let $\pi$ be an optimal transport plan for $\Wasssymbol{p}(\mu, \nu)$. 
Then
\begin{align*}
&\left( \int \Wasssymbol{1}^p(\locprob{\mu}{x}, \locprob{\nu}{y}) \dd {\pi}(x,y) \right)^\frac{1}{p} \\
&~~~~~~~~~~~~~\leq 2^\frac{p-1}{p} \left(
c_{\ref{Paper2:lem:Wstabxxsqrt:index}} r^\frac{1}{2} \mu(\lambda^r)^\frac{1}{p} \alpha^\frac{1}{2}
+ \left(2 r^d + c_{\ref{Paper2:lem:Wstabxysqrt:index}} r^\frac{d+1}{2} + c_{\ref{Paper2:lem:Wstabxx:index}}\right)\alpha + (1+c_{\ref{Paper2:lem:Wstabxy:index}}) w
\right).
\end{align*}
If we suppose that $r\leq 1$, then
\begin{align*}
\left( \int \Wasssymbol{1}^p(\locprob{\mu}{x}, \locprob{\nu}{y}) \dd {\pi}(x,y) \right)^\frac{1}{p}
&\leq 2^\frac{p-1}{p} \left(
c_{\ref{Paper2:lem:Wstabxxsqrt:index}}r^\frac{1}{2} \mu(\lambda^r)^\frac{1}{p} \alpha^\frac{1}{2} + c_{\ref{Paper2:lem:intlocwasserstein:index}}\alpha
\right)
\end{align*}
with $c_{\ref{Paper2:lem:intlocwasserstein:index}} = 3 +c_{\ref{Paper2:lem:Wstabxy:index}} + c_{\ref{Paper2:lem:Wstabxysqrt:index}} + c_{\ref{Paper2:lem:Wstabxx:index}} $.
\label{Paper2:lem:intlocwasserstein}
\end{lemma}

\begin{proof}
We denote $w = \Wasssymbol{p}(\mu, \nu)$ and $\alpha = (\frac{w}{r^{d-1}})^\frac{1}{2}$.
Let us subdivide the integral as follows:
\begin{align}
\label{terms:intlocwasserstein}
\int\Wasssymbol{1}^p(\overline{\mu_x}, \overline{\nu_y}) d \pi(x,y)
= \int_{A}
+ \int_{B}
+ \int_{C}
\Wasssymbol{1}^p(\overline{\mu_x}, \overline{\nu_y}) \dd \pi(x,y)
\end{align} 
where 
\begin{itemize}
\itemsep.05cm
\item $A = \{(x,y) \mid \eucN{x-y} \geq \alpha\}$,
\item $B = \{(x,y) \mid \eucN{x-y} < \alpha ~~\mathrm{and}~~ \lambda(x) > r\}$,
\item $C = \{(x,y) \mid \eucN{x-y} < \alpha ~~\mathrm{and}~~ \lambda(x) \leq r\}$.
\end{itemize}

\medbreak \noindent \emph{Term \hyperref[terms:intlocwasserstein]{$A$}.} 
We use the following simple upper bound:
\begin{align*}
\Wasssymbol{1}(\overline{\mu_x}, \overline{\nu_y}) 
&\leq \Wasssymbol{1}(\overline{\mu_x}, \delta_x) + \Wasssymbol{1}(\delta_x, \delta_y) 
+ \Wasssymbol{1}(\delta_y, \overline{\nu_y}) \\
&\leq r + \eucN{x-y} + r
\end{align*} 
to obtain $\Wasssymbol{1}^p(\overline{\mu_x}, \overline{\nu_y}) \leq 2^{p-1}\big( (2r)^p+\eucN{x-y}^p \big)$ and 
\begin{align*}
\int_{A} \Wasssymbol{1}^p(\overline{\mu_x}, \overline{\nu_y}) \dd \pi(x,y)
&\leq \int_{A} 2^{p-1}\big( (2r)^p+\eucN{x-y}^p \big) \dd \pi(x,y) \\
&\leq  2^{p-1}(2r)^p \pi(A) + \int 2^{p-1} \eucN{x-y}^p \dd \pi(x,y) \\
&= 2^{p-1} (2r)^p \pi(A) + 2^{p-1} w^p.
\end{align*} 
Besides, Markov inequality yields 
$$\pi(A) = \pi(\{(x,y) \mid \eucN{x-y}>\alpha) = \pi\left(\left\{(x,y) \mid \eucN{x-y}^p>\alpha^p)\right\}\right) \leq \left(\frac{w}{a}\right)^p.$$
Therefore,
\begin{align*}
\int_{A} \Wasssymbol{1}^p(\overline{\mu_x}, \overline{\nu_y}) \dd \pi(x,y)
&\leq   2^{p-1} (2r)^p \left(\frac{w}{\alpha}\right)^p +  2^{p-1}w^p \\
&= 2^{p-1} (2 r^d \alpha)^p + 2^{p-1} w^p,
\end{align*} 
where we used $ r\frac{w}{\alpha} = r^d \alpha$ on the last line.

\smallbreak \noindent \emph{Term \hyperref[terms:intlocwasserstein]{$B$}.} 
On the event $B$, we write
\[ \Wasssymbol{1}(\overline{\mu_x}, \overline{\nu_y}) 
\leq \Wasssymbol{1}(\overline{\mu_x}, \overline{\mu_y}) +\Wasssymbol{1}(\overline{\mu_y}, \overline{\nu_y}). \]
Since $\lambda(x)>r$, Lemma \ref{Paper2:lem:Wstabxy} and Lemma \ref{Paper2:lem:Wstabxx} give $\Wasssymbol{1}(\locprob{\mu}{x}, \locprob{\mu}{y}) \leq c_{\ref{Paper2:lem:Wstabxy:index}} \eucN{x-y}$ and $\Wasssymbol{1}(\overline{\mu_y}, \overline{\nu_y}) \leq c_{\ref{Paper2:lem:Wstabxx:index}} \alpha$.
We deduce that
\begin{align*}
\int_{B} \Wasssymbol{1}^p(\overline{\mu_x}, \overline{\nu_y}) \dd \pi(x,y)
&\leq2^{p-1} \int_{B} (c_{\ref{Paper2:lem:Wstabxy:index}} \eucN{x-y})^p + (c_{\ref{Paper2:lem:Wstabxx:index}} \alpha)^p \dd \pi(x,y) \\
&\leq 2^{p-1} (c_{\ref{Paper2:lem:Wstabxy:index}} w)^p + 2^{p-1} (c_{\ref{Paper2:lem:Wstabxx:index}} \alpha)^p. 
\end{align*}

\smallbreak \noindent \emph{Term \hyperref[terms:intlocwasserstein]{C}.} 
We proceed as for Term \hyperref[terms:intlocwasserstein]{B}, but using Lemmas \ref{Paper2:lem:Wstabxysqrt} and \ref{Paper2:lem:Wstabxxsqrt} instead of Lemmas \ref{Paper2:lem:Wstabxy} and \ref{Paper2:lem:Wstabxx}. This yields
\begin{align*}
\Wasssymbol{1}(\overline{\mu_x}, \overline{\nu_y}) 
&\leq \Wasssymbol{1}(\overline{\mu_x}, \overline{\mu_y}) + \Wasssymbol{1}(\overline{\mu_y}, \overline{\nu_y}) \\
&\leq c_{\ref{Paper2:lem:Wstabxysqrt:index}} r^\frac{1}{2} \eucN{x-y}^\frac{1}{2} + c_{\ref{Paper2:lem:Wstabxxsqrt:index}} r^\frac{1}{2} \alpha^\frac{1}{2},
\end{align*}
and we deduce that
\begin{align}
\int_{C} \Wasssymbol{1}^p(\overline{\mu_x}, \overline{\nu_y}) \dd \pi(x,y)
&\leq \int_{C} 2^{p-1} \left(c_{\ref{Paper2:lem:Wstabxysqrt:index}} r^\frac{1}{2}\eucN{x-y}^\frac{1}{2} \right)^p\dd \pi(x,y) \nonumber \\
&~~~~+ 2^{p-1} \pi(C) \left( c_{\ref{Paper2:lem:Wstabxxsqrt:index}}r^\frac{1}{2} \alpha^\frac{1}{2}\right)^p. 
\label{Paper2:eq:intergral_w1p}
\end{align}
On the one hand, we have $\int_{C} \eucN{x-y}^\frac{p}{2}\dd \pi(x,y)\leq \int_{E\times E} \eucN{x-y}^\frac{p}{2}\dd \pi(x,y)$, and by Jensen's inequality, 
\begin{align*}
\int_{E\times E} \eucN{x-y}^\frac{p}{2}\dd \pi(x,y) \leq (w^p)^\frac{1}{2}.
\end{align*} 
On the other hand, by definition of $C$, we have $\pi(C) \leq \mu(\lambda^r)$. 
Hence Equation \eqref{Paper2:eq:intergral_w1p} yields
\[\int_{C} \Wasssymbol{1}(\overline{\mu_x}, \overline{\nu_y}) \dd \pi(x,y)
\leq 2^{p-1} \left(c_{\ref{Paper2:lem:Wstabxysqrt:index}} r^\frac{1}{2}w^\frac{1}{2}\right)^p + 2^{p-1} \mu(\lambda^r) \left(c_{\ref{Paper2:lem:Wstabxxsqrt:index}} r^\frac{1}{2} \alpha^\frac{1}{2}\right)^p. \]
To conclude the proof, we write
\begin{align*}
\int \Wasssymbol{1}(\locprob{\mu}{x}, \locprob{\nu}{y}) \dd {\pi}(x,y) 
&=  \int_{A} + \int_{B} + \int_{C} \Wasssymbol{1}(\overline{\mu_x}, \overline{\nu_y}) \dd \pi(x,y)\\
&\leq 2^{p-1} (2 r^d \alpha)^p + 2^{p-1} w^p + 2^{p-1} (c_{\ref{Paper2:lem:Wstabxy:index}} w)^p + 2^{p-1} (c_{\ref{Paper2:lem:Wstabxx:index}} \alpha)^p \\
&~~~~~+ 2^{p-1} \left(c_{\ref{Paper2:lem:Wstabxysqrt:index}} r^\frac{1}{2}w^\frac{1}{2}\right)^p + 2^{p-1} \mu(\lambda^r) \left(c_{\ref{Paper2:lem:Wstabxxsqrt:index}} r^\frac{1}{2} \alpha^\frac{1}{2}\right)^p.
\end{align*}
We use the inequality $(a+b)^\frac{1}{p} \leq a^\frac{1}{p} + b^\frac{1}{p}$, where $a,b\geq 0$, to get
\begin{align*}
&\bigg( \int \Wasssymbol{1}(\locprob{\mu}{x}, \locprob{\nu}{y}) \dd {\pi}(x,y) \bigg)^\frac{1}{p} \\
&\leq 2^\frac{p-1}{p} \bigg( 2 r^d \alpha + w +c_{\ref{Paper2:lem:Wstabxy:index}} w + c_{\ref{Paper2:lem:Wstabxx:index}} \alpha +c_{\ref{Paper2:lem:Wstabxysqrt:index}} r^\frac{1}{2}w^\frac{1}{2} + \mu(\lambda^r)^\frac{1}{p} c_{\ref{Paper2:lem:Wstabxxsqrt:index}} r^\frac{1}{2} \alpha^\frac{1}{2} \bigg) \\
&\leq 2^\frac{p-1}{p} \bigg(
c_{\ref{Paper2:lem:Wstabxxsqrt:index}} r^\frac{1}{2} \mu(\lambda^r)^\frac{1}{p} \alpha^\frac{1}{2}
+ \left(2 r^d + c_{\ref{Paper2:lem:Wstabxysqrt:index}} r^\frac{d+1}{2} + c_{\ref{Paper2:lem:Wstabxx:index}}\right)\alpha
+ (1+c_{\ref{Paper2:lem:Wstabxy:index}}) w
\bigg),
\end{align*}
where we used $c_{\ref{Paper2:lem:Wstabxysqrt:index}} r^\frac{1}{2} w^\frac{1}{2}= c_{\ref{Paper2:lem:Wstabxysqrt:index}} r^\frac{d+1}{2} \alpha$ on the the last line.
This proves the first result.

If we suppose $r\leq 1$, we can use the inequalities 
$r^d \leq r^\frac{d+1}{2} \leq 1$ and $w = \alpha r^\frac{d-1}{2}w^\frac{1}{2} \leq \alpha$
to obtain the simplified expression
\begin{align*}
\bigg( \int \Wasssymbol{1}(\locprob{\mu}{x}, \locprob{\nu}{y}) \dd {\pi}(x,y) \bigg)^\frac{1}{p}
&\leq 2^\frac{p-1}{p} \bigg(
c_{\ref{Paper2:lem:Wstabxxsqrt:index}}r^\frac{1}{2} \mu(\lambda^r)^\frac{1}{p} \alpha^\frac{1}{2} + \left(3 +c_{\ref{Paper2:lem:Wstabxy:index}} + c_{\ref{Paper2:lem:Wstabxysqrt:index}} + c_{\ref{Paper2:lem:Wstabxx:index}}\right )\alpha
\bigg),
\end{align*}
as wanted.
\end{proof}

\subsection{An approximation theorem}
\label{subsec:approx_th}
We are now able to state that the lifted measure $\nucheck$ is close to the exact lifted measure $\muchecko$, that is, $\nucheck$ is a consistent estimator of $\muchecko$, in Wasserstein distance.

\begin{theorem}
\label{Paper2:th:estimation}
Assume that $\MMo$ and $\muo$ satisfy Hypotheses \hyperref[hyp:1]{1}, \hyperref[hyp:2]{2}, \hyperref[hyp:3]{3}.
Let $\nu$ be any probability measure. Denote $w=\Wasssymbol{p}(\mu, \nu)$. Suppose that $r \leq \mini{\frac{1}{2\rho}, 1}$ and $w\leq \mini{c_{\ref{Paper2:hyp:muA:index}}, 1}(\frac{r}{4})^{d+1}$.
Then
\begin{align*}
\gammaWassersteinp(\check{\nu}, \check{\mu}_0) 
&\leq \gamma c_{\ref{Paper2:th:estimation:index}}  \mu(\lambda^r)^\frac{1}{p} 
+ \gamma c_{\ref{Paper2:lem:loccovconsistency:index}} r
+ \gamma c_{\ref{Paper2:prop:stability:index}} \left(\frac{w}{r^{d+1}}\right)^\frac{1}{2}
+ 2w
\end{align*}
where $c_{\ref{Paper2:th:estimation:index}} = 2 + \frac{1}{2}c_{\ref{Paper2:prop:stability:index}} '$.
\end{theorem}

\begin{proof}
By using the triangle inequality $\gammaWassersteinp(\check{\nu}, \check{\mu}_0) \leq\gammaWassersteinp(\check{\nu}, \check{\mu})+\gammaWassersteinp(\check{\mu}, \check{\mu}_0)$, we see that the result is a direct consequence of Propositions \ref{Paper2:prop:consistencycheck} and \ref{Paper2:prop:stability}.
\end{proof}

\begin{remark}
\label{rem:optimal_choice_r}
The quality of the bound given by Theorem \ref{Paper2:th:estimation} is balanced by the contributions of Propositions \ref{Paper2:prop:consistencycheck} and \ref{Paper2:prop:stability}. According to the first one, the quantity $\gammaWassersteinp(\check{\mu}, \check{\mu}_0)$ is minimized when $r$ is as small as possible, and according to the second one, the distance $\gammaWassersteinp(\check{\nu}, \check{\mu})$ is minimized when $r$ is chosen large.
Roughly speaking, to optimize the bound given by the theorem, we have to pick a $r$ given by equating the terms $r$ and $\left(\frac{w}{r^{d+1}}\right)^\frac{1}{2}$, that is, $r = w^\frac{1}{d+3}$.
We will make this choice in the following corollary. 
\end{remark}

\begin{remark}
\label{Paper2:rem:estimation}
In the case where $\MMo$ is embedded, we have seen in Proposition \ref{Paper2:prop:normalreachtoreach} that the normal reach $\lambda$ is bounded below by $\reach{\MM}>0$.
In particular, $\mu(\lambda^r)$ is zero for $r$ small enough.
In this case, Theorem \ref{Paper2:th:estimation} reads
\begin{align*}
\gammaWassersteinp(\check{\nu}, \check{\mu}_0) 
&\leq \gamma c_{\ref{Paper2:lem:loccovconsistency:index}} r
+ \gamma c_{\ref{Paper2:prop:stability:index}} \left(\frac{w}{r^{d+1}}\right)^\frac{1}{2}
+ 2w
\end{align*}
We deduce an approximation result: if $(\nu_i)_{i\geq 0}$ is a sequence of probability measures such that $w_i=\Wasssymbol{p}(\mu, \nu_i)$ goes to zero, and if we choose a sequence of radii $(r_i)_{i\geq 0}$ such that 
$(r_i)_{i\geq 0}$ and $\left(w_i/r_i^{d+1}\right)_{i\geq 0}$ 
go to zero, then $\gammaWassersteinp(\check{\nu}_i, \check{\mu}_0)$ goes to zero too.

More generally, $\gammaWassersteinp(\check{\nu}_i, \check{\mu}_0) $ goes to zero if we only assume that $\MMo$ satisfies Hypothesis \hyperref[hyp:4]{4}. This is stated in the following corollary, which is a weaker version of the theorem, that we shall use in the following section to simplify the results.
\end{remark}
\medbreak

\begin{corollary}
\label{Paper2:cor:approximation}
Let $r>0$.
Assume that $\MMo$ and $\mu_0$ satisfy Hypotheses \hyperref[hyp:1]{1}, \hyperref[hyp:2]{2} and \hyperref[hyp:4]{4}.
Let $\nu$ be any probability measure. Denote $w=\Wasssymbol{p}(\mu, \nu)$. Suppose that $r < \mini{\frac{1}{2\rho},r_{\ref{Paper2:hyp:normalreach:index}}, 1}$ and $w\leq \mini{c_{\ref{Paper2:hyp:muA:index}}, 1} (\frac{r}{4})^{d+3}$.
Then
\begin{align*}
\gammaWassersteinp(\check{\nu}, \check{\mu}_0) 
&\leq\big(1 + \gamma c_{\ref{Paper2:cor:approximation:index}}\big) r^\frac{1}{p}
\end{align*}
with $c_{\ref{Paper2:cor:approximation:index}} = c_{\ref{Paper2:th:estimation:index}}(c_{\ref{Paper2:hyp:normalreach:index}})^\frac{1}{p} + c_{\ref{Paper2:prop:stability:index}}  + c_{\ref{Paper2:lem:loccovconsistency:index}}$.
\end{corollary}

\begin{proof}
According to Theorem \ref{Paper2:th:estimation}, we have 
\begin{align*}
\gammaWassersteinp(\check{\nu}, \check{\mu}_0) 
&\leq \gamma c_{\ref{Paper2:th:estimation:index}}  \mu(\lambda^r)^\frac{1}{p} 
+ \gamma c_{\ref{Paper2:lem:loccovconsistency:index}} r
+ \gamma c_{\ref{Paper2:prop:stability:index}} \left(\frac{w}{r^{d+1}}\right)^\frac{1}{2}
+ 2w.
\end{align*}
Note that the assumption $w\leq (\frac{r}{4})^{d+3}$ yields $\left(\frac{w}{r^{d+1}}\right)^\frac{1}{2} \leq r$. Besides, $r \leq 1$ yields $w \leq \left(\frac{r}{4}\right)^{d+3} \leq \frac{r}{2}$.
Finally, Hypothesis \hyperref[hyp:4]{4} gives $\mu(\lambda^r) \leq c_{\ref{Paper2:hyp:normalreach:index}} r$, and we deduce the result:
\begin{align*}
\gammaWassersteinp(\check{\nu}, \check{\mu}_0) 
&\leq \gamma c_{\ref{Paper2:th:estimation:index}}  (c_{\ref{Paper2:hyp:normalreach:index}}r)^\frac{1}{p} 
+ \gamma c_{\ref{Paper2:lem:loccovconsistency:index}} r
+ \gamma c_{\ref{Paper2:prop:stability:index}} r
+ r \\
&\leq \left( \gamma c_{\ref{Paper2:th:estimation:index}}  (c_{\ref{Paper2:hyp:normalreach:index}})^\frac{1}{p} 
+ \gamma c_{\ref{Paper2:lem:loccovconsistency:index}} 
+ \gamma c_{\ref{Paper2:prop:stability:index}} 
+ 1 \right) r^\frac{1}{p} 
\end{align*}
where we used to the weak upper bound $r\leq r^\frac{1}{p}$ on the last line.
\end{proof}

\section{Topological inference with the lifted measure}
\label{Paper2:sec:topoinference}
Based on the results of the last section, we show how the lifted measure $\nucheck$ can be used to infer the homotopy type of $\MMcheck$ or its homology groups.

\subsection{Overview of the method}
\label{Paper2:subsec:overview}
Let us recall the results obtained so far.
Assume that the immersion $\imm\colon \MMo \rightarrow \MM$ and the measure $\muo$ satisfy the Hypotheses \hyperref[hyp:1]{1}, \hyperref[hyp:2]{2} and \hyperref[hyp:3]{3}.
Our goal is to estimate the exact lifted measure $\muchecko$ on $E \times \matrixspace{E}$, since its support is the submanifold $\MMcheck$, which is diffeomorphic to $\MMo$.
To do so, we suppose that we are observing a measure $\nu$ on $E$.
No assumptions are made on $\nu$. 
Our results only depends on the Wasserstein distance $w = \Wasssymbol{p}(\mu, \nu)$, where $\mu = \imm_* \muo$.
Recall that the measure $\muchecko$ is defined as (see Equation \eqref{Paper2:eq:muchecko}):
\begin{align*}
\muchecko = (\imm_* \mu_0)(x_0) \otimes \left\{\delta_{ \frac{1}{d+2} p_{T_{x} \MM} }\right\}.
\end{align*}
To approximate $\muchecko$, pick a parameter $r>0$ and consider the lifted measure $\nucheck$ built on $\nu$ (see Definition \ref{Paper2:def:lifted_measure}):
\begin{align*}
\check \nu = \nu(x) \otimes \left\{\delta_{\loccovnorm{\nu}{x}}\right\}.
\end{align*}
Choose $\gamma > 0$.
Endow the space $E \times \matrixspace{E}$ with the norm $\gammaN{\cdot}$ (see Equation \eqref{Paper2:eq:gammaN}), and consider the Wasserstein distance $\gammaWassersteinp(\cdot, \cdot)$ between measures on $E \times \matrixspace{E}$ (see Equation \eqref{Paper2:eq:gammaWassersteinp}).
We quantify the quality of the approximation by the Wasserstein distance $\gammaWassersteinp(\muchecko, \nucheck)$.
According to Theorem \ref{Paper2:th:estimation}, we have
\begin{align*}
\gammaWassersteinp(\check{\nu}, \check{\mu}_0) 
&\leq \gamma c_{\ref{Paper2:th:estimation:index}}  \mu(\lambda^r)^\frac{1}{p} 
+ \gamma c_{\ref{Paper2:lem:loccovconsistency:index}} r
+ \gamma c_{\ref{Paper2:prop:stability:index}} \left(\frac{w}{r^{d+1}} \right)^\frac{1}{2}
+ 2w
\end{align*}
as long as the parameter $r$ satisfies
\begin{align*}
4\left(\frac{w}{\mini{c_{\ref{Paper2:hyp:muA:index}}, 1}}\right)^\frac{1}{d+1} 
~\leq~ r 
~\leq~ \mini{\frac{1}{2\rho} , 1}. 
\end{align*}
Under Hypothesis \hyperref[hyp:4]{4}, Corollary \ref{Paper2:cor:approximation} gives a weaker form of this result. 
We have
\begin{align*}
\gammaWassersteinp(\check{\nu}, \check{\mu}_0) 
&\leq \big(1 + \gamma c_{\ref{Paper2:cor:approximation:index}}\big) r^\frac{1}{p}
\end{align*}
as long as the parameter $r$ satisfies
\begin{align*}
4\left(\frac{w}{\mini{c_{\ref{Paper2:hyp:muA:index}}, 1}}\right)^\frac{1}{d+3} 
~\leq~ r 
~\leq~ \mini{\frac{1}{2\rho} , r_{\ref{Paper2:hyp:normalreach:index}} , 1}. 
\end{align*}
In the following subsections, we show how these results lead to consistent estimations of $\MMo$ and its homology.
Namely, we can estimate the homotopy type of $\MMcheck$, and hence of $\MMo$, by considering the sublevel sets of the DTM $\DTM{\nucheck,m,\gamma}$ (see Corollary \ref{Paper2:cor:homotopytypeDTMcheck}).
The notation $\DTM{\nucheck,m,\gamma}$ corresponds to the DTM, as defined in Subsect. \ref{subsec:pers_measures}, with measure $\nucheck$, parameter $m$, and seen in the ambient space $\left(E\times\matrixspace{E}, \gammaN{\cdot}\right)$.
Besides, we can estimate the persistent homology of the DTM-filtration $W_\gamma[\muchecko]$ with the filtration $W_\gamma[\nucheck]$ (see Corollary \ref{Paper2:cor:DTMfiltrcheck}).
Here, $W_\gamma[\cdot]$ corresponds to the DTM-filtration in the ambient space $\left(E\times\matrixspace{E}, \gammaN{\cdot}\right)$.

\begin{example}
\label{Paper2:ex:lemniscate_overview}
Let $\MM$ be the lemniscate of Bernoulli of diameter 2.
It is the immersion of a circle $\MMo$.
We observe a 100-sample $X$ of $\MM$ (Figure \ref{Paper2:fig:lemniscate_sample}).
Experimentally, we computed the Hausdorff distance $\Hdist{\MM}{X} \approx 0{.}026$.
Let $\mu$ be the Hausdorff measure on $\MM$ and $\nu$ the empirical measure on $X$.
We choose the parameter $p=2$.
Their Wasserstein distance is approximately $\Wasssymbol{2}(\mu,\nu) \approx 0{.}015$.

\begin{figure}[H]
\centering
\begin{minipage}{.49\linewidth}
\centering
\includegraphics[width=.7\linewidth]{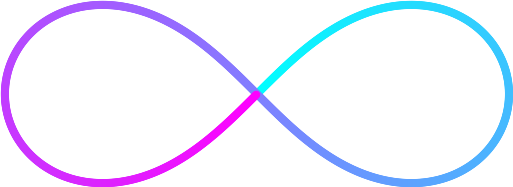}
\end{minipage}
\begin{minipage}{.49\linewidth}
\centering
\includegraphics[width=.7\linewidth]{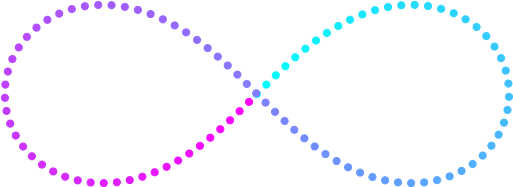}
\end{minipage}
\caption{Left: The lemniscate $\MM$.
Right: The set $X$, a 100-sample of $\MM$.}
\label{Paper2:fig:lemniscate_sample}
\end{figure}
\noindent
For each point $x$ of $X$, we compute the normalized local covariance matrix $\loccovnorm{\nu}{x}$ with parameter $r = 0{.}5$ and $0{.}1$.
This matrix is used as an estimator of the tangent space $T_x \MM$.
In order to observe the quality of this estimation, we represent on Figure \ref{Paper2:fig:lemniscate_parameter_r} (first row) the principal axes of $\loccovnorm{\nu}{x}$ for some $x$. 
On the second row are represented the distances $\frobN{ \loccovnorm{\nu}{x} - \frac{1}{d+2} p_{T_x \MM} }$.
One sees that $r = 0{.}1$ yields a better approximation.
However, the estimation is still biased next to the self-intersection points of $\MM$.

\begin{figure}[H]
\centering
\begin{minipage}{.49\linewidth}
\centering
\includegraphics[width=.75\linewidth]{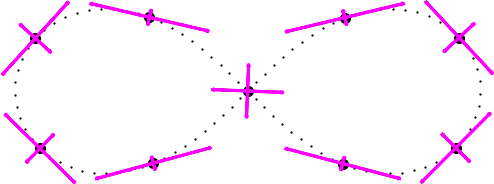}
\\~\\
~~~~~~~~~~~~
\includegraphics[width=.85\linewidth]{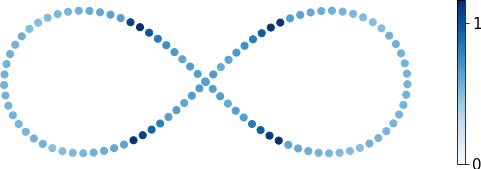}
$r=0{.}5$
\end{minipage}
\begin{minipage}{.49\linewidth}
\centering
\includegraphics[width=.75\linewidth]{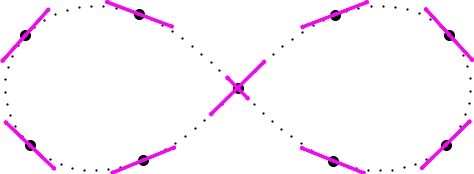}
\\~\\
~~~~~~~~~~~~
\includegraphics[width=.85\linewidth]{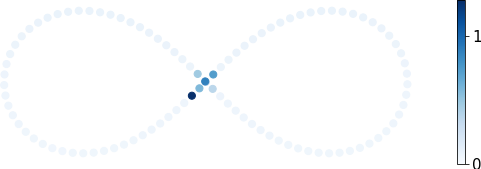}
$r=0{.}1$
\end{minipage}
\caption{First row: The eigenvectors of $\loccovnorm{\nu}{x}$ for some $x \in X$, weighted with their corresponding eigeinvalue.
Second row: color representation of the distances $\frobN{ \loccovnorm{\nu}{x} - \frac{1}{d+2} p_{T_x \MM} }$.}
\label{Paper2:fig:lemniscate_parameter_r}
\end{figure}
\noindent
Now we choose the parameter $\gamma = 2$. 
For $r = 0{.}5$ and $0{.}1$, we consider the lifted measures built on $\nu$, respectively denoted $\nucheck^{0{.}5}$ and $\nucheck^{0{.}1}$.
They are measure on the lift space $\R^2 \times \matrixspace{\R^2}$, which is endowed with the norm $\gammaN{\cdot}$.
We computed the Wasserstein distances:
\begin{center}
 $\Wdistgamma{2}{\muchecko}{\nucheck^{0{.}5}} \approx 0{.}674$ ~~~~~~and~~~~~~ 
 $\Wdistgamma{2}{\muchecko}{\nucheck^{0{.}1}} \approx 0{.}200$.
\end{center}
In comparison, even with a small parameter $r$, the Hausdorff distance between their support is still large: 
\begin{center}
$\Hdist{\MMcheck}{\supp{\nucheck^{0{.}5}}} \approx 1{.}142$ 
~~~~~~and~~~~~~ 
$\Hdist{\MMcheck}{\supp{\nucheck^{0{.}1}}} \approx 1{.}273$.
\end{center}
These sets are represented in Figure \ref{Paper2:fig:lemniscate_sample_check}.
Observe that, at the center of the graphs, the measures $\nucheck^{0{.}5}$ and $\nucheck^{0{.}1}$ deviate from the set $\MMcheck$. 

\begin{figure}[H]
\centering
\begin{minipage}{.32\linewidth}
\centering
\includegraphics[width=.9\linewidth]{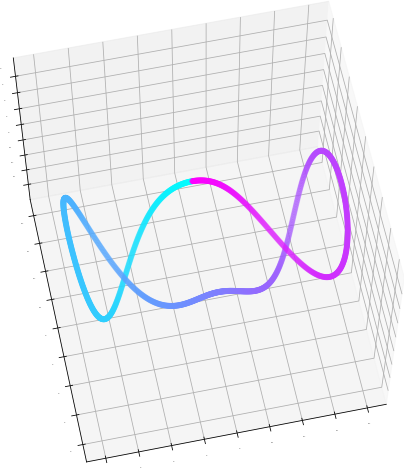}
\end{minipage}
\begin{minipage}{.32\linewidth}
\centering
\includegraphics[width=.9\linewidth]{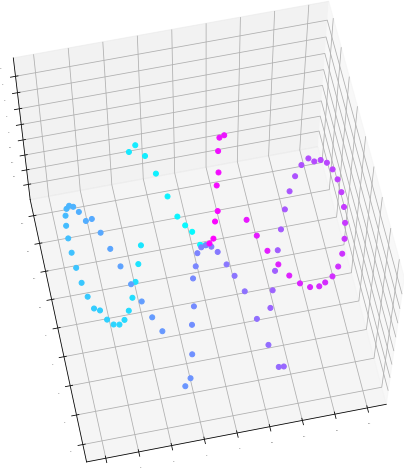}
\end{minipage}
\begin{minipage}{.32\linewidth}
\centering
\includegraphics[width=.9\linewidth]{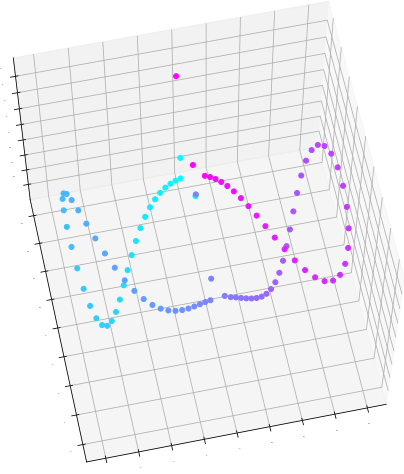}
\end{minipage}
\caption{Left: The lifted lemniscate $\MMcheck$, projected in a 3-dimensional subspace via PCA.
Center: The set $\supp{\nucheck^{0{.}5}}$ projected in the same 3-dimensional subspace.
Right: Same for $\supp{\nucheck^{0{.}1}}$.}
\label{Paper2:fig:lemniscate_sample_check}
\end{figure}
\end{example}

\begin{example}
\label{ex:torus}
Let $u\colon \MMo \rightarrow \MM$ be the figure-8 immersion of the torus in $\R^3$, represented in Figure \ref{Paper2:fig:torus_triangulation}.
It can be parametrized by rotating a lemniscate around an axis, while forming a full twist.
The self-intersection points of this immersion corresponds to the inner circle formed by the center of the lemniscate. 
These are the points $x$ of $\MM$ such that their normal reach $\lambda(x)$ is zero.
\begin{figure}[H]
\centering
\begin{minipage}{.49\linewidth}
\centering
\includegraphics[width=.62\linewidth]{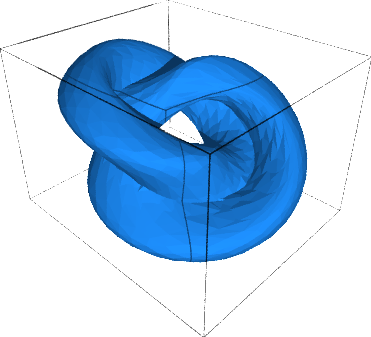}
\end{minipage}
\begin{minipage}{.49\linewidth}
\centering
\includegraphics[width=.62\linewidth]{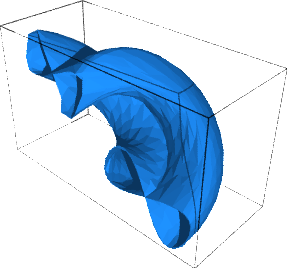}
\end{minipage}
\caption{Left: The immersion $\MM$ of the torus.
Right: A section of $\MM$. One sees the inner lemniscate.}
\label{Paper2:fig:torus_triangulation}
\end{figure}
\noindent
Let $\MMcheck$ be the lift of $\MMo$. It is a submanifold of $\R^3 \times \matrixspace{\R^3} \simeq \R^{12}$.
One cannot embed $\MMcheck$ in $\R^3$ by performing a PCA.
However, we can try to visualize $\MMcheck$ by considering a small section of it. Figure \ref{Paper2:fig:torus_section} represents a subset of $\MMcheck$, projected in a 3-dimensional subspace via PCA. One sees that it does not self-intersect. 
\begin{figure}[H]
\centering
\begin{minipage}{.49\linewidth}
\centering
\includegraphics[width=.6\linewidth]{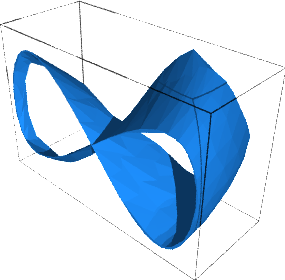}
\end{minipage}
\begin{minipage}{.49\linewidth}
\centering
\includegraphics[width=.65\linewidth]{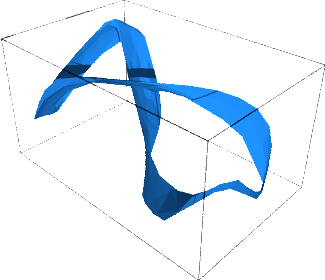}
\end{minipage}
\caption{
Left: A section of $\MM$.
Right: The corresponding section of $\MMcheck$, projected in a 3-dimensional subspace via PCA. Observe that it does not self-intersect.}
\label{Paper2:fig:torus_section}
\end{figure}
\noindent
In order to fit in the context of our study, let $\mu$ be the Hausdorff measure on $\MM$.
We observe a 9000-sample $X$ of $\MM$, and consider its empirical measure $\nu$. The set $X$ is depicted in Figure \ref{Paper2:fig:torus_sample}.
Choose the parameter $p=1$.
We compute the Wasserstein distance $\Wasssymbol{1}(\mu, \nu) \approx 0{.}070$ and the Hausdorff distance $\Hdist{\MM}{X} = 0{.}083$. 
Let $r = 0{.}09$.
In order to observe the estimation of tangent spaces by local covariance matrices $\loccovnorm{\nu}{x}$ with parameter $r$, we represent on Figure \ref{Paper2:fig:torus_sample} the points $x$ such that the distance $\frobN{ \loccovnorm{\nu}{x} - \frac{1}{d+2} p_{T_x \MM} }$ is greater than 1.
Observe that the estimation is biased next to the self-intersection circle of $\MM$.

Last, let us choose the parameter $\gamma = 2$, and consider the lifted measure $\nucheck$. 
We have $\Wasssymbol{1}(\muchecko, \nucheck) \approx 0{.}986$.
In comparison, the Hausdorff distance between their support is large: $\Hdist{\MMcheck}{\supp{\nucheck}} \approx 2{.}188$.

\begin{figure}[H]
\centering
\begin{minipage}{.49\linewidth}
\centering
\includegraphics[width=.75\linewidth]{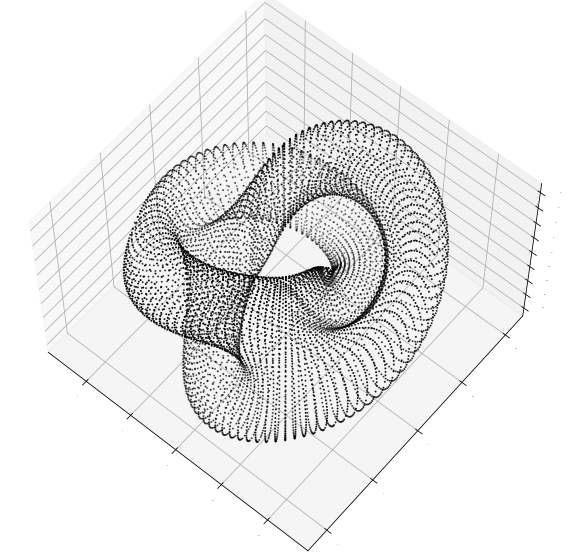}
\end{minipage}
\begin{minipage}{.49\linewidth}
\centering
\includegraphics[width=.75\linewidth]{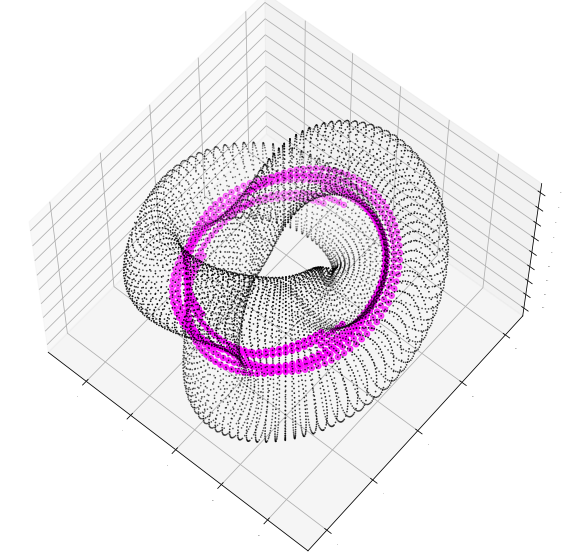}
\end{minipage}
\caption{Left: The set $X$, a sample of $\MM$.
Right: The set $X$, where $x\in X$ is colored in magenta if $\frobN{ \loccovnorm{\nu}{x} - \frac{1}{d+2} p_{T_x \MM} } \geq 1$.}
\label{Paper2:fig:torus_sample}
\end{figure}
\end{example}

\subsection{Homotopy type estimation with the DTM}
\label{Paper2:subsec:homotopy_type_estimation}
In this subsection, we use the DTM, as defined in Subsect. \ref{subsec:pers_measures}, to infer the homotopy type of $\MMcheck$ from the lifted measure $\nucheck$.
We shall use the DTM on $\nucheck$, which lives in the space $E \times \matrixspace{E}$ endowed with the norm $\gammaN{\cdot}$.
It is denoted $d_{\nucheck,m,\gamma}$.

In order to apply Theorem \ref{th:DTM-homotopytype} in our setting, we have to consider geometric quantities associated to the submanifold $\MMcheck$. 
Note that the map $\check \imm$ itself satisfies the Hypotheses \hyperref[hyp:2]{2} and \hyperref[hyp:3]{3}, since the immersion $\imm$ does.
Hence we can consider the following quantities: for every $\gamma > 0$, we denote by 
\begin{itemize}
\item $\reachgamma{\MMcheck}$ the reach of $\MMcheck$ (for the norm $\gammaN{\cdot}$),
\item $\check{\rho}_\gamma$, $\check{L}_{0,\gamma}$, $\check{f}_{\text{min},\gamma}$ and $\check{f}_{\text{max},\gamma}$ the constants given by Hypotheses \hyperref[hyp:2]{2} and \hyperref[hyp:3]{3} applied to $\MMcheck$,
\item $\check c_{\ref{subsec:index}, \gamma} = \check{f}_{\text{min},\gamma} \Jmin \volball{d}$ the constant given by Proposition \ref{Paper2:prop:probabilityboundssqrt} Point \ref{prop:probabilityboundssqrt:point1} applied to $\muchecko$.
\end{itemize}
According to Subsect. \ref{subsec:riemannian}, a sufficient condition for $\MMcheck$ to satisfy $\reachgamma{\MMcheck}>0$ is that it is a $\mathcal{C}^2$-submanifold. 
This would be the case if $\MMo$ and $u$ were $\mathcal{C}^3$. 
Also, we point out that the constant $\check{\rho}_\gamma$ cannot be deduced from $\rho$: the first one can be arbitrary large or small compared to the second one, even with $\gamma$ being fixed. This remark holds for the other constants.


These constants being given, we propose a way to tune the parameters $r$, $\gamma$, $m$ and $t$ in such a way that the $t$-sublevel set $d_{\nucheck,m,\gamma}^t$ of the DTM captures the homotopy type of $\MMcheck$, or equivalently, of $\MMo$.

\begin{corollary}
\label{Paper2:cor:homotopytypeDTMcheck}
Assume that $\MMo$ and $\mu_0$ satisfy Hypotheses \hyperref[hyp:1]{1}, \hyperref[hyp:2]{2}, \hyperref[hyp:3]{3} and \hyperref[hyp:4]{4}.
Let $\nu$ be any probability measure on $E$. Denote $w=\Wasssymbol{2}(\mu, \nu)$. Choose $r>0$, $\gamma>0$ and $m \in (0,1)$ such that
\begin{itemize}
\item $4\left(\frac{w}{\mini{c_{\ref{Paper2:hyp:muA:index}}, 1}}\right)^\frac{1}{d+3} \leq r \leq \mini{\frac{1}{2\rho}, r_{\ref{Paper2:hyp:normalreach:index}} , 1}$
\item $m \leq \frac{\check c_{\ref{subsec:index}, \gamma}}{\left(2\check{\rho}_\gamma\right)^d}$ and
\item $(1 + \gamma c_{\ref{Paper2:cor:approximation:index}}) r^\frac{1}{2} 
 \leq  m^\frac{1}{2} \left(\frac{\reachgamma{\MMcheck}}{9} - \left(\frac{m}{\check c_{\ref{subsec:index}, \gamma}}\right)^{\frac{1}{d}} \right)$.
\end{itemize}
Define $\epsilon$ and choose $t$ as follows:
\begin{align*}
\epsilon = \left(\frac{m}{\check c_{\ref{subsec:index}, \gamma}}\right)^\frac{1}{d} + (1 + \gamma c_{\ref{Paper2:cor:approximation:index}}) \left(\frac{r}{m}\right)^\frac{1}{2}
~~~~~~~~
\text{and}
~~~~~~~~
t \in \left[4 \epsilon, \reachgamma{\MMcheck} - 3\epsilon\right].
\end{align*}
Then the sublevel set of the DTM $\DTM{\nucheck,m,\gamma}^t$ is homotopy equivalent to $\MMo$.
\end{corollary}

\begin{proof}
In order to fit in the context of Theorem \ref{th:DTM-homotopytype}, we have to consider the usual Euclidean norm $\eucN{\cdot}$ on $E \times \matrixspace{E}$. 
It corresponds to the norm $\gammaN{\cdot}$ with $\gamma = 1$. 
For a general parameter $\gamma>0$, consider the dilatation map $i_\gamma\colon E \times \matrixspace{E} \rightarrow E \times \matrixspace{E}$ defined as 
\begin{align*}
i_\gamma\colon (x, A) \mapsto (x, \gamma A).
\end{align*}
A computation shows that, for every probability measures $\alpha, \beta$ on $E \times \matrixspace{E}$, we have
\begin{align*}
\gammaWassersteindeux(\alpha, \beta)
= \Wasssymbol{2}\big((i_\gamma)_* \alpha, (i_\gamma)_* \beta \big),
\end{align*}
where $\Wasssymbol{2}$ denotes the 2-Wasserstein distance on $E \times \matrixspace{E}$ endowed with the usual Euclidean norm $\eucN{\cdot}$.
Corollary \ref{Paper2:cor:approximation} then reads 
\begin{align*}
\Wasssymbol{2}\big( (i_\gamma)_* \muchecko, (i_\gamma)_* \nucheck \big) \leq (1 + \gamma c_{\ref{Paper2:cor:approximation:index}}) r^\frac{1}{2},
\end{align*}
where $(i_\gamma)_* \muchecko$ and $(i_\gamma)_* \nucheck$ are the push-forwards of $\muchecko$ and $\nucheck$ by the map $i_\gamma$.
Besides, consider the set 
\begin{align*}
\MMcheck_\gamma 
= i_\gamma(\MMcheck)
= \{ (x, \gamma A) \mid (x,A) \in \MMcheck \}.
\end{align*}
It is clear that
\begin{align*}
\reachgamma{\MMcheck} = \reach{\MMcheck_\gamma}, 
\end{align*}
where we recall that $\reachgamma{\MMcheck}$ is the reach of $\MMcheck$ with respect to the norm $\gammaN{\cdot}$, and $\reach{\MMcheck_\gamma}$ is the reach of $\MMcheck_\gamma$ with respect to the usual norm $\eucN{\cdot}$ on $E \times \matrixspace{E}$.
Finally, consider the DTM $\DTM{(i_\gamma)_*\nucheck, m}$ with respect to the usual Euclidean norm.
Observe that, for every $t \geq 0$, the sublevel sets of the DTM $\DTM{(i_\gamma)_*\nucheck,m}$ and $\DTM{\nucheck,m,\gamma}$ are linked via
\begin{align*}
\DTM{\nucheck,m}^t
= i_\gamma\left( \DTM{\nucheck,m,\gamma}^t \right).
\end{align*}
In particular, they share the same homotopy type.
Now we obtain the result as a consequence of Theorem \ref{th:DTM-homotopytype} applied to the measures $(i_\gamma)_* \muchecko$ and $(i_\gamma)_* \nucheck$.
Let us verify that the assumptions of the theorem are satisfied.
Our assumption about $m$ ensures that 
\[\left(\frac{m}{\check c_{\ref{subsec:index}, \gamma}}\right)^\frac{1}{d} \leq \frac{1}{2\check{\rho}_\gamma},\] 
hence by Proposition \ref{Paper2:prop:probabilityboundssqrt} Point \ref{prop:probabilityboundssqrt:point1} we get $\muchecko(\openball{x}{r}) \geq \check c_{\ref{subsec:index}, \gamma} r^d$ for all $x \in \supp{\muchecko}$ and $r < \left(\frac{m}{\check c_{\ref{subsec:index}, \gamma}}\right)^\frac{1}{d}$. 
Moreover, the assumption about $(1 + \gamma c_{\ref{Paper2:cor:approximation:index}})r^{\frac{1}{2}}$ ensures that 
\[\Wasssymbol{2}\big( (i_\gamma)_* \muchecko, (i_\gamma)_* \nucheck \big) 
 \leq  m^\frac{1}{2} \left(\frac{\reachgamma{\MMcheck}}{9} - \left(\frac{m}{\check c_{\ref{subsec:index}, \gamma}}\right)^{\frac{1}{d}} \right)\]
is satisfied, since $\Wasssymbol{2}\big( (i_\gamma)_* \muchecko, (i_\gamma)_* \nucheck \big) \leq (1 + \gamma c_{\ref{Paper2:cor:approximation:index}}) r^\frac{1}{2}$ by Corollary \ref{Paper2:cor:approximation}.
\end{proof}

\begin{example}
\label{ex:lemniscate_DTM}
Let $\MM$ be the lemniscate of Bernoulli, as in Example \ref{Paper2:ex:lemniscate_overview}.
Suppose that $\mu$ is the uniform distribution on $\MM$, and $\nu$ is the empirical measure on a 500-sample of $\MM$.
We choose the parameters $\gamma=2$, $r=0{.}03$ and $m=0{.}01$.
Let $\nucheck$ be the lifted measure associated to $\nu$.
Figure \ref{Paper2:fig:DTM_lemniscate} represents set the $\supp{\nucheck}$, and the values of the DTM $\DTM{\nucheck,m,\gamma}$ on it.
Observe that the anomalous points, i.e., points for which the local covariance matrix is not well estimated, have large DTM values.

\begin{figure}[H]
\centering
\begin{minipage}{.49\linewidth}
\centering
\includegraphics[width=.6\linewidth]{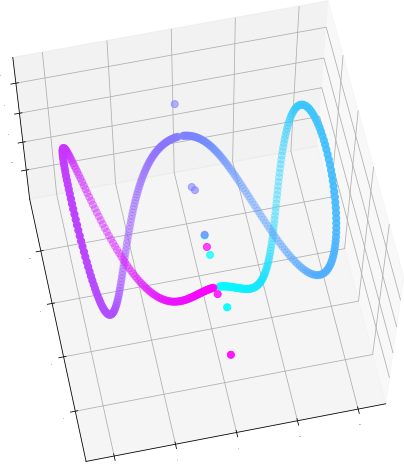}
\end{minipage}
\begin{minipage}{.49\linewidth}
\centering
\includegraphics[width=.8\linewidth]{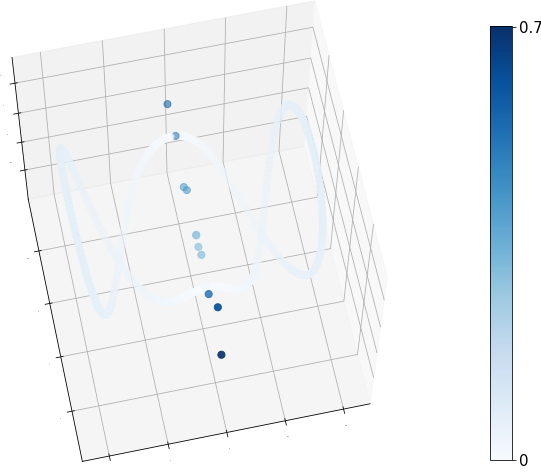}
\end{minipage}
\caption{Left: The set $\supp{\nucheck} \subset \mathbb{R}^2 \times\matrixspace{\R^2}$, projected in a 3-dimensional subspace via PCA.
Right: The set $\supp{\nucheck}$ with colors indicating the value of the DTM $\DTM{\nucheck, m, \gamma}$.}
\label{Paper2:fig:DTM_lemniscate}
\end{figure}
\end{example}

\subsection{Persistent homology with DTM-filtrations}
\label{Paper2:subsec:DTM_filtrations}
In this subsection, we aim to estimate the DTM-filtration of $\muchecko$, as defined in Subsect. \ref{subsec:pers_measures}, from $\nu$.
We shall use the DTM-filtration on $\nucheck$, denoted $W_\gamma[\nucheck]$, with respect to the ambient norm $\gammaN{\cdot}$ on $E \times \matrixspace{E}$. 
We use the notations $\check{\rho}_\gamma$ and $\check c_{\ref{subsec:index}, \gamma}$ of the previous subsection.

\begin{corollary}
\label{Paper2:cor:DTMfiltrcheck}
Assume that $\MMo$ and $\mu_0$ satisfy Hypotheses \hyperref[hyp:1]{1}, \hyperref[hyp:2]{2}, \hyperref[hyp:3]{3} and \hyperref[hyp:4]{4}.
Let $\nu$ be any probability measure. Denote $\Wasssymbol{2}(\mu, \nu) = w$. 
Choose $r>0$, $\gamma>0$ and $m \in (0,1)$ such that 
\begin{itemize}
\item $4\left(\frac{w}{\mini{c_{\ref{Paper2:hyp:muA:index}}, 1}}\right)^\frac{1}{d+3} \leq r \leq \mini{\frac{1}{2\rho} , r_{\ref{Paper2:hyp:normalreach:index}} , 1}$,
\item $m \leq \frac{\check c_{\ref{subsec:index}, \gamma}}{\left(2\check{\rho}_\gamma\right)^d}$,
\item $\big(1 + \gamma c_{\ref{Paper2:cor:approximation:index}}\big) r^\frac{1}{2} \leq \frac{1}{4}$.
\end{itemize}
Then we have a bound on the interleaving distance between the DTM-filtrations:
\begin{align*}
\idist{W_\gamma[\muchecko]}{W_\gamma[\nucheck]}
\leq \check c_{\ref{Paper2:cor:DTM-filtrations:index}, \gamma} (1 + \gamma c_{\ref{Paper2:cor:approximation:index}})^\frac{1}{2} m^{-\frac{1}{2}}r^\frac{1}{4} + \check c_{\ref{Paper2:cor:DTM-filtrations:index}, \gamma}' m^\frac{1}{d},
\end{align*}
where $\check c_{\ref{Paper2:cor:DTM-filtrations:index}, \gamma} = 8\mathrm{diam}(\MM) + 8\gamma + 5$
and $\check c_{\ref{Paper2:cor:DTM-filtrations:index}, \gamma}' = 2\left(\check c_{\ref{subsec:index}, \gamma}\right)^{-\frac{1}{d}}$.
\end{corollary}

\begin{proof}
As in the proof of Corollary \ref{Paper2:cor:homotopytypeDTMcheck}, let $i_\gamma$ be the map $i_\gamma\colon (x, A) \mapsto (x, \gamma A)$.
Let $W[\cdot]$ denotes the DTM-filtration on $\nucheck$ with respect to the usual Euclidean norm. That is, the filtration $W[\cdot]$ corresponds to $W_\gamma[\cdot]$ with $\gamma=1$.
A computation shows that the filtration $W[(i_\gamma)_*\nucheck]$ and $W_\gamma[\nucheck]$ are linked via
\begin{align*}
W[(i_\gamma)_*\nucheck]
= i_\gamma\left( W_\gamma[\nucheck] \right).
\end{align*}
Now let $\check w = \Wasssymbol{2}((i^\gamma)_*\muchecko, (i^\gamma)_*\nucheck)$.
We have $\check w = \Wasssymbol{2,\gamma}(\muchecko, \nucheck)$, hence Corollary \ref{Paper2:cor:approximation} gives
\begin{equation}
\label{Paper2:eq:DTMfiltrcheck_1}
\check w  \leq \big(1 + \gamma c_{\ref{Paper2:cor:approximation:index}}\big) r^\frac{1}{2}.
\end{equation}
Moreover, we can apply Corollary \ref{Paper2:cor:DTM-filtrations} to $\mu=(i^\gamma)_*\muchecko$ and $\nu=(i^\gamma)_*\nucheck$ to get
\begin{equation}
\label{Paper2:eq:DTMfiltrcheck_2}
\idist{W[(i^\gamma)_*\muchecko]}{W[(i^\gamma)_*\nucheck]}
\leq \check c_{\ref{Paper2:cor:DTM-filtrations:index}, \gamma} \left(\frac{\check w}{m}\right)^\frac{1}{2} + \check c_{\ref{Paper2:cor:DTM-filtrations:index}, \gamma}' m^\frac{1}{d},
\end{equation}
where $\check c_{\ref{Paper2:cor:DTM-filtrations:index}, \gamma} = \left(8\mathrm{diam}(\MMcheck) + 5\right)$
and $\check c_{\ref{Paper2:cor:DTM-filtrations:index}, \gamma}' = 2\left(\check c_{\ref{subsec:index}, \gamma}\right)^{-\frac{1}{d}}$.
Note that
\begin{align*}
\mathrm{diam}(\MMcheck) \leq \left(\mathrm{diam}(\MM)^2 + \gamma^2 \left(2\frac{1}{2}\right)^2 \right)^\frac{1}{2} \leq \mathrm{diam}(\MM) + \gamma
\end{align*}
since the matrices $\frac{1}{d+2}p_{T_x \MM}$ have norm $\frobN{\frac{1}{d+2}p_{T_x \MM}} = \frac{\sqrt{d}}{d+2} \leq \frac{1}{2}$.
Our assumption $m \leq \frac{\check c_{\ref{subsec:index}, \gamma}}{\left(2\check{\rho}_\gamma\right)^d}$ ensures that the condition $\muchecko(\openball{x}{r}) \geq \check c_{\ref{subsec:index}, \gamma} r^d$ of the theorem is satisfied.
Similarly, the assumption $\big(1 + \gamma c_{\ref{Paper2:cor:approximation:index}}\big) r^\frac{1}{2} \leq \frac{1}{4}$ yields $\check w \leq \frac{1}{4}$.

Combining Equations \eqref{Paper2:eq:DTMfiltrcheck_1} and \eqref{Paper2:eq:DTMfiltrcheck_2} we get
\begin{align*}
\idist{W[(i^\gamma)_*\muchecko]}{W[(i^\gamma)_*\nucheck]}
 \leq \check c_{\ref{Paper2:cor:DTM-filtrations:index}, \gamma} \big(1 + \gamma c_{\ref{Paper2:cor:approximation:index}}\big)^\frac{1}{2} m^{-\frac{1}{2}}r^\frac{1}{4} + \check c_{\ref{Paper2:cor:DTM-filtrations:index}, \gamma}' m^\frac{1}{d}.
\end{align*}
Now, by using the definition of an interleaving of filtrations, one proves that 
\[\idist{W_\gamma[\muchecko]}{W_\gamma[\nucheck]} = \idist{W[(i^\gamma)_*\muchecko]}{W[(i^\gamma)_*\nucheck]},\] 
and we obtain the result.
\end{proof}

\begin{example}
\label{ex:olympics}
Say that $\mu$ is the uniform measure on the union of five intersecting circles of radius 1.
We observe $\nu$, the empirical measure on the point cloud $X$ drawn in Figure \ref{Paper2:fig:DTMF2}. It consists of 300 points per circle, and 100 anomalous points.
Let $p=1$.
Experimentally, we have $\Wdist{1}{\mu}{\nu} \approx 0{.}044$.

\begin{figure}[H]
\centering
\begin{minipage}{.49\linewidth}
\centering
\includegraphics[width=.75\linewidth]{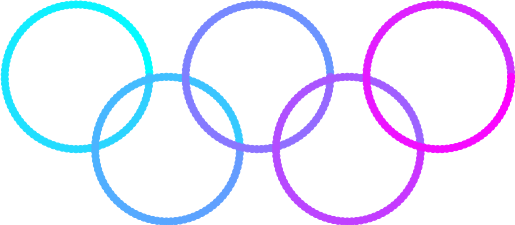}
\end{minipage}
\begin{minipage}{.49\linewidth}
\centering
\includegraphics[width=.8\linewidth]{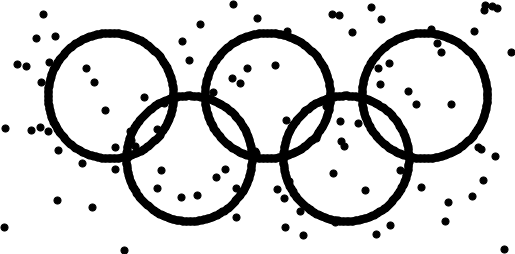}
\end{minipage}
\caption{Left: the set $\MM = \supp{\mu}$.
Right: The set $X = \supp{\nu}$.}
\label{Paper2:fig:DTMF2}
\end{figure}
\noindent
Let $\gamma = 1$. 
Observe that the barcodes of the DTM-filtration of the exact lifted measure $W_\gamma[\muchecko]$, represented in Figure \ref{Paper2:fig:DTMF4}, reveal the homology of the disjoint union of five circles---which is the set $\MMo$.
Only bars of length larger than $0{.}1$ are displayed.
We consider the construction of the lifted $\nucheck$ with parameter $r = 0{.}03$, and the DTM-filtration with $m = 0{.}01$.
The barcodes of the DTM-filtration $W_\gamma[\nucheck]$ are close to the barcodes of $W_\gamma[\muchecko]$.
To compare, we also plot the persistence barcodes of the usual \v{C}ech filtration on $\supp{\nucheck}$. Observe that the five connected components do not appear clearly anymore.

\begin{figure}[H]
\centering
\begin{minipage}{.15\linewidth}
\centering
$W_\gamma[\muchecko]$
\vspace{.9cm}
\\~\\
$W_\gamma[\nucheck]$
\vspace{.9cm}
\\~\\
\v{C}ech filtration
\vspace{.1cm}
\end{minipage}
\begin{minipage}{.40\linewidth}
\centering
\includegraphics[width=.99\linewidth]{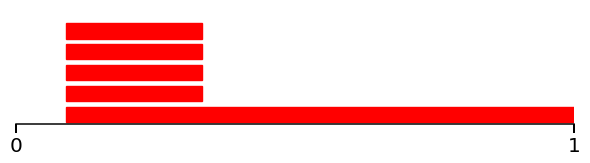}
\\~\\
\includegraphics[width=.99\linewidth]{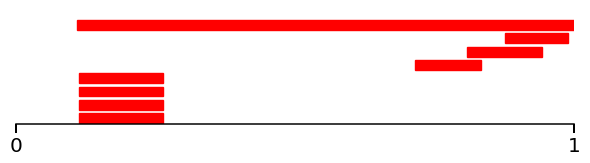}
\\~\\
\includegraphics[width=.99\linewidth]{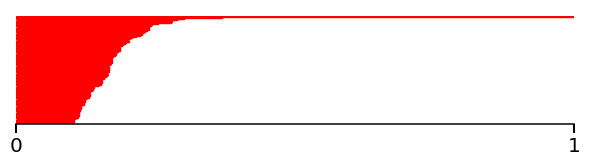}
\end{minipage}
\begin{minipage}{.40\linewidth}
\centering
\includegraphics[width=.99\linewidth]{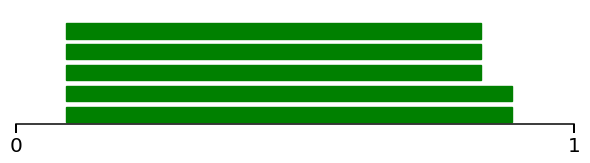}
\\~\\
\includegraphics[width=.99\linewidth]{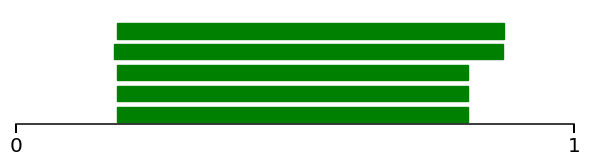}
\\~\\
\includegraphics[width=.99\linewidth]{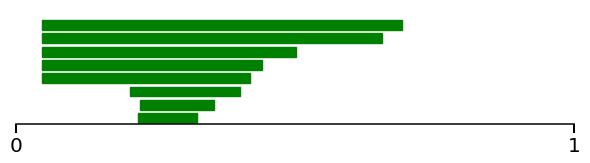}
\end{minipage}
\caption{First row: Persistence barcode of the 0- and 1-homology of the DTM-filtration on $\muchecko$.
Second row: Same for $\nucheck$.
Third row: Persistence barcodes of the usual \v{C}ech filtration on $\supp{\nucheck}$.}
\label{Paper2:fig:DTMF4}
\end{figure}
\noindent
At this point, we can propose a clustering procedure based on $W_\gamma[\nucheck]$. 
First, select a $t \in [0, +\infty)$. Then, extract the connected components of the set $W^t_\gamma[\nucheck]$ of the DTM-filtration.
We show in Figure \ref{Paper2:fig:DTMF5} the components we obtain for several values of $t$. We see that there exists a value for which the five circles a well clustered ($t=0.2$). Besides, observe that small values of $t$ (resp. large) may lead to more connected components than wanted (resp. less).

\begin{figure}[H]
\centering
\begin{minipage}{.32\linewidth}
\centering
\includegraphics[width=.95\linewidth]{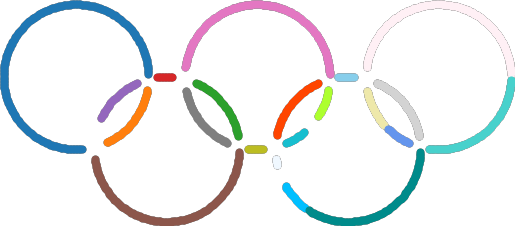}
\\
$t=0.13$
\end{minipage}
\begin{minipage}{.32\linewidth}
\centering
\includegraphics[width=.95\linewidth]{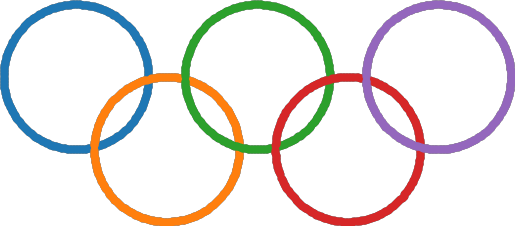}
\\
$t=0.2$
\end{minipage}
\begin{minipage}{.32\linewidth}
\centering
\includegraphics[width=.95\linewidth]{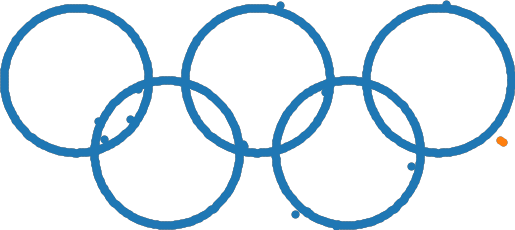}
\\
$t=0.4$
\end{minipage}
\caption{Components obtained by the clustering procedure, where each color correspond to a cluster. The clusterings consist respectively in $21$, $5$ and $2$ clusters.}
\label{Paper2:fig:DTMF5}
\end{figure}
\end{example}
\noindent
From an algorithmic viewpoint, this clustering can be obtained by computing the connected components of the nerve of the set $W^t_\gamma[\nucheck]$ or, equivalently, the connected components of its underlying graph $G$. As we see from the definition of the DTM-filtration (Equation \eqref{eq:def_DTM_filtration}), the vertices of $G$ are the points $\check x \in \supp{\nucheck}$ with DTM value $\DTM{\nucheck, m, \gamma}(\check x)$ not greater than $t$, and where an edge $[\check x,\check y]$ is added if $$\eucN{x-y} + \DTM{\nucheck, m, \gamma}(\check x) + \DTM{\nucheck, m, \gamma}(\check y)\leq t.$$

\begin{example}
\label{ex:Klein}
Consider the immersion of the Klein bottle in $\R^3$ represented in Figure \ref{fig:exampledim2_1}.
Note that the self-intersection of this immersion forms a circle.
We consider a $22'092$-sample $X$ of it.

\begin{figure}[H]
\centering
\includegraphics[width=.4\linewidth]{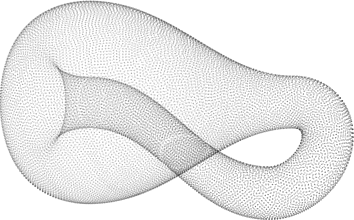}
\caption{Sample of the Klein bottle immersed in $\R^3$.}
\label{fig:exampledim2_1}
\end{figure}
\noindent
Let $\nu$ be the empirical measure on this point cloud. We build the lifted measure $\nucheck$ with parameters $r=0.08$ and $\gamma=3$, and we consider the DTM-filtration $W_\gamma[\nucheck]$ with parameter $m = 0.0001$.
The barcodes of this filtration are depicted in Figure \ref{fig:exampledim2_2}, with coefficients in two finite field: $\Z/2\Z$ and $\Z/3\Z$. We also plot the barcodes of the usual \v{C}ech filtration of $X$ in $\R^3$. Only bars of length larger than 0.4 are displayed.

\begin{figure}[H]
\centering
\begin{minipage}{.15\linewidth}
\centering
$W[(i^\gamma)_*\nucheck]$
\vspace{1.2cm}
\\~\\
\v{C}ech filtration
\\
~
\end{minipage}
\begin{minipage}{.40\linewidth}
\centering
\includegraphics[width=.99\linewidth]{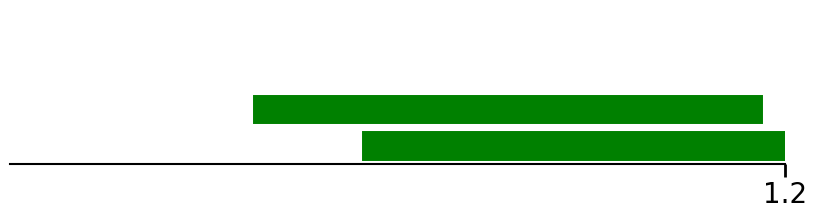}
\\~\\
\includegraphics[width=.99\linewidth]{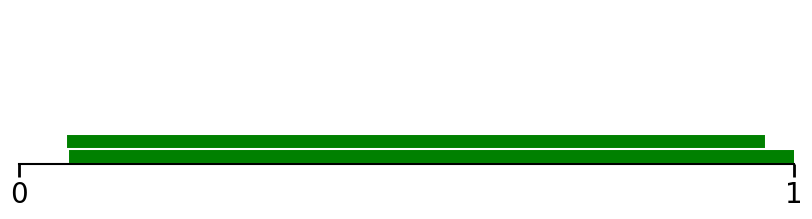}
\\
over $\Z/2\Z$
\end{minipage}
\begin{minipage}{.40\linewidth}
\centering
\includegraphics[width=.99\linewidth]{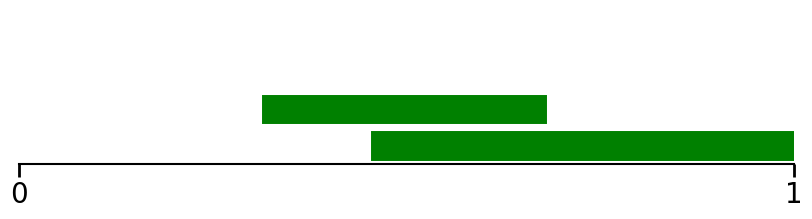}
\\~\\
\includegraphics[width=.99\linewidth]{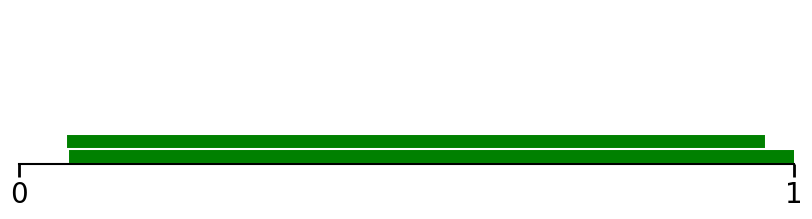}
\\
over $\Z/3\Z$
\end{minipage}
\caption{First row: Persistence barcode of the 1-homology of the DTM-filtration on $\nucheck$.
Second row: Persistence barcodes of the usual \v{C}ech filtration on $X$.}
\label{fig:exampledim2_2}
\end{figure}
\noindent
We see that the barcodes of $W_\gamma[\nucheck]$ over $\Z/2\Z$ and $\Z/3\Z$ differ. This is a consequence of the homology of the Klein bottle itself, which depends on the field of coefficients. Over $\Z/2\Z$, its first homology group is $(\Z/2\Z)^2$, while over $\Z/3\Z$ it is $\Z/3\Z$.
These homology groups can be read on the right part of the barcodes.
In comparison, the barcodes of the usual \v{C}ech filtration are the same.
\end{example}

\begin{example}
\label{ex:cube}
As a last example, we consider two datasets that do not satisfy the hypotheses we studied. Hence the present paper does not provide theoretical guarantees, although our method gives interesting results.
The first point cloud, denoted $X_1$, is sampled on the unit cube of $\R^3$. It is made up of $6\times 2000$ points. It can be seen as the immersion of six squares. Note that this immersion does not satisfy the model considered in this paper since the squares are manifolds with boundaries.
The second point cloud, $X_2$, is sampled on the union of three spheres and a circle. It is made up of $4 \times 2000$ points. This subset can be seen as the immersion of the disjoint union of three spheres and a circle. Again, this does not fit in our model, since these manifolds have different dimensions.
These point clouds are represented in Figure \ref{fig:cube1}.

\begin{figure}[H]
\centering
\begin{minipage}{.49\linewidth}
\centering
\includegraphics[width=.7\linewidth]{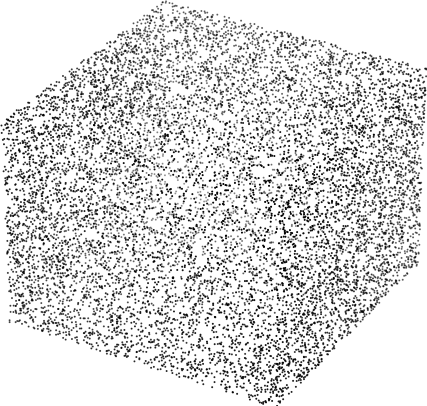}
\end{minipage}
\begin{minipage}{.49\linewidth}
\centering
\includegraphics[width=.8\linewidth]{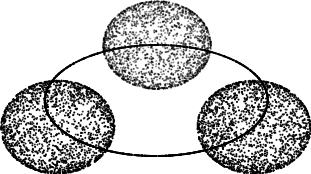}
\end{minipage}
\caption{Left: $X_1$ is a $12'000$-sample of the cube. 
Right: $X_2$ is a $8'000$-sample of the immersion of three spheres and a circle.}
\label{fig:cube1}
\end{figure}
\noindent
We represent on Figure \ref{fig:cube2} the $0$-persistence diagrams of the DTM-filtrations of their lifted measures. We choose the parameters $r = 0.05$, $\gamma = 2$, $m =0.01$ for $X_1$, and $r = 0.2$, $\gamma = 2$, $m =0.01$ for $X_2$.
Observe that the first barcode contains six long bars, corresponding to the six faces of the cube. Similarly, the second barcode contains four long bars, corresponding to the three spheres and the circle.

\begin{figure}[H]
\centering
\begin{minipage}{.49\linewidth}
\centering
\includegraphics[width=.99\linewidth]{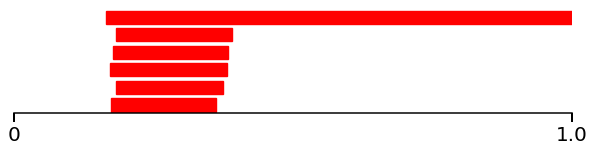}
\end{minipage}
\begin{minipage}{.49\linewidth}
\centering
\includegraphics[width=.99\linewidth]{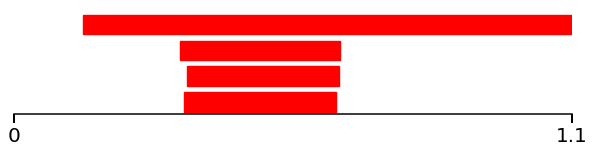}
\end{minipage}
\caption{Left: Persistence barcode of the 0-homology of the DTM-filtration of the lifted measure built from $X_1$.
Right: Same for $X_2$.}
\label{fig:cube2}
\end{figure}

\noindent
In Figure \ref{fig:cube3}, we apply the clustering procedure described in Example \ref{ex:olympics}.
For $t = 0.35$, $X_1$ is clustered into 83 connected components. We see that there are six main connected components, represented by the faces, and a few outliers.
Similarly, we chose $t=0.6$ for $X_2$, and obtained 20 connected components, four of them representing the four underlying objects.

\begin{figure}[H]
\centering
\begin{minipage}{.49\linewidth}
\centering
\includegraphics[width=.7\linewidth]{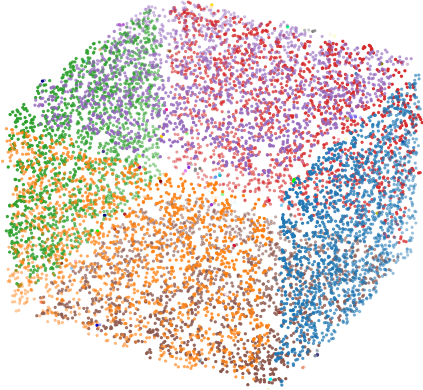}
\end{minipage}
\begin{minipage}{.49\linewidth}
\centering
\includegraphics[width=.8\linewidth]{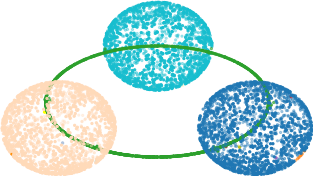}
\end{minipage}
\caption{Left: the clustering procedure applied to $X_1$ at $t = 0.35$.
Right: Same for $X_2$ at $t = 0.6$.}
\label{fig:cube3}
\end{figure}
\end{example}

\section{Conclusion}
In this paper we described a method to estimate the tangent bundle of a manifold $\MMo$ immersed in a Euclidean space, based on a sample of its image. This estimation is stable in Wasserstein distance.
Using the DTM, we are able to estimate the homotopy type of $\MMo$.
Moreover, via the DTM-filtrations, we can define a filtration of the space $\R^n\times \matrixspace{\R^n}$ whose persistence module contains information about the homology of $\MMo$.

The robust estimation of tangent bundles of manifolds opens the way to the estimation of other topological invariants than homology groups---such as characteristic classes---a problem that will be addressed in further works.

Also, as we pointed out in Subsect. \ref{Paper2:subsec:homotopy_type_estimation}, it would be interesting to understand the geometric quantities associated to the lifted manifold $\MMcheck$ (such as $\check{\rho}_\gamma$, $\check{L}_{0,\gamma}$, $\check{f}_{\text{min},\gamma}$ and $\check{f}_{\text{max},\gamma}$) as a function of those associated with the initial manifold $\MMo$ ($\rho$, $L_0$, $\fmin$ and $\fmax$).

\paragraph{Acknowledgements.}
I would like to thank Frédéric Chazal, Marc Glisse and Théo Lacombe for fruitful discussions and corrections. I also thank the anonymous reviewers for their precious corrections and suggestions.

\appendix
\section{Notations}
\label{app:notations}
We adopt the following notations:
\begin{itemize}
\itemsep.1cm
\item $n, d > 0$ are integers.
\item If $x, y \in \R$, $\mini{x,y}$ is the minimum of $x$ and $y$.
\item $I$ is the interval $[0, +\infty)$ or $[0,T]$ for $T \geq 0$.
\item $E = \R^n$ is the Euclidean space, $\matrixspace{E}$ the vector space of $n \times n$ matrices, $\Grass{d}{E}$ the Grassmannian of $d$-planes in $E$.
\item $A$ is a subset of $E$, $\med{A}$ denotes its medial axis, $\reach{A}$ its reach. For every $x\in E$, $\dist{x}{A}$ is the distance from $x$ to $A$.
\item For $x,y \in E$, $x \bot y$ denotes the orthogonality of $x$ and $y$
\item If $x,y\in E$, $x \otimes y = x \transp{y} \in \matrixspace{E}$ is the outer product, and $\outerP{x} = x \otimes x$.
\item $\eucN{\cdot}$ is the Euclidean norm on $E$ and $\eucP{\cdot}{\cdot}$ the corresponding inner product, $\frobN{\cdot}$ the Frobenius norm on $\matrixspace{E}$, $\gammaN{\cdot}$ the $\gamma$-norm on $E \times \matrixspace{E}$ (defined in Subsect. \ref{Paper2:subsec:defgammaN}).
\item $\Wasssymbol{p}$ is the $p$-Wasserstein distance between measures on $E$, $\Wasssymbol{p,\gamma}$ is the $(p,\gamma)$-Wasserstein distance between measures on $E \times \matrixspace{E}$ (defined in Subsect. \ref{Paper2:subsec:defgammaN}).
\item $\HH^d$ is the $d$-dimensional Hausdorff measure on $E$ or on a subspace $T \subset E$ (not renormalized).
\item If $\mu$ is a measure of positive finite mass, $|\mu|$ denotes its mass, $\overline{\mu} = \frac{1}{|\mu|} \mu$ is the associated probability measure, $\mucheck$ denotes the associated lifted measure (introduced in Subsect. \ref{Paper2:subsec:defgammaN}).
\item $1_A$ is the indicator function of a measurable set $A$.
\item If $T$ is a subspace of $E$, $p_T$ denotes the orthogonal projection matrix on $T$.
\item $\openball{x}{r}$ and $\closedball{x}{r}$ denote the open and closed balls of $E$, $\sphere{x}{r}$ the sphere. $\volball{d}$ and $\volsphere{d-1}$ denote $\HH^d(\openball{0}{1})$ and $\HH^{d-1}(\sphere{0}{1})$ (note that $\volsphere{d-1} = d\volball{d}$).
\item $\MMo$ is a Riemannian manifold, and $\openballM{x}{r}{\MMo}$ and $\closedballM{x}{r}{\MMo}$ denote the open and closed geodesics balls.
For $x_0, y_0 \in \MMo$, $\geoD{x_0}{y_0}{\MMo}$ denotes the geodesic distance.
\item If $T$ is a subspace of $E$, $\openballM{x}{r}{T}$ and $\closedballM{x}{r}{T}$ denote the open and closed balls of $T$ for the Euclidean distance.
\item If $f$ is a map with values in $\R$ and $t \in \R$, $f^t$ denotes the sublevel set $f^t = f^{-1}\left((-\infty, t]\right)$.
\end{itemize}

\section{Table of constants}
\label{app:constants}
In the following table, each constant is preceded by the result where it appeared first.
If a constant is defined from the others, it is indicated here.
The indices are arbitrary and only reflect the order of apparition of each result.

\renewcommand{\arraystretch}{1.5}
\begin{center}
\begin{tabular}{p{1cm}p{4.5cm}p{8cm}}
\hline
Index & Result & Constant \\
\hline
\begin{minipage}{5in} \begin{enumerate}\setcounter{enumi}{0} \item 
\label{Paper2:hyp:regularityimm} \label{Paper2:cor:DTM-filtrations:index}
\end{enumerate} \end{minipage}
&Corollary \ref{Paper2:cor:DTM-filtrations}
&$a$, ~~~$c_{\ref{Paper2:cor:DTM-filtrations:index}} = 8\mathrm{diam}(\supp{\mu}) + 5$, ~~~$c_{\ref{Paper2:cor:DTM-filtrations:index}}' = 2a^{-\frac{1}{d}}$
\\

\begin{minipage}{5in} \begin{enumerate}\setcounter{enumi}{1} \item 
\label{Paper2:hyp:curvature} 
\end{enumerate} \end{minipage}
&Hypothesis \hyperref[hyp:2]{2}
&$\rho$
\\

\begin{minipage}{5in} \begin{enumerate}\setcounter{enumi}{2} \item 
\label{Paper2:hyp:mu:index}\label{Paper2:hyp:mu} 
\end{enumerate} \end{minipage}
&Hypothesis \hyperref[hyp:3]{3}
&$L_0$, ~~~$f_{\mathrm{min}}$, ~~~$f_{\mathrm{max}}$
\\

\begin{minipage}{5in} \begin{enumerate}\setcounter{enumi}{3} \item 
\label{Paper2:hyp:normalreach:index}\label{Paper2:hyp:normalreach}
\end{enumerate} \end{minipage}
&Hypothesis \hyperref[hyp:4]{4}
&$c_{\ref{Paper2:hyp:normalreach:index}}$, ~~~$r_{\ref{Paper2:hyp:normalreach:index}}$
\\

\begin{minipage}{5in} \begin{enumerate}\setcounter{enumi}{4} \item 
\label{Paper2:lem:comparisoneucgeod2:index}
\end{enumerate} \end{minipage}
&Lemma \ref{Paper2:lem:comparisoneucgeod2}
&$c_{\ref{Paper2:lem:comparisoneucgeod2:index}}\colon t \mapsto \frac{1}{t}\left(1 - \sqrt{1 - 2t}\right)$
\\

\begin{minipage}{5in} \begin{enumerate}\setcounter{enumi}{5} \item 

\end{enumerate} \end{minipage}
&Lemma \ref{Paper2:lem:regularityexp}
&$\Jmin = (\frac{23}{24})^d$, ~~ $\Jmax = (\frac{5}{4})^d$
\\


\begin{minipage}{5in} \begin{enumerate}\setcounter{enumi}{6} \item 
\label{Paper2:hyp:muB} \label{Paper2:lem:densityg:index}
\end{enumerate} \end{minipage}
&Lemma \ref{Paper2:lem:densityg}
&$c_{\ref{Paper2:lem:densityg:index}} = 4 L_0 \Jmax + \frac{d}{2}\rho \fmax$
\\

\begin{minipage}{5in} \begin{enumerate}\setcounter{enumi}{7} \item 
\label{Paper2:hyp:muBsqrt} \label{Paper2:prop:probabilitybounds:index}
\end{enumerate} \end{minipage}
&Proposition \ref{Paper2:prop:probabilitybounds}
&$c_{\ref{Paper2:prop:probabilitybounds:index}} = c_{\ref{Paper2:lem:densityg:index}} + \fmax \Jmax  d 2^{d} \rho$
\\

\begin{minipage}{5in} \begin{enumerate}\setcounter{enumi}{8} \item 
\label{Paper2:hyp:muA:index}
\end{enumerate} \end{minipage}
&Proposition \ref{Paper2:prop:probabilitybounds}, Hypothesis \hyperref[hyp:5]{5}
&$c_{\ref{Paper2:hyp:muA:index}} = \fmin \Jmin \volball{d}$
\\

\begin{minipage}{5in} \begin{enumerate}\setcounter{enumi}{9} \item 
\label{Paper2:hyp:muB:index}
\end{enumerate} \end{minipage}
&Proposition \ref{Paper2:prop:probabilitybounds}, Hypothesis \hyperref[hyp:6]{6}
&$c_{\ref{Paper2:hyp:muB:index}} = d 2^{d}\fmax \Jmax \volball{d}$
\\

\begin{minipage}{5in} \begin{enumerate}\setcounter{enumi}{10} \item 
\label{Paper2:hyp:muBsqrt:index}
\end{enumerate} \end{minipage}
&Proposition \ref{Paper2:prop:probabilityboundssqrt}, Hypothesis \hyperref[hyp:7]{7}
&$c_{\ref{Paper2:hyp:muBsqrt:index}}= \frac{\fmax \Jmax}{\fmin \Jmin }\left( \frac{\rho}{\sqrt{4 - \sqrt{13}}}\right)^d d 2^{2d} \sqrt{3}$
\\

\begin{minipage}{5in} \begin{enumerate}\setcounter{enumi}{11} \item 
\end{enumerate} \end{minipage}
&Subsect. \ref{Paper2:subsec:quantif_normal_reach}
&$\Delta$, ~~~$\Delta_0$, ~~~$\Theta$
\\


\begin{minipage}{5in} \begin{enumerate}\setcounter{enumi}{12} \item 
\label{prop:quantif_normal_reach:index} 
\end{enumerate} \end{minipage}
&Proposition \ref{prop:quantif_normal_reach}
&$r_{\ref{prop:quantif_normal_reach:index}} = \mini{\frac{\sin(\Theta)}{8\rho}, \frac{\sin(\Theta)^2}{4}, \frac{\Delta_0\sin(\Theta)}{4},\Delta}$ \\ &&$c_{\ref{prop:quantif_normal_reach:index}} =   \left(\frac{2}{\sin(\theta)}\right)^\alpha V_{\alpha}\fmax\HH^{d'}_{\MMo}(\NNo)$
\\
 
\begin{minipage}{5in} \begin{enumerate}\setcounter{enumi}{13} \item 
\label{Paper2:lem:loccovconsistency:index}
\end{enumerate} \end{minipage}
&Proposition \ref{Paper2:prop:consistency}
&$c_{\ref{Paper2:lem:loccovconsistency:index}} = 6 \rho + \frac{1}{\fmin \Jmin}\left( 4c_{\ref{Paper2:lem:densityg:index}} + \fmax 2^d d \rho + c_{\ref{Paper2:prop:probabilitybounds:index}}\right)$,
\\

\begin{minipage}{5in} \begin{enumerate}\setcounter{enumi}{14} \item 
\label{Paper2:lem:Wstabxy:index}
\end{enumerate} \end{minipage}
&Lemma \ref{Paper2:lem:Wstabxy}
&$c_{\mathrm{\ref{Paper2:lem:Wstabxy:index}}} = 2\left(1 + 4\frac{5^{d-1}}{3^d}\right)\frac{c_{\ref{Paper2:hyp:muB:index}}}{c_{\ref{Paper2:hyp:muA:index}}}$
\\

\begin{minipage}{5in} \begin{enumerate}\setcounter{enumi}{15} \item 
\label{Paper2:lem:Wstabxysqrt:index}
\end{enumerate} \end{minipage}
&Lemma \ref{Paper2:lem:Wstabxysqrt}
&$c_{\mathrm{\ref{Paper2:lem:Wstabxysqrt:index}}} = \left(2 + \frac{2^\frac{5}{2} 5^{d-\frac{1}{2}} }{3^d} \right)\frac{c_{\ref{Paper2:hyp:muBsqrt:index}}}{c_{\ref{Paper2:hyp:muA:index}}}$
\\

\begin{minipage}{5in} \begin{enumerate}\setcounter{enumi}{16} \item 
\label{Paper2:lem:Wstabxx:index}
\end{enumerate} \end{minipage}
&Lemma \ref{Paper2:lem:Wstabxx}
&$c_{\ref{Paper2:lem:Wstabxx:index}} = 
\frac{2^{d-1}}{c_{\ref{Paper2:hyp:muA:index}}}
+ 2\frac{12 \cdot 5^{d-1}c_{\ref{Paper2:hyp:muB:index}}+1}{3^d c_{\ref{Paper2:hyp:muA:index}}}
+ 2^{d+3} \frac{(\frac{3}{2})^{d-1}c_{\ref{Paper2:hyp:muB:index}}+1}{c_{\ref{Paper2:hyp:muA:index}}}$
\\

\begin{minipage}{5in} \begin{enumerate}\setcounter{enumi}{17} \item 
\label{Paper2:lem:Wstabxxsqrt:index}
\end{enumerate} \end{minipage}
&Lemma \ref{Paper2:lem:Wstabxxsqrt}
&$c_{\ref{Paper2:lem:Wstabxxsqrt:index}} = 
\frac{2^{d-2}}{c_{\ref{Paper2:hyp:muA:index}}}
+ \frac{4 \cdot 3^\frac{1}{2} 5^{d-\frac{1}{2}} c_{\ref{Paper2:hyp:muBsqrt:index}}+4^{d-\frac{1}{2}} }{3^d c_{\ref{Paper2:hyp:muA:index}}}
+2\cdot 4^d\frac{2 c_{\ref{Paper2:hyp:muBsqrt:index}} (\frac{3}{2})^{d-\frac{1}{2}}+1 }{3^d c_{\ref{Paper2:hyp:muA:index}}}$
\\

\begin{minipage}{5in} \begin{enumerate}\setcounter{enumi}{18} \item 
\label{Paper2:prop:stability:index}
\end{enumerate} \end{minipage}
&Proposition \ref{Paper2:prop:stability}
&$c_{\ref{Paper2:prop:stability:index}} = 4(1 + c_{\ref{Paper2:lem:intlocwasserstein:index}})$, ~~ $c_{\ref{Paper2:prop:stability:index}} ' =4c_{\ref{Paper2:lem:Wstabxxsqrt:index}}$
\\

\begin{minipage}{5in} \begin{enumerate}\setcounter{enumi}{19} \item 
\label{Paper2:lem:intlocwasserstein:index}
\end{enumerate} \end{minipage}
&Lemma \ref{Paper2:lem:intlocwasserstein}
&$c_{\ref{Paper2:lem:intlocwasserstein:index}} = 3 +c_{\ref{Paper2:lem:Wstabxy:index}} + c_{\ref{Paper2:lem:Wstabxysqrt:index}} + c_{\ref{Paper2:lem:Wstabxx:index}}$
\\


\begin{minipage}{5in} \begin{enumerate}\setcounter{enumi}{20} \item 
\label{Paper2:th:estimation:index}
\end{enumerate} \end{minipage}
&Theorem \ref{Paper2:th:estimation}
&$c_{\ref{Paper2:th:estimation:index}} = 2 + \frac{1}{2}c_{\ref{Paper2:prop:stability:index}}'
= 2(1 + c_{\ref{Paper2:lem:Wstabxxsqrt:index}})$
\\

\begin{minipage}{5in} \begin{enumerate}\setcounter{enumi}{21} \item 
\label{Paper2:cor:approximation:index}
\end{enumerate} \end{minipage}
&Corollary \ref{Paper2:cor:approximation}
&$c_{\ref{Paper2:cor:approximation:index}} = c_{\ref{Paper2:th:estimation:index}}(c_{\ref{Paper2:hyp:normalreach:index}})^\frac{1}{p} + c_{\ref{Paper2:prop:stability:index}}  + c_{\ref{Paper2:lem:loccovconsistency:index}}$
\\

\begin{minipage}{5in} \begin{enumerate}\setcounter{enumi}{22} \item 
\label{subsec:index}
\end{enumerate} \end{minipage}
&Subsect. \ref{Paper2:subsec:homotopy_type_estimation}
&$\check{\rho}_\gamma$, ~~~$\check{f}_{\text{min},\gamma}$, ~~~$\check c_{\ref{subsec:index}, \gamma} = \check{f}_{\text{min},\gamma} \Jmin \volball{d}$
\\

\begin{minipage}{5in} \begin{enumerate}\setcounter{enumi}{23} \item 
\label{Paper2:cor:DTMfiltrcheck:index}
\end{enumerate} \end{minipage}
&Corollary \ref{Paper2:cor:DTMfiltrcheck}
&$\check c_{\ref{Paper2:cor:DTM-filtrations:index}, \gamma} = 8\mathrm{diam}(\MM) + 8\gamma + 5$, ~~~$\check c_{\ref{Paper2:cor:DTM-filtrations:index}, \gamma}' = 2\left(\check c_{\ref{subsec:index}, \gamma}\right)^{-\frac{1}{d}}$
\end{tabular}
\end{center}

\section{Supplementary material for Sect. \ref{Paper2:sec:intro}}
\label{sec:appendix_intro}

\begin{proof}[Proof of Lemma \ref{Paper2:lem:distancetocenter}]
The proof is based on the following observations.
We can use the triangle inequality, then the Pythagorean Theorem with $\eucP{v}{y-x} = 0$ and Lemma \ref{Paper2:lem:Federergeod} Point \ref{lem:Federergeod:point1} to get
\begin{align*}
\eucN{\gamma(t) - x} 
&\leq \eucN{(y+tv) - x} + \eucN{\gamma(t) - (y+tv)} \\
&\leq \sqrt{ \eucN{tv}^2+\eucN{y-x}^2} + \frac{\rho}{2}t^2\\
&= \sqrt{t^2+l^2} + \frac{\rho}{2}t^2.
\end{align*}
For any $r \leq \frac{1}{\rho}$, consider the equation
\begin{equation}
\label{eq:plus}
\sqrt{t^2+l^2} + \frac{\rho}{2} t^2 = r.
\end{equation}
By squaring this equality, we get $\left(\frac{\rho}{2}\right)^2t^4 - (1+\rho r)t^2+(r^2-t^2)=0$.
By considering the polynomial $T \mapsto\left(\frac{\rho}{2}\right)^2 T^2 - (1+\rho r) T +(r^2-t^2)$, whose discriminant is $1 + 2\rho r + (\rho t)^2>0$, we see that the solutions of Equation \eqref{eq:plus} are
$$T_1 = \frac{\sqrt{2}}{\rho}\sqrt{1+\rho r - \sqrt{ 1 + 2\rho r + \rho^2 l^2}} ~~~~\text{and}~~~~ T_1' = \frac{\sqrt{2}}{\rho}\sqrt{1+\rho r + \sqrt{ 1 + 2\rho r + \rho^2 l^2}}.$$
Following the same ideas, one obtains
\begin{align*}
\eucN{\gamma(t) - x} 
&\geq\sqrt{t^2+l^2} - \frac{\rho}{2}t^2.
\end{align*}
Moreover, the equation 
\begin{equation}
\label{eq:moins}
\sqrt{t^2+l^2} - \frac{\rho}{2} t^2 = r
\end{equation}
admits the following roots:
$$T_2 = \frac{\sqrt{2}}{\rho}\sqrt{1-\rho r - \sqrt{ 1 - 2\rho r + \rho^2 l^2}} ~~~~\text{and}~~~~ T_2' = \frac{\sqrt{2}}{\rho}\sqrt{1-\rho r + \sqrt{ 1 - 2\rho r + \rho^2 l^2}}.$$
We now prove the five points successively.

\medbreak \noindent \emph{Point \ref{lem:distancetocenter:point1}.}
Observe that $\dot \phi(t) = 2\eucP{\dot \gamma(t)}{\gamma(t) - x}$, and that
\[\ddot \phi(t) = 2\eucP{\dot \gamma(t)}{\dot \gamma(t)} + 2\eucP{\ddot \gamma(t)}{\gamma(t) - x}.\]
By Cauchy-Schwarz inequality, $\eucP{\ddot \gamma(t)}{\gamma(t) - x} \geq -\eucN{\ddot \gamma(t)}\eucN{\gamma(t) - x}$.
Note that $\eucP{\dot \gamma(t)}{\dot \gamma(t)} = 1$ since $\gamma$ is parametrized by arc-length, and that $\eucN{\ddot \gamma(t)} \leq \rho$ by Equation \eqref{eq:supnormiscurvature}. 
Hence we get 
\begin{equation}
\ddot \phi(t) \geq 2(1-\rho\eucN{\gamma(t) - x}).
\label{Paper2:eq:ddotphi}
\end{equation}
Consider Equation \eqref{eq:plus} with $r = \frac{1}{\rho}$. We see that $\eucN{\gamma(t) - x} \leq \frac{1}{\rho}$ when $t$ is lower than
$$
T_1 = \frac{\sqrt{2}}{\rho} \sqrt{2-\sqrt{3+\rho^2 l^2}}.
$$
In this case, $\ddot \phi(t) \geq 0$ according to Equation \eqref{Paper2:eq:ddotphi}.
Since $\dot \phi(0) = 0$, we deduce that $\phi$ is increasing on $[0,T_1]$.

\medbreak \noindent \emph{Point \ref{lem:distancetocenter:point2}.}
As we have seen with Equation \eqref{eq:moins}, we have $\eucN{\gamma(t)-x} > r$ when $t \in (T_2, T_2')$.
In order to give an upper bound on $T_2$, we use the inequality $\sqrt{b} - \sqrt{a} = \frac{1}{\sqrt{a}+\sqrt{b}}(b-a) \leq \frac{1}{\sqrt{b}}(b-a)$, where $a<b$, to get 
\begin{align*}
1-\rho r - \sqrt{ 1 - 2\rho r + \rho^2 l^2 } 
&\leq \frac{1}{1-\rho r} \rho^2(r^2 - l^2) 
\end{align*}
and we conclude that $T_2 \leq \frac{\sqrt{2}}{\sqrt{1-\rho r}} \sqrt{r^2 - l^2}$. Since $r \leq \frac{1}{2\rho}$, we obtain $T_2 \leq 2 \sqrt{r^2 - l^2}$.

\medbreak \noindent \emph{Point \ref{lem:distancetocenter:point4}.}
When $l = 0$, algebraic manipulations show that $T_2 = \frac{1}{\rho}(1-\sqrt{1-2\rho r})$ and $T_2' = \frac{1}{\rho}(1+\sqrt{1-2\rho r})$.

\medbreak \noindent \emph{Point \ref{lem:distancetocenter:point5}.}
Consider the map $\phi\colon t \mapsto \eucN{\gamma(t) - x}^2$. 
By definition of $b$, for all $t \in (0,b)$, we have $\eucN{\gamma(t) - x} \leq r$. Hence Equation \eqref{Paper2:eq:ddotphi} gives $\ddot \phi(t) \geq 2(1-\rho r)$.
It follows that $\dot \phi(t) \geq 2(1-\rho r) t$, and that
\begin{align*}
\phi(b)-\phi(a)
= \int_{a}^{b} \dot \phi(t) \dd t
&\geq \int_{a}^{b} 2(1-\rho r) t \dd t\\
&= (1-\rho r)(b^2-a^2).
\end{align*}
Note that $r^2 = \phi(b)$. Besides,  $s^2 = \phi(a)$ or $s^2 < \phi(a)$, depending on whether $s \geq l$ or $s < l$. In both cases, we have $r^2 - s^2 \geq \phi(b)-\phi(a)$, and we deduce that 
\begin{equation*}
r^2 - s^2 \geq (1-\rho r)(b^2-a^2).
\end{equation*}
Writing $r^2-s^2 = (r+s)(r-s)$ and $b^2-a^2 = \big(b+a\big)\big(b-a\big)$ leads to 
\begin{equation}
\label{eq:upperboundab}
b-a \leq \frac{r+s}{b+a}\frac{1}{1-\rho r}(r-s).
\end{equation}

Now, let us give a lower bound on $b$. According to Equation \eqref{eq:plus}, $b$ is lower bounded by $T_1 = \frac{\sqrt{2}}{\rho}\sqrt{1+\rho r - \sqrt{ 1 + 2\rho r + \rho^2 l^2}}$.
Using the inequality $\sqrt{b} - \sqrt{a} = \frac{1}{\sqrt{b}+\sqrt{a}}(b-a) \geq \frac{1}{2\sqrt{b}}(b-a)$, where $a<b$, we get
\begin{align*}
1+\rho r - \sqrt{ 1 + 2\rho r + \rho^2 l^2 }
\geq \frac{1}{2(1+\rho r)} \rho^2(r^2 - l^2),
\end{align*}
and we conclude that $b \geq (1+\rho r)^{-\frac{1}{2}}\sqrt{r^2 - s^2}$.
Injecting $b+a \geq b \geq (1+\rho r)^{-\frac{1}{2}}\sqrt{r^2 - s^2}$ in Equation \eqref{eq:upperboundab} yields
\begin{equation*}
b(v) - a(v) \leq \frac{(1+\rho r)^{\frac{1}{2}}}{1-\rho r} \sqrt{r^2-s^2}.
\end{equation*}
Under the hypothesis $r \leq \frac{1}{2\rho}$, we get $b-a \leq \sqrt{6}\sqrt{r^2-s^2}$.
 
\medbreak \noindent \emph{Point \ref{lem:distancetocenter:point6}.}
When $l=0$, we have $b(v) + a(v) \geq r + s$. Hence Equation \eqref{eq:upperboundab} yields $b(v)-a(v) \leq \frac{1}{1-\rho r}(r-s).$
Using $r \leq \frac{1}{2\rho}$, we obtain $b(v)-a(v) \leq 2(r-s).$
\end{proof}

\begin{proof}[of Corollary \ref{Paper2:cor:DTM-filtrations}]
We shall first study an intermediate quantity.
Let $\mu$ be a probability measure on $E=\R^n$, $m\in(0,1)$, and $\DTM{\mu,m}$ the corresponding DTM.
Consider the quantity $c(\mu, m)$ is defined as
\[c(\mu, m) = \sup_{x\in \supp{\mu}} \DTM{\mu, m}(x).\]
Suppose that $\mu$ satisfies the following for $r < \left(\frac{m}{a}\right)^\frac{1}{d}$: $\forall x \in \supp{\mu}, \mu(\openball{x}{r}) \geq a r^d$.
Let us show that $c(\mu, m) \leq C m^\frac{1}{d}$
with $C = a^{-\frac{1}{d}}$.
By definition, 
\begin{align*}
\delta_{\mu, t}(x) = \inf \left\{r\geq0 \mid\mu\left(\closedball{x}{r}\right)>t\right\} ~~~~\text{ and }~~~~ \DTM{\mu,m}^2(x) = \frac{1}{m}\int_0^{m} \delta_{\mu,t}^2(x)\dd t.
\end{align*}
Using the assumption $\mu(\openball{x}{r}) \geq a r^d$ for all $x \in\supp{\mu}$, we get $\delta_{\mu, t}(x) \leq (\frac{t}{a})^\frac{1}{d}$, and a simple computation yields 
\begin{align*}
\DTM{\mu,m}^2(x) 
&\leq \frac{d}{d+2} \left(\frac{t}{a}\right)^{\frac{2}{d}} 
\leq  \left(\frac{t}{a}\right)^{\frac{2}{d}},
\end{align*}
which yields the result.

We can now prove the corollary.
Let $\pi$ be an optimal transport plan for $w=\Wassdeux{\mu}{\nu}$.
Denote $\alpha = w^\frac{1}{2}$ and $D = \mathrm{diam}(\supp{\mu})$.
Define $\pi'$ to be $\pi$ restricted to the set $\{x,y \in E \mid \eucN{x-y}<\alpha\}$. We denote its marginals $\mu'$ and $\nu'$.
By Markov inequality, $1-|\pi'| \leq \frac{w^2}{\alpha^2} = w$, where we recall that $|\pi'|$ denotes the total mass of $\pi'$.
Consider the probability measures $\overline{\mu'}=\frac{1}{|\mu'|}\mu'$ and $\overline{\nu'}=\frac{1}{|\nu'|}\nu'$.
Let us show that we have
\begin{equation}
\label{Paper2:eq:proof_DTMfiltr_without_intermediate}
\Wassdeux{\mu}{\overline{\mu'}} = 2D \alpha,
~~~~~~~~~
\Wassdeux{\overline{\mu'}}{\overline{\nu'}} \leq \alpha
~~~~~~~~~\text{and}~~~~~~~~~
\Wassdeux{\nu}{\overline{\nu'}} \leq 2(1+D) \alpha.
\end{equation}
The first inequality is an application of Lemma \ref{Paper2:lem:transportsubmeasure}: 
\begin{align*}
\Wassdeux{\mu}{\overline{\mu'}} \leq  2(1-|\mu'|)^\frac{1}{2} D =  2(1-|\pi'|)^\frac{1}{2} D \leq 2 w^\frac{1}{2} D.
\end{align*}
To obtain the second inequality, we write
\begin{align*}
\Wasssymbol{2}^2(\overline{\mu'}, \overline{\nu'})
= \int \eucN{x-y}^2 \dd  \overline{ \pi'}(x,y)
&= \int \eucN{x-y}  \frac{\dd \pi'(x,y)}{|\pi'|} \\
&\leq \frac{1}{|\pi'|}\int \eucN{x-y}\dd \pi(x,y).
\end{align*}
Hence Jensen inequality leads to $\Wassdeux{\overline{\mu'}}{\overline{\nu'}}
\leq \frac{w}{|\pi'|^\frac{1}{2}}$. 
Since $1-|\pi'| \leq w$, we have $\frac{w}{|\pi'|^\frac{1}{2}} \leq \frac{w}{1-w}$, 
and the assumption $w \leq \frac{1}{4}$ yields $\frac{w}{1-w} \leq \alpha$. This proves the second point.
Finally, we obtain the third inequality by applying the triangle inequality:
\begin{align*}
\Wassdeux{\nu}{\overline{\nu'}} \leq \Wassdeux{\nu}{\mu} + \Wassdeux{\mu}{ \overline{\mu'}} + \Wassdeux{\overline{\mu'}}{\overline{\nu'}}.
\end{align*}

Next, let us deduce that
\begin{align}
&c(\overline{\mu'},m) \leq c(\mu) + m^{-\frac{1}{2}} 2D \alpha \nonumber \\
\text{and}~~~~~~~&c(\overline{\nu'},m) \leq c(\mu,m) + \left(m^{-\frac{1}{2}}+m^{-\frac{1}{2}}2D+1\right)\alpha.
\label{Paper2:eq:proof_DTMfiltr_without_intermediate2}
\end{align}
The first inequality follows from the stability of the DTM (see Equation \eqref{eq:stab_DTM}):
\begin{align*}
c(\overline{\mu'},m) 
= \sup_{x \in \supp{\overline{\mu'}}} \DTM{\overline{\mu'}}(x)
\leq \sup_{x \in \supp{\overline{\mu'}}} \DTM{\mu}(x) + m^{-\frac{1}{2}}\Wassdeux{\overline{\mu'}}{\mu},
\end{align*}
and we conclude with $\Wassdeux{\mu}{\overline{\mu'}} = 2D \alpha$.
In order to prove the second inequality, we also use Equation \eqref{eq:stab_DTM}:
\begin{align*}
c(\overline{\nu'},m) 
&= \sup_{x \in \supp{\overline{\nu'}}} \DTM{\overline{\nu'}}(x) 
\leq \sup_{x \in \supp{\overline{\nu'}}} \DTM{\overline{\mu'}}(x) + m^{-\frac{1}{2}} \Wassdeux{\overline{\mu'}}{\overline{\nu'}}.
\end{align*}
Since $\pi'$ has support included in $\{x,y \in E \mid \eucN{x-y}<\alpha\}$, we use the fact that the DTM is 1-Lipschitz to obtain
\begin{align*}
\sup_{x \in \supp{\overline{\nu'}}} \DTM{\overline{\mu'}}(x)
\leq \sup_{x \in \supp{\overline{\mu'}}} \DTM{\overline{\mu'}}(x) + \alpha
= c(\mu',m) + \alpha
\end{align*}
and we deduce 
\begin{align*}
c(\overline{\nu'},m) 
&\leq c(\mu',m) + \alpha + m^{-\frac{1}{2}} \Wassdeux{\overline{\mu'}}{ \overline{\nu'}}\\
&\leq c(\mu,m) + (m^{-\frac{1}{2}}+m^{-\frac{1}{2}}2D+1)\alpha. 
\end{align*}

We can now conclude with Theorem \ref{th:DTM-filtr-stability}.
In our context, it reads
\begin{align*}
d_i(\DTMF{\mu},\DTMF{\nu})
&\leq  m^{-\frac{1}{2}}\Wassdeux{\mu}{\overline{\mu'}} + m^{-\frac{1}{2}} \Wassdeux{\overline{\mu'}}{\overline{\nu'}} + m^{-1} \Wassdeux{\nu}{\overline{\nu'}} + c(\overline{\mu'},m) + c(\overline{\nu'},m)  \\
&\leq \big(m^{-\frac{1}{2}}(4D + 1)+4(D+1)\big) \alpha + 2 c(\mu,m),
\end{align*}
where we used Equations \eqref{Paper2:eq:proof_DTMfiltr_without_intermediate} and \eqref{Paper2:eq:proof_DTMfiltr_without_intermediate2} on the last line.
Since $m \leq 1$, we can simplify this expression into
\begin{align*}
\idist{\DTMF{\mu}}{\DTMF{\nu}}
&\leq m^{-\frac{1}{2}}(8D+5) \alpha + 2 c(\mu,m).
\end{align*}
We conclude the proof by using the inequality $c(\mu,m) \leq a^{-\frac{1}{d}} m^\frac{1}{d}$ shown at the beginning of the proof.
\end{proof}

\bibliographystyle{unsrt}      
\bibliography{biblio_DCG}   

\end{document}